\documentclass[10pt]{article}
\usepackage[margin=1in]{geometry}
\usepackage{amsthm,amsfonts,amssymb,amsmath,comment,slashed,mathtools,enumitem}
\usepackage{empheq} 
\usepackage[T1]{fontenc} 
\usepackage[utf8]{inputenc}

\usepackage[affil-it]{authblk} 
\usepackage[hidelinks]{hyperref}    
\usepackage{graphicx} 
\usepackage{epstopdf} 
\usepackage{thm-restate} 
\usepackage{xargs}          
\usepackage{mathrsfs} 
\usepackage{xcolor}  
\usepackage{bm}
\usepackage{bbm}
\usepackage{here} 
\usepackage[capitalize]{cleveref} 
\usepackage{tikz-cd} 

\allowdisplaybreaks

\usepackage[colorinlistoftodos,prependcaption,textsize=tiny]{todonotes}
\newcommandx{\unsure}[2][1=]{\todo[linecolor=red,backgroundcolor=red!25,bordercolor=red,#1]{#2}}
\newcommandx{\change}[2][1=]{\todo[linecolor=blue,backgroundcolor=blue!25,bordercolor=blue,#1]{#2}}
\newcommandx{\info}[2][1=]{\todo[linecolor=OliveGreen,backgroundcolor=OliveGreen!25,bordercolor=OliveGreen,#1]{#2}}
\newcommandx{\improvement}[2][1=]{\todo[linecolor=Plum,backgroundcolor=Plum!25,bordercolor=Plum,#1]{#2}}
\newcommandx{\thiswillnotshow}[2][1=]{\todo[disable,#1]{#2}}
\usetikzlibrary{decorations.pathmorphing}

\theoremstyle{plain}
\newtheorem{definition}{Definition} 
\newtheorem{prop}{Proposition}
\newtheorem{lemma}{Lemma}
\newtheorem{cor}{Corollary}

\newtheorem{rmk}{Remark}
\declaretheorem[name=Theorem]{theorem}

\newcommand{\mr}{\mathring}
\newcommand{\wh}{\widehat}
\newcommand{\wt}{\widetilde}
\newcommand{\br}{\Breve}
\newcommand{\hz}{\hat{z}}

\makeatletter
\renewcommand{\paragraph}[1]{%
	\par 
	\addvspace{\medskipamount}
	\textit{#1\@addpunct{.}}\enspace\ignorespaces
}
\makeatother
\numberwithin{equation}{section}
\numberwithin{prop}{section}
\numberwithin{rmk}{section}
\numberwithin{lemma}{section}
\numberwithin{definition}{section}
\numberwithin{cor}{section}
\title{High regularity waves on self-similar naked singularity \\ interiors: decay and the role of blue-shift}
\author[1]{Jaydeep Singh\thanks{jaydeeps@math.princeton.edu}}
\affil[1]{\small  Department of Mathematics, Princeton University, Washington~Road,~Princeton,~NJ~08544,~United~States~of~America \vskip.1pc \ }

\date{\today}

\begin{document}
\maketitle
\begin{abstract}
    We consider solutions to the linear wave equation $\Box_{g}\varphi = 0$ on a class of approximately $k$-self-similar naked singularity interiors. This equation models the blue-shift effect, an instability exploited by Christodoulou \cite{chris3} in the proof of low-regularity weak cosmic censorship. Using a combination of resonance expansions and multiplier estimates, we find in the small-mass regime $k^2 \ll 1$ that the asymptotics of solutions are strongly sensitive to the regularity assumed on outgoing, characteristic initial data across the past light-cone of the singularity. Above a threshold regularity set by the $k$-self-similar scalar field, solutions are shown to always obey self-similar bounds, indicating that the blue-shift instability \textit{competes} with the stabilizing influence of high regularity. We conclude that a proper statement of weak cosmic censorship, as well as an understanding of the role of naked singularities in phenomena such as critical collapse, may depend on the topology of initial data.
\end{abstract}
\tableofcontents
\section{Introduction}
The spherically symmetric Einstein-scalar field system in $(3\!+\!1)$-dimensions is a geometric system of equations for a Lorentzian spacetime $({\mathcal{M}}, {g})$ and an associated real-valued scalar field ${\mathcal{\phi}}: {\mathcal{M}} \rightarrow \mathbb{R}$. Denoting by $\textrm{Ric}_{\mu \nu}, R$ the Ricci curvature tensor and scalar curvature respectively, the system takes the form 
\begin{equation}
    \label{eq:0}
    \begin{cases}
        \textrm{Ric}_{\mu \nu}[{g}] - \frac12 {g}_{\mu \nu} \textrm{R}[g]   = 2 \textrm{T}_{\mu \nu}[{g}, {\mathcal{\phi}}], \\[2\jot]
        \textrm{T}_{\mu \nu}[{g}, {\mathcal{\phi}}] = \partial_\mu \phi \partial_\nu \phi - \frac12 {g}_{\mu \nu} \partial^\alpha \phi \partial_\alpha \phi,\\[2\jot]
        \Box_{g} {\mathcal{\phi}} = 0.
    \end{cases}
\end{equation}
The foundational works \cite{chris0.5}--\cite{chris3} established (\ref{eq:0}) as a useful model for studying a variety of problems in general relativity: low regularity well-posedness \cite{chris1}, the behavior of general dispersive solutions \cite{chris0.5,lukoh1,lukohyang1}, trapped surface formation and black hole stability \cite{chris1.5, IgorMihalis_Price}, 
the structure of spacelike singularities \cite{anzhang2,angajic}, and---most relevant to the current work---the existence and stability of naked singularities \cite{chris2, chris3, liuli, singh}.

Naked singularities are a class of spacetimes containing gravitational or matter singularities, without a corresponding event horizon; see Figure \ref{fig1}. In black hole spacetimes, the horizon plays an important role in the lives of observers who remain in the exterior region, preventing signals from the black hole interior (and any singularity lurking inside) from reaching them. In contrast, the breakdown of spacetime near a naked singularity propagates to all observers, with potentially disastrous consequences \cite{wald_WCC}.
\begin{figure}[t]
\centering
\includegraphics{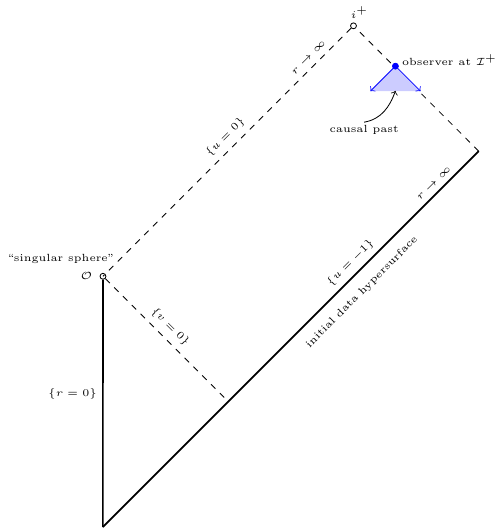}
\caption{Penrose diagram representation of a spherically symmetric, globally naked singularity. To the future of initial data, the spacetime is foliated by outgoing null hypersurfaces, each extending from the center ${(r=0)}$ to null infinity ${(r\rightarrow \infty)}.$ Within finite $u$-coordinate time all observers intersect the future light-cone of the first singularity $\mathcal{O},$ which lies at the center. 
Also shown are the conventions for double-null coordinates adopted in this paper; the initial data surface is $\{u\!=\!-1\}$, and the singularity is normalized to ${(u,v)=(0,0)}$. }
\label{fig1}
\end{figure}

To make the discussion more precise, define\footnote{This terminology is non-standard, and is meant to distinguish the class of naked singularities considered here from the well-known examples of singular spacetimes that lack an event horizon, e.g. negative mass or super-extremal black holes.} a \textit{globally naked singularity arising from collapse} as a solution to (\ref{eq:0}) which is the maximal development of regular, asymptotically flat initial data, and which moreover admits a future incomplete null infinity. A modern formulation of Penrose's weak cosmic censorship conjecture asserts that this class of spacetimes should be \textit{non-generic} \cite{penroseWCC, chris5}. It is thus of great interest to understand the space of possible naked singularities, and the mechanism for their (in-)stability.

The most definitive result concerning the nonlinear instability of naked singularity solutions to (\ref{eq:0}) is \cite{chris3} (see also \cite{liuli}). One statement proved there is the following: any globally naked singularity in the low-regularity bounded variation (BV) class\footnote{This class was introduced in \cite{chris1}, where local existence and small-data global existence was shown.} with $\frac{m}{r} \not\rightarrow 0$ along the past light-cone of the singular point is unstable. Explicit unstable perturbations of the scalar field are identified, which have the properties that they (a) are generically only BV regular, and (b) are supported in the exterior region of the singularity. 

The mechanism serving as the ``cosmic censor'' in \cite{chris3} is the blue-shift instability. In contrast to the celebrated red-shift effect along black hole horizons (cf. \cite{redshifteffect, lectures_linearwaves}), which provides a \textit{stabilizing} mechanism for waves, the blue-shift is an instability mechanism\footnote{Although we will only view the blue-shift in the context of naked singularities, it plays a significant role in the study of strong cosmic censorship, dating back to Penrose's investigations of Cauchy horizons in the interior of black holes \cite{penroseSCC}.}. Importantly, the proof in \cite{chris3} quantifies a blue-shift localized to the past light-cone of the singular point, which makes heavy use of both the regularity and support of the initial data perturbations.

In the present paper we initiate the study of stability of naked singularity solutions to (\ref{eq:0}) under perturbations of general regularity and support. Our analysis applies to $k$-self-similar naked singularity solutions $(\mathcal{M}, g_k, \phi_k)$, originally studied in \cite{brady,chris2}, as well as the related examples constructed in \cite{singh}. We moreover restrict to the setting of small mass spacetimes, equivalent to taking\footnote{Without loss of generality we can restrict to $k>0$.} \underline{$k^2 \ll 1$}. For any such spacetime with metric $g$, we analyze solutions to the linear wave equation 
\begin{equation}
    \Box_{g} \varphi = 0, \label{eq:1}
\end{equation}
which can be viewed as a crude linearization of the full Einstein-scalar field system (\ref{eq:0}), neglecting linear backreaction on the geometry. The approach of considering (\ref{eq:1}) as a first step to understanding nonlinear instability has been fruitful in other general relativity contexts, e.g. the stability problem for black hole exteriors \cite{lectures_linearwaves}. Similar studies of linear field equations on negative mass and overcharged black holes spacetimes have been considered, yielding a complex picture of the stability of these singular spacetimes \cite{GibbHartIshib,GleisDotti,GleisDottiPull,shadi_stalker}. See also \cite{nolan} for a study of (\ref{eq:1}) in a wide class of naked singularities arising from collapse, in a region close to the future lightcone of the singularity.

We consider the characteristic initial value problem with data posed on an outgoing null cone\footnote{For coordinate conventions, refer to Figure \ref{fig1}.} ${\{u=-1\}}$, and establish boundedness and sharp decay for $\varphi$ and first derivatives, as a function of the \underline{regularity} and \underline{support} of data. Here, regularity refers to the behavior of data towards the past light-cone of the singular point, $\{v=0\}$, which is quantified by Hölder-type spaces. Support refers to a classification of data as exterior data or interior data, according to whether $\varphi$ is supported entirely in the exterior region of the naked singularity, $\{v \geq 0\}$, or has non-trivial support in the causal past of the singularity, $\{v < 0 \}$.

Given that the background spacetime is itself singular as $(u,v) \rightarrow (0,0)$, we must specify what behavior for solutions to (\ref{eq:1}) is considered indicative of ``stability.'' The approach here is to compare with the $k$-self-similar scalar field $\phi_k$, which obeys the self-similar asymptotic ${\big|\frac{1}{\partial_v r_k} \partial_v \phi_k\big| \sim |u|^{-1}}$. As discussed in Section \ref{subsec:intro_blueshift}, it is known that generic \textit{low-regularity, exterior} solutions to (\ref{eq:1}) do not obey such a bound, but rather experience a blue-shift instability localized to $\{v=0\}$, resulting in the lower bound ${\big|\frac{1}{\partial_v r_k} \partial_v \varphi\big| \gtrsim |u|^{-1-k^2}}$. We therefore adopt the perspective that stability refers to the presence of self-similar bounds; see Definition \ref{defn:self-similarbounds}.

The behavior of (linear and nonlinear) waves on the exteriors of self-similar naked singularities was first studied in \cite{igoryak2}. There, it is shown for vacuum self-similar spacetimes that high-regularity can act as a stabilizing influence against the blue-shift; the same was shown to be true for the system (\ref{eq:0}) in \cite{singh}. We discuss the analogs of these results for the linear problem (\ref{eq:1}) below; see Proposition \ref{prop:extinst}. The behavior in this case is characterized by a threshold phenomenon: below a given regularity, set by that of the $k$-self-similar scalar field $\phi_k$ across the past light-cone of the singularity, a weaker form of the blue-shift instability continues to hold. However, above this threshold the instability suddenly ceases, and all solutions obey self-similar bounds.  

The case of interior perturbations, however, has not been previously considered in either the linear or nonlinear context. Heuristic arguments reviewed in Section \ref{subsec:intro_blueshift} suggest that the dependence of the blue-shift on regularity could be weaker in the interior, leading to the existence of high-regularity instabilities. Our main result is that this is not the case. As stated in Theorem \ref{thm:introrough1}, a similar threshold regularity exists for interior perturbations, at and above which self-similar bounds always hold. We moreover establish precise expansions for the self-similar components of the solution, up to (faster decaying) errors. As an application of the methods of this paper, we also show (Theorem \ref{thm:introrough2}) that a similar stability result holds for non-spherically symmetric solutions, provided the regularity is sufficiently above this threshold.

We conclude by outlining the remainder of the introduction. Section \ref{subsec:intro_statement} states our main result concerning solutions to (\ref{eq:1}) in the interior. Section \ref{subsec:intro_blueshift} makes precise the notion of the blue-shift instability, and gives additional motivation to the study of the naked singularity interiors. Section \ref{subsec:intro_proofoutline} sketches the proof, and Sections \ref{subsec:intro_outlook}--\ref{subsec:intro_related} interpret our results in relation to weak cosmic censorship and other problems in general relativity.

\subsection{Rough statement of results}
\label{subsec:intro_statement}
Let $(u,v)$ denote renormalized double-null coordinates defined in Section \ref{sec:renormalizedgauge}, $\omega \in \mathbb{S}^2$ angular coordinates, and $r_k(u,v)$ the area radius function associated to a $k$-self-similar spacetime $(\mathcal{M},g_k,\phi_k)$. With respect to this gauge, the singularity $\mathcal{O}$ will be represented by the idealized sphere  $\{u=v=0\}$, the interior region by $\{v \leq 0\}$, and the exterior by $\{v \geq 0\}$.

The following definition makes the notion of unstable solutions to (\ref{eq:1}) precise, as those whose $C^1$ norm grows strictly faster than dictated by self-similarity: 
\begin{definition} 
    \label{defn:self-similarbounds}
    A solution $\varphi(u,v,\omega) \in C^2({\mathcal{M}} \setminus \{v=0\}) \cap C^1(\mathcal{M})$ to (\ref{eq:1}) in renormalized gauge is said to satisfy \boldmath $C^1$ \unboldmath \textbf{self-similar bounds} if there exists a constant $A$ such that 
    \begin{equation}
        \label{eq:intro1}
       \max_{i+j+l = 1} \sup_{{\mathcal{M}}} \Big| |u|^{i+j+l} \Big(\frac{1}{\partial_u r_k}\partial_u \Big)^i \Big( \frac{1}{\partial_v r_k}\partial_v\Big)^j  \slashed{\nabla}^l \varphi \Big| \leq A. 
    \end{equation}
    Solutions for which (\ref{eq:intro1}) fails to hold are said to exhibit \textbf{faster than self-similar growth}, or more simply are said to exhibit \textbf{instability}.
\end{definition}

Our first result concerns spherically symmetric solutions to (\ref{eq:1}). To contextualize the statement, let $\mr{\phi}(v): [-1,0] \rightarrow \mathbb{R}$ denote the restriction of the $k$-self-similar scalar field $\phi_k(u,v)$ to $\{u\!=\!-1\}$. The function $\mr{\phi}(v)$ has a finite Hölder regularity $\mr{\phi}(v) \in C^{1,\frac{k^2}{1-k^2}}([-1,0])$ (cf. Proposition \ref{lemma:backgroundprelims1}), which will serve as the ``threshold regularity'' separating instability from self-similar bounds.
\begin{theorem}[Spherically symmetric solutions]
    \label{thm:introrough1}
    Fix a value of $k$ sufficiently small, and let $(\mathcal{M},g,\phi)$ denote a naked singularity interior that is ``sufficiently close'' to $(\mathcal{M}, g_k, \phi_k)$. Assume spherically symmteric, characteristic initial data $\varphi_0(v): \{u\!=\!-1\} \rightarrow \mathbb{R}$ is given for (\ref{eq:1}), which is supported in the interior and sufficiently regular on $\{v<0\}$. Then the following holds:
    \begin{enumerate}[label=(\alph*)]
        \item (Polynomial convergence to constants above threshold) Let $\varphi_0(v)$ have regularity as ${v \rightarrow 0}$ which is ``strictly better'' than that of $\mr{\phi}(v)$. The solution $\varphi$ to (\ref{eq:1}) then satisfies $C^1$ self-similar bounds. Moreover, $\varphi$ converges polynomially in $|u|$ to a constant, and the first derivatives of $\varphi$ obey bounds which are polynomially in $|u|$ better than self-similar. The polynomial improvement over self-similarity is dictated by the regularity gap between that of the initial data, and the threshold regularity.
        
        \item (Self-similar bounds at threshold) Suppose there exists a constant $c \in \mathbb{R}$, and data $\varphi_{reg}(v)$ with regularity ``strictly better'' than that of $\mr{\phi}(v)$, such that $\varphi_0(v) \sim c |v|^{1+\frac{k^2}{1-k^2}} + \varphi_{reg}(v)$ holds as $v \rightarrow 0$. The solution $\varphi$ to (\ref{eq:1}) then satisfies $C^1$ self-similar bounds. Moreover, $\varphi$ converges polynomially in $|u|$ to a linear combination of constants and $\phi_k(u,v)$, and similarly for the first derivatives of $\varphi$.
        \item (Instability below threshold) Let $\varphi_0(v)$ have regularity ``strictly between'' $\text{BV}([-1,0])$ and that of $\mr{\phi}(v)$. Generic solutions $\varphi$ to (\ref{eq:1}) with this data are unstable, with $\frac{1}{\partial_v r_k}\partial_v \varphi$ growing at a rate strictly between $|u|^{-1}$ (self-similar rate) and $|u|^{-1-k^2}$ (blue-shift rate).
    \end{enumerate}
\end{theorem}

For non-spherically symmetric solutions, we establish a corresponding stability statement for solutions with sufficiently high regularity.
\begin{theorem}[Non-spherically symmetric solutions]
    \label{thm:introrough2}
    With the assumptions of Theorem \ref{thm:introrough1} on the background spacetime, let ${\varphi_0(v,\omega)\!: \{u\!=\!-1\} \rightarrow \mathbb{R}}$ denote characteristic initial data supported on angular modes $\ell \geq 1$. Assume the data is supported in the interior region, sufficiently regular in angular directions, and has null regularity ``sufficiently above'' that of $\mr{\phi}(v)$. Then the solution $\varphi$ to (\ref{eq:1}) obeys $C^1$ self-similar bounds. Moreover, $\varphi$ converges to zero at a polynomial rate in $|u|$, and the first derivatives of $\varphi$ obey bounds which are polynomially in $|u|$ better than self-similar.
\end{theorem}

\subsection{Exterior perturbations and the blue-shift}
\label{subsec:intro_blueshift}
In this section we argue that sufficiently low regularity solutions to (\ref{eq:1}) supported in the exterior region experience a blue-shift effect. We then state a result, Proposition \ref{prop:extinst}, which is the analog of Theorem \ref{thm:introrough1} for the exterior region. 

\subsubsection*{The blue-shift heuristic}

For naked singularity solutions to (\ref{eq:0}), or more generally first central singularities, the blue-shift triggers growth of transversal derivatives of the gravitational degrees of freedom along the ingoing null cone of the singularity. On the $k$-self-similar metric\footnote{The arguments of this section apply to asymptotically $k$-self-similar metrics as well.} $g_k$ this effect is explicit: restricting (\ref{eq:1}) to $\{v=0\}$ and inserting the values fixed by self-similarity, we find  
    \begin{equation}
        \label{eq:intro2}
        \partial_u \Big(\frac{1}{\partial_v r_k} \partial_v \varphi \Big) - \frac{1+k^2}{|u|} \Big(\frac{1}{\partial_v r_k} \partial_v \varphi \Big) = -\frac{1}{r_k}\partial_u \varphi.
    \end{equation}
This equation can be interpreted as a differential equation for $\partial_v \varphi$. Directly integrating yields 
\begin{align}
    \label{eq:intro2.4}
    \Big(\frac{1}{\partial_v r_k} \partial_v \varphi \Big)(u,0) &= |u|^{-1-k^2} I_k(u),
\end{align}
where
\begin{align}
    \label{eq:intro2.5}
    I_k(u) \doteq \Big(\frac{1}{\partial_v r_k} \partial_v \varphi \Big)(-1,0) -  \int_{-1}^u |u'|^{1+k^2}\frac{1}{r_k}\partial_u \varphi (u',0)du'.
\end{align}
Recall that self-similar bounds require (\ref{eq:intro2.4}) to be bounded by $|u|^{-1}$ as $u \rightarrow 0$, i.e. for $|I_k(u)| \lesssim |u|^{k^2}$. However, the individual terms in (\ref{eq:intro2.5}) are generically non-vanishing, and thus the asymptotic for $\frac{1}{\partial_v r_k} \partial_v \varphi$ depends on the degree of cancellation (if any) between the terms defining $I_k(u)$. 

We will, somewhat informally, term the expectation that generic data to (\ref{eq:1}) leads to non-self-similar growth along $\{v=0\}$ with rate $\frac{1}{\partial_v r_k} \partial_v \varphi \sim |u|^{-1-k^2}$ the blue-shift heuristic. The explicit rate $|u|^{-1-k^2}$ is termed the \textbf{blue-shift rate}, in contrast to the \textbf{self-similar rate} $|u|^{-1}$. In light of Definition \ref{defn:self-similarbounds}, any rate strictly between the two is also considered an instability.

\subsubsection*{Exterior perturbations}
For solutions with support in the exterior region, the value of the ingoing scalar field, $\partial_u \varphi(u,0)$, and therefore the integrand in (\ref{eq:intro2.5}), is fixed by the choice of initial data. Assuming this data satisfies self-similar bounds, it follows that for \textit{generic} choices of $\partial_v \varphi (-1,0)$, the solution grows with the blue-shift rate along $\{v=0\}$. For example, in the low regularity bounded variation (BV) class\footnote{Local well-posedness is known to hold in this class, by arguments analogous to \cite{chris1}.}, it is sufficient to take as data $\partial_u \varphi(u,0) =0, \ \partial_v \varphi(-1,v) = \mathbbm{1}_{\{v \in [0,1]\}}(v).$ 

If we require higher regularity for the scalar field data---for example, even continuity of $\partial_v \varphi(-1,v)$ across $\{v=0\}$---the situation changes considerably. Regularity implies $\partial_v\varphi (-1,v) \rightarrow 0$ as $v\rightarrow 0^+$, and thus the analysis of (\ref{eq:intro2.5}) is inconclusive, assuming as before the ingoing scalar field data to vanish. One still expects the blue-shift instability to exist, but its strength is now modified by the degree of localization to $\{v=0\}$; the latter is directly tied to the degree of regularity assumed on initial data. 

The following proposition quantifies the strength of the blue-shift for a family of exterior data whose regularity interpolates between below-threshold, threshold, and above-threshold regularities as the index $\alpha$ varies. Compare the statement below with Theorem \ref{thm:introrough1}.
 
\begin{prop}
    \label{prop:extinst}
    Fix any $k$-self-similar spacetime\footnote{This proposition is valid for any $(\epsilon_0,k)$-admissible spacetime, as defined in Section \ref{sec:admissiblespacetimes}.} with $k^2 \in (0,\frac13)$, and consider spherically symmetric, exterior characteristic initial data of the form 
        \begin{equation}
            \label{thm:extinst_eqn1}
            \begin{cases}
                \partial_v \varphi_0(-1,v) = g(v) |v|^{\alpha-1}\mathbbm{1}_{\{v \in [0,1]\}}(v) , \\[\jot]
                \partial_u \varphi_0(u,0) = 0, \\[\jot]
                \varphi_0(-1,0) = 0,
            \end{cases}
        \end{equation}
        where $g(v) \in C_v^1([0,1])$ is bounded above and below by positive constants, and $\alpha \in (1,2)$ is a fixed parameter. The following dichotomy holds:
        \begin{enumerate}[label=(\alph*)]
            \item (Self-similar bounds at and above threshold) If $\alpha \geq \frac{1}{1-k^2} $, the solution $\varphi$ to (\ref{eq:1}) with data (\ref{thm:extinst_eqn1}) satisfies $C^1$ self-similar bounds. Moreover, if $\alpha > \frac{1}{1-k^2}$, then $\varphi$ vanishes polynomially in $|u|$ as $u \rightarrow 0$, and its first derivatives obey bounds which are polynomially better than the self-similar rate.
            \item (Instability below threshold) If $1 < \alpha < \frac{1}{1-k^2} $, the solution $\varphi$ to (\ref{eq:1}) with data (\ref{thm:extinst_eqn1}) satisfies  
            \begin{equation*}
                 \big\| \frac{1}{\partial_v r_k}\partial_v \varphi \big\|_{L^\infty(\Sigma_u)} \sim |u|^{-1 -k^2 + (1-k^2)(\alpha-1)},
            \end{equation*}
            where $\Sigma_u$ denotes the outgoing null surfaces $\{u=\text{const}.\}$. This rate approaches the blue-shift rate as $\alpha \rightarrow 1$, and the self-similar rate as $\alpha \rightarrow \frac{1}{1-k^2}$.
        \end{enumerate}
\end{prop}

This result is a straightforward application of the techniques in \cite{igoryak1, igoryak2, singh}, and is proved in Appendix \ref{section:app1}. In fact, applying similar methods gives a stronger nonlinear result, stated informally here:
\begin{prop}
        \label{prop:extinst2}
        Fix an $(\epsilon_0,k)$-admissible spacetime, and consider sufficiently small, spherically symmetric, exterior perturbations of the outgoing scalar field data as in (\ref{thm:extinst_eqn1}), with $\alpha \geq \frac{1}{1-k^2}$. Then the solution to the Einstein-scalar field system exists in any  self-similar neighborhood\footnote{i.e. a neighborhood of the form $\{0 \leq \frac{v}{|u|^{1-k^2}} \leq c\}$, with $c$ an arbitrary positive constant.} of $\{v=0\}$. The spacetime is free of trapped surfaces, and the scalar field satisfies $C^1$ self-similar bounds. Moreover, if $\alpha > \frac{1}{1-k^2}$, then the background spacetime is ``asymptotically stable'' as $(u,v) \rightarrow (0,0)$.
\end{prop}

\subsubsection*{Critical collapse and the existence of dispersive perturbations}
Above, we have motivated our use of interior perturbations as a problem of extending \cite{chris3} to show nonlinear instability to trapped surface formation in high-regularity. There is, however, a complementary motivation coming from the literature surrounding ``critical collapse.'' For a detailed review, see \cite{gundreview}. 

Numerical simulations of Einstein-matter systems interpolating between small-data and large-data regimes have consistently found that naked singularities (typically with a continuous or discrete self-similarity) appear as local attractors for solutions on the verge of either dispersing or forming black holes. These attractors are termed ``critical solutions,'' and it is of interest to identify these solutions. Candidates for the system (\ref{eq:0}) have been numerically constructed in \cite{gund1,gund2}, and analytically for the Einstein-$\textrm{SU}(2)$ system in \cite{bizon_2002}. However, the full critical collapse picture has not been established for any model.

A consequence of this picture, if valid, would be the existence of perturbations of certain naked singularities which lead to dispersion, i.e. to the singularity itself disappearing! Such an instability cannot be captured by arguments of the type given in \cite{chris3}, which concern the exterior region (and thus cannot affect the dynamics of singularity formation). Even in the case of $k$-self-similar solutions, it is not known if dispersive perturbations could exist. The methods for studying interior perturbations introduced in the present work, although linear in nature, may be useful in addressing these questions.

\subsection{Proof outline}
\label{subsec:intro_proofoutline}
We focus this discussion on the proof of Theorem \ref{thm:introrough1}(a), which contains the key techniques of the paper. Recall this theorem asserts $C^1$ self-similar bounds and convergence to constants for solutions to (\ref{eq:1}) arising from data with regularity \textit{above threshold}. When relevant, we interperse remarks on the threshold and below threshold settings, as well as the case of non-spherically symmetric solutions. Therefore, unless otherwise specified we work with spherically symmetric solutions $\varphi(u,v)$ supported in the interior region of a fixed spacetime $(\mathcal{M}, g, \phi)$. For concreteness, this spacetime can be assumed to be $k$-self-similar for some $0 < k \ll 1$.

For a fixed value of $k$, define the constants $q_k \doteq 1-k^2$, $p_k \doteq (1-k^2)^{-1}$. Observe $q_k < 1 < p_k.$

\subsubsection*{Coordinates and geometric setup}
We first introduce two non-double null coordinate systems: similarity coordinates $(s,z)$ and hyperbolic coordinates $(t,x)$. Here $s, t$ serve as ``time'' coordinates, and $z, x$ as ``space'' coordinates. Similarity coordinates are regular across $\{v=0\}$, and surfaces of constant $s$ are outgoing null. Hyperbolic coordinates, on the other hand, are built out of an asymptotically (past) null slicing, with surfaces of constant $t$ tracing out asymptotically (past) null hyperboloids $\{|u| |v|^{p_k} = \text{const.}\}$.

Define $\psi \doteq r \varphi$, where $r(u,v)$ is the area radius function of the spacetime. It follows that (\ref{eq:1}) reduces to
\begin{equation}
        \label{eq:proofout1}
        \partial_s \partial_z \psi - q_k |z| \partial_z^2 \psi + q_k \partial_z \psi + V(s,z) \psi = 0,
\end{equation}
where $V(s,z) \sim k^2$ is a positive potential with small amplitude. For a $k$-self-similar spacetime, $V(s,z) \doteq V_k(z)$ reduces to a function of the spatial coordinate $z$. 

To prove self-similar bounds for the solution, it is equivalent to establish $|\partial_s \psi|, |\partial_z \psi| \lesssim e^{-s}$ as ${s \rightarrow \infty}$. The main obstacle to proving these estimates is the blue-shift effect, which we see by restricting (\ref{eq:proofout1}) to $\{z=0\}$ as an equation for $\partial_z \psi(s,0)$. Directly integrating the equation as in Section \ref{subsec:intro_blueshift} suggests the asymptotic $\partial_z \psi \sim e^{-(1-k^2) s}$, which is the blue-shift rate\footnote{Recall we are estimating $r\varphi$, rather than $\varphi$ itself. The bound consistent with self-similarity is then $|r\varphi| \lesssim |u|$, and the blue-shift bound is $|r\varphi| \lesssim |u|^{1-k^2}$.} expressed in similarity coordinates. The challenge of the proof is in extracting additional \textit{decay} from the term $V(s,z)\psi$. As ${V \sim k^2}$ is small, it is natural that this balance is delicate. 

\begin{rmk}
    For non-spherically symmetric solutions, the projection onto a fixed $(m,\ell)$-mode satisfies (\ref{eq:proofout1}) with an additional potential $L_{\ell}(s,z) \approx \frac{\ell(\ell+1)|u|^2}{r^2}$. Unlike the potential $V(s,z)$, the angular potential is of size $\sim 1$, and so for small $k$ one expects the blue-shift term to be dwarfed by the contribution of the angular terms.
\end{rmk}

The argument for Theorem \ref{thm:introrough1}(a) is naturally divided into three steps. We devote a section to each below.

\subsubsection*{Step 1: Backwards scattering}
The regularity of characteristic initial data $\psi_0(z)$ is described by a scale of Hölder-type spaces $\mathcal{C}^{\alpha}_{(hor)}(\{s=0\})$, $\alpha \in (1,2)$, modeled on the functions $|z|^{\alpha} \in \mathcal{C}^{\alpha}_{(hor)}(\{s=0\})$. These spaces only distinguish regularity as $z \rightarrow 0$. To the past of this null cone, we require all data to be sufficiently smooth (e.g., $C^m$ for $m \geq 5$). The $k$-self-similar scalar field lies in $\mathcal{C}^{p_k}_{(hor)}(\{s=0\})$, and we thus take $\alpha \in (1,p_k)$ as the range of below-threshold regularity, $\alpha = p_k$ as threshold regularity, and $\alpha \in (p_k,2)$ as above-threshold regularity.

Although we are ultimately interested in the asymptotics of solutions in $\{s \geq 0\}$, it is helpful to first translate the problem from one with outgoing characteristic data $\psi_0(z)$ along $\{s=0\}$ to an equivalent problem with data $(\psi(x), \partial_t \psi(x)) = (f_0(x), f_1(x))$ along the spacelike slice $\{t = 0\}$. This induced data is quantitatively more regular than the null data; the tradeoff, however, is that we must allow for data with exponential tails as $x \rightarrow \infty$. See Figure \ref{fig2}.

We build a dictionary between the regularity of $\psi_0(z)\in \mathcal{C}^{\alpha}_{(hor)}(\{s=0\})$ and the optimal decay achievable for spacelike data $(f_0(x),f_1(x))$, with the latter recorded in a scale of function spaces $\mathcal{D}^{\alpha,\bullet}_{\infty}(\{t=0\})$. This result, which can be viewed as a backwards scattering statement, follows by solving (\ref{eq:proofout1}) in a region $\{s \leq 0\}$ subject to prescribed data along $\{s=0\}$ and free data on $\{z=0, s \leq 0\}$. Exploiting the choice of free data, we use multiplier estimates to establish the optimal decay rates $|f_0|, |f_1| \lesssim e^{-\alpha q_k x}$. The estimates rely on the use of a multiplier vector field $\partial_s$, a multiple of the self-similar vector field. As the underlying spacetimes are asymptotically $k$-self-similar, this vector field is suitable for producing energy-type estimates.

\begin{figure}[h]
    \centering
    \includegraphics[scale=.8]{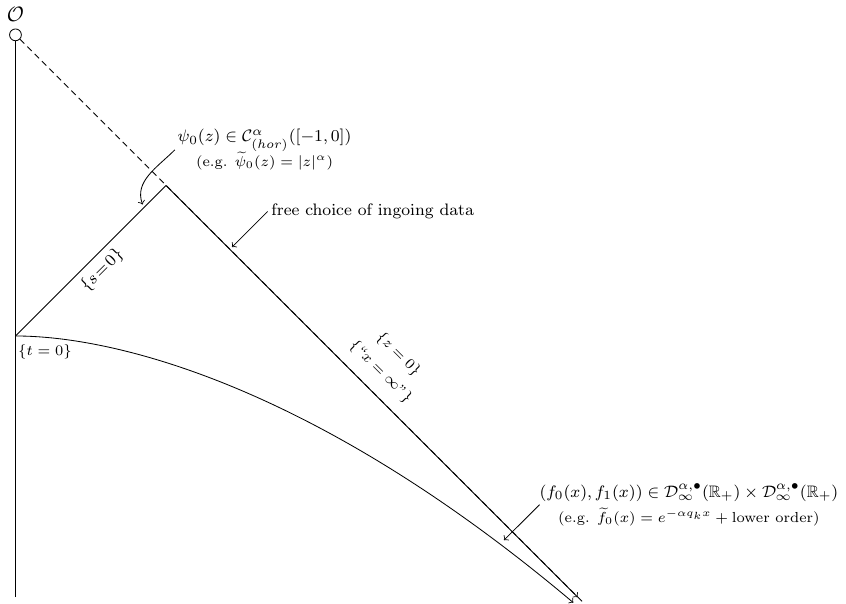}
    \caption{Schematic representation of Step 1. The diagram should be read ``towards the past,'' with the result being the construction of data along $\{t=0\}$. Also shown is an explicit example of null data $\wt{\psi}_0$, and a component of the associated spacelike data.}
    \label{fig2}
    \end{figure}

\subsubsection*{Step 2: Resonance expansion and solution theory in hyperbolic coordinates}
Having constructed spacelike data, we transform (\ref{eq:proofout1}) to hyperbolic coordinates and study the resulting wave equation. For simplicity, consider the case of exact $k$-self-similarity. The equation for $\psi$ is schematically
    \begin{equation}
        \label{eq:proofout2}
        \partial_t^2 \psi - \partial_x^2 \psi + 4q_k V_k(x)e^{-2q_k x} \psi = F(t,x),
    \end{equation}
where we allow for a forcing $F(t,x)$, and the potential $V_k(x)$ is a function of the spatial coordinate alone. Define the spectral family 
    \begin{equation*}
        P_k(\sigma) \doteq  -\partial_x^2 + (4q_k V_k(x)e^{-2q_k x} - \sigma^2)
    \end{equation*}
on the half-line $\mathbb{R}_+$. Taking a formal Fourier-Laplace transform of (\ref{eq:proofout2}), it is natural to study the spectral theory of $P_k(\sigma)$, and in particular the existence and analytic properties of a scattering resolvent $R(\sigma)$. 

The procedure for constructing $R(\sigma)$ goes through the methods of $1$-dimensional scattering theory on the half-line, as discussed in \cite{yafaev,dyatlovzworski}. Let $\rho_{x_i}$ denote a cutoff to $\{x \lesssim x_i\}$, and define the cutoff resolvent $R_{x_0,x_1}(\sigma) \doteq \rho_{x_0} R(\sigma)\rho_{x_1}$. It is straightforward given the exponential decay of the potential to construct $R_{x_0,x_1}(\sigma)$ as a meromorphic family of operators on $L^2_x(\mathbb{R}_+)$, for $\{\Im \sigma > -q_k\}$. The construction can be made explicit by the introduction of outgoing and Dirichlet solutions $f_{(out),\sigma}(x), f_{(dir),\sigma}(x)$ to the equation $P_k(\sigma)f = 0$, cf. the definition (\ref{eq:resolvkernel}). See also Figure \ref{fig3}.

However, in order to establish self-similar bounds via a resonance expansion, we require the existence of a meromorphic extension of $R_{x_0,x_1}(\sigma)$ on a full neighborhood of $\{\Im \sigma = -1\}$. To achieve this we apply the small-$k$ expansion of the metric $g_k$ in Appendix \ref{section:app3} to produce a splitting 
    \begin{equation*}
        P_k(\sigma) = \underbrace{-\partial_x^2 + (w_k^2 k^2 e^{-2q_kx})}_{P_k^{(0)}(\sigma)} + O_{L^\infty}(k^2e^{-4q_k(1-\epsilon)x}),
    \end{equation*}
where the leading order operator $P_k^{(0)}(\sigma)$ has an explicit exponential tail, and the constant $w_k$ admits an asymptotic expansion in $k$. With this splitting, we are able to define a cutoff resolvent $R^{(0)}_{x_0,x_1}(\sigma)$ for the operator $P_k^{(0)}(\sigma)$ which is meromorphic on $\mathbb{C}$. By gaining sufficiently good control on the outgoing and Dirichlet solutions for the leading order operator, a perturbation argument yields the existence of a meromorphic cutoff resolvent $R_{x_0,x_1}(\sigma)$ in the domain $\mathbb{I}_{[-\frac32, \frac12]} \doteq \{-\frac32 \leq \Im \sigma \leq \frac12 \}$.

This argument rests on having sufficient regularity on the background metric to extract the leading order term in an expansion of $V_k(x)$ near infinity. Such an expansion does not obviously continue to higher order, and therefore the existence of a meromorphic extension to $\mathbb{C}$ is unclear. The limited regularity motivates the construction of a resolvent by hand, and distinguishes this problem from (much more general) resolvent constructions in the asymptotically hyperbolic case (see \cite{dyatlovzworski}).

With the construction of a cutoff resolvent in $\mathbb{I}_{[-\frac32, \frac12]}$, two problems remain. The first consists in identifying the locations and multiplicities of any poles of $R_{x_0,x_1}(\sigma)$. This problem may be equivalently stated as one involving the location of zeros of the Wronskian $\mathcal{W}(\sigma) \doteq W[f_{(out),\sigma}, f_{(dir),\sigma}]$. It is here that the delicate balance in (\ref{eq:proofout1}) between the blue-shift growth mechanism and the potential damping due to $V(s,z)\psi$ is manifest. A perturbation argument using the explicit form of the leading order operator yields that there is a zero of $\mathcal{W}(\sigma)$ in a neighborhood of ${\sigma = -i q_k}$, and we have  
    \begin{equation*}
        \mathcal{W}(\sigma) \sim \frac{1}{\Gamma(\sigma + i q_k)} + O(k^2).
    \end{equation*}
    The $O(k^2)$ term depends on the asymptotic expansion for $w_k$, as well as the sub-leading terms we dropped by considering the leading order operator. As ${\sigma = -i q_k}$ corresponds to a resonance at the blue-shift rate, and ${\sigma = -i}$ the self-similar rate, it follows that distinguishing stability and instability rests on the sign of the small $O(k^2)$ terms. The assumption of small $k$ allows us to prove that the region $\Im \sigma \geq -1-O(k^2)$ contains exactly one, simple zero of $\mathcal{W}(\sigma)$. However, it offers no insight into whether this zero lies in the stable ($\Im \sigma \leq -1$) or unstable ($\Im \sigma > -1$) region. 

    The approach we take towards identifying the pole is motivated by \cite{warn}, which emphasizes the relationship between scattering resonances defined with repect to spacelike slices, and the regularity of the respective mode solutions along outgoing null slices. Exploiting this relationship, we are able to show that the unique zero of $\mathcal{W}(\sigma)$ lies at $\sigma = -i$, and corresponds to solutions of (\ref{eq:1}) with $\varphi = \text{const.}$ In a residue expansion, it is the residue at this pole that is responsible for the appearance of constants in the statement of Theorem \ref{thm:introrough1}(a).

    With the mode identified, the remaining problem concerns the applicability of a leading order resonance expansion for initial data that is not compactly supported, but rather decays with exponential tails. We show that in the high regularity case $\alpha > p_k$, the decay established by the backwards scattering result is fast enough to define a meromorphic extension of the resolvent with the same pole structure in $\{\Im \sigma \geq -1-\delta \}$, for some $\delta \in (0,  \alpha q_k - 1)$ sufficiently small depending on the regularity gap $|\alpha - p_k|$. A standard resonance expansion then yields the desired decay of $\varphi$ to constants, in regions $\{x \leq \text{const.}\}$. 

    We note that although the spectral theory is performed on a $k$-self-similar background, the flexibility to include a (suitably decaying) forcing $F(t,x)$ in (\ref{eq:proofout2}) allows the argument to go through on general asymptotically $k$-self-similar spacetimes. 

    \begin{figure}[h]
        \centering
        \includegraphics[scale=.75]{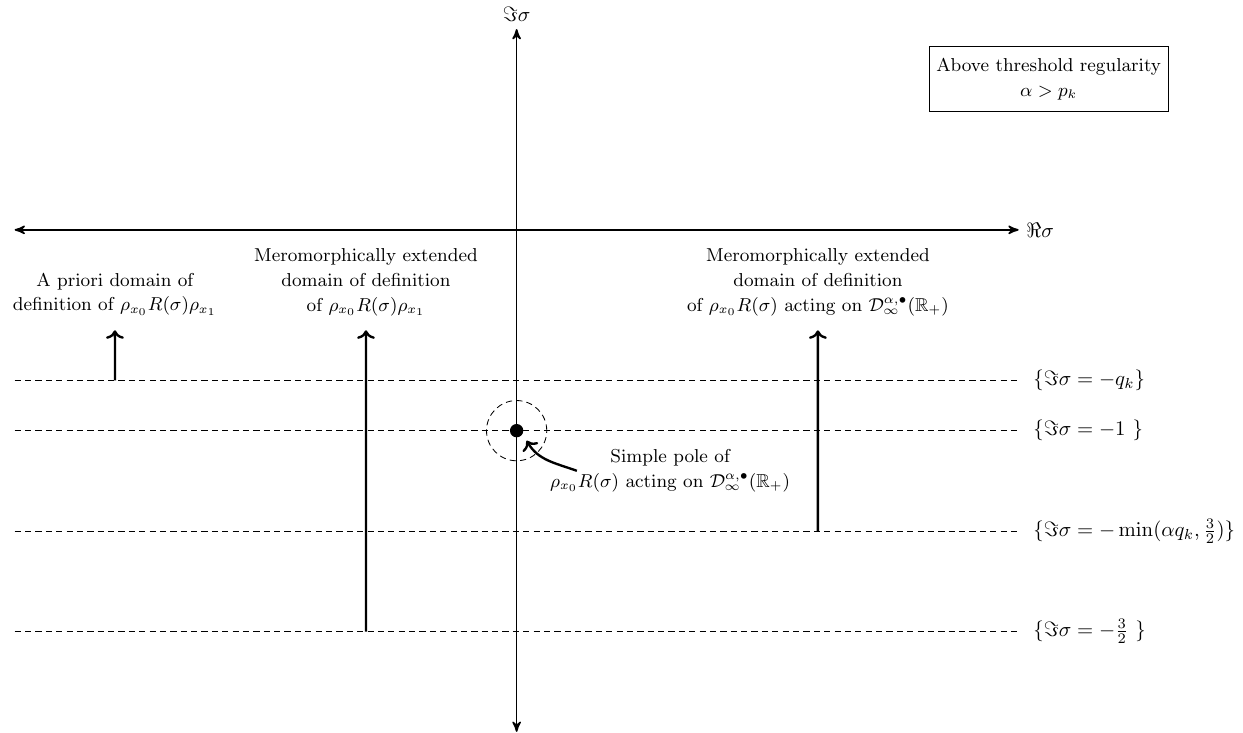}
        \caption{Domains in the complex plane for which various resolvent operators are defined. In the context of a resonance expansion, the imaginary part of $\sigma$ corresponds to growth/decay of waves at a rate $e^{(\Im \sigma) t}$ in regions $\{x \leq \text{const.}\}$.
        This diagram is specific to setting of Theorem \ref{thm:introrough1}(a).    } 
        \label{fig3}
        \end{figure}

\begin{rmk}
    \label{rmk:introthresholdreg}
    Steps 1--2 illustrate that the \textit{regularity} of null initial data determines the \textit{decay} of induced spacelike data, which in turns limits the \textit{domain of definition} of the resolvent operator. For regularities above threshold, i.e. ${\alpha>p_k}$, this procedure allows for a straightforward definition of the resolvent on a strict neighborhood of $\{\Im \sigma = -1\}$, corresponding in a resonance expansion to the self-similar rate. 
    
    For regularities at or below threshold, however, the situation is more complicated. Taking null data $\psi_0 \in \mathcal{C}^{p_k}_{(hor)}(\{s=0\})$, for example, the associated spacelike data decays at a rate which only allows for a resolvent to be defined in the half plane $\{\Im \sigma > -1 \}$. To address this issue, we further restrict the class of threshold regularity data to include functions in the space $\mathcal{C}^{\alpha,\delta}_{(hor)} \subset \mathcal{C}^{\alpha}_{(hor)}$. Data in this class admits an expansion 
    \begin{equation*}
        \varphi_0(z) \sim |z|^{p_k} + O(|z|^{p_k+\delta}),
    \end{equation*}
    This condition is verified for the $k$-self-similar scalar field, and the proof of backwards scattering in Step 1 can be extended to show that the induced data for $\varphi$ lies in $\wt{\mathcal{D}}_\infty^{\alpha,\bullet,\delta,c_0}(\mathbb{R}_+) \subset \mathcal{D}^{\alpha,\bullet}_{\infty}(\mathbb{R}_+)$, corresponding to functions with the exact exponential tail 
    \begin{equation*}
        f_0(x) = c_0 e^{-x} + O(e^{-(1+\delta q_k) x}).
    \end{equation*}
    It will follow from the spectral analysis that the resolvent, applied to \textit{this class of functions}, extends meromorphically to a strict neighborhood of $\{\Im \sigma = -1\}$. 
    
    An interesting consequence of this extension (cf. Proposition \ref{lem:extR_sigma_2}) is that the previously identified single pole at $\sigma=-i$ becomes a double pole, with the $k$-self-similar scalar field function serving as a generalized resonance. This serves to explain the origin of the logarithmic growth (in $|u|$) of $\phi_k(u,v)$; the resonance expansion for solutions in this threshold regularity class encounters a double pole at $\sigma=-i$, the residue of which generates an additional logarithmically growing term. 
\end{rmk}

\subsubsection*{Step 3: Multiplier estimates}
Expressing the results of the previous step in similarity coordinates yields the estimates 
\begin{equation}
    \label{eq:proofout3}
|\partial_s \big(r(\varphi - \varphi_\infty) \big)|, \  |\partial_z \big(r(\varphi - \varphi_\infty) \big)| \lesssim e^{-(1+\delta) s},
\end{equation}
holding in a region $\{z \leq z_0 <0\}$ bounded strictly away from the cone $\{z=0\}$, for some constants $\varphi_{\infty}$ and $\delta > 0$. The remaining step is to propagate these bounds globally in the interior region. The argument proceeds in physical space, using vector field multipliers and commutators for (\ref{eq:proofout1}). See Figure \ref{fig4}.

Define $\psi_{\delta} \doteq  e^{(1+\delta) s} r(\varphi - \varphi_\infty)$, which satisfies
\begin{equation*}
    \partial_s \partial_z \psi_{\delta} - q_k |z| \partial_z^2 \psi_{\delta} - (\delta+k^2)\partial_z \psi_{\delta} + V(s,z)\psi_{\delta} = 0.
\end{equation*} 
The challenge is finding a way to absorb the remaining blue-shift term $(\delta+k^2)\partial_z \psi_{\delta}$, which can contribute unfavorable bulk terms to a multiplier estimate. 
The key structure we use is that under commutation by $\partial_z$, this blue-shift term gains a good sign provided $\delta, k$ are small:
\begin{equation}
    \label{eq:proofout4}
        \partial_s \partial_z^2 \psi_{\delta} - q_k |z| \partial_z^3 \psi_{\delta} + (1-\delta-2k^2)\partial_z^2 \psi_{\delta} + V(s,z)\partial_z \psi_{\delta} + \partial_z V(s,z) \psi_{\delta}=0.
\end{equation} 
If $\psi_{\delta} \in W^{2,2}_z(\{s=0\})$, we may multiply by $\partial_z^2 \psi_{\delta}$ and integrate by parts in $\mathcal{R}(0,s_0) \doteq \{0 \leq s \leq s_0 \}$. To control the error terms, we observe that any term with support away from $\{z=0\}$ is already controlled, and that the potential $V$ is small in integrated norms: 
\begin{equation}
    \label{eq:proofout5}
    \sum_{0 \leq i+j \leq 1}\|\partial_s^i \partial_z^j V(s,z)\|_{L^p_z(\{s=s_0\})} \lesssim_p k^2,
\end{equation}
which holds for all sufficiently small $k$ and any $p < \infty$. The commuted energy estimate then closes, yielding control on the top order quantity $\|\partial_z^2 \psi_\delta \|_{L^2_z(\{s=s_0\})}$. Combined with the decay (\ref{eq:proofout3}) on lower order quantities near the axis and integration along characteristics, we will conclude the desired bounds.

An important complication arises, however, for less regular data. We allow for general data $\psi_{\delta} \in \mathcal{C}^{\alpha}_{(hor)}(\{s=0\}),$ and therefore for $\alpha < \frac{3}{2}$ it is not the case that $\psi_{\delta} \in W^{2,2}_z(\{s=0\})$. To adjust our estimates in this low regularity setting, we instead multiply (\ref{eq:proofout4}) by $|z|^{\omega} \partial_z^2 \psi_{\delta}$ and integrate by parts, for an $\omega>0$ depending on $\alpha$. The favorable bulk term in (\ref{eq:proofout4}) appearing in the energy estimate now depends on both $\delta$ and $\omega$, however, and one must be careful to choose the parameters appropriately for the estimate to close. 

\begin{figure}[h]
    \centering
    \includegraphics[scale=.7]{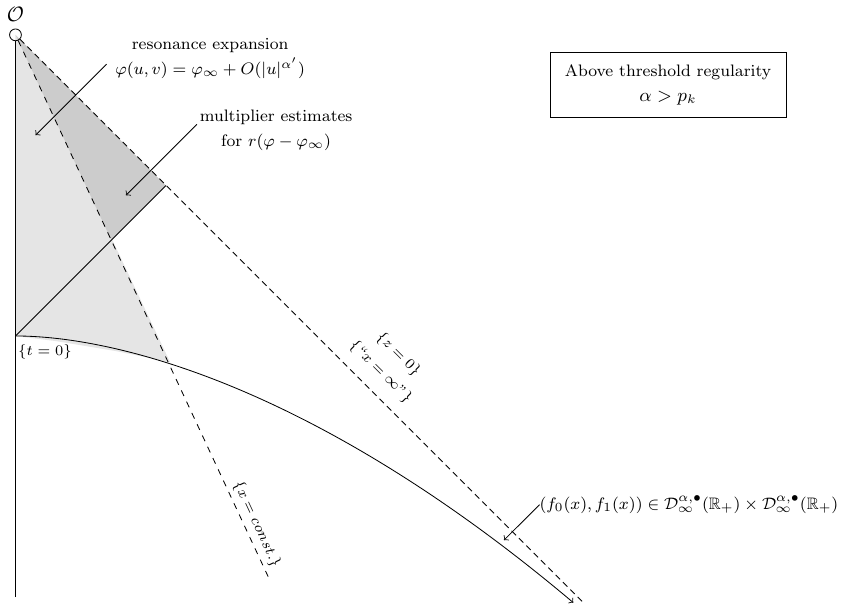}
    \caption{The conclusion of the resonance expansion (active in light shaded area) and the multiplier estimates (active additionally in the darker shaded area).}
    \label{fig4}
    \end{figure}

\begin{rmk}
The proof of Theorem \ref{thm:introrough2} concerning non-spherically symmetric solutions is carried out entirely at the level of multiplier estimates for the individual projections onto angular modes, and Steps 1--2 are not required. The primary novelty in this case is the presence of the angular potential $L_{\ell}(s,z)$. We must ensure that the estimates associated to $\partial_z$ multipliers and commutators are independent of $\ell$, as the angular terms are not small in $k$. It turns out that the derivatives of this potential have favorable sign (Lemma \ref{lem:propertiesofL_kl}), and so the estimates proceed much as in the spherically symmetric case. 

For regularities above, but within $O(k^2)$, of the threshold $\alpha = p_k$, the estimate (\ref{eq:proofout5}) turns out to be too weak to close the multiplier estimates. As a result, the precise statement of our result (cf. Theorem \ref{theorem:angularthm}) is restricted to regularities $\alpha \geq p_k + O(k^2)$. A similar obstacle is faced for spherically symmetric solutions; however, in that case we use additional estimates arising from the resonance expansion.
\end{rmk}

\subsection{Outlook}
    \label{subsec:intro_outlook}
In this section we bring together our main results for interior perturbations, Theorems \ref{thm:introrough1}--\ref{thm:introrough2}, with the corresponding statements Propositions \ref{prop:extinst}--\ref{prop:extinst2} in the exterior, and discuss the implications for weak cosmic censorship. We also remark on various open problems.

\vspace{.5em}
\noindent
\textbf{(Nonlinear) Stability above threshold regularity: } Theorem \ref{thm:introrough1}(a) implies that solutions to (\ref{eq:1}) arising from interior, sufficiently high-regularity initial data obey $C^1$ self-similar bounds. In fact, $\partial \varphi$ obeys bounds which are polynomially in $|u|$ \textit{better} than self-similar. This behavior suggests that similar decay should hold for classes of nonlinear wave equations on $(\epsilon_0,k)$-admissible backgrounds, which (a) respect self-similar scaling, (b) have a linearization given by (\ref{eq:1}), and (c) have only derivative nonlinearities. 

Ultimately we are interested in the stability of $(\mathcal{M}, g_k, \phi_k)$ as a solution to the spherically symmetric Einstein-scalar field system. Although this system does reduce (in an appropriate gauge) to a nonlinear wave equation, (\ref{eq:1}) is not the linearization. There is an additional nonlocal term capturing the linear backreaction on the geometry. Robust decay results for equations of this type are the main obstacle to proving a true stability statement in regularities above threshold, for interior perturbations. 

Under exterior, high-regularity perturbations, it is known that these spacetimes are asymptotically stable (Proposition \ref{prop:extinst2}). The techniques introduced in \cite{igoryak1,igoryak2}, when translated to the spherically symmetric system (\ref{eq:0}), allow one to bypass the linearization altogether, and directly apply the analysis of the linear wave equation (\ref{eq:1}). It is not clear if similar techniques apply to the study of the interior region, or if there are novel effects associated to the linear backreaction.

Results showing nonlinear stability of the interior region of finite regularity, singularity forming solutions to nonlinear wave equations have been obtained in other contexts, cf. \cite{burzio_krieger, kriegermiao, kriegermiaoschlag}. We emphasize that an analogous result for $k$-self-similar naked singularities, if true, would not pose any direct challenge to weak cosmic censorship. The most natural formulation of the conjecture is to require only that finite regularity naked singularities be unstable to perturbations of commensurate regularity. Still, such a result would signal that the conjecture depends sensitively on the topology of initial data considered, and may even depend on the scale of spaces in which one measures the ``threshold'' regularity. 

\vspace{.5em}
\noindent
\textbf{(Nonlinear) Stability/instability at threshold regularity: } In light of the stabilizing influence of high-regularity perturbations, one can instead ask if $k$-self-similar spacetimes are nonlinearly unstable to perturbations at precisely the regularity of the background scalar field. Our (mostly linear) analysis offers few expectations in this case.

Proposition \ref{prop:extinst2} establishes that the exterior region is \textit{orbitally stable} to perturbations of threshold regularity\footnote{In fact we do not expect asymptotic stability, a consequence of an additional $1$-parameter family of $k$-self-similar spacetimes with fixed $k$. The latter correspond to a different choice of Neumann data in a self-similar expansion; for more details, we refer to \cite{igoryak1, igoryak2}.}. Despite the absence of a quantitative asymptotic stability statement, we rule out any instability leading to trapped surface formation or to dispersion.

The linear result Theorem \ref{thm:introrough1}(b) establishes convergence of $\varphi$ to a leading order self-similar profile, comprised of a linear combination of constants and the $k$-self-similar scalar field $\phi_k$. Whether this decay is preserved for solutions to the full linearized system, and whether it is sufficient to close a nonlinear argument, remains to be understood. 

\vspace{.5em}
\noindent
\textbf{(Nonlinear) Instability below threshold regularity: }
Sufficiently irregular solutions to (\ref{eq:1}) experience an instability in both the interior and exterior regions. The explicit bound\footnote{A similar, but slightly weaker, bound holds for interior solutions.} in Proposition \ref{prop:extinst}(b) shows that the instability rate is confined to a range between the self-similar rate and the blue-shift rate, with the latter achieved only for BV perturbations. 

A consequence is that the blue-shift heuristic, stated quantitatively, does not hold for spherically symmetric solutions to (\ref{eq:1}) with data in any topology that is of the form $C^{1,\beta}([-1,0])$, $\beta > 0 $, near ${v=0}$. Such a result is essentially contained in \cite{igoryak1,igoryak2,singh} for perturbations supported in the exterior; however, that this should be true in the interior region as well, where both terms in (\ref{eq:intro2.5}) are generically non-trivial, is surprising. It follows that the existence of a cancellation along $\{v=0\}$ in (\ref{eq:intro2.5}) is a robust phenomenom, applying even in low regularities. 

We nevertheless have that generic soluions to (\ref{eq:1}) are unstable in all regularities strictly below threshold. For exterior perturbations one expects instability to hold nonlinearly as well, and to lead to trapped surface formation. For interior perturbations, the situation is unclear for similar reasons as sketched above. It is not clear, for example, whether trapped surface formation is the generic endstate of instability, or whether interior perturbations of regularity at, or strictly below, threshold could lead to a globally dispersive spacetime.

\vspace{1em}
    \noindent
    \textbf{Extensions to other naked singularity backgrounds: } A natural extension of this work is to consider solutions to (\ref{eq:1}) on different naked singularity interiors. With the explosion of interest in the subject in recent decades, there are many examples to consider \cite{igoryak2,yakov1,guohadzicjang,anzhang,choptuik1,gund1,gund2}. We briefly comment here on the issues that may arise in extending the analysis of this paper to new settings.

    The construction \cite{igoryak2,yakov1} of vacuum naked singularities in ${3\!+\!1}$ dimensions shares many features with Christodoulou's $k$-self-similar spacetimes, and are in fact built on an analogous $\kappa$-self-similarity adapted to the vacuum equations in double-null gauge. In particular, the spacetimes share a limited regularity near $\{v=0\}$ and a global smallness of all double-null quantities away from $\{v=0\}$, both of which are crucial to the setup of this paper. The interior construction in \cite{yakov1} is highly non-spherically symmetric, however, and one expects more sophisticated tools from scattering theory to be required in proving an analog of Theorem \ref{thm:introrough1}.

    In the setting of the spherically symmetric Einstein-scalar field system, the discretely self-similar solutions numerically constructed in \cite{choptuik1, gund1} are a natural counterpart to the continuously self-similar solutions considered here. These solutions are believed to be smooth, linearly stable to non-spherically symmetric perturbations \cite{gund2}, and linearly (mode) unstable to smooth, spherically symmetric perturbations \cite{gund1}. It would be instructive to consider the relationship of the blue-shift heuristic to the observed mode instability \cite{gund1}, as well as to lower regularity instabilities. Related comments apply to the continuously self-similar naked singularities constructed in \cite{bizon_2002} for the Einstein--$\textrm{SU}(2)$ system.

     One can also hope to extend the results of this paper to $k$-self-similar naked singularities in the full range $k^2 \in (0,\frac13).$ A closer inspection of the proof of Theorem \ref{thm:introrough1}(a) shows that the assumption of small-$k$ is only required in analyzing the locations of poles of the cutoff resolvent $\rho_{x_0}R(\sigma)\rho_{x_1}$. It is in principle possible that additional mode instabilities (i.e. poles of the cutoff resolvent with imaginary part greater than $-1$) could arise for small, but finite, values of $k$. 
    
    Finally, we point out a setting in which the role of blue-shift remains unclear, even for low-regularity perturbations. The work \cite{guohadzicjang} constructs self-similar Einstein-Euler spacetimes with singularities arising due to fluid blowup, rather than the collapse of gravitational degrees of freedom. The blue-shift heuristic no longer holds directly along the backwards light cone of the singularity, and further study into the nature of instability (or perhaps stability) of these solutions is required. 

\subsection{Related works}
\label{subsec:intro_related}

    \vspace{1em}
    \noindent
    \textbf{Waves on black hole exteriors: } Spectral theoretic methods, of the type applied in the current paper, have had significant application in the study of linear (and nonlinear) wave equations on stationary black hole exteriors. In the asymptotically de Sitter ($\Lambda > 0$) case, a characteristic feature of massless waves is the exponential convergence to constants, with the rate determined by a spectral gap between a constant mode and shallow scattering resonances \cite{bonyhafner,dyat1,dyat2}. For a physical space perspective on exponential decay, see \cite{mavr1}. 

    In the asymptotically flat case, decay for waves is characterized by polynomial Price-law type tails, with the scattering resonances playing a secondary role to branch cut singularities in the scattering resolvent. This theory is well developed on Schwarzschild in the massless and massive cases, see \cite{hintz1,YakovFedMax}.

    The conclusions of our analysis share many similarities with the positive cosmological constant case. However, the obstacles to proving decay in self-similar naked singularity and black hole spacetimes are quite different. In our context the (potential) instability is restricted to the spherically symmetric mode, and higher $\ell$-modes are strictly easier to study; in particular, there is no trapping of null geodesics. The main novelty of the problem here is the implicit and highly non-analytic nature of the naked singularity metrics.

    \vspace{1em}
    \noindent
    \textbf{Waves on black hole interiors: } 
    The blue-shift instability appearing in the study of (\ref{eq:1}) has an analog in the study of the strong cosmic censorship conjecture; see \cite{dafluk} for a detailed introduction. As a linear analog of the conjecture, many works have investigated conditions for $C^1$ (or $H^1$) blowup of transversal derivatives of linear waves near Cauchy horizons appearing in black hole interiors. In the $\Lambda = 0$ case  a sample includes \cite{lukoh2,luksb,YakovChristoph}. The blowup is induced by a blue-shift effect associated to the positive surface gravity of the future Cauchy horizon $\mathcal{CH}^+$. However, there is inevitably a competition between the \textit{decay} of the solution along the future event horizon $\mathcal{H}^+$ due to dispersion of waves in the exterior region, and the \textit{growth} near $\mathcal{CH}^+$ due to the blue-shift. In the $\Lambda = 0$ case, it appears that the latter mechanism is decisive.
    
    An intriguing picture emerges in the $\Lambda > 0$ case, however, which has surprising parallels with the present work. We refer to \cite{YakovDaf} for a discussion of the novel issues associated to Reissner-Nordström de-Sitter and Kerr de-Sitter black holes. Based on heuristic work \cite{Cardoso_etal}, there appears to be a regime of near-extremal Reissner-Nordström de-Sitter black holes for which \textit{sufficiently regular} Cauchy data does not lead to any instability along $\mathcal{CH}^+$! In the high regularity, $\Lambda > 0$ setting, the decay of waves along $\mathcal{H}^+$ is exponential in an appropriate $v$ coordinate, and is determined by the location of quasinormal modes associated to the black hole exterior. These modes (particularly, the shallow modes determining the late-time tails) are difficult to determine analytically, and have little a priori relation to the geometric surface gravities associated to the event and Cauchy horizons, which determine the scale of the red-shift and blue-shift effects respectively. For near-extremal Reissner-Nordström de-Sitter black holes these modes are located such that the exponential decay along $\mathcal{H}^+$ is too fast to trigger $H^1$ blowup along the Cauchy horizon. The fate of strong cosmic censorship for such black holes thus appears unclear. Further complicating the story, it seems that for Kerr de-Sitter black holes this issue does not arise \cite{Dias_etal}.

    One suggestion for recovering strong cosmic censorship in the Reissner-Nordström de-Sitter case is given in \cite{YakovDaf}. By dropping the requirement that initial data be highly regular (say, $C^\infty$), one can choose a class of lower regularity data for which the decay along $\mathcal{H}^+$ is shown to generically be much slower; in fact, at an exponential rate determined by the surface gravity of $\mathcal{H}^+$. It is moreover shown that generic data in this class does blow up in $H^1$ at the Cauchy horizon. 

    A correspondence can be made with the $k$-self-similar naked singularity case. The present study aims to identify a blue-shift instability along the ingoing null surface $\{v=0\}$, which however fails to manifest for all threshold and above-threshold regularity data. That this should be true reflects the fact that decay for sufficiently regular solutions to (\ref{eq:1}) is determined by scattering resonances lying closest to the real-axis. In contrast, in regularities strictly below that of the background an analogous mode construction as in \cite{YakovDaf} can be performed. One can then construct mode solutions with growth arbitrarily close to that predicted by the blue-shift heuristic. Such a construction is in fact contained in Theorem \ref{thm:introrough1}(c).

    We conclude from this connection that it may be unphysical to study the instability of finite regularity naked singularities in regularities strictly higher than a threshold set by the background. One would still hope to show instability at the threshold regularity; in light of Theorem \ref{thm:introrough1}, it is unclear what the role of the blue-shift would be in such a result.

\subsection*{Guide to the paper}
    \label{subsec:overviewofpaper}

    In Section \ref{sec:prelims}, we introduce the $k$-self-similar spacetime and define the full class of $(\epsilon_0,k)$-admissible spacetimes to which our main results apply. We also discuss various coordinate systems, the form of (\ref{eq:1}) with respect to such coordinates, and translate the assumptions on $(\epsilon_0,k)$-admissible spacetimes to statements on the coefficients of (\ref{eq:1}).

    Section \ref{section:mainresults} gives precise statements of the main results, Theorems \ref{theorem:mainthm}--\ref{theorem:angularthm}.

    Section \ref{sec:multipliers} derives a series of multiplier estimates for (\ref{eq:1}) in similarity coordinates. As a corollary, we close weak decay bounds for the spherically symmteric part of the solution (Propositions \ref{prop:basicdecayphysicalspace}, \ref{prop:terrible_boundedness}).
    
    Section \ref{sec:scattering} discusses the relevant scattering theory construction on $k$-self-similar spacetimes. In particular, the meromorphic extension of the resolvent $R(\sigma)$ on a variety of function spaces is discussed, as well as the structure and location of its poles.
    
    Section \ref{sec:proof} completes the proofs of Theorems \ref{theorem:mainthm}--\ref{theorem:angularthm}. 
    
    Appendix \ref{section:app1} contains the proof of Proposition \ref{prop:extinst}. Finally, Appendix \ref{section:app3} derives the first terms in a small-$k$ asymptotic expansion of all metric and scalar field quantities on a $k$-self-similar interior.

\subsection*{Acknowledgements}
The author benefited extensively from conversations with Mihalis Dafermos, Igor Rodnianski, and Yakov Shlapentokh-Rothman.

\section{Preliminaries}
\label{sec:prelims}
In Section \ref{section:sphersymm} we review the language of spherically symmetric spacetimes, and in Sections \ref{section:christodoulou}--\ref{sec:renormalizedgauge} we review Christodoulou's $k$-self-similar metrics in different double-null gauges. We collect useful notation in Section \ref{sec:notation}.

In Section \ref{sec:admissiblespacetimes} we define $(\epsilon_0, k)$-admissible spacetimes, and Section \ref{subsec:coordinatesystems} introduces additional (non-double null) coordinate systems. The remaining Sections \ref{subsec:waveeqnandseparationofvars}--\ref{section:functspaces} are devoted to a discussion of the wave equation (\ref{eq:1}), including well-posedness in appropriate function spaces.

\subsection{Spherically symmetric spacetimes with scalar field}
\label{section:sphersymm}
We will study solutions to (\ref{eq:1}) on an underlying spacetime $({\mathcal{M}}, {g})$, where ${g}$ is a Lorentzian metric. The assumption of spherical symmetry implies we can define the quotient manifold $\mathcal{Q} ={\mathcal{M}} / \textrm{SO}(3)$, a ${(1\!+\!1)}$-dimensional Lorentzian manifold with metric also denoted $g$, and with a boundary $\Gamma$ comprised of fixed points of the $\textrm{SO}(3)$ action. This boundary is alternatively called the \textit{center} or \textit{axis}, and will be assumed to be a timelike curve.  

A function associated to the orbits of the symmetry action is the \textit{area radius}, defined geometrically for $p \in \mathcal{Q}$ by  
\begin{equation*}
    r(p) \doteq \sqrt{ \frac{\text{Area}(\text{proj}^{-1}(p))}{4\pi}}.
\end{equation*}
Here, $\text{proj}: {\mathcal{M}} \rightarrow \mathcal{Q}$ is the quotient map. In terms of the area radius, the center $\Gamma$ is given by $\{r(p)=0\}$. 

The quotient spacetimes considered here admit a global double-null gauge $(u,v)$, with respect to which the quotient metric $g$ assumes the form 
\begin{equation}
    g = -\Omega^2(u,v) du dv,
\end{equation}
for a gauge-dependent quantity $\Omega$ called the \textit{null lapse}. The ${(3\!+\!1)}$-dimensional metric is determed by the pair of functions $\Omega(u,v), r(u,v)$. We also define the \textit{Hawking mass}, 
\begin{equation}
    \label{eq:defnofm}
    m \doteq \frac{r}{2}\big(1 - g(\nabla r, \nabla r)\big) = \frac{r}{2}\bigg(1 + \frac{4\partial_u r \partial_v r}{\Omega^2}\bigg),
\end{equation}
as well as the \textit{mass ratio}
\begin{equation}
    \label{eq:defnofmu}
    \mu \doteq \frac{2m}{r}
\end{equation}
and the null derivatives of the area radius
\begin{equation}
    \label{eq:defnoflambda}
\nu \doteq \partial_u r, \quad \lambda \doteq \partial_v r.
\end{equation}
The wave equation (\ref{eq:1}) on a fixed spherically symmetric spacetime depends on the underlying geometry through the metric quantities $\Omega(u,v),r(u,v)$ (or equivalently, through $m(u,v),r(u,v)$). The $k$-self-similar quotient spacetimes, and more generally the class of $(\epsilon_0,k)$-admissible spacetimes defined in Section \ref{sec:admissiblespacetimes}, carry an additional real-valued scalar field $\phi(u,v)$, which is dynamically coupled to the metric via the spherically symmetric Einstein-scalar field system. In particular, this scalar field $\phi(u,v)$ is \textit{itself} a solution to (\ref{eq:1}). For details on the Einstein-scalar field system we refer to \cite{chris1.5,lukoh1,singh}. In the present paper, the scalar field associated to the background spacetime is denoted $\phi$, in contrast to the notation $\varphi$ used for general solutions to (\ref{eq:1}). We denote a given spherically symmetric spacetime with scalar field by the tuple $(\mathcal{Q},g,r,\phi).$

\subsection{$k$-self-similarity and self-similar gauge}
\label{section:christodoulou}
We do not attempt a detailed motivation of $k$-self-similarity here, and refer to \cite{chris2,igoryak2,singh}. Our focus is on the realization of $k$-self-similarity in double-null gauge, as well as quantiative bounds for double-null quantities. In the following, functions associated to an exactly $k$-self-similar spacetime are denoted with a subscript $k$. 

\begin{definition}
Fix a parameter $k \in \mathbb{R}$. A spherically symmetric solution $(\mathcal{Q}, g_k, r_k, \phi_k)$ to the Einstein-scalar field system is \textbf{k-self-similar} if there exists a one-parameter family of scaling diffeomorphisms $f_a: \mathcal{Q}\rightarrow \mathcal{Q}$, $a \in \mathbb{R}_+$, with respect to which
\begin{equation}
    f_a^{*}g_k = a^2 g_k, \ \ f_a^* r_k = a r_k, \ \ f_a^* \phi_k = \phi_k - k \log a.
\end{equation}
The vector field generating the family of diffeomorphisms is denoted $K$, and is a conformal Killing vector field of the $4$-dimensional metric. Denote by $\mathcal{O}$ the scaling origin, where the vector field $K$ vanishes.
\end{definition}

The assumption of $k$-self-similarity imposes a simple functional form for the metric and scalar field quantities, recorded in the following lemma. Each double-null quantity is determined by its value along a single outgoing null hypersurface, up to scaling by a power of $\hat{u}$. This scaling property is moreover preserved under differentiation.

The coordinate $\hat{z} \doteq -\frac{\hat{v}}{\hat{u}}$ will appear frequently in the following, and parameterizes integral curves of $K$. With respect to this coordinate, the regular center will be given by $\{\hat{z}=-1\},$ and the past light-cone of $\mathcal{O}$ by $\{\hat{z}=0\}$. 

\begin{lemma}[\cite{singh}]
    \label{lem:meaningof_Psi_selfsimilar}
    For any double-null quantity $\Psi_k \in \{r_k, \nu_k, \lambda_k, \Omega_k, m_k, \mu_k\},$ there exist integers $s_{\Psi_k}$ and functions\footnote{The restriction on the domain of the functions $\mr{\Psi}(\hat{z})$ is discussed further in Section \ref{sec:renormalizedgauge} below.} $\mr{\Psi}(\hat{z}): [-1,\infty) \setminus \{0\} \rightarrow \mathbb{R}$ such that 
    $$\Psi_k(\hat{u},\hat{v}) = |\hat{u}|^{s_{\Psi_k}}\mr{\Psi}(\hat{z})$$
    holds on $\mathcal{Q} \setminus \{\hat{v}=0\}$. Explicitly, 
    \begin{align*}
        (s_{r_k}, s_{\nu_k}, s_{\lambda_k}, s_{\Omega_k}, s_{m_k}, s_{\mu_k}) = (1,0,0,0,1,0).
    \end{align*}
    There moreover exists a function $\mr{\phi}(\hat{z}): [-1,\infty) \setminus \{0\} \rightarrow \mathbb{R}$ with derivative denoted $\mr{\phi}'(\hat{z}) = \frac{d}{d\hat{z}}\mr{\phi}(\hat{z})$ such that 
    \begin{gather*}
        \phi_k(\hat{u},\hat{v}) = \mr{\phi}(\hat{z}) - k \ln |\hat{u}|, \\[.5em]
        \partial_u \phi_k(\hat{u},\hat{v}) = |\hat{u}|^{-1}\big(\hat{z} \mr{\phi}'(\hat{z}) + k\big), \ \ \partial_v \phi_k(\hat{u},\hat{v}) = |\hat{u}|^{-1}\mr{\phi}'(\hat{z}),
    \end{gather*}
    hold on $\mathcal{Q} \setminus \{\hat{v}=0\}$. Denoting by $\mr{\Psi}(\hat{z})$ the restrictions all double-null quantities (including $\phi_k$) to $\{\hat{u} =\! -1\}$, the system (\ref{ss:SSESF1})--(\ref{ss:SSESF7}) of coupled differential equations holds.
\end{lemma}

The seminal work \cite{chris2} considers the global behavior of $k$-self-similar solutions by fixing an outgoing Bondi gauge\footnote{This gauge is constructed such that the conformal Killing field takes the simple form $K = r\partial_r + u\partial_u$.}, and reducing the Einstein-scalar field system to a coupled system of differential equations. An important observation is that the range\footnote{To be more precise, for each value of $k$ there is a further $1$-parameter subfamily of exterior regions associated to a free choice of data in the exterior construction. We will not work with the exterior further in this paper, and assume an arbitrary member of this family has been chosen.} $k^2 \in (0, \frac{1}{3})$ corresponds -- after a suitable asymptotically flat truncation -- to globally naked singularities. We therefore restrict to this range of $k$ in the following; without loss of generality we may also assume $k > 0$. 

It will be more convenient to work with a double-null gauge. A natural such choice is given by \textbf{self-similar double-null coordinates} $(\hat{u},\hat{v})$, in which $\mathcal{O}$ lies at $(u,v) = (0,0)$, and the self-similar vector field assumes the simple form 
\begin{equation*}
    K = \hat{u}\partial_{\hat{u}} + \hat{v}\partial_{\hat{v}}.
\end{equation*}
The remaining gauge freedom is spent by identifying the axis $\Gamma$ with $\{\hat{u} = \hat{v}\}$, and normalizing the lapse at the axis, e.g. by setting $\Omega^2 \big|_\Gamma = 1$. In this gauge, the spacetimes constructed in \cite{chris2} exist on the coordinate domain 
\begin{equation}
    \label{eq:defnofsolnmnfldQFULL}
    \mathcal{Q}^{(full)} \doteq  \{ (\hat{u},\hat{v}) :  \hat{u} < 0, \ -1 \leq -\frac{\hat{v}}{\hat{u}} < \infty \}.
\end{equation}
See \cite[Appendix A]{singh} for details. For the purposes of this paper, we primarily work with a subset $\mathcal{Q} \subset \mathcal{Q}^{(full)}$ of the full quotient manifold, defined by 
\begin{equation}
    \label{eq:defnofsolnmnfldQ}
    \mathcal{Q} \doteq  \{ (\hat{u},\hat{v}) : -1 \leq \hat{u} < 0, \ -1 \leq -\frac{\hat{v}}{\hat{u}} \leq 1 \}.
\end{equation}

\subsection{Renormalized gauge}
\label{sec:renormalizedgauge}
In this section we motivate the use of an alternative double-null gauge for studying $k$-self-similar spacetimes near $\{\hat{v}=0\}$, and record useful estimates on the metric $g_k$ and scalar field $\phi_k$.

As discussed in \cite[Appendix A]{singh}, an analysis of the local behavior of solutions to the self-similar system (\ref{ss:SSESF1})--(\ref{ss:SSESF7}) reveals that solutions arising from regular data at the axis become singular as $\hat{z} \rightarrow 0$. Each of $\mr{\phi}'(\hat{z}), \ \mr{\Omega}^2(\hat{z}), \ \mr{\lambda}(\hat{z})$ behaves like $|\hat{z}|^{-k^2}$, and thus the metric and scalar field fail to be $C^0$ and $C^1$, respectively. This behavior is in fact only a coordinate singularity, and is resolved by transforming to \textbf{k-renormalized double-null coordinates} $(u,v)$, where 
\begin{equation*}
    (u,v) \doteq (\hat{u}, -|\hat{v}|^{q_k}), \quad  \big(\frac{\partial}{\partial \hat{u}}, \frac{\partial}{\partial \hat{v}}\big) \doteq \big(\frac{\partial}{\partial u}, {q_k} |v|^{-p_k k^2} \frac{\partial}{\partial v}\big).
\end{equation*}
Here we have introduced the constants 
\begin{equation*}
    q_k \doteq 1-k^2, \ \ p_k \doteq (1-k^2)^{-1}.
\end{equation*}
Moreover, define a renormalized version of the coordinate $\hat{z}$ by 
\begin{equation}
    \label{coords:ztozhat}
    z \doteq -|\hat{z}|^{q_k} = v |u|^{-q_k},
\end{equation}
with respect to which $\mathcal{Q} \subset \{z \in [-1,1]\}.$  The manifolds $\mathcal{Q}, \mathcal{Q}^{(full)}$ may be expressed in terms of the $(u,v)$ coordinates; however, these representations depend explicitly on $k$. See Figure \ref{fig5} for a Penrose diagram representation.

The following lemma records the regularity of $k$-self-similar metric and scalar field quantities in $k$-renormalized coordinates.
\begin{prop}[\cite{singh}, Appendix \ref{section:app3}]
    \label{lemma:backgroundprelims1}
    Let $\mr{\Psi}(z)$ denote the restrictions of double-null quantities to $\{\hat{u} \!=\!-1\}$, computed with respect to $k$-renormalized double-null gauge. Then the following regularity holds:
    \begin{itemize}
        \item $\mr{\Psi}(z) \in C^{\infty}_z\big(\{z < 0 \} \cup \{0 < z \leq 1\}\big)$.
        \item $\mr{m}(z), \mr{\Omega}(z), \mr{\phi}(z) \in C^1_z\big(\{z \leq 1 \} \big),$ $\mr{r}(z) \in  C^2_z\big(\{z \leq 1 \} \big).$ Moreover, 
        \begin{equation*}
            \partial_z \mr{\phi}(z) \sim k^{-1} \ \text{as} \ z \rightarrow 0,
        \end{equation*}
        and we have the identities 
        \begin{align}
            \mr{\mu}(0) &= \frac{k^2}{1+k^2}, \\[.5\jot] 
            \mr{\nu}(0) &= - \mr{r}(0). \label{ident:rvsnu}
        \end{align}
        \item For $k$ sufficiently small the above regularity as $z\rightarrow 0$ is \underline{sharp}. We have $\mr{m}(z), \mr{\Omega}(z), \mr{\phi}(z) \notin C^2_z\big(\{z \leq 1 \} \big),$ $\mr{r}(z) \notin  C^3_z\big(\{z \leq 1 \} \big),$ and 
        \begin{equation}
            \partial_z^2 \mr{\phi}(z) \sim  |z|^{-1+p_k k^2} \ \text{as} \ z \rightarrow 0. \label{eq:dvsquaredblowup}
        \end{equation}
    \end{itemize}
    The spacetime $(\mathcal{Q}, g_k, r_k, \phi_k)$ lies in the BV solution class (cf. \cite{chris1}) strictly to the past of $(u,v)=(0,0)$, and for $k$ sufficiently small the induced characteristic data has BV norm of size $k^{-1}$. 
\end{prop}
\begin{rmk}
The sharp regularity statement (\ref{eq:dvsquaredblowup}) is proved in Appendix \ref{section:app3}, cf. Proposition \ref{prop:dvsquaredblowupproof}. One expects this property to hold in the full naked singularity range $k^2 \in (0,\frac13)$, but current methods are limited to the regime of small-$k$.
\end{rmk}

Translating back to the spacetime picture, regularity for all double-null quantities follow in $\mathcal{Q}$. We emphasize that in the $k$-renormalized gauge, the metric quantity $\Omega^2$, as well as all $v$ coordinate derviatives of double-null quantities, do not have the same scaling behavior as in self-similar gauge (cf. Lemma \ref{lem:meaningof_Psi_selfsimilar}). Such quantities do not in general scale along generators of $K$ by integer powers of $u$. As an illustration, consider the behavior of the renormalized coordinate derivative $\lambda(u,v) \doteq \partial_v r(u,v)$:
\begin{align*}
    \partial_v r(u,v) \sim |\hat{v}|^{k^2} \mr{\lambda}(\hz) \sim  |\hat{v}|^{k^2} |\hz|^{-k^2} \sim |u|^{k^2}.
\end{align*}
In particular, $\lambda(u,v)$ vanishes as $|u|\rightarrow 0$. Similar computations determine the scaling of remaining double-null quantities, and higher coordinate derivatives. For the purposes of exposition, we record the scaling behavior of relevant low-order quantities here:
\begin{lemma}
\label{lem:christodouloulowerorder}
The spacetime double-null quantities satisfy the following bounds in $\mathcal{Q}$, where implied constants are allowed to depend on $k$.
\begin{gather*}
    |r_k| \lesssim |u|, \ \ (-\nu_k) \sim 1, \ \ \lambda_k \sim |u|^{k^2}, \\[2\jot]
    |m_k| \lesssim r_k^3 |u|^{-2}, \ \Omega^2_k \sim |u|^{k^2},\\[2\jot]
    |\partial_u \phi_k| \lesssim |u|^{-1}, \ \ |\partial_v \phi_k| \lesssim |u|^{-q_k}.
\end{gather*}
\end{lemma}

\begin{figure}
    \centering
    \includegraphics[scale=.35]{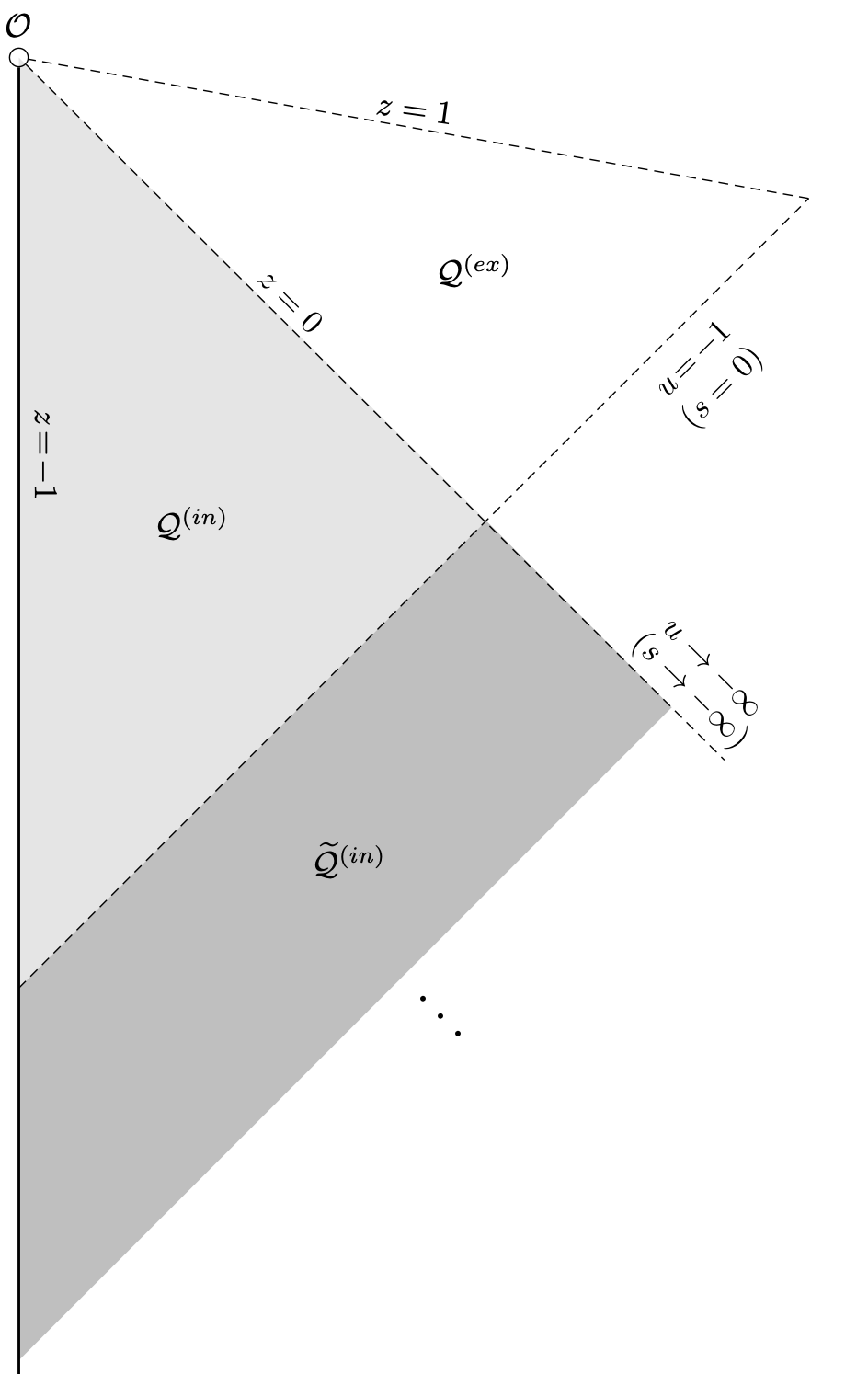}
    \caption{Penrose diagram of the relevant subdomains of the interior and exterior regions. The interior region $\mathcal{Q}^{(in)}$ (lightly shaded) is a subset of  $\{u\geq -1\}$, and the extended interior $\wt{\mathcal{Q}}^{(in)}$ (lightly shaded and darkly shaded) extends to $\{u > -\infty \}$.}
    \label{fig5}
\end{figure}

\subsection{Notation and conventions}
\label{sec:notation}
Recall the definitions (\ref{eq:defnofsolnmnfldQFULL})--(\ref{eq:defnofsolnmnfldQ}) of the quotient spacetimes $\mathcal{Q}^{(full)}, \mathcal{Q}$. It will be helpful to distinguish the \textit{interior} and \textit{exterior} regions, given respectively as 
\begin{equation*}
    \mathcal{Q}^{(in)} \doteq \mathcal{Q} \cap \{v \leq 0\}, \ \ \mathcal{Q}^{(ex)} \doteq \mathcal{Q} \cap \{v \geq 0\}.
\end{equation*}
We may also define the \textit{extended interior} by allowing the $u$ coordinate to take arbitrarily large, negative values: 
\begin{equation*}
    \wt{\mathcal{Q}}^{(in)} \doteq \mathcal{Q}^{(full)} \cap \{v \leq 0\}. 
\end{equation*}

The majority of the analysis takes place on the quotient manifold $\mathcal{Q}$ (or subsets thereof). When we wish to define explicitly a subset of the ambient $(3\!+\!1)$-dimensional spacetime $\mathcal{M}$, the angular coordinates $\omega \in \mathbb{S}^2$ will be explicitly written, or a notation will be used to indicate the dimensionality of the set. For example, $\Sigma_{u_0} \doteq \{u = u_0\}$ denotes the one-dimensional null hypersurface within $\mathcal{Q}$, and $^{(3)}\Sigma_{u_0} \!\subset \! {\mathcal{M}}$ the corresponding three-dimensional null hypersurface.

\vspace{1em}
For subsets $U \! \subset \! \mathcal{Q}$, denote by $L^p(U)$ the standard Lebesgue spaces, and $W^{k,p}(U)$ the Sobolev spaces. When the underlying coordinates are unclear, we expressly include the coordinates as subscripts, e.g. $L^p_{z}([-1,0])$ or $L^p_{\hz}([-1,0])$. When writing mixed norms $W^{k,p}_{u,v}W_\omega^{k',p'}(U \times \mathbb{S}^2)$ of functions defined on subsets of ${\mathcal{M}}$, the volume form is given by the standard volume form on $\mathbb{R}^2_{u,v} \times \mathbb{S}^2_\omega$, where $\mathbb{S}^2$ is the round, unit sphere.
 
\vspace{1em}
In addition to the parameter $k$, it is convenient to define the constants 
\begin{equation*}
    q_k \doteq 1-k^2, \ \  p_k \doteq (1-k^2)^{-1}.
\end{equation*}
One has $q_k < 1 < p_k$ as well as the identities $q_k p_k = 1$, $1 + p_k k^2 = p_k$. 

\vspace{1em}
Geometric quantities associated to the background solution are notated in various ways to differentiate between coordinate systems and scalings. For example, in this paper $r$ appears in the forms $\mr{r}, \br{r}, r_k, r$. The $k$-self-similar function $\mr{r}(z)$ depends on a single coordinate $z$ (or $\hz$), and $r_k(u,v) = \mr{r}(z)|u|$ is the extension to a self-similar function on spacetime. $r(u,v)$ is the double-null area radius function on the $(\epsilon_0,k)$-admissible spacetime, which is assumed to be close (in terms of $\epsilon_0$) to $r_k(u,v)$. Finally, $\br{r}(u,v) = |u|^{-1}r(u,v)$ is the same quantity with self-similar scaling removed. Note that $\br{r}(u,v)$ is not necessarily a function of $z$ alone. 

Similarly, one has $\lambda(u,v) = \partial_v r(u,v)$, \ $\br{\lambda}(u,v) = |u|^{-k^2}\lambda(u,v)$, \ and $\lambda_k(u,v) = \partial_v r_k(u,v) $.

\vspace{1em}
We use standard big-O notation $O(\cdot)$ and the relations $\lesssim, \gtrsim$. When the dependence of a given estimate on a parameter (often $k$, spectral parameters $\sigma$, or cutoffs $x_0$) is important, we use appropriate subscripts. We also allow for error terms of the form $O_{L^p}(\cdot)$, for which only a bound on the $L^p$ norm is tracked.

\subsection{Admissible spacetimes}
\label{sec:admissiblespacetimes}
In this section we define the class of approximately $k$-self-similar spacetimes. Assume a value of $k^2\in(0,\frac13)$ is fixed, as well as a spacetime $(\mathcal{Q}, g, r, \phi)$ defined in $k$-renormalized gauge. For any double-null quantity $\Psi(u,v)$ on the spacetime, let $\Psi_k(u,v)$ denote the coordinate expression of the same quantity in a fixed $k$-self-similar spacetime, and define $\Psi_p(u,v)$ by 
\begin{equation}
    \label{eq:psivspsi_p}
    \Psi(u,v) = \Psi_k(u,v) + \Psi_p(u,v).
\end{equation}
This definition naturally extends to rational functions of $\Psi$. 

\begin{definition}
    \label{dfn:admissiblespacetimesDEFN}
    Fix parameters ${\epsilon_0 \ll 1}$, $k^2 \in (0, \frac13)$. A spacetime $(\mathcal{Q}, g, r, \phi)$ with $r \in C^5(\mathcal{Q} \setminus \{v=0\}), m \in C^4(\mathcal{Q} \setminus \{v=0\}), \phi \in C^3(\mathcal{Q} \setminus \{v=0\})$ is an \textbf{$\mathbf{(\epsilon_0, k)}$-admissible spacetime} if the following conditions hold:
\begin{enumerate}
    \item For all $\delta > 0$ small, $(\mathcal{Q} \cap \{u \leq -\delta\}, g, r, \phi)$ is a BV solution to the spherically symmetric Einstein-scalar field system.
    \item (Normalization) Along $\{v=0\}$ the condition $r_p(u,0) = 0$ holds\footnote{By (\ref{ident:rvsnu}), this implies $\nu(u,0) = -r(u,0)$.}.
    \item (Axis regularity) 
    \begin{align}
       &\sup_{\{\frac{v}{|u|^{q_k}} \leq -\frac12\}} \Big| \partial_u^i \partial_v^j \Big( \frac{m}{r^3}\Big)_p\Big| \lesssim \epsilon_0 k^2|u|^{-i-q_k j}, &0 \leq i+j \leq 4 \label{eq:admissiblebounds1}
    \end{align}
    \item (Ingoing bounds) 
    \begin{align}
        |r_p| &\lesssim \epsilon_0 k^2 r_k |u|^2, \ \ \label{eq:admissiblebounds2}\\[2\jot]
        |\partial_u^i \partial_v^j r_p| &\lesssim \epsilon_0 k^2 |u|^{3-i-q_k j},  &1 \leq i+j \leq 5, \ j \leq 2 \label{eq:admissiblebounds3}\\[2\jot] 
        |\partial_u^i \partial_v^j m_p| &\lesssim \epsilon_0 k^2  |u|^{3-i-q_k j} ,   &0 \leq i+j \leq 4, \ j \leq 1  \label{eq:admissiblebounds4}
    \end{align}
    \item (Outgoing bounds) 
    \begin{align}
        |\partial_u^i \partial_v^j r_p| &\lesssim_{\epsilon_0,k} |u|^{q_k-i} |v|^{-1+p_k k^2+(j-3)}, \ & 3 \leq i+j \leq 5, \ j \geq 3 \label{eq:admissiblebounds5}\\[2\jot] 
        |\partial_u^i \partial_v^2 m_p| &\lesssim_{\epsilon_0,k} |u|^{2-i} |v|^{-1+p_k k^2+(j-2)}. \ & 2 \leq i+j \leq 4, \ j \geq 2 \label{eq:admissiblebounds6}
    \end{align}
\end{enumerate}
Given an $(\epsilon_0, k)$-admissible spacetime, we define an \textbf{extended $\mathbf{(\epsilon_0, k)}$-admissible spacetime} $(\wt{\mathcal{Q}}^{(in)} \cup \mathcal{Q}^{(ex)}, g, r, \phi)$ as follows. 
Extend $\Psi_k(u,v)$ to $\wt{\mathcal{Q}}^{(in)}$ via self-similarity, and fix an arbitrary extension of the $\Psi_p$ subject to the regularity requirements, the conditions (2)--(5), and the condition that the support of all $\Psi_p$ is contained in $\{u \geq -2 \}$. By (\ref{eq:psivspsi_p}), this procedure defines the double-null quantities for the extended spacetime.
\end{definition}

\subsection{Coordinate systems in $\wt{\mathcal{Q}}^{(in)}$}
\label{subsec:coordinatesystems}
The bulk of the analysis of the (extended) interior region $\wt{\mathcal{Q}}^{(in)}$ takes place in non-double null gauges. In this section we introduce two such gauges: similarity coordinates, adapted to the multiplier estimates of Section \ref{sec:multipliers}, and hyperbolic coordinates, adapted to the scattering theory constructions in Section \ref{sec:scattering}.

Define \textbf{similarity coordinates} $(s,z)$ by 
\begin{equation}
    \label{coords:uvtosim}
    (u,v) \doteq (-e^{-s}, e^{-q_k s}z), \quad (s,z) \doteq (-\log |u|, \frac{v}{|u|^{q_k}}).
\end{equation} 
Here, $s$ is a time coordinate serving to push the singularity to ${s = +\infty}$. It follows that power dependence on $u$ translates to exponential dependence on $s$, and surfaces of constant $s$ are reparameterizations of surfaces of constant $u$, and are therefore null. Surfaces of constant $z$ parameterize integral curves of the conformal Killing field, and the interior region corresponds to the range $z \in [-1,0]$.

We next define \textbf{hyperbolic coordinates} $(t,x)$. These coordinates cover $\wt{\mathcal{Q}}^{(in)} \cap \{v < 0 \}$, and are defined by 
\begin{align}
    \label{coords:simtohyper}
    &(s,z) \doteq (t-x, -e^{-2q_k x}), \quad (t,x) \doteq (s-\frac{1}{2q_k}\ln |z|, -\frac{1}{2q_k}\ln |z| ).
\end{align}
The null surface $\{v=0\}$ formally corresponds to the set $\{x=\infty\}$. Level sets $\{t = t_0\}$ trace out hyperbolas $\{|\hat{u}\hat{v}| = e^{-2t_0}\}$ in the $(\hat{u},\hat{v})$ plane. Moreover, note that the ``time'' coordinate $t$ is not equivalent to $s$ (or $-\ln|u|$), except in compact regions $\{x \leq \text{const}.\}$.

A summary of the coordinate systems introduced thus far, as well as useful formulas for relating coordinate derivatives, is given in Table \ref{table:1} below.
\begin{table}[h!]
\begin{center}
    \begin{tabular}{ |c||c|c|c| } 
     \hline
     Coordinates & $K$ & Transformation & Coordinate derivatives \\[2\jot] 
     \hline 
     $(\hat{u},\hat{v})$ & $\hat{u}\partial_{\hat{u}} + \hat{v}\partial_{\hat{v}}$ & $(u,v) = (\hat{u},-|\hat{v}|^{q_k})$ & $\big(\partial_{u},\partial_v\big) = \big(\partial_{\hat{u}}, p_k|\hat{v}|^{k^2}\partial_{\hat{v}} \big)$  \\[2\jot] 
     $(u,v)$ & $u\partial_u + q_k v\partial_v $ & Id.  & Id. \\[\jot] 
     $(s,z)$ & $-\partial_s $ & $(u,v) = (-e^{-s},e^{-q_k s}z)$ & $\big(\partial_{u},\partial_v\big) = \big(e^s \partial_s-q_k|z|e^{s}\partial_z, e^{q_k s}\partial_z \big) $ \\[2\jot]
     $(t,x)$ & $-\partial_t $ & $(u,v) = (-e^{x-t}, -e^{-q(t+x)}) $ & $\big(\partial_{u},\partial_v\big) = \big(\frac12 e^{t-x}(\partial_t -\partial_x), \frac{1}{2q_k}e^{q_k(t+x)}(\partial_t +\partial_x) \big) $ \\[\jot]
     \hline
    \end{tabular}
    \end{center}
    \caption{Relations between various coordinate systems.}
\label{table:1}
\end{table}

\subsection{The wave equation (\ref{eq:1}) and separation of variables}
\label{subsec:waveeqnandseparationofvars}
In a general double-null gauge, the linear wave equation (\ref{eq:1}) assumes the following two forms, for variables $r\varphi$ and $\varphi$ respectively:
\begin{equation}
    \label{sec2.5:eq1}
    \partial_u \partial_v (r \varphi) + \frac{ \lambda (-\nu)}{(1-\mu)r^2}\big( \slashed{\Delta}_{\mathbb{S}^2} + \mu \big)(r \varphi) =0.
\end{equation} 
\begin{equation}
    \label{sec2.5:eq1.5}
    \partial_u \partial_v \varphi + \frac{\lambda}{r}\partial_u \varphi + \frac{\nu}{r}\partial_v \varphi + \frac{\lambda (-\nu)}{(1-\mu)r^2} \slashed{\Delta}_{\mathbb{S}^2}\varphi = 0.
\end{equation} 
Here, $\slashed{\Delta}_{\mathbb{S}^2}$ is the Laplacian on the round, unit sphere. 

With the exception of the analysis in the exterior region (cf. Appendix \ref{section:app1}), it is convenient to work with the weighted quantity $r\varphi$ instead of the wave $\varphi$ itself. Recall $r(u,v)$ is assumed to be a given function associated to the background spacetime; away from the axis, it follows from Definition \ref{dfn:admissiblespacetimesDEFN} that ${r \sim |u|}$, and therefore the $r$ factor is roughly equivalent to a $|u|$ weight. Near the axis however, the structure of (\ref{sec2.5:eq1}) makes it simpler to establish regularity of quantities formed out of $r\varphi$.

We first record the form of (\ref{sec2.5:eq1}) in similarity coordinates.
\begin{lemma}
    The wave equation (\ref{sec2.5:eq1}) is equivalent to 
\begin{equation}
    \label{eq:wavesimfull}
    \partial_s \partial_z (r \varphi) - q_k |z| \partial_z^2 (r\varphi) + q_k \partial_z (r\varphi) + e^{-q_k(1+s)}\frac{\lambda(-\nu)}{(1-\mu)r^2}( \slashed{\Delta}_{\mathbb{S}^2}+\mu)(r\varphi)=0,
\end{equation}
where double-null quantities are viewed as functions of similarity coordinates via (\ref{coords:uvtosim}).
\end{lemma}
\begin{proof}
    Introduce the quantity in double-null gauge
\begin{equation*}
    H(u,v) \doteq \Big(\frac{\lambda(-\nu)}{(1-\mu)r^2}\Big)(u,v),
\end{equation*}
appearing (up to the factor $ \slashed{\Delta}_{\mathbb{S}^2}+\mu $) as the coefficient of the zeroth order term in (\ref{sec2.5:eq1}). Decompose $H(u,v) = H_k(u,v) + H_p(u,v)$ as in (\ref{eq:psivspsi_p}). By the scalings (\ref{lem:christodouloulowerorder}), we may write 
\begin{align*}
    H_k(u,v) = \mr{H}_k(z)\mr{\mu}(z)^{-1}e^{(1+q_k)s},
\end{align*}
for an appropriate function $\mr{H}_k(z)$. Similarly, transforming the wave operator to similarity coordinates using Table \ref{table:1} gives 
\begin{align*}
    \partial_u \partial_v (r \varphi) &= e^s(\partial_s - q_k|z|\partial_z)(e^{q_k s}\partial_z)(r \varphi) \\
    &=e^{(1+q_k)s}(\partial_s \partial_z - q_k|z|\partial_z^2 +q_k \partial_z)(r \varphi).
\end{align*}
It now suffices to insert these expressions in similarity coordinates into (\ref{sec2.5:eq1}) and cancel $s$ weights to arrive at (\ref{eq:wavesimfull}).
\end{proof}

The assumption of spherical symmetry allows for a separation of variables in (\ref{sec2.5:eq1})--(\ref{sec2.5:eq1.5}). Denote by $Y_{m \ell}(\omega),\  \omega \in \mathbb{S}^2$ the standard spherical harmonics, and let $P_{m \ell}: L^2_{\omega}(\mathbb{S}^2) \rightarrow \mathbb{R}$ denote the projection operator onto the coefficient of the $(m, \ell)$-th mode: 
\begin{equation}
    P_{m \ell}f = \langle f, Y_{m \ell}\rangle_{L^2_\omega(\mathbb{S}^2)}.
\end{equation}
We introduce the following spherical harmonic decomposition for functions $\varphi: {\mathcal{M}} \rightarrow \mathbb{R}$ with sufficient regularity in the angular coordinates:
\begin{align*}
    \varphi(u,v,\omega) &= \varphi_{0}(u,v) + \sum_{ \substack{ \ell > 0 \\ |m| \leq \ell }  }\varphi_{m \ell}(u,v,\omega) \\
    &\doteq \varphi_{0}(u,v) + \sum_{ \substack{ \ell > 0 \\ |m| \leq \ell }  } \hat{\varphi}_{m \ell }(u,v) Y_{m \ell}(\omega),
\end{align*}
where the coefficient functions $\hat{\varphi}_{0}(u,v), \hat{\varphi}_{m \ell }(u,v)$ are defined on the quotient spacetime $\mathcal{Q}$. Extending $P_{m \ell}$ in the natural way to spacetime functions, we have $P_{m \ell}\varphi_{m \ell} = \hat{\varphi}_{m \ell}$. We often abuse notation by using the same symbol $\varphi_{m \ell}$ for both the function on $\mathcal{M}$, and the projection onto a function on $\mathcal{Q}$.

Consider a fixed projection $\hat{\varphi}_{m \ell}(u,v)$, and define   
\begin{equation}
    \label{def:psivar}
    \psi_{m \ell} \doteq r \hat{\varphi}_{m \ell}.
\end{equation}
The next lemma records various forms of (\ref{sec2.5:eq1}) for this mode-reduced quantity.
\begin{lemma}
    \label{lem:formsofwaveequation}
In $k$-renormalized double-null coordinates, $\psi_{m \ell}$ satisfies
\begin{equation}
    \label{sec2.5:eq2}
     \partial_u \partial_v \psi_{m \ell} + \frac{ \lambda (-\nu)}{(1-\mu)r^2}\big( \ell(\ell+1) + \mu \big)\psi_{m \ell} =0.
\end{equation}
In similarity coordinates we have 
\begin{equation}
\label{eq:wavesim}
\partial_s \partial_z \psi_{m \ell} - q_k|z| \partial_z^2 \psi_{m \ell} + q_k \partial_z \psi_{m \ell} + (V_k(z) + L_{k,\ell}(z))\psi_{m \ell} = \mathcal{E}_{p,\ell}(s,z),
\end{equation}   
where
\begin{equation}
V_k(z) \doteq p_k \frac{\mr{\mu}(\hat{z})(\mr{\lambda}(\hat{z})|\hat{z}|^{k^2})(-\mr{\nu}(\hat{z})) }{(1-\mr{\mu}(\hat{z}))\mr{r}(\hat{z})^2},
\end{equation}
\begin{equation}
\label{eq:defn_Lkell}
L_{k,\ell}(z) \doteq p_k \frac{(\mr{\lambda}(\hat{z})|\hat{z}|^{k^2})(-\mr{\nu}(\hat{z})) }{(1-\mr{\mu}(\hat{z}))\mr{r}(\hat{z})^2}\ell(\ell+1),
\end{equation}
and 
\begin{align}
    \label{eq:defn_Vkp}
\mathcal{E}_{p,\ell}(s,z) \doteq -\underbrace{e^{-(1+q_k) s}\bigg(\frac{\mu\lambda(-\nu)}{(1-\mu)r^2} \bigg)_p \psi_{m \ell}}_{V_{k,p}(s,z)\psi_{m\ell}}- \underbrace{e^{-(1+q_k) s}\bigg(\frac{\lambda(-\nu)}{(1-\mu)r^2} \bigg)_p \ell(\ell+1) \psi_{m \ell}}_{L_{k,\ell,p}(s,z)\psi_{m\ell}}.
\end{align}
The terms appearing in the above expression are recast as functions of $z$ via (\ref{coords:ztozhat}). 
Finally, the following equation holds in hyperbolic coordinates:
\begin{equation}
    \label{eq:wavehyper}
    \partial_t^2 \psi_{m \ell} - \partial_x^2 \psi_{m \ell} + 4q_k e^{-2q_k x}(V_k(x) + L_{k,\ell}(x))\psi_{m \ell} =  4q_k e^{-2q_k x} \mathcal{E}_{p,\ell}(t,x),
\end{equation}
where $V_k, L_{k,\ell}, \mathcal{E}_{p,\ell}$ are the functions appearing in (\ref{eq:wavesim}), viewed in hyperbolic coordinates via (\ref{coords:simtohyper}).
\end{lemma}

\begin{proof}
    The equations (\ref{sec2.5:eq2}), (\ref{eq:wavesim}) follow from (\ref{sec2.5:eq1}), (\ref{eq:wavesimfull}) respectively after setting $\slashed{\Delta}_{\mathbb{S}^2}\psi_{m \ell} = \ell(\ell+1)\psi_{m \ell}$, which follows for functions $\psi_{m \ell}$ supported on a single $(m,\ell)$-mode. Moreover, (\ref{eq:wavehyper}) is a direct consequence of (\ref{eq:wavesim}) and the coordinate transformations given in Table \ref{table:1}.
\end{proof}

For convenience, define the combined potentials
\begin{equation}
    \label{eq:combinedpotentials}
    V(s,z) \doteq V_k(z) + V_{k,p}(s,z), \ \ L_{\ell}(s,z) \doteq L_{k,\ell}(z) + L_{k,\ell,p}(s,z).
\end{equation}

\subsection{Properties of the potentials $V_k(z), L_{k,\ell}(z)$}
\label{sec:propertiesofpotentials}
On a fixed $k$-self-similar background, the geometric properties of the background enter the wave equation through the pair of potentials $V_k, L_{k,\ell}$. Self-similarity implies these potentials are function of a single ``spatial'' variable $z$ (or $x$). Applying the results of Appendix \ref{section:app1} allows us to control these quantities quantitatively for $k$ sufficiently small. In this section we record the regularity and estimates we shall need in the following.

\begin{prop}
    \label{lem:propertiesofV}
    For $k$ sufficiently small and $p \in [1,\infty)$, the following bound holds:
    \begin{equation}
        \label{eq:potentialest1}
        \|V_k\|_{L^\infty([-1,0])} + \|\partial_z V_k\|_{L^p_z([-1,0])} +  \|\partial_z V_k\|_{L^\infty([-1,-\frac{1}{2}])}  \lesssim_p k^2.
    \end{equation}
    Moreover, there exists a constant $\gamma_k$ such that for any $\epsilon >  0$, 
    \begin{align}
        V_k(z) &= \gamma_k k^2 + k^2 E_{k}(z) |z|^{1-\epsilon}, \label{eq:potentialest2} \\[1.5\jot]
        V_k(x) &= \gamma_k k^2 + k^2 E_k(z(x))e^{-2q_k(1-\epsilon)}. \label{eq:potentialest3}
    \end{align}
    Here, $E_k(z)$ depends on $\epsilon$, and satisfies $\|E_k\|_{L^\infty} \lesssim_\epsilon 1$. The constant $\gamma_k$ satisfies 
    \begin{equation*} 
        \gamma_k = 1 + O(k^2).
    \end{equation*} 
    The higher derivatives of $H_k$ satisfy the following bounds, which degenerate as $z \rightarrow 0$:
    \begin{equation}
        \label{eq:potentialest3.5}
        \big\||z|^{1-p_k k^2 + (j-2)}\frac{d^j}{dz^j}V_k(z) \big\|_{L^\infty([-1,0])} \lesssim_k 1, \ \ 2 \leq j \leq 5.
    \end{equation}
    Finally, for fixed $k$ and ${\omega > 0}$, there exists a ${z_k < 0}$ and ${c_k > 0}$ such that for any $(\epsilon_0,k)$-admissible background, the following repulsivity statement holds:
    \begin{equation}
        \label{eq:potentialest_repulsivity}
       \sup_{\{ s \geq 0, \ z_k \leq z \leq 0 \}} \partial_z(|z|^{\omega} V(s,z)) < - c_k.
    \end{equation}
\end{prop}
\begin{proof}
    The estimate (\ref{eq:potentialest1}) is a consequence of the collection of bounds (\ref{backgroundbd:1}), (\ref{eq:smallkbound8}), (\ref{eq:smallkbound10}), (\ref{eq:smallkbound10.5}), (\ref{eq:smallkbound10.6}), (\ref{eq:smallkbound14}), (\ref{eq:smallkbound15}). Similarly, (\ref{eq:potentialest3.5}) follows from (\ref{eq:smallkbound10.7}). The repulsivity estimate (\ref{eq:potentialest_repulsivity}) is direct for $V_k(z)$, and follows for $V(s,z) = V_k(z) + V_{k,p}(s,z)$ by the assumptions on admissible spacetimes.
    
    We next turn to the expansions (\ref{eq:potentialest2})--(\ref{eq:potentialest3}). The latter is a direct translation of the former in hyperbolic coordinates, so we only discuss (\ref{eq:potentialest2}). Write $V_k(z) = V_k(0) + \mathcal{E}_k(z),$ and estimate $\mathcal{E}_k$ using Hölder and (\ref{eq:potentialest1}): 
    \begin{align*}
        |\mathcal{E}_k(z)| &\leq \int_{z}^{0} |\partial_z V_k(z')|dz' \\[\jot]
        & \lesssim |z|^{1-\epsilon}\|\partial_z V_k\|_{L^{\frac{1}{\epsilon}}_z([-1,0])} \lesssim_\epsilon k^2 |z|^{1-\epsilon}.
    \end{align*}
    Defining $E_k(z) \doteq k^{-2}|z|^{-(1-\epsilon)}\mathcal{E}_k(z)$ gives the result.
    \end{proof}

    \begin{figure}[t]
        \centering
        \includegraphics[scale=.7]{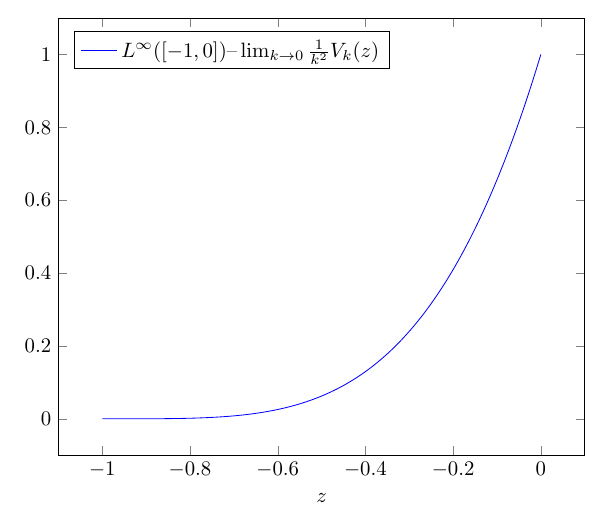}
        \caption{A plot of the leading order behavior of $V_k(z)$ as $k \rightarrow 0$, using the asymptotics of Appendix \ref{app:eq3}.    }
        \label{fig6}
        \end{figure}

The next proposition establishes similar results for the angular potential $L_{k,\ell}$. Note that $L_{k,\ell}$ lacks the smallness present in $V_k$ (compare the factors of $\mu$). It is therefore important that this angular potential (and its derivatives) carry definite signs.
\begin{prop}
    \label{lem:propertiesofL_kl}
    For all $\ell \geq 1$, the inequality $L_{k,\ell} \geq 0$ holds. Moreover, 
    \begin{align}
        \frac{d}{dz}L_{k, \ell}(z) = \underbrace{-2\ell(\ell+1) \frac{(\partial_z \mr{r})^2(-\mr{\nu})}{(1-\mr{\mu})\mr{r}^3}}_{L^{(0)}_{k,\ell}(z)} + L^{(1)}_{k,\ell}(z), \label{eq:angmompotentialest1}
    \end{align}
    where 
    \begin{align}
    \|\mr{r}^2 L^{(1)}_{k,\ell}\|_{L^\infty([-1,-\frac{1}{2}])} \lesssim k^2 \ell(\ell+1), \label{eq:angmompot_temp1}\\
    \|\mr{r}^2 L^{(1)}_{k,\ell}\|_{L^p_z([-1,0])} \lesssim_p k^2 \ell(\ell+1),\label{eq:angmompot_temp2}
    \end{align}
    and 
    \begin{equation}
        \frac{d}{dz}L_{k, \ell}(z) \leq L^{(0)}_{k,\ell}(z) \lesssim -\frac{\ell(\ell+1)}{\mr{r}^3}. \label{eq:L_kl_derivative_sign}
    \end{equation}
    \end{prop}
    \begin{proof}
        The non-negativity of $L_{k,\ell}$ is immediate from the definition (\ref{eq:defn_Lkell}). Similarly, the bounds (\ref{eq:angmompot_temp1})--(\ref{eq:angmompot_temp2}) follow from estimates on the derivative of background double-null quantities, precisely as in the proof of Proposition \ref{lem:propertiesofV} above.
        
        We finally establish a pointwise sign for $\frac{d}{dz}L_{k,\ell}$. By (\ref{eq:angmompotentialest1}) it suffices to consider $L^{(1)}_{k,\ell}$, which is given explicitly by 
        \begin{equation}
            L^{(1)}_{k,\ell} = \partial_z\bigg(\frac{\partial_z \mr{r}(-\mr{\nu})}{(1-\mr{\mu})} \bigg)\frac{\ell(\ell+1)}{\mr{r}^2} = \partial_z\big(p_k |\hz|^{k^2}\mr{\Omega}^2 \big)(z) \frac{\ell(\ell+1)}{4\mr{r}^2}.
        \end{equation}
        It follows from (\ref{eq:monotonic_omega}) that $|\hz|^{k^2}\mr{\Omega}^2$ is \textit{decreasing}, implying $L^{(1)}_{k,\ell} \leq 0$ as desired.

        \end{proof}

\begin{rmk}
    The estimate (\ref{eq:L_kl_derivative_sign}) above is an instance where the $k$-self-similar solution may not be treated as a perturbation of corresponding flat space potential $L_{\ell}^{(flat)}\doteq \frac{\ell(\ell+1)}{(1+z)^2}$. Derivatives of $L_{k,\ell}-L_{\ell}^{(flat)}$ are of size $O(1)$ as $z \rightarrow 0$, for all $k$ small. If this derivative did not have an appropriate sign, there would be the possibility of slowly decaying solutions to (\ref{eq:1}) supported on $\ell \gg 1$ localized near the cone $\{z=0\}$.  
\end{rmk}

\subsection{Function spaces and local well-posedness}
\label{section:functspaces}
In this section we discuss local well-posedness for (\ref{sec2.5:eq1}) with characteristic initial data along $\{u\!=\!-1\}$ (equivalently, along $\{s=0\}$). The main subtlety concerns the limited regularity of the background geometry as $v \rightarrow 0$. It follows from Proposition \ref{lemma:backgroundprelims1} that the coefficients $V_k, L_{k,\ell} \in C^1_v \setminus C^2_v$, and therefore we cannot hope to propagate arbitrarily high regularity on data. It will prove convenient to allow for limited outgoing regularity already in initial data. 

These considerations lead to the following definitions. Let $\alpha \in (1, 2), \delta \in (0,1), \gamma \in (0,\frac{1}{2}]$ be constants, and ${I \subset [-1,1]}$ a closed subinterval containing $0$. We introduce the following spaces:
\begin{align}
     &\mathcal{C}^{\alpha,\delta}_{(hor)}(I) \doteq  \big\{f(z): I \rightarrow \mathbb{R} \mid  f(z) \in C^{5}_z(I \setminus \{0\} ) \cap C^{1}_z(I), \nonumber \\[\jot]
     & \hspace{15em} |z|^{j - \alpha } \frac{d^j}{dz^j}f(z) \in C^{0,\delta}_z(I), \ \text{for} \ 2 \leq j \leq 5 \big\}, \label{defn:Chorizonspaces} \\[3\jot]
     &\mathcal{H}^{1,\gamma}_{(hor)}(I) \doteq \{f(z): I \rightarrow \mathbb{R} \mid \ f(z) \in W^{5,2}_z(I \setminus \{0\}) \cap  W^{1,2}_z(I), \ |z|^{\frac12 - \gamma} \partial_z^{2} f(z) \in L^2_z(I) \}, \label{defn:Chorizonspaces2.0}\\[3\jot]
     &\mathcal{H}^{1,\gamma}_{(hor)}(I \times \mathbb{S}^2) \doteq \{f(z,\omega): I\times\mathbb{S}^2 \rightarrow \mathbb{R} \mid \ f(z,\omega) \in W^{5,2}_z W^{4,2}_\omega( (I\setminus \{0\}) \times \mathbb{S}^2) \cap W^{1,2}_z W^{4,2}_\omega(I \times \mathbb{S}^2), \nonumber \\[\jot]
     &\hspace{15em} |z|^{\frac12 - \gamma} \partial_z^{2} f(z,\omega) \in L^2_z W^{4,2}_\omega(I \times \mathbb{S}^2)  \} . \label{defn:Chorizonspaces2}
\end{align}
Define the norms
\begin{align*}
    \|f \|_{\mathcal{C}^{\alpha,\delta}_{(hor)}(I)} &\doteq \sum_{j=0}^{1} \big\|\frac{d^j}{dz^j}f \big\|_{L^\infty([-1,0])} + \sum_{j=2}^{5} \big\||z|^{j-\alpha}\frac{d^j}{dz^j}f \big\|_{C^{0,\delta}_z([-1,0])}, \\[2\jot]
    \|f \|_{\mathcal{H}^{1,\gamma}_{(hor)}(I)} &\doteq \|f \|_{W^{1,2}_z(I)} + \big\||z|^{\frac12 - \gamma} \partial_z^2 f \big\|_{L^2_z(I)}, \\[2\jot]
    \|f \|_{\mathcal{H}^{1,\gamma}_{(hor)}(I\times \mathbb{S}^2)} &\doteq \|f \|_{W^{1,2}_zW^{4,2}_\omega(I\times \mathbb{S}^2)} + \big\||z|^{\frac12 - \gamma} \partial_z^2 f \big\|_{L^2_zW^{4,2}_\omega(I\times\mathbb{S}^2)}.
\end{align*}
where $\|f\|_{C^{0,\delta}_z([-1,0])}$ is the standard Hölder norm. We shall often work with $C^{\alpha,\delta}_{(hor)}(I)$ spaces with $\delta = 0$, which we abbreviate as $C^{\alpha}_{(hor)}(I) \doteq  C^{\alpha,0}_{(hor)}(I)$.

\begin{rmk}
    \label{rmk:functionspaces}
    The regularity of a general spherically-symmetric function $f(z) \in \mathcal{C}^{\alpha,0}_{(hor)}(I)$, $\alpha\in (1,2)$ can be thought of as a generalization of that of the explicit function $|z|^{\alpha}$. Low order derivatives $f, f'$ are pointwise bounded, the second order derivative $f''$ is permitted to be singular, but integrable, as $z \rightarrow 0$, and higher derivatives $f^{(3)}, f^{(4)}, f^{(5)}$ each lose at most an additional power of $|z|^{-1}$. Note that $|z|^{\alpha} \in C^{\alpha,\delta}_{(hor)}([-1,1])$ for all $\delta \in (0,1)$. 

    A similar intepretation applies to $\mathcal{H}^{1,\gamma}_{(hor)}(I)$. A direct computation shows $|z|^{\alpha} \in \mathcal{H}^{1,\frac{1}{2}}_{(hor)}(I)$ when $\alpha > \frac32$, and $|z|^{\alpha} \in \mathcal{H}^{1,\gamma}_{(hor)}(I)$ for any $\gamma \in (0, \alpha-1)$ when $\alpha \leq \frac32$.

    In this scale of spaces, the $k$-self-similar scalar $\mr{\phi}(z)$ satisfies 
    \begin{equation*}
        \mr{\phi}(z) \in \mathcal{C}^{p_k, p_k k^2}_{(hor)}([-1,1]) \cap \mathcal{H}^{1,\gamma}_{(hor)}([-1,1]), \ \gamma \in (0, p_k k^2).
    \end{equation*}
    This is a consequence of (\ref{eq:newexpansions2}).
\end{rmk}

\begin{rmk}
    The pointwise $\mathcal{C}^{\alpha,\delta}_{(hor)}(I)$ and integrated $\mathcal{H}^{1,\gamma}_{(hor)}(I)$ spaces will be used in the study of the $\ell = 0$ and $\ell > 0$ components of $\varphi$, respectively. 
\end{rmk}

\begin{rmk}
    Along the null surface $\{u=-1\}$, the $\mathcal{C}^{\alpha}_{(hor)}(I)$ regularity of a function $f(v)$ in double-null coordinates is equivalent to that of $f(z)$ in similarity coordinates. The same is true for $\mathcal{H}^{1,\gamma}_{(hor)}(I)$.
\end{rmk}

The following lemma establishes a useful decomposition for functions in $\mathcal{C}^{\alpha,\delta}_{(hor)}(I)$.

\begin{lemma}
    \label{lem:holderspacedecomp}
    Fix $\alpha \in (1,2)$, $\delta \in (0,1)$, and $f(z) \in \mathcal{C}^{\alpha,\delta}_{(hor)}(I)$. Then there exists a decomposition 
    \begin{equation}
        \label{eq:holderspacedecomp}
        f(z) = c|z|^{\alpha} + f_1(z),
    \end{equation}
    where $f_1(z) \in \mathcal{C}^{\alpha+\delta'}_{(hor)}(I)$ for any $\delta' < \delta$ and we have 
    \begin{equation}
        \label{eq:holderspacedecomp_est}
        |c| + \|f_1\|_{\mathcal{C}^{\alpha+\delta'}_{(hor)}(I)} \lesssim_{\alpha} \|f\|_{\mathcal{C}^{\alpha,\delta}_{(hor)}(I)}.
    \end{equation} 
    For $k$ sufficiently small and $f(z) \in C^{p_k k^2, \delta}_{(hor)}(I)$, there exists a decomposition 
    \begin{equation}
        \label{eq:holderspacedecomp_threshold}
        f(z) = c\mr{\phi}(z) + f_1(z),
    \end{equation}
    where $f_1(z) \in \mathcal{C}^{p_k k^2 + \delta'}_{(hor)}(I)$ for any $\delta' < \min(\delta,p_k k^2)$. The estimate (\ref{eq:holderspacedecomp_est}) continues to hold. 
\end{lemma}
\begin{proof}
Define the constants $c_j \doteq \lim_{z \rightarrow 0} |z|^{j-\alpha}\frac{d^j}{dz^j}f(z)$ for $2 \leq j \leq 5$, which exist and are finite by assumption on $f$. We first derive a relationship between the $c_j$.

As $f(z) \in C^5(I \setminus \{0\})$, it follows that $F_j(z) \doteq |z|^{j-\alpha}\frac{d^j}{dz^j}f(z)$ satisfies
\begin{align*}
    \frac{d}{dz}F_j(z) &= \frac{\alpha-j}{|z|} F_j(z) + \frac{1}{|z|}F_{j+1}(z) \\ 
    &= \frac{c_j(\alpha-j) + c_{j+1}}{|z|}  + O( \frac{1}{|z|^{1-\delta}}).
\end{align*}
By assumption $F_j(z)$ is continuous, and assumes a finite limit at $z=0$. Therefore the coefficient of the $|z|^{-1}$ term must vanish, giving
\begin{equation}
    c_j(j-\alpha) = c_{j+1}.
\end{equation} 
With this relationship in hand, we show that 
\begin{equation*}
    f(z) = \frac{c_2}{\alpha(\alpha-1)}|z|^{\alpha} + \underbrace{(f(z) - \frac{c_2}{\alpha(\alpha-1)}|z|^{\alpha})}_{f_1(z)}
\end{equation*}
is the required decomposition. It is immediate that $f_1(z) \in C^1(I)$, and thus it suffices to check 
\begin{equation*}
    |z|^{j-\alpha-\delta'}\frac{d^j}{dz^j} f_1(z) = |z|^{-\delta'}\big(F_j(z) - c_{j} \big) \in C^0(I), \quad  2 \leq j \leq 5.
\end{equation*}
This condition follows from the regularity of $F_j(z)$ in $I \setminus \{0\}$, and the Hölder continuity with index $\delta$.

We next consider the case $\alpha = p_k k^2$. Given the decomposition (\ref{eq:holderspacedecomp}), it suffices to show that there exists a constant $c$ such that  $|z|^{p_k k^2} = c \mr{\phi}(z) + g(z)$, where $g(z) \in \mathcal{C}^{p_k k^2 + \delta'}_{(hor)}(I)$ for all $\delta' < p_k k^2$. This latter statement is a consequence of (\ref{eq:newexpansions2}).
\end{proof}

We conclude this section with a well-posedness statement for the wave equation (\ref{eq:wavesimfull}), given data in similarity coordinates with finite regularity. 

\begin{prop}
\label{prop:localexist}
Fix an $(\epsilon_0, k)$-admissible spacetime $(\mathcal{Q},g,r,\phi)$. Let initial data $\varphi|_{\{s=0\}} = f(z,\omega)$ to (\ref{eq:wavesimfull}) be given. Decomposing into $\ell=0$ and $\ell>0$ components as $f(z,\omega) = f_0(z) + f_{> 0}(z,\omega)$, assume 
\begin{align*}
    f_0(z) \in \mathcal{C}^{\alpha}_{(hor)}([-1,1]), \quad f_{> 0}(z,\omega) \in \mathcal{H}_{(hor)}^{1,\gamma}([-1,1] \times \mathbb{S}^2),
\end{align*} 
for parameters $\alpha \in (1, 2)$, $\gamma \in (0,\frac12]$. Assume the projections to fixed $(m,\ell)$-mode, $\ell > 0$ satisfy 
\begin{align}
    \label{eq:localext_higherlmodereg}
    (P_{m 1}f)(z) = O(r(z)), \quad (P_{m \ell}f)(z) = O(r(z)^2), \ \ell \geq 2, \ |m| \leq \ell.
\end{align}
Then there exists a unique solution $\varphi(s,z,\omega)$ to (\ref{eq:wavesimfull}) on $\{s \geq 0\}\times\mathbb{S}^2 $ with spherical harmonic decomposition $\varphi(s,z,\omega) = \varphi_{0}(s,z) + \varphi_{>0}(s,z,\omega)$ satisfying 
\begin{align*}
    r\varphi_0(s,z) \in C^0_s(\mathbb{R}_{+}; \mathcal{C}^{\alpha}_{(hor)}([-1,1])), \quad r\varphi_{> 0}(s,z,\omega) \in C^0_s(\mathbb{R}_+; \mathcal{H}^{1,\gamma}_{(hor)}([-1,1] \times \mathbb{S}^2)).
\end{align*} 
The solution has the regularity 
\begin{align}
    &r \varphi_0 \in C^5_{s,z}(\mathcal{Q}\setminus\{z=0\}), \text{and} \ \partial_s^{i}\partial_z^{j}(r \varphi_0) \in C^0_{s,z}(\mathcal{Q}), \ i+j \leq 5, \ j \leq 1, \label{eq:localexist_regular_1} \\[\jot]
    &P_{m \ell} (r \varphi) \in C^4_{s,z}(\mathcal{Q}\setminus\{z=0\}), \ \ \ell \geq 1,\label{eq:localexist_regular_2}
\end{align}
and for any $s_0 > 0$ satisfies the bound
\begin{align}
    \sup_{s' \in [0,s_0]}\|r\varphi_0 \|_{C^\alpha_{(hor)}(\{s=s'\})}+\sup_{s' \in [0,s_0]}\|r\varphi_{> 0} &\|_{\mathcal{H}^{1,\gamma}_{(hor)}(\{s=s'\}\times\mathbb{S}^2)} \nonumber \\[\jot]
    &\leq C_{s_0} \big(\|f_0 \|_{C^\alpha_{(hor)}(\{s=0\})} + \|f_{>0} \|_{\mathcal{H}^{1,\gamma}_{(hor)}(\{s=0\} \times \mathbb{S}^2)} \big). \label{eq:localexist_bound}
\end{align}
Along any fixed $\{s = s_0\}$, (\ref{eq:localext_higherlmodereg}) moreover holds.
\end{prop}

\begin{proof}[Proof sketch]

    Existence of solutions with the prescribed regularity will follow from appropriate a priori estimates. For simplicity, we restrict attention to the interior region.

    Beginning with the spherically symmetric component $\varphi_0$, we observe that the main obstacle to closing estimates is the limited regularity of the coefficients of (\ref{eq:wavesim}), cf. Proposition \ref{lem:propertiesofV}. As the data and background solution are at least $C^5_{s,z}$ in $\{z < 0\}$, standard well-posedness for the wave equation implies the existence of a solution in $\{z < 0\}$ satisfying $r \varphi_0 \in C^5_{s,z}(\mathcal{Q}\setminus\{z=0\})$. Moreover, it is straightforward to see that $\partial_z (r\varphi_0) \in C^0(\mathcal{Q})$, and thus by commuting with $\partial_s$ repeatedly that the statement (\ref{eq:localexist_regular_1}) holds.

    To show that the norm (\ref{eq:localexist_bound}) is controlled on finite $s$ intervals, commute with the vector fields $\partial_s^j$, $|z|^{j+1-\alpha}\partial_z^j$, $j \leq 4$. The coefficients remain bounded due to the presence of $|z|$ weights, and it thus suffices to integrate (\ref{eq:wavesim}) along the integral curves of $\partial_s - q_k |z|\partial_z$, and apply Grönwall. It follows that $\mathcal{C}^{\alpha}_{(hor)}([-1,0])$ bounds on data are propagated to the future.

    We next consider the non-spherically symmetric component, and estimate individual projections $\varphi_{m \ell}$. Local existence in $\{z < 0\}$ follows by standard arguments, and by propagation of regularity and Sobolev inequalities, (\ref{eq:localexist_regular_2}) follows. It remains to control the $\mathcal{H}^{1,\gamma}_{(hor)}([-1,0])$ norms. We use multiplier vector fields $X \in \{\partial_s, \chi_{z_0}(z) \partial_z\}$, where $\chi(z)$ is an increasing cutoff with support on $\{-\frac12 < z < 0\}$.

    Multiplying (\ref{eq:wavesim}) by $X(r\varphi_{m \ell})$, integrating by parts in $\mathcal{R}(s_0) \doteq \{0 \leq s \leq s_0, \ z \leq 0 \}$, adding the resulting estimates, and applying Grönwall, gives 
    \begin{align*}
        \sup_{s'\in[0,s_0]} \|r\varphi_{m \ell} \|_{W^{1,2}_z(\{s=s'\})} \lesssim_{s_0} \|r\varphi_{m \ell} \|_{W^{1,2}_z(\{s=0\})}.
     \end{align*} 
     An analogous estimate at second order follows by commuting with the set $\{\partial_s, |z|^{1-2\gamma}\chi_{z_0}(z) \partial_z\}$ and applying the same multipliers. It now remains to collect these estimates for individual angular modes to conclude (\ref{eq:localexist_bound}). 

     We finally sketch an argument that (\ref{eq:localext_higherlmodereg}) holds for fixed angular projections $\varphi_{m \ell}$. Observe that by (\ref{eq:localexist_regular_2}) and the averaging estimate (\ref{eq:averagingest}), we have $\varphi_{m \ell} \in C^3_{s,z}(\mathcal{Q} \setminus \{z=0\})$. With this pointwise control, a direct inspection of (\ref{eq:wavesim}) shows that we must have $\mr{r}^{-1} \varphi_{m \ell} \in L^\infty(\{s=s_0\})$, and thus $\varphi_{m \ell} = O(\mr{r})$.

     For $\ell \geq 2$ this vanishing can be improved. Denote by $\Phi_1$ the quantity $\mr{r}^{-1}(r\varphi_{m\ell})$, which satisfies 
     \begin{align}
        \partial_s \partial_z \Phi_1 - q_k |z|\partial_z^2  \Phi_1 &= q_k(2|z|\partial_z\mr{r} \frac{1}{\mr{r}} - 1)\partial_z \Phi_1  - \frac{1}{\mr{r}}\partial_z \mr{r} \partial_s \Phi_1 \nonumber \\[\jot]
        &\hspace{1em}+ \big(-V(s,z)-L(s,z) +q_k|z|\frac{1}{\mr{r}}\partial_z^2 \mr{r}  - q_k \frac{1}{\mr{r}}\partial_z \mr{r} \big)\Phi_1. \label{eq:LWPeq1}
    \end{align}
    Restricting to a given $\{s=s_0\}$, we have established that $\Phi_1 \in C^3_z$, and that $\Phi_1 = O(\mr{r})$. It therefore follows from (\ref{eq:LWPeq1}) that 
    \begin{equation}
        \label{eq:LWPeq2}
       \big( 2q_k |z| \partial_z \mr{r} \partial_z \Phi_1 - \mr{r}L\Phi_1 \big)\big|_{\{s=s_0\}}= O(\mr{r}).
    \end{equation}
    Calculate 
    \begin{align*}
        2q_k |z| \partial_z \mr{r}(z) \partial_z \Phi_1(s_0,z)  &= \partial_z \Phi_1(s_0,-1) + O(\mr{r}). \\
        \mr{r}(z)L(s_0,z)\Phi_1(s_0,z) &=  \frac12 \ell(\ell+1)\partial_z \Phi_1 (s_0,-1) + O(\mr{r})
    \end{align*}
    For (\ref{eq:LWPeq2}) to hold for $\ell \geq 2$, it must be the case that $\partial_z \Phi_1(s_0,-1) = 0$, implying $\Phi_1 = O(\mr{r}^2)$.
     
    \end{proof}

\section{Main Results}
\label{section:mainresults}

In this section we state precise versions of our main results. Compare with the informal statements in Theorems \ref{thm:introrough1}--\ref{thm:introrough2}.

\begin{theorem}[Spherically symmetric solutions]
    \label{theorem:mainthm}
    Fix parameters $\epsilon_0, k$ sufficiently small, and $\alpha \in (1,2)$. Let $(\mathcal{Q}, g, r, \phi)$ be an $(\epsilon_0, k)$-admissible background, and $\varphi_0(v) \in \mathcal{C}^{\alpha}_{(hor)}([-1,0])$ spherically symmetric, characteristic initial data for (\ref{eq:wavesim}). 

    \vspace{.5em}
    \noindent
    \textbf{Above threshold regularity:}
    Assume ${\alpha \in (p_k,2)}$. There exist constants $\varphi_{\infty}, C$ depending on $\|\varphi_0 \|_{\mathcal{C}^{\alpha}_{(hor)}([-1,0])} $, and a constant $\alpha' \in \big(p_k, \min(\frac32, \alpha)\big)$ such that the unique solution $\varphi(u,v)$ to (\ref{eq:wavesim}) satisfies the following pointwise bound in $\{u \geq -1\}$:
    \begin{equation}
        \label{eq:mainthmbound1}
           \max_{i+j \leq 1} \| \partial_u^i \partial_v^j (\varphi-\varphi_\infty)\|_{L^\infty(\Sigma_u)} \leq C |u|^{\alpha'q_k -1 - i - q_k j }.
        \end{equation}

    \vspace{.5em}
    \noindent
    \textbf{Threshold regularity:} Fix $\alpha = p_k, \ \delta \in (0,1)$, and assume $\varphi_0(v) \in \mathcal{C}^{p_k,\delta}_{(hor)}([-1,0])$. There exist constants $\varphi_{\infty}^{(i)}, C$ depending on $\|\varphi_0 \|_{\mathcal{C}^{p_k,\delta}_{(hor)}([-1,0])} $, and a constant $\delta' \in \big(0, \min(\frac{1}{2},\delta) \big)$ such that the unique solution $\varphi(u,v)$ to (\ref{eq:wavesim}) satisfies the following pointwise bound in $\{u \geq -1\}$:
    \begin{equation}
        \label{eq:mainthmbound2}
        \max_{i+j \leq 1} \| \partial_u^i \partial_v^j (\varphi - \varphi_\infty^{1} \phi_k(u,v) - \varphi_\infty^{2} )\|_{L^\infty(\Sigma_u)} \leq C |u|^{\delta'q_k - i - q_k j }.
    \end{equation}

    \vspace{.5em}
    \noindent
    \textbf{Below threshold regularity:} Assume $\alpha \in (1,p_k)$. For any $\epsilon > 0$, there exists a constant $C_\epsilon$ depending on $\|\varphi_0 \|_{\mathcal{C}^{\alpha}_{(hor)}([-1,0])} $ and $\epsilon$ such that the unique solution $\varphi(u,v)$ to (\ref{eq:wavesim}) satisfies the following pointwise bound in $\{u \geq -1\}$:
    \begin{equation}
        \label{eq:mainthmbound2.5}
        \max_{i+j \leq 1} \| \partial_u^i \partial_v^j \varphi\|_{L^\infty(\Sigma_u)} \leq C_\epsilon |u|^{-\epsilon -\alpha q_k - i - q_k j }.
    \end{equation}
    There moreover exists $\delta > 0$ and a choice of initial data $\wt{\varphi}_0(v)\in \mathcal{C}^{\alpha,\delta}_{(hor)}([-1,0])$, as well as a constant $C$, such the unique solution $\wt{\varphi}(u,v)$ to (\ref{eq:wavesim}) satisfies the lower bound, for any $0 \leq i+j \leq 1$:
    \begin{equation}
        \label{eq:mainthmbound2.7}
        \|\partial_u^i \partial_v^j \wt{\varphi}\|_{L^\infty(\Sigma_u)} \geq C |u|^{-\alpha q_k - i - q_k j }.
    \end{equation} 
    The statements in this case hold also for $\alpha = p_k$.

\end{theorem}

\begin{theorem}[Non-spherically symmetric solutions]
    \label{theorem:angularthm}
    Fix parameters $\epsilon_0, k$ sufficiently small, and $B$ sufficiently large independently of $k$. Let $(\mathcal{Q}, g, r, \phi)$ be an $(\epsilon_0, k)$-admissible background, and $\varphi_0(v,\omega) \in \mathcal{H}^{1, B k^2}_{(hor)}([-1,0] \times \mathbb{S}^2)$ characteristic initial data for (\ref{eq:wavesimfull}) which is supported on angular modes ${\ell\geq 1}$. Then there exist constants $C$, depending on $\|\varphi_0 \|_{\mathcal{H}^{1, B k^2}_{(hor)}([-1,0] \times \mathbb{S}^2)}$, and ${\epsilon >0}$, such that 
    \begin{equation}
        \label{eq:mainthmbound3}
        \max_{i+j+l \leq 1} \| \partial_u^i \partial_v^j \slashed{\nabla}^l  (\br{r}\varphi) \|_{L^\infty(^{(3)}\Sigma_u)} \leq C |u|^{\epsilon - i - q_k j -l }.
    \end{equation} 
\end{theorem}

\section{Multiplier estimates}
\label{sec:multipliers}

\subsection{Integral inequalities}

We start with a basic one-dimensional integral estimate:
\begin{lemma}
    \label{lemma:basic1dintegral}
    Let $f(t) \in C^1([0,1])$ be given, and $\omega,  \omega'  \in [0,\frac{1}{2})$ with $\omega \leq \omega'$. Define $a_{\omega,\omega'} \doteq \frac{a}{1-2(\omega'-\omega)}$. Then for all $a \in (0,a_{\omega,\omega'})$, the following estimates hold:
    \begin{equation}
        \int_{0}^{1}(1-t)^{2\omega} f(t)^2 dt \leq \sup_{t\in [0,1]}|(1-t)^{\omega} f(t)|^2 \leq \frac{1}{1-a_{\omega,\omega'}}f(0)^2 + \frac{1}{a(1-a_{\omega,\omega'})}\int_0^1 (1-t)^{2\omega'}  f'(t)^2 dt. \label{1d_est:Linf}
    \end{equation}
\end{lemma}
\begin{proof}
        The first inequality is immediate. For the second, assume $\omega > 0$ and write 
        \begin{align}
            (1-t)^{2\omega} f(t)^2 &= f(0)^2 + \int_0^t \frac{d}{dt}\big( (1-t)^{2\omega}f^2 \big) (t)dt \nonumber \\
            & =f(0)^2  +2\int_0^t (1-t)^{2\omega}f(t)f'(t) dt - 2\omega \int_0^t (1-t)^{2\omega-1}f(t)^2 dt  \nonumber \\
            &\leq f(0)^2 + a \int_0^1 (1-t)^{2(\omega-\omega')} (1-t)^{2\omega} f(t)^2 dt + \frac{1}{a} \int_0^1 (1-t)^{2\omega'} f'(t)^2 dt \nonumber \\
            &\leq f(0)^2 + a_{\omega,\omega'} \sup_{t \in [0,1]}|(1-t)^{\omega} f(t)|^2 + \frac{1}{a} \int_0^1 (1-t)^{2\omega'} f'(t)^2 dt
        \end{align}
        which holds for any $a >0$. Taking the supremum over $t$ and absorbing terms to the left hand side yields the desired result. 
        
        For $\omega = 0$ the computation is similar.
\end{proof}
We next turn to a Hardy-type inequality that will used for absorbing low order terms with singular $\mr{r}$ weights. 
\begin{lemma}
    Let $f(z) \in C^1([0,1])$ be given, and $\nu \geq 2$ a parameter. Assume 
    \begin{equation}
        \label{hardy1condition}
        \lim_{z\rightarrow -1}\mr{r}(z)^{-\frac{1}{2}(\nu-1)}f(z)=0.
    \end{equation}
    For $k$ sufficiently small, there exists a constant $C_\nu > 0$ such that 
    \begin{equation}
        \label{hardy1}
        \int_{-1}^{-\frac12} \frac{f(z)^2}{\mr{r}^\nu} dz \leq C_\nu \int_{-1}^{-\frac12} \frac{(f'(z))^2}{\mr{r}^{\nu-2}} dz. 
    \end{equation}
    The same identity holds with $\br{r}$ in place of $\mr{r}$.
\end{lemma}
\begin{proof}
Begin the proof of (\ref{hardy1}) by writing (we allow $C_\nu$ to change from line to line)
\begin{align*}
    \int_{-1}^{-\frac12} \frac{f(z)^2}{\mr{r}^\nu} dz &= -\frac{1}{\nu-1}\int_{-1}^{-\frac12} \partial_z \bigg(\frac{1}{\mr{r}^{\nu-1}} \bigg) \frac{1}{\partial_z \mr{r}} f(z)^2 dz \\[\jot]
    &= -\frac{1}{\nu-1}\frac{1}{\big(\mr{r}^{\nu-1} \partial_z\mr{r}\big)(-\frac12)}f\big({-\frac12}\big)^2 + \frac{1}{\nu-1}\int_{-1}^{-\frac12}\frac{1}{\mr{r}^{\nu-1}} \partial_z \bigg(\frac{1}{\partial_z \mr{r}} f(z)^2  \bigg)dz \\[\jot]
    &\leq C_\nu k^2 \int_{-1}^{-\frac12} \frac{f(z)^2}{\mr{r}^{\nu-1}}dz + C_\nu \int_{-1}^{-\frac12} \frac{1}{\mr{r}^{\nu-1}}f(z)f'(z)dz \\[\jot] 
    &\leq C_\nu k^2 \int_{-1}^{-\frac12} \frac{f(z)^2}{\mr{r}^{\nu-1}}dz + \delta C_\nu \int_{-1}^{-\frac12} \frac{f(z)^2}{\mr{r}^{\nu}}dz + \delta^{-1} C_\nu \int_{-1}^{-\frac12} \frac{f'(z)^2}{\mr{r}^{\nu-2}}dz \\[\jot]
    &\leq C_\nu (\delta + k^2)\int_{-1}^{-\frac12} \frac{f(z)^2}{\mr{r}^{\nu}}dz + \delta^{-1} C_\nu \int_{-1}^{-\frac12} \frac{f'(z)^2}{\mr{r}^{\nu-2}}dz.
\end{align*}
We have used the condition (\ref{hardy1condition}) in order to drop boundary terms at $z = -1$, as well as uniform (in $k$) bounds for derivatives of $\mr{r}$ in the interval $z \in [-1,-\frac12]$. Choosing $\delta, k$ sufficiently small and absorbing the first integral on the right gives the stated estimate.
\end{proof}

The next estimate is key to deriving non-degenerate $C^k$ bounds on $\varphi_{m\ell}$ near the axis, given estimates on $r \varphi_{m\ell}$. Such estimates necessarily lose a derivative; this loss is quantified by \textit{averaging estimates}. A general framework for arbitrary order averaging estimates is given in \cite{lukoh1,lukohyang1}. In the context of $k$-self-similar spacetimes, and for the estimate below, we refer to \cite{singh}.

\begin{lemma}
Let $f(s,z): \mathcal{Q} \rightarrow \mathbb{R}$ be a given $C^2_{s,z}$ function. Uniformly in $s$, the following holds: 
\begin{equation}
    \label{eq:averagingest}
    \|\partial_z f(s,z)\|_{L^\infty([-1,-\frac{1}{2}])} \lesssim_k e^{-s}\|\partial_z (r_k f)(s,z)\|_{L^\infty([-1,-\frac{1}{2}])} + e^{-s}\|\partial_z^2 (r_k f)(s,z)\|_{L^\infty([-1,-\frac{1}{2}])}.
\end{equation} 
The same inequality holds with $r$ in place of $r_k$.
\end{lemma}

\subsection{First order estimates}
We now proceed to derive multiplier estimates for the wave equation (\ref{eq:wavesim}). Recall from (\ref{def:psivar}) the definition of the $r$ weighted quantity $\psi_{m \ell}(u,v)$, a function on the quotient spacetime. To simplify notation, we assume a $(m,\ell)$-mode is fixed and write $\psi = \psi_{m \ell}$.

To study the behavior of solutions to (\ref{eq:wavesim}) as $s \rightarrow \infty$, we consider $\psi_\rho \doteq e^{(q_k-\rho)s}\psi$ for $\rho \in \mathbb{R}$, which solves 
\begin{equation}
    \label{eq:wavesim_alpha}
    \partial_s \partial_z \psi_\rho - q_k |z| \partial_z^2 \psi_\rho + \rho \partial_z \psi_\rho + (V_k(z) + L_{k,\ell}(z))\psi_\rho = e^{(q_k-\rho)s}\mathcal{E}_{p,\ell}(s,z). 
\end{equation}
As suggested by the first order term with coefficient $\rho$, the availability of multiplier estimates will depend heavily on the value of $\rho$. For $\rho < 0$, a bound $|\psi_{\rho}| \lesssim 1$ asserts an exponential \textit{improved} decay rate with respect to the blue-shift rate ${\psi \sim e^{-q_k s}}$. Setting $\rho = 0$ corresponds to the blue-shift rate, and ${\rho > 0}$ to slower decay. 

In this section we prove a multiplier estimate at the level of $\partial \psi_\rho$. For $s_0, s_1 \in \mathbb{R}$ define  
\begin{align*}
    \mathcal{R}(s_0,s_1) &\doteq \mathcal{Q} \cap \{s_0 \leq s \leq s_1\}, \\
    \Gamma_{s_0, s_1} &\doteq \Gamma \cap \{s_0 \leq s \leq s_1\}, \\
    H_{s_0,s_1} &\doteq \{z=0\} \cap \{s_0 \leq s \leq s_1\}.
\end{align*}
We denote by $\mathcal{R}(s_1)$ the set $\mathcal{R}(0,s_1)$ provided $s_1 > 0$. Similarly we define $\Gamma_{s_1}, H_{s_1}$. 

\begin{prop}
    \label{prop:firstordermultiplierest}
Fix an $(\epsilon_0, k)$-admissible background spacetime, and a parameter $\rho \in [-1,1]$. There exists $k$ sufficiently small (independently of $\ell$) and constants $C_0$, $C_1 = C_1(\|V_k \|_{W^{1,2}_z([-1,0])})$, such that the following estimate holds in $\mathcal{R}(s_0)$ for sufficiently regular solutions to (\ref{eq:wavesim_alpha}):
\begin{align}
    \|\partial_z \psi_\rho\|^2_{L^2_z(\{s=s_0\})} +  \|\partial_z \psi_\rho \|^2_{L^2_s(\Gamma_{s_0})} &+\|(-L_{k,\ell}')^{\frac{1}{2}}\psi_\rho \|^2_{L^2_{s,z}(\mathcal{R}(s_0))}    + (k^2+\ell^2)\|\psi_\rho\|^2_{L^2_s(H_{s_0})} \nonumber \\[2\jot]
    &\leq  C_0 \|\partial_z \psi_\rho\|^2_{L^2_z(\{s=0\})} + (1 + C_1 k^2 - 2\rho)\|\partial_z \psi_\rho \|^2_{L^2_{s,z}(\mathcal{R}(s_0))}. \label{eq:energyest1}
\end{align}
\end{prop}
\begin{proof}
Multiplying (\ref{eq:wavesim_alpha}) by $\partial_z \psi_\rho$ and integrating by parts in $\mathcal{R}(s_0)$ yields 

\begin{align*}
    \frac{1}{2}\iint_{\mathcal{R}(s_0)}&\partial_s(\partial_z\psi_\rho)^2dzds - \frac{1}{2}q_k \iint_{\mathcal{R}(s_0)}|z|\partial_z(\partial_z\psi_\rho)^2dzds \\[\jot]
    &+ \rho \iint_{\mathcal{R}(s_0)}(\partial_z\psi_\rho)^2dzds + \frac{1}{2}\iint_{\mathcal{R}(s_0)} (V_k + L_{k,\ell}) \partial_z \psi_\rho^2 dzds = \iint_{\mathcal{R}(s_0)} e^{(q_k-\rho)s}\partial_z \psi_\rho \mathcal{E}_{p,\ell}(s,z) dzds
\end{align*}
\begin{align}
    \implies \frac{1}{2}\int_{\{s=s_0\}}(\partial_z\psi_\rho)^2 dz &+\frac{1}{2}q_k\int_{\Gamma_{s_0}}(\partial_z \psi_\rho)^2 ds + \frac{1}{2}\int_{H_{s_0}}(V_k + L_{k,\ell}) \psi_\rho^2 ds \nonumber \nonumber\\[\jot]
    &= \frac{1}{2}\int_{\{s=0\}}(\partial_z\psi_\rho)^2 dz + (\frac{1}{2}q_k-\rho)\iint_{\mathcal{R}(s_0)} (\partial_z \psi_\rho)^2 dzds \nonumber \\[\jot]
    &+ \frac{1}{2}\iint_{\mathcal{R}(s_0)} (V_k' + L_{k,\ell}') \psi_\rho^2 dzds + \iint_{\mathcal{R}(s_0)} e^{(q_k-\rho)s}\partial_z \psi_\rho \mathcal{E}_{p,\ell}(s,z) dz ds. \label{eq:zerothorderest_temp1}
\end{align}
As $V_k, L_{k,\ell} \geq 0$, the boundary term along $H_{s_0}$ is non-negative. To complete the estimate it remains to understand the bulk terms appearing on the right hand side. The term proportional to $L_{k,\ell}'$ has a good sign by (\ref{eq:L_kl_derivative_sign}). To handle the term proportional to $V_k'$, apply (\ref{eq:potentialest1}), (\ref{1d_est:Linf}), and the Dirichlet boundary condition for $\psi_\rho$ to compute  
\begin{align*}
    \frac{1}{2}\iint_{\mathcal{R}(s_0)} V_k' \psi_\rho^2 dzds &\leq \frac{1}{2}\int_{0}^{s_0}\big(\sup_{z'\in[-1,0]}|\psi_\rho|^2(z',s)\big) \int_{-1}^{0}V_k'(z) dz ds\\[\jot]
    &\lesssim C_1 k^2 \int_{0}^{s_0}\big(\sup_{z'\in[-1,0]}|\psi_\rho|^2(z',s)\big) ds \\[\jot]
    &\lesssim C_1k^2 \iint_{\mathcal{R}(s_0)} (\partial_z \psi_{\rho})^2(z,s)dzds.
\end{align*}
Finally, we consider the term proportional to $\mathcal{E}_{p,\ell},$ containing the perturbation from exact $k$-self-similarity. These terms are small in terms of $k$, and are handled by integration by parts:
\begin{align}
    \iint_{\mathcal{R}(s_0)} e^{(q_k-\rho)s}\partial_z \psi_\rho \mathcal{E}_{p,\ell}(s,z) dz ds &= -\frac{1}{2} \iint_{\mathcal{R}(s_0)} \frac{e^{-q_k s}}{\mr{r}}\Big(\bigg(\frac{\mu\lambda(-\nu)}{(1-\mu)r} \bigg)_p + \bigg(\frac{\lambda(-\nu)}{(1-\mu)r} \bigg)_p\ell(\ell+1)\Big)\partial_z \psi_\rho^2 \nonumber\\[\jot]
    &= O(\epsilon_0 k^2) \int_{H_{s_0}}(1 + \ell(\ell+1))\psi_\rho^2 + O(\epsilon_0 k^2) \iint_{\mathcal{R}(s_0)} \bigg(1 + \frac{\ell(\ell+1)}{\mr{r}^3}\bigg)\psi_\rho^2,\nonumber \\[\jot]
    &\leq O(\epsilon_0 k^2) \int_{H_{s_0}}(1 + \ell(\ell+1))\psi_\rho^2 + O(\epsilon_0 k^2) \iint_{\mathcal{R}(s_0)} (\partial_z \psi_{\rho})^2 \\[\jot]
    &\hspace{2em}+ O(\epsilon_0 k^2) \iint_{\mathcal{R}(s_0)} \frac{\ell(\ell+1)}{\mr{r}^3}\psi_\rho^2,
\end{align} 
where we have used the regularity of the background spacetime to ensure that $z$-derivatives of double-null quantities are bounded. The remaining terms carrying $\ell$-dependent constants can be absorbed for $\epsilon_0$ sufficiently small, and the bulk term proportional to $(\partial_z \psi_\rho)^2$ contributes to the right hand side of (\ref{eq:energyest1}).

Combining the analyses of bulk terms with the integrated estimate (\ref{eq:zerothorderest_temp1}), and choosing $C_1$ sufficiently large, concludes the proof. 
\end{proof}

\subsection{Second order estimates: $\ell = 0$}
In order to close the multiplier estimate, we will have to absorb the unfavorable bulk term appearing in (\ref{eq:energyest1}). The necessary structure to do so emerges after commuting (\ref{eq:wavesim_alpha}) by $\partial_z$, leading to estimates at the level of two derivatives of the solution. In this section we pursue such an estimate for the spherical component of the solution, and drop all terms proportional to $\ell$. A slight complication arises due to the limited regularity assumed on the background solution and the initial data -- recall $\partial_z^2 \psi$ is \textit{not} guaranteed to be bounded pointwise. We thus incorporate singular $|z|$ weights into the analysis.

Commuting (\ref{eq:wavesim_alpha}) by $\partial_z$ and setting $\ell =0$ yields
\begin{equation}
    \label{eq:wavesim_alpha_comm0}
    \partial_s \partial_z^2 \psi_\rho - q_k |z| \partial_z^3 \psi_\rho + (q_k + \rho)\partial_z^2 \psi_\rho + V_k \partial_z \psi_\rho + V_k' \psi_\rho = e^{(q_k-\rho)s}\partial_z \mathcal{E}_{p,0}(s,z).
\end{equation}
\begin{prop}
    \label{prop:secondordermultiplierest_l0}
    Fix an $(\epsilon_0, k)$-admissible background spacetime and parameters $\omega \in [0,\frac{1}{2}) $, $\rho \in [-1,1]$. There exists $k$ sufficiently small and constants $C_0$, $C_1 = C_1(\|V_k \|_{W^{1,2}_z([-1,0])})$, $C_2$ such that the following estimate holds in $\mathcal{R}(s_0)$ for sufficiently regular solutions to (\ref{eq:wavesim_alpha_comm0}):
    \begin{align}
         \label{eq:energyest2_0}
        \||z|^{\omega}\partial_z^2 \psi_\rho\|^2_{L^2_z(\{s=s_0\})} + q_k\|\partial_z^2 \psi_\rho \|^2_{L^2_s(\Gamma_{s_0})} + &\big(q_k(1-2\omega) - C_2 k^2 + 2\rho \big)\||z|^{\omega}\partial_z^2 \psi_\rho \|^2_{L^2_{s,z}(\mathcal{R}(s_0))} \nonumber \\[2\jot]
        &\leq C_0\||z|^{\omega}\partial_z^2 \psi_\rho\|^2_{L^2_z(\{s=0\})} + C_1 k^2 \|\partial_z \psi_\rho \|^2_{L^2_s(\Gamma_{s_0})} 
    \end{align}   
    For any constant $\delta \in (0,1)$, there moreover exists $k$ sufficiently small (depending on $\delta$) such that (\ref{eq:energyest2_0}) holds with the replacements $C_2 k^2 \rightarrow \delta$, and $C_1 k^2  \rightarrow C_1 k^4$.
\end{prop}
\begin{proof}
    To simplify the number of cases handled in the proof, assume without much loss of generality that $\omega > 0$. The case $\omega = 0$ introduces additional boundary terms that are easily handled.
    
    Multiplying (\ref{eq:wavesim_alpha_comm0}) by $|z|^{2 \omega}\partial_z^2 \psi_\rho$ and integrating by parts in $\mathcal{R}(s_0)$ yields
    \begin{align}
        \frac{1}{2}\int_{\{s=s_0\}}&|z|^{2\omega}(\partial_z^2\psi_\rho)^2 dz + \frac{1}{2}q_k\int_{\Gamma_{s_0}}(\partial_z^2 \psi_\rho)^2 ds + (\frac{1}{2}q_k(1-2\omega)+\rho)\iint_{\mathcal{R}(s_0)} |z|^{2\omega}(\partial_z^2 \psi_\rho)^2 dsdz  \nonumber \\[\jot]
        &= \frac{1}{2}\int_{\{s=0\}}|z|^{2\omega}(\partial_z^2\psi_\rho)^2 dz - \iint_{\mathcal{R}(s_0)} V_k |z|^{2\omega}\partial_z \psi_\rho \partial_z^2 \psi_\rho dzds - \iint_{\mathcal{R}(s_0)} V_k' |z|^{2\omega}\psi \partial_z^2 \psi_\rho dzds \nonumber \\[\jot]
        &+ \iint_{\mathcal{R}(s_0)} e^{(q_k-\rho)s} |z|^{2\omega}\partial_z^2 \psi_\rho \partial_z \mathcal{E}_{p,0}(s,z) dzds. \label{eq:temp2ndordermultiplier}
    \end{align}
    For any $\delta_1 \in (0,1)$ we estimate the first bulk term on the right hand side using Lemma \ref{lemma:basic1dintegral} as
    \begin{align*}
        \Big|\iint_{\mathcal{R}(s_0)}V_k |z|^{2\omega}\partial_z\psi_\rho \partial_z^2 \psi_\rho \Big| &\leq \frac{1}{2\delta_1}\iint_{\mathcal{R}(s_0)}V_k^2 |z|^{2\omega} (\partial_z\psi_\rho)^2 + \frac{\delta_1}{2}\iint_{\mathcal{R}(s_0)}|z|^{2\omega}(\partial_z^2\psi_\rho)^2 \\[\jot]
        &\leq \frac{1}{2\delta_1}\|V_k\|_{L^2_z}^2 \int_0^{s_0}\sup_{z' \in [-1,0]}|z'|^{2\omega} |\partial_z \psi_\rho|^2(z',s) ds + \frac{\delta_1}{2}\iint_{\mathcal{R}(s_0)}|z|^{2\omega}(\partial_z^2\psi_\rho)^2 \\[\jot]
        &\leq \frac{1}{\delta_1}C_1^2k^4 \int_{\Gamma_{s_0}}(\partial_z \psi_\rho)^2 + \big(\frac{2}{\delta_1}C_1^2k^4+\frac{\delta_1}{2}\big)\iint_{\mathcal{R}(s_0)}|z|^{2\omega}(\partial_z^2\psi_\rho)^2
    \end{align*}
    Similarly, the second bulk term may be estimated
    \begin{align*}
        \Bigg| \iint_{\mathcal{R}(s_0)} V_k' |z|^{2\omega} \psi_\rho \partial_z^2 \psi_\rho  \Bigg|  &\leq \frac{1}{2\delta_1}\iint_{\mathcal{R}(s_0)} |z|^{2\omega} (V_k' \psi_\rho)^2  + \frac{\delta_1}{2} \iint_{\mathcal{R}(s_0)} |z|^{2\omega}(\partial_z^2 \psi_\rho)^2 \\[\jot]
        &\leq   \frac{1}{2\delta_1}\|V_k'\|^2_{L^2_z}\int_{0}^{s_0}\sup_{z'\in[-1,0]}|z'|^{2\omega}|\psi_\rho|^2(z',s)ds + \frac{\delta_1}{2} \iint_{\mathcal{R}(s_0)} |z|^{2\omega}(\partial_z^2 \psi_\rho)^2 \\[\jot]
        &\leq   \frac{2}{\delta_1}C_1^2 k^4 \iint_{\mathcal{R}(s_0)} |z|^{2\omega} (\partial_z \psi_\rho)^2  + \frac{\delta_1}{2} \iint_{\mathcal{R}(s_0)} |z|^{2\omega}(\partial_z^2 \psi_\rho)^2 \\[\jot]
        &\leq \frac{4}{\delta_1}C_1^2 k^4 \int_{\Gamma_{s_0}}(\partial_z \psi_\rho)^2 ds + \big(\frac{8}{\delta_1}C_1^2k^4 +\frac{\delta_1}{2}\big)\iint_{\mathcal{R}(s_0)}|z|^{2\omega} (\partial_z^2 \psi_\rho)^2.
    \end{align*}
    It remains to consider the bulk term containing $\partial_z \mathcal{E}_{p,0}$.  
    \begin{align}
        \Bigg|\iint_{\mathcal{R}(s_0)} e^{(q_k-\rho)s}|z|^{2\omega}\partial_z^2 \psi_\rho \partial_z \mathcal{E}_{p,0}     \Bigg| &\lesssim O(\epsilon_0 k^2)  \iint_{\mathcal{R}(s_0)} |z|^{2\omega} |\partial_z^2 \psi_\rho| \big( |\psi_\rho| + |\partial_z \psi_\rho| \big) \label{3.3:eq1} \\[\jot]
        &\lesssim O(\epsilon_0 k^2) \int_{\Gamma_{s_0}}(\partial_z \psi_\rho)^2 + \big(O(\epsilon_0 k^2) + \frac{\delta_1}{2}\big)\iint_{\mathcal{R}(s_0)} |z|^{2\omega}(\partial_z^2 \psi_\rho)^2. \label{3.3:eq2}
    \end{align}
    Choosing $\delta_1 \sim k^2$, $C_1$ large enough and $k$ small, we arrive at (\ref{eq:energyest2_0}). The alternative form follows by choosing $\delta_1 \in (0,1)$ arbitrary, and then choosing $k$ sufficiently small. 
    \end{proof}

\subsection{Second order estimates: $\ell > 0$}
For $\ell > 0$ we first rewrite (\ref{eq:wavesim_alpha}) to collect all $\ell$-dependent quantities:
\begin{equation}
    \label{3.4:eq1}
    \partial_s \partial_z \psi_\rho - q_k |z| \partial_z^2 \psi_\rho + \rho \partial_z\psi_\rho + (V_k(z) + L_{\ell}(s,z))\psi_\rho = e^{(q_k-\rho)s}\mathcal{E}_{p,0}(s,z),
\end{equation}
where 
\begin{equation*}
    L_{\ell}(s,z) \doteq \frac{\br{\lambda} (-\nu)}{(1-\mu)\br{r}^2}\ell(\ell+1),
\end{equation*}
and for convenience we have set $\br{\lambda} \doteq e^{k^2 s}\lambda$, \ $\br{r} \doteq e^{s}r$. A consequence of (\ref{3.4:eq1}) is that the right hand side no longer has any $\ell$-dependence. Define a weight $w(s,z) \doteq \frac{(1-\mu)}{\br{\lambda}(-\nu)}$, and commute (\ref{3.4:eq1}) by $\partial_z(w \br{r}^2 \cdot)$. After a straightforward computation, we arrive at the following equation for $\psi_\rho^{(1)} \doteq w \br{r}^2 \partial_z \psi_\rho$:
\begin{align}
    \partial_s \partial_z &\psi_\rho^{(1)} - q_k |z| \partial_z^2 \psi_\rho^{(1)} + \underbrace{\big(q_k +\rho + q_k |z| \partial_z \log(w \br{r}^2\big) - \partial_s \log(w \br{r}^2))}_{\mathcal{A}_1}\partial_z \psi_\rho^{(1)} \nonumber\\[2\jot]
    &+ \underbrace{\bigg(\frac{\ell(\ell+1)}{w\br{r}^2} + V_k - q_k\partial_z \log(w \br{r}^2) + q_k |z| \partial_z^2 \log(w \br{r}^2) - \partial_z\partial_s \log(w \br{r}^2) \bigg)}_{\mathcal{A}_2}\psi_\rho^{(1)} + \partial_z\big(w\br{r}^2V_k\big)\psi_\rho \nonumber \\[2\jot]
    &\hspace{250pt} = \partial_z\big(w \br{r}^2e^{(q_k-\rho)s}\mathcal{E}_{p,0}\big). \label{eq:wavesim_alpha_comm1}
\end{align}

The key multiplier estimate is contained in the following proposition.
\begin{prop}
\label{prop:secondordermultiplierest_l1}
Fix an $(\epsilon_0, k)$-admissible background spacetime and parameters $\omega \in [0,\frac{1}{2})$, $\rho \in [-1,1]$. Assume $\varphi$ is supported on a single $(m,\ell)$-mode with $\ell > 0$. There exists $k$ sufficiently small (independent of $\ell$ and $\omega$), and constants $C_0, C_1$ such that the following estimate holds in $\mathcal{R}(s_0)$ for sufficiently regular solutions to (\ref{eq:wavesim_alpha_comm1}):
\begin{align}
    \label{eq:energyest2_1}
   \||z|^{\omega}\br{r}^{-2}\partial_z \psi_\rho^{(1)}\|^2_{L^2_z(\{s=s_0\})}  + \big(q_k(1-2\omega) - C_1 k^2 + 2\rho \big)&\||z|^{\omega}\br{r}^{-2}\partial_z \psi_\rho^{(1)} \|^2_{L^2_{s,z}(\mathcal{R}(s_0))} \nonumber \\[2\jot]
   &\leq C_0 \||z|^{\omega}\br{r}^{-2}\partial_z \psi_\rho^{(1)}\|^2_{L^2_z(\{s=0\})}.
\end{align}   
\end{prop}

Before beginning the proof, we state estimates on certain coefficients appearing in (\ref{eq:wavesim_alpha_comm1}).
\begin{lemma}
For all $\ell \geq 1$ and $k$ sufficiently small, we have
\begin{align}
    \partial_z \log w &\geq 0, \label{34:est1}\\[2\jot]
    \sup_{s \geq 0}\|\partial_s (w \br{r}^2) \|_{L^\infty([-1,0])} \lesssim k^2, \ \  \sup_{s \geq 0}\|\partial_z\partial_s (w \br{r}^2) \|_{L^\infty([-1,0])}  &\lesssim k^2, \label{34:est2}\\[2\jot] 
    \sup_{s \geq 0}\| \partial_z w\|_{L^p_z([-1,0])} \lesssim_p k^2, \ \ \sup_{s \geq 0}\| |z|\partial_z^2 w\|_{L^p_z([-1,0])} &\lesssim_p k^2, \label{34:est3}\\[2\jot]
    \sup_{s \geq 0}\big\|\br{r}^{-1}\big(1 - q_k w \br{r} \partial_z \br{r} - q_k |z| w (\partial_z \br{r})^2\big) \big\|_{L^\infty([-1,-\frac12])}  &\lesssim k^2 . \label{34:est4} 
\end{align}
\end{lemma}
\begin{proof}
    Decompose the weight $w(s,z)$ as $w(s,z) = w_k(z) + w_p(s,z)$, where $w_k(z)$ is the $k$-self-similar component. The latter may be rewritten as
    \begin{equation*}
        w_k(z) =  \frac{1-\mr{\mu}(z)}{(p_k |\hz|^{k^2}\mr{\lambda})(-\mr{\mu}(z)) } = \frac{4q_k}{ (|\hz|^{k^2}\mr{\Omega}^2)(z)}.
    \end{equation*} 
    By (\ref{eq:monotonic_omega}), $(|\hz|^{k^2}\mr{\Omega}^2)(z)$ is decreasing in $z$, with derivative bounded above by a multiple of ${- k^2}$. We may also estimate 
    \begin{align*}
        |\partial_z w_p| \lesssim \epsilon_0 k^2. 
    \end{align*} 
    Applying the triangle inequality allows us to conclude (\ref{34:est1}).

    To see (\ref{34:est2}), note that both estimates have at least one $\partial_s$ derivative. These derivatives vanish on the $k$-self-similar contribution, and therefore only see contributions arising from the background perturbation. Applying the bounds (\ref{eq:admissiblebounds1})--(\ref{eq:admissiblebounds4}) associated to admissible spacetimes gives the statement.

    The first bound in (\ref{34:est3}) is a consequence of the $L^p_z$ smallness manifest in (\ref{eq:smallkbound8})--(\ref{eq:smallkbound10}), (\ref{eq:smallkbound10.6}), as well as estimates on the perturbations. Similarly, the second bound in (\ref{34:est3}) follows from (\ref{eq:smallkbound15}).

    We finally consider (\ref{34:est4}). Define $f(s,z) \doteq 1- q_k w \br{r}\partial_z \br{r} + q_k |z| w (\partial_z \br{r})^2$. A direct, if tedious, calculation gives that 
    \begin{align}
        \label{eq:2ndorde_lest_temp1}
        \br{r}^{-1}f(s,z) = \Delta^{-1}(q_k \mu  + \frac{\mu}{\br{r}}q_k|z|\partial_z\br{r} + k^2 - \partial_s \log \br{r}),
    \end{align}
    where 
    \begin{align*}
        \Delta \doteq q_k |z| \partial_z \br{r} + \br{r} - \partial_s \br{r}.
    \end{align*}
    Each term appearing in (\ref{eq:2ndorde_lest_temp1}) is of size $k^2$ in the region $z \in [-1,-\frac12]$. (\ref{34:est4}) therefore follows.
\end{proof}

\begin{proof}[Proof of Proposition \ref{prop:secondordermultiplierest_l1}]
    To simplify the number of cases handled in the proof, assume that $\omega > 0$. The case $\omega = 0$ introduces additional boundary terms that are easily handled.
     
    Multiplying (\ref{eq:wavesim_alpha_comm1}) by $|z|^{2\omega}\br{r}^{-4}\partial_z \psi^{(1)}_\rho$ and integrating by parts in $\mathcal{R}(s_0)$ yields
    \begin{align}
        \frac{1}{2}\int_{\{s=s_0\}}&\frac{|z|^{2\omega}}{\br{r}^4}(\partial_z \psi^{(1)}_\rho)^2 + \frac{1}{2}\int_{\Gamma_{s_0}} \frac{\big(\partial_z \psi_\rho^{(1)} \big)^2 }{\br{r}^4} \nonumber \\[2\jot]
        +\iint_{\mathcal{R}(s_0)} &\big(\frac12 q_k(1-2\omega) + \rho + q_k |z| \partial_z \log w - \partial_s \log w \big) \frac{|z|^{2\omega}}{\br{r}^4}(\partial_z \psi^{(1)}_\rho)^2   \nonumber\\[2\jot]
        &= \frac{1}{2}\int_{\{s=0\}}\frac{|z|^{2\omega}}{\br{r}^4}(\partial_z\psi^{(1)}_\rho)^2  - \iint_{\mathcal{R}(s_0)} \mathcal{A}_2 \frac{|z|^{2\omega}}{\br{r}^4}\partial_z \psi^{(1)}_\rho  \psi^{(1)}_\rho  - \iint_{\mathcal{R}(s_0)} \partial_z(w \br{r}^2V_k) \frac{|z|^{2\omega}}{\br{r}^4}\partial_z \psi^{(1)}_\rho \psi_\rho  \nonumber\\[2\jot]
        &+ \iint_{\mathcal{R}(s_0)} e^{(q_k-\rho)s} \frac{|z|^{2\omega}}{\br{r}^4}\partial_z \psi^{(1)}_\rho \partial_z (w \br{r}^2 \mathcal{E}_{p,0}) .\label{eq:prop33:temp1.0}
    \end{align}
    We turn to estimating the various bulk terms, starting with the term proportional to $\mathcal{A}_2$. Since $\ell \geq 1$ we have $\ell(\ell+1) \geq 2$. Define $\ell^2_r \doteq \ell(\ell+1)-2$. Applying (\ref{34:est2})--(\ref{34:est4}) gives the pair of localized estimates  
    \begin{align}
        \sup_{s \geq 0} \| \mr{r}\bigg( \mathcal{A}_2 - \frac{\ell^2_r}{w\br{r}^2}\bigg)\|_{L^\infty([-1,-\frac{1}{2}])} &\lesssim k^2, \label{eq:prop33:temp2.0}\\
        \sup_{s \geq 0} \| \mathcal{A}_2 - \frac{\ell^2_r}{w\br{r}^2} \|_{L^2_z([-\frac{1}{2},0])} &\lesssim k^2, \label{eq:prop33:temp3.0}
    \end{align}
    Observe that (\ref{eq:prop33:temp2.0}) captures a cancellation in top order powers of $\mr{r}$ near the axis. Therefore
    \begin{align}
        \iint_{\mathcal{R}(s_0)} &\mathcal{A}_2 \frac{|z|^{2\omega}}{\br{r}^4}\partial_z \psi^{(1)}_\rho  \psi^{(1)}_\rho \nonumber  \\[\jot]
        &=  \frac{1}{2}\iint_{\mathcal{R}(s_0)} \frac{\ell^2_r}{w \br{r}^2} \frac{|z|^{2\omega}}{\br{r}^4}\partial_z \big(\psi_\rho^{(1)}\big)^2 + \iint_{\mathcal{R}(s_0)} \bigg(\mathcal{A}_2-\frac{\ell^2_r}{w \br{r}^2} \bigg)\frac{|z|^{2\omega}}{\br{r}^4}\partial_z \psi_\rho^{(1)} \psi_\rho^{(1)}. \label{eq:prop33:temp4.0}
      \end{align} 
      The first term may be integrated by parts to produce favorable bulk terms (cf. (\ref{34:est1}))
      \begin{align*}
          \omega \iint_{\mathcal{R}(s_0)}  \frac{\ell_r^2}{w\br{r}^2}\frac{|z|^{2\omega-1}}{\br{r}^4} \big(\psi_\rho^{(1)}\big)^2 + \frac12 \iint_{\mathcal{R}(s_0)}  \frac{\ell_r^2}{w\br{r}^2} \frac{|z|^{2\omega}}{\br{r}^4}\partial_z \log(w\br{r}^6) \big(\psi_\rho^{(1)}\big)^2.
      \end{align*}
      We have dropped boundary terms at the axis, noting that these terms appear only for $\ell \geq 2$, for which we have $|\psi_{\rho}^{(1)}| \lesssim \br{r}^4.$ This decay is fast enough to overwhelm the singular powers of $\br{r}$ appearing the integrand. 
    
      The latter term in (\ref{eq:prop33:temp4.0}) gives, after applying the localized estimates (\ref{eq:prop33:temp2.0})--(\ref{eq:prop33:temp3.0}), the Hardy inequality (\ref{hardy1}), and Lemma \ref{lemma:basic1dintegral}, 
    \begin{align*}
        \Bigg| \iint_{\mathcal{R}(s_0)} &\bigg(\mathcal{A}_2-\frac{\ell^2_r}{w \br{r}^2} \bigg)\frac{|z|^{2\omega}}{\br{r}^4}\partial_z \psi_\rho^{(1)} \psi_\rho^{(1)} \Bigg| \\[2\jot]
        &\lesssim k^2 \iint\limits_{\substack{\mathcal{R}(s_0) \\ \{z \leq -\frac{1}{2} \} }} \frac{1}{\br{r}^5} \big|\partial_z \psi_\rho^{(1)} \big| \big| \psi_\rho^{(1)}\big| + \iint\limits_{\substack{\mathcal{R}(s_0) \\ \{z \geq -\frac{1}{2} \} }} \bigg(\mathcal{A}_2-\frac{\ell^2_r}{w \br{r}^2} \bigg)|z|^{2\omega}\big|\partial_z \psi_\rho^{(1)} \big| \big|\psi_\rho^{(1)}\big| \\[2\jot]
        &\lesssim k^2 \iint\limits_{\substack{\mathcal{R}(s_0) \\ \{z \leq -\frac{1}{2} \} }} \frac{1}{\br{r}^4} \big(\partial_z \psi_\rho^{(1)} \big)^2 + k^2 \iint\limits_{\substack{\mathcal{R}(s_0) \\ \{z \leq -\frac{1}{2} \} }} \frac{1}{\br{r}^6}\big( \psi_\rho^{(1)} \big)^2 \\[2\jot]
        &+ \delta_1^{-1}\iint\limits_{\substack{\mathcal{R}(s_0) \\ \{z \geq -\frac{1}{2} \} }} \bigg(\mathcal{A}_2-\frac{\ell^2_r}{w \br{r}^2} \bigg)^2 |z|^{2\omega}\big( \psi_\rho^{(1)} \big)^2 + \delta_1 \iint\limits_{\substack{\mathcal{R}(s_0) \\ \{z \geq -\frac{1}{2} \} }} |z|^{2\omega}\big(\partial_z \psi_\rho^{(1)} \big)^2 \\[2\jot]
        &\lesssim (\delta_1+ k^2) \iint_{\mathcal{R}(s_0)} \frac{|z|^{2\omega}}{\br{r}^4} \big(\partial_z \psi_\rho^{(1)} \big)^2 + \delta_1^{-1}\int_{0}^{s_0} \sup_{z \in [-\frac12, 0]} \big| |z|^{\omega}\psi_\rho^{(1)} \big|^2 \bigg( \int_{-\frac12}^{0}\bigg(\mathcal{A}_2-\frac{\ell^2_r}{w \br{r}^2} \bigg)^2 dz\bigg) ds \\[2\jot] 
        &\lesssim (\delta_1+ k^2) \iint_{\mathcal{R}(s_0)} \frac{|z|^{2\omega}}{\br{r}^4} \big(\partial_z \psi_\rho^{(1)} \big)^2 + \delta_1^{-1}k^4 \int_0^{s_0}\sup_{z \in [-\frac12, 0]} \big| |z|^{\omega}\psi_\rho^{(1)} \big|^2 ds \\[2\jot]
        &\lesssim (\delta_1+ k^2) \iint_{\mathcal{R}(s_0)} \frac{|z|^{2\omega}}{\br{r}^4} \big(\partial_z \psi_\rho^{(1)} \big)^2 + \delta_1^{-1}k^4 \iint_{\mathcal{R}(s_0)} |z|^{2\omega}\big(\partial_z \psi_\rho^{(1)} \big)^2 \\[2\jot]
        &\lesssim (\delta_1+ k^2 + \delta_1^{-1}k^4) \iint_{\mathcal{R}(s_0)} \frac{|z|^{2\omega}}{\br{r}^4} \big(\partial_z \psi_\rho^{(1)} \big)^2.
    \end{align*}
    A similar analysis using (\ref{eq:potentialest1}), (\ref{34:est3}) gives for the second bulk term in (\ref{eq:prop33:temp1.0})
    \begin{align*}
    \Bigg|\iint_{\mathcal{R}(s_0)} &\partial_z(w \br{r}^2V_k) \frac{|z|^{2\omega}}{\br{r}^4}\partial_z \psi^{(1)}_\rho \psi_\rho  \Bigg| \nonumber\\
    &\lesssim \delta_1^{-1} \iint_{\mathcal{R}(s_0)} \big(\partial_z(w \br{r}^2V_k)\big)^2\frac{|z|^{2\omega}}{\br{r}^4} (\psi_\rho)^2 + \delta_1\iint_{\mathcal{R}(s_0)} \frac{|z|^{2\omega}}{\br{r}^4} (\partial_z \psi^{(1)}_\rho)^2 \nonumber\\[2\jot]
    &\lesssim \delta_1^{-1}k^4 \iint\limits_{\substack{\mathcal{R}(s_0) \\ \{z \leq -\frac{1}{2} \} }} \frac{1}{\br{r}^2} (\psi_\rho)^2 + \delta_1^{-1} \iint\limits_{\substack{\mathcal{R}(s_0) \\ \{z \geq -\frac{1}{2} \} }}(\partial_z(w V_k))^2 |z|^{2\omega} (\psi_\rho)^2  + \delta_1 \iint_{\mathcal{R}(s_0)} \frac{|z|^{2\omega}}{\br{r}^4} (\partial_z \psi^{(1)}_\rho)^2  \\
    &\lesssim \delta_1^{-1}k^4 \iint\limits_{\substack{\mathcal{R}(s_0) \\ \{z \leq -\frac{1}{2} \} }} \frac{1}{\br{r}^4} (\psi_\rho^{(1)})^2 + \delta_1^{-1} k^4 \int_{0}^{s_0} \sup_{z \in [-1,0]} \big||z|^\omega  \psi_\rho \big|^2 + \delta_1 \iint_{\mathcal{R}(s_0)} \frac{|z|^{2\omega}}{\br{r}^4} (\partial_z \psi^{(1)}_\rho)^2  \\
    &\lesssim \delta_1^{-1} k^4 \iint_{\mathcal{R}(s_0)} \frac{|z|^{2\omega}}{\br{r}^4} ( \psi^{(1)}_\rho)^2 + \delta_1 \iint_{\mathcal{R}(s_0)} \frac{|z|^{2\omega}}{\br{r}^4} (\partial_z \psi^{(1)}_\rho)^2 \\
    &\lesssim (\delta_1+\delta_1^{-1}k^4) \iint_{\mathcal{R}(s_0)} \frac{|z|^{2\omega}}{\br{r}^4} (\partial_z \psi^{(1)}_\rho)^2.
    \end{align*}
    
    The remaining bulk term in (\ref{eq:prop33:temp1.0}) depends on $\mathcal{E}_{p,0}$. The calculation proceeds by the strategy above, in order to handle the $\br{r}$ weights near the axis. The term here is not borderline in terms of $\br{r}$ weights, has all terms proportional to $k^2$, and has no $\ell$-dependence. Therefore one may simply take $k$ small to absorb this term, at the cost of a small loss in the $\partial_z \psi_\rho^{(1)}$ bulk term.
    
    Collecting terms and setting $\delta_1 \sim k^2$, we arrive at the stated estimate.
    \end{proof}

    \subsection{Non-sharp decay: $\ell = 0$}
    For the spherically symmetric component of the solution to (\ref{eq:1}), denoted $\varphi(s,z)$, the multiplier estimates (\ref{eq:energyest1}), (\ref{eq:energyest2_0}) alone do not yield the sharp decay stated in Theorem \ref{theorem:mainthm}. Still, these physical space methods are able to give an upper bound on the growth rate of the scalar field in the case $\varphi\big|_{\{s=0\}} \in \mathcal{C}^{\alpha}_{(hor)}([-1,0])$ for $\alpha > p_k$ independently of $k$. In this section we will allow $k$ to be chosen small in terms of $\alpha$, and thus we are working in a ``high above threshold regularity'' setting.

Subject to this regularity condition, we show that the blue-shift rate $\frac{1}{\partial_v r_k}\partial_v\varphi \sim |u|^{-1-k^2}$ provides the sharp scaling with respect to $k$ for any solutions to (\ref{eq:1}) that are unstable in the sense of Definition \ref{defn:self-similarbounds}. The argument is unable to detect the difference between solutions growing at rates $|u|^{-1}$ (self-similar), $|u|^{-1-k^2}$ (blue-shift), or $|u|^{-1-B k^2}$ (multiplies of the blue-shift), as our multiplier estimates have used only weak information on $V_k$ (i.e., smallness in $W^{1,2}_z$). Therefore, to show sharper decay (and necessarily, to rule out unstable modes) requires either energy estimates of a more refined nature, or appeal to spectral theory as done in the remainder of the paper.

\begin{prop}[Growth at a multiple of the blue-shift rate]
    \label{prop:basicdecayphysicalspace}
    Fix a parameter $\alpha \in (1,2)$. There exists $k$ \underline{small depending on $\alpha$}, such that for any $(\epsilon_0,k)$-admissible spacetime and spherically symmetric initial data to (\ref{eq:wavesim}) with regularity $\varphi_0(z) \in C^{\alpha}_{(hor)}([-1,0])$, there exists a constant $B > 1$ depending on $\alpha$, and a constant $C$ depending on the data such that for all $u \in [-1,0)$ we have the pointwise bounds
    \begin{align}
        \|\varphi\|_{L^\infty(\Sigma_u)} &\lesssim C |u|^{-B k^2}, \label{eq:prelimgrowthbound}\\
        \|\partial_{u}(\br{r} \varphi)\|_{L^\infty(\Sigma_u)} &\lesssim C |u|^{-1-B k^2}, \label{eq:prelimgrowthbound2}\\
        \|\partial_{v}(\br{r} \varphi)\|_{L^\infty(\Sigma_u)} &\lesssim C |u|^{-1-(B-1) k^2}. \label{eq:prelimgrowthbound3}
    \end{align}
\end{prop}
\begin{proof}
The strategy is to close the pair of estimates (\ref{eq:energyest1}), (\ref{eq:energyest2_0}) with $\rho \doteq B k^2$, for an appropriately chosen constant $B > 0$. Unpacking the regularity assumption $\varphi_0 \in C^{\alpha}_{(hor)}([-1,0])$ shows 
\begin{equation*}
    |z|^{\omega} \partial_z^2 (\br{r}\varphi_0) \in L^2_z([-1,0]),
\end{equation*}
where $\omega = 0$ if $\alpha > \frac{3}{2}$, and $\omega = \frac32 - \alpha + \epsilon$ for $\alpha \leq \frac{3}{2}$ and $\epsilon \ll 1 $ sufficiently small.

We focus on the bulk term $\|\partial_z  \psi_\rho \|^2_{L^2_{s,z}(\mathcal{R}(s_0))}$. By Lemma \ref{lemma:basic1dintegral} this term is estimated as follows, for any $a \in (0,1-2\omega)$,
\begin{align*}
    \|\partial_z  \psi_\rho \|^2_{L^2_{s,z}(\mathcal{R}(s_0))} & \leq \frac{1}{1-a_{0,\omega}} \|\partial_z \psi_\rho\|^2_{L^2_s(\Gamma_{s_0})} + \frac{1}{a(1-a_{0,\omega})} \||z|^\omega \partial_z^2 \psi_\rho \|^2_{L^2_{s,z}(\mathcal{R}(s_0))}. 
\end{align*}
To control the latter bulk term, we apply (\ref{eq:energyest2_0}) (in particular, the alternative form stated in Proposition \ref{prop:secondordermultiplierest_l0}) with parameter $\delta \leq \frac{1}{4}(1-2\omega)$. It follows that for $k$ sufficiently small depending on $\delta$ (and therefore, $\alpha$),  
\begin{align*}
    \|\partial_z  &\psi_\rho \|^2_{L^2_{s,z}(\mathcal{R}(s_0))} \\[\jot]
    &\leq \frac{1}{1-a_{0,\omega}} \|\partial_z \psi_\rho \|^2_{L^2_s(\Gamma_{s_0})} + \frac{1}{(\frac{1}{2}q_k(1-2\omega) - \delta +\rho)a(1-a_{0,\omega})}\bigg(\||z|^{\omega} \partial_z^2 \psi_\rho\|^2_{L^2_z(\{s=0\})} + C_0 k^4 \|\partial_z \psi_\rho\|^2_{L^2_s(\Gamma_{s_0})} \bigg) \\[\jot]
    &\leq C(a,\rho, \omega, k) \|\partial_z \psi_\rho \|_{W^{1,2}_z(\{s=0\})} +  D(a,\rho,\omega, k) \| \partial_z \psi_\rho \|^2_{L^2_{s,z}(\mathcal{R}(0,s_0))},
\end{align*} 
where 
\begin{equation*}
    D(a,\rho,\omega, k) \doteq \frac{1}{1-a_{0,\omega}}\Big(1 + \frac{C_0 k^4}{a(\frac{1}{2}q_k(1-2\omega) - \delta + \rho)} \Big)\Big(1 + C_1 k^2 - 2 \rho \Big).
\end{equation*}
To close the estimate, it suffices to check $D(a,\rho,\omega,k) < 1$ for suitable choices of $a$, $B$. Provided $\rho \leq \frac{1}{4}$ it is straightforward to check that the second factor in the definition of $D(a,\rho,\omega, k)$ is decreasing in $\rho$ and the third factor non-negative. Therefore,

\begin{equation*}
    D(a,\rho,\omega, k) \leq \frac{1}{1-a_{0,\omega}}\Big(1 + \frac{C_0 k^4}{a(\frac{1}{2}q_k(1-2\omega) - \delta)} \Big)\Big(1 + C_1 k^2 - 2 \rho \Big).
\end{equation*}
Set $a = k^2$. For $k$ sufficiently small we have $k^2 < 1 - 2\omega$, and thus this is a permissible choice for $a$. With $c_\alpha$ denoting positive, $\alpha$-dependent constants, we may write
$a_{0,\omega} \leq c_\alpha k^2$ and $\frac12 q_k (1-2\omega) - \delta \geq c_\alpha. $
It follows that for $k$ sufficiently small depending on $c_\alpha$, 
\begin{equation*}
    D(a,\rho,\omega,k) \leq (1 + 2 c_\alpha k^2)(1 + \frac{C_0}{c_\alpha}k^2 ) ( 1 + C_1 k^2 - 2 B k^2).
\end{equation*}
It now suffices to choose $B$ large in relation to $C_0, C_1, c_\alpha$ in order to render this term strictly less than $1$ for all $k$ small.

To conclude the proof, we translate control on the bulk integral to the stated pointwise bounds. By (\ref{eq:energyest1})--(\ref{eq:energyest2_0}), we bound
\begin{align*}
    \sup_{s\geq0 }\|\partial_z(r \varphi) \|_{L^2_z(\{s=s_0\})} + \sup_{s\geq0 }\||z|^{\omega} \partial_z^2(r \varphi) \|_{L^2_z(\{s=s_0\})} \lesssim  e^{-(1-Bk^2)s_0}.
\end{align*}
One-dimensional Sobolev embedding on $\{s=s_0, z \leq z_0 < 0\}$ gives pointwise decay for $\partial_z (r \varphi)$ in $L^\infty(\{s=s_0\})$ in the near-axis region. The $L^2_z$ control on $|z|^{\omega}\partial_z^2(r\varphi)$ is sufficient to extend this near-axis bound to one on the whole interval $z \in [-1,0]$. We are using here that $|z|^{-\omega} \in L^2_z([-1,0])$.

Integrating from the axis and applying the boundary condition $r\varphi\big|_{\Gamma}=0$, we conclude an identical pointwise decay bound for $r \varphi$. 

To see the bound for $\partial_s (r\varphi)$, we turn to the wave equation (\ref{eq:wavesimfull}) and use the boundary condition $\partial_s(r\varphi)\big|_{\Gamma}=0$, as $\partial_s$ is tangent to the axis. Note that the angular derivative terms in (\ref{eq:wavesimfull}) drop out, as we are working with the spherically symmetric component of the solution. It follows that 
\begin{align*}
    |\partial_s(r \phi)(s_0,z)| &\lesssim \| |z| \partial_z^2(r\phi)\|_{L^1_z(\{s=s_0\})} + \|  \partial_z(r\phi)\|_{L^1_z(\{s=s_0\})} + \|  r\phi\|_{L^1_z(\{s=s_0\})} \\
    &\lesssim e^{-(1-Bk^2)s_0}.
\end{align*}
Applying the transformation rules in Table \ref{table:1}, the desired bounds in double-null coordinates follow.
\end{proof}

The energy estimates also imply a boundedness statement for spherically symmetric solutions to (\ref{eq:wavesim}).

\begin{prop}
    \label{prop:terrible_boundedness}
    Fix a parameter $\alpha \in (1,2)$. There exists $k$ small independent of $\alpha$ such that for any $(\epsilon_0,k)$-admissible spacetime and spherically symmetric initial data to (\ref{eq:wavesim}) with regularity $\varphi_0(z) \in C^{\alpha}_{(hor)}([-1,0])$, there exists a constant $C$ depending on the data such that for all $u \in [-1,0)$ we have the pointwise bounds
    \begin{align}
        \|\br{r} \varphi\|_{L^\infty(\Sigma_u)} &\lesssim C|u|^{-1} , \label{eq:terribleboundednessbd1}\\
        \|\partial_{u}(\br{r} \varphi)\|_{L^\infty(\Sigma_u)} &\lesssim C |u|^{-2}, \label{eq:terribleboundednessbd2}\\
        \|\partial_{v}(\br{r} \varphi)\|_{L^\infty(\Sigma_u)} &\lesssim C |u|^{-2+ k^2}. \label{eq:terribleboundednessbd3}
    \end{align}
    Moreover, the higher derivative bounds hold
    \begin{align}
        \|\partial_u^i \partial_v^j (\br{r} \varphi) \|_{L^\infty(\Sigma_u)} \lesssim C |u|^{-1-i-q_k j}, \quad 1 \leq i+j \leq 5, \ j \leq 1.
    \end{align}
\end{prop}
\begin{proof}
The proof follows by the same strategy as in Proposition \ref{prop:basicdecayphysicalspace}---close the multiplier estimates for some value of $\rho$, and repeatedly integrate the resulting $L^2$ estimates (at high orders of derivatives) to control quantities in $L^\infty$ (with the loss of a derivative). We therefore only sketch the details.

Choose $\rho = q_k$. It follows that for $k$ sufficiently small, the bulk term appearing on the right hand side of (\ref{eq:energyest1}) becomes negative, and we thus conclude control on $\partial_z \psi_{\rho} \in L^2_z(\{s=s_0\})$ in terms of data. This control allows us to close the second order estimate (\ref{eq:energyest2_0}), controlling $|z|^{\omega}\partial_z^2 \psi_{\rho} \in L^2_z(\{s=s_0\})$ for some power $\omega \in [0,\frac12]$. Commuting repeatedly with $\partial_s$, applying the multiplier estimates, using the fundamental theorem of calculus and the equation (\ref{eq:wavesim}), and transforming back to double-null coordinates, we conclude the stated estimates.

\end{proof}

\subsection{Elementary bound for a wave equation in hyperbolic coordinates}
In this section we prove a basic pointwise bound for solutions to the inhomogeneous, spherically symmetric equation
\begin{equation}
    \label{eq:elementarygrowthbound_hyper}
    \begin{cases}
        \partial_t^2 \psi - \partial_x^2 \psi + 4q_k e^{-2q_k x}V(t,x)\psi = F(t,x),\\[\jot]
        \psi(t,0) = 0\\[\jot]
        (\psi(0,x), \partial_t\psi(0,x)) = (f_1(x),f_2(x)) \in W^{3,2}_x(\mathbb{R}_+) \times W^{2,2}_x(\mathbb{R}_+),
    \end{cases}
\end{equation}
where the forcing satisfies 
\begin{equation*}
     F(t,x) \in W^{3,1}_t L^2_x(\mathbb{R}_+ \times \mathbb{R}_+),
\end{equation*}
and $V(t,x)$ is defined in (\ref{eq:combinedpotentials}).
\begin{lemma}
    \label{lem:pointwiseboundforeasywaveequation}
    Fix an $(\epsilon_0, k)$-admissible extended background. The solution to (\ref{eq:elementarygrowthbound_hyper}) satisfies the pointwise bound
    \begin{equation}
        \label{eq:pointwisebdforcedeqn}
        \max_{0 \leq j \leq 3}\sup_{ t \geq 0 }|\partial_t^j \psi(t,x)| \lesssim C \sqrt{x},
    \end{equation}
    for a constant $C>0$ satisfying 
    \begin{equation*}
        C \lesssim \|f_1 \|_{W^{3,2}_x(\mathbb{R}_+)} + \|f_2 \|_{W^{2,2}_x(\mathbb{R}_+)} + \sum_{j=0}^3 \|\partial_t^j F\|_{L^1_t L^2_x(\mathbb{R}_+ \times \mathbb{R}_+)}.
    \end{equation*}
\end{lemma}
\begin{proof}
    Fix a $t_0 > 0$. Multiplying (\ref{eq:elementarygrowthbound_hyper}) by $\partial_t \psi$ and integrating by parts in the truncated lightcone $\mathcal{B}(t,t_0)\doteq\{0 \leq t \leq t_0, \ 0 \leq x \leq 2t_0-t\}$ yields the energy estimate 
\begin{align*}
    \sup_{0 \leq t' \leq t_0} \int_{\{t=t',\ 0 \leq x \leq  2t_0-t'\}} &\big((\partial_t \psi)^2 + (\partial_x \psi)^2  + e^{-2q_k x} \psi^2 \big)dx \\[\jot]
    &\lesssim \int_{\{t=0\}} \big((\partial_t \psi)^2 + (\partial_x \psi)^2  + e^{-2q_k x} \psi^2 \big)dx +  \int_{0}^{t_0} \|F\|_{L^2_x([0,2t_0-t])} dt  \\[\jot]
    &\hspace{2em}+ \iint_{\mathcal{B}(t,t_0)}\big|\partial_t V_{k,p} \big| e^{-2q_k x}\psi^2 dtdx  \\[\jot]
    \lesssim &\int_{\{t=0\}} \big((\partial_t \psi)^2 + (\partial_x \psi)^2  + e^{-2q_k x} \psi^2 \big)dx +  \int_{0}^{t_0} \|F\|_{L^2_x([0,2t_0-t])} dt  \\[\jot]
    &\hspace{2em}+ \epsilon_0 \sup_{0 \leq t' \leq t_0}   \int_{\{t=t',\ 0 \leq x \leq  2t_0-t'\}}  e^{-2q_k x} \psi^2 dx
 \end{align*}
We have dropped the boundary term along $\{x = 2t_0 - t\},$ which has a favorable sign. For $\epsilon_0$ sufficiently small, we may absorb the remaining error term and conclude the estimate on $\sup_{t' \geq 0}\|\partial_x \psi \|_{L^2_x(\{t=t'\})}$. As $\psi(t,0)=0$ we may integrate outwards in $x$ to conclude for any $t'$, 
\begin{align*}
    \|\psi\|_{L^\infty(\{t=t'\})} &\leq \|\partial_x \psi\|_{L^1_x(\{t=t'\})} \\
    &\leq \|\partial_x \psi\|_{L^2_x(\{t=t'\})}\sqrt{x} \\
    &\leq C(\|F\|_{L^1_t L^2_x}, \|f_1\|_{W^{1,2}_x}, \|f_2\|_{L^2_x})\sqrt{x}.
\end{align*}
Commuting (\ref{eq:elementarygrowthbound_hyper}) by $\partial_t^{j}$ and applying an identical argument gives the remaining cases of (\ref{eq:pointwisebdforcedeqn}).
\end{proof}
\section{Scattering theory on $k$-self-similar backgrounds}
\label{sec:scattering}

\subsection{Special functions}
\label{subsec:specialfns}
In this section we recall definitions for a class of special functions, including the modified Bessel functions and digamma functions. References for the material here include \cite{olver, NIST:DLMF}.

We first introduce modified Bessel's equation with complex order $\hat{\sigma} \in \mathbb{C}$. 

\begin{definition}
    \label{def:propertiesofBessel}
    The modified Bessel equation is the following second order ordinary differential equation for a function $f(y):(0,\infty) \rightarrow \mathbb{C}$:
    \begin{equation}
        \label{eq:modifiedbessel}
        y^2 \frac{d^2f}{dy^2} + y \frac{df}{dy} - (y^2 + \hat{\sigma}^2)f = 0.
    \end{equation}
    An independent set of solutions to (\ref{eq:modifiedbessel}) is given by the modified Bessel functions of the first and second kind, denoted $I_{\hat{\sigma}}(y), K_{\hat{\sigma}}(y)$ respectively.

    These solutions enjoy the following properties:
    \begin{itemize}
        \item For fixed $\hat{\sigma} \in \mathbb{C}$, $I_{\hat{\sigma}}(y), K_{\hat{\sigma}}(y) $ are smooth functions of $y \in (0,\infty)$.
        \item For fixed $y \in (0,\infty)$, $I_{\hat{\sigma}}(y), K_{\hat{\sigma}}(y) $ are entire functions of $\hat{\sigma}$.
        \item For $\hat{\sigma} \notin \mathbb{Z}$, the set $\{I_{\pm \hat{\sigma}}(y)\}$ is linearly independent, with Wronskian\footnote{For single-variable functions $f(x), g(x) \in C^1(I)$ defined on an open interval $I$, the Wronskian is defined as $W[f,g](x) \doteq f(x)g'(x) - f'(x)g(x)$.}
        \begin{equation}
            \label{bessel:wronskian}
            W[I_{\hat{\sigma}}, I_{-\hat{\sigma}}](y) = -\frac{2 \sin(\pi \hat{\sigma})}{\pi y}.
        \end{equation}
        \item For $\hat{\sigma} \notin -\mathbb{N}$, $I_{\hat{\sigma}}(y)$ admits a convergent series expansion 
        \begin{equation}
            \label{eq:model1}
            I_{\hat{\sigma}}(y) = \Big(\frac{y}{2}\Big)^{\hat{\sigma}} \sum_{m=0}^{\infty} \frac{1}{m!\Gamma(m +\hat{\sigma}+1)}\Big(\frac{y}{2}\Big)^{2m},
        \end{equation} 
        where $\Gamma(z)$ denotes the gamma function. In particular, $I_{\hat{\sigma}}(y) \sim y^{\hat{\sigma}}$ as $y \rightarrow 0$.
        \item For $\hat{\sigma} \in - \mathbb{N}$, (\ref{eq:model1}) holds formally after dropping the first $|\hat{\sigma}|$ terms. In this case $\{I_{\pm \hat{\sigma}}(y)\}$ are linearly dependent, and $I_{\hat{\sigma}}(y) \sim y^{-\hat{\sigma}}$ as $y \rightarrow 0$.
        \item For $\hat{\sigma} \notin \mathbb{Z}$, we can relate $I_{\pm \hat{\sigma}}, K_{\hat{\sigma}}$ by the formula 
        \begin{equation}
            \label{eq:KvsI_Bessel}
            K_{\hat{\sigma}}(y) = \frac{\pi}{2}\frac{I_{-\hat{\sigma}}(y)-I_{\hat{\sigma}}(y)   }{\sin(\pi \hat{\sigma})}.
        \end{equation}
    \end{itemize}

\end{definition}

Provided $\hat{\sigma} \notin \mathbb{Z}$, (\ref{bessel:wronskian}) implies that $\{I_{\pm \hat{\sigma}}(y)\}$ form a basis of solutions to (\ref{eq:modifiedbessel}). We will largely use this basis for computations, given the series expansion (\ref{eq:model1}). In a neighborhood of $\hat{\sigma}=1$ this basis is no longer valid, however, and we will require certain additional estimates on $K_{\hat{\sigma}}(y)$. The remainder of this section is dedicated to this goal.

\begin{definition}
    The digamma function $\Psi^{(0)}(\sigma)$ is the complex derivative of the logarithm of the gamma function,
    \begin{equation*}
        \Psi^{(0)}(\sigma) \doteq \frac{d}{d\sigma} \ln \Gamma(\sigma) = \frac{\Gamma'(\sigma)}{\Gamma(\sigma)}.
    \end{equation*}
    In $\{\Re \sigma > 0\}$ the digamma function is holomorphic, and satisfies the asymptotic estimate
    \begin{align}
        \label{est:digammafun}
        \Psi^{(0)}(\sigma) &\sim \ln \sigma + O(|\sigma|^{-1}).
    \end{align}
\end{definition}

\vspace{1em}
\noindent
In the following, let $\mathbb{B}_{\frac{1}{2}}(1)$ denote the open ball of radius $\frac12$ in $\mathbb{C}$, centered on $1$.

\begin{lemma}
The modified Bessel function of the second kind $K_{\hat{\sigma}}(y)$ satisfies the estimates 
\begin{align}
    \label{eq:BesselK_estimates}
    \sup_{\hat{\sigma} \in \mathbb{B}_{\frac{1}{2}}(1)} \sup_{y \in (0,1]}\bigg| \big(y \partial_y\big)^j (y^{\hat{\sigma}}K_{\hat{\sigma}})(y)\bigg| \lesssim 1, \ \ (0 \leq j \leq 2).
\end{align}
\end{lemma}
\begin{proof}
The strategy will be to apply the series expansions and (\ref{eq:KvsI_Bessel}) for $\{I_{\pm \hat{\sigma}}\}$, $\hat{\sigma} \in \mathbb{B}_{\frac{1}{2}}(1) \setminus \{1\}$. We may relate $\sin (\pi \hat{\sigma})$ with gamma functions using the reflection formula 
\begin{equation}
    \label{eq:reflectionGamma}
    \sin(\pi \hat{\sigma}) = \frac{\pi}{\Gamma(\hat{\sigma})\Gamma(1-\hat{\sigma})}, \ \ (\hat{\sigma} \notin \mathbb{Z}).
\end{equation}  
This yields
\begin{align*}
\Big(\frac{y}{2}\Big)^{\hat{\sigma}}K_{\hat{\sigma}}(y) &=  \sum_{m=0}^{\infty} \frac{\Gamma(\hat{\sigma})\Gamma(1-\hat{\sigma})}{2m!} \bigg(\frac{1}{\Gamma(m +1 - \hat{\sigma})} - \Big(\frac{y}{2} \Big)^{2\hat{\sigma}}\frac{1}{\Gamma(m+1+\hat{\sigma})} \bigg)\Big(\frac{y}{2} \Big)^{2m} \\[2\jot]
&= \frac{\Gamma(\hat{\sigma})}{2} + \sum_{m=0}^{\infty} \frac{\Gamma(\hat{\sigma})\Gamma(1-\hat{\sigma})}{2m!} \underbrace{\bigg(\Big(\frac{y}{2} \Big)^{2}\frac{1}{(m+1)\Gamma(m +2 - \hat{\sigma})} - \Big(\frac{y}{2} \Big)^{2\hat{\sigma}}\frac{1}{\Gamma(m+1+\hat{\sigma})} \bigg)}_{\text{I}_m(y)}\Big(\frac{y}{2} \Big)^{2m}.
\end{align*}
In order to cancel the simple pole of $\Gamma(1-\hat{\sigma})$, we need to establish that $\text{I}_m(y)$ vanishes linearly in $\hat{\sigma}$ as $\hat{\sigma} \rightarrow 1$, for all $y$. Define 
\begin{equation*}
    f_{m}(\hat{\sigma}) \doteq \Gamma(m+1+\hat{\sigma}) - (m+1)\Gamma(m+2-\hat{\sigma}),
\end{equation*}
which for $m \in \mathbb{Z}_{\geq 0}$ is holomorphic in $\mathbb{B}_{\frac12}(1)$. Compute 
\begin{align}
    |\text{I}_m(y)| &= \bigg| \Big(\frac{y}{2} \Big)^{2}\frac{\Gamma(m+1+\hat{\sigma}) - (m+1)\Gamma(m+2-\hat{\sigma})}{(m+1)\Gamma(m+2-\hat{\sigma})\Gamma(m+1+\hat{\sigma})} + \frac{1}{\Gamma(m+1+\hat{\sigma})}\Big(\Big(\frac{y}{2} \Big)^{2\hat{\sigma}}- \Big(\frac{y}{2} \Big)^{2}\Big)\bigg| \nonumber \\[2\jot]
    &\lesssim \frac{1}{(m+1)}\frac{\sup_{\mathbb{B}_{\frac12}(1)}|f_m'(\hat{\sigma})|}{\Gamma(m+2-\hat{\sigma})\Gamma(m+1+\hat{\sigma})}|1-\hat{\sigma}| + |1-\hat{\sigma}|.\label{4.1:eq:temp1}
\end{align}
We have used the inequality 
\begin{align*}
    \sup_{y\in [0,1]}\sup_{\mathbb{B}_{\frac{1}{2}}(1)}\Big|\Big(\frac{y}{2} \Big)^{2\hat{\sigma}}- \Big(\frac{y}{2} \Big)^{2} \Big| \lesssim |1-\hat{\sigma}|,
\end{align*}
which follows by considering the holomorphic function $g(\hat{\sigma}) = \big(\frac{y}{2} \big)^{2\hat{\sigma}}- \big(\frac{y}{2} \big)^{2}$ for \textit{fixed} $y \in [0,1]$, and applying Taylor's theorem for $\hat{\sigma} \in \mathbb{B}_{\frac{1}{2}}(1)$. It is easily verified that the derivative estimate is uniform in $y$.

To complete the $j=0$ case of (\ref{eq:BesselK_estimates}), we estimate $f_m'(\sigma)$. Differentiating and applying the bound (\ref{est:digammafun}) for the digamma function gives
\begin{align*}
    |f_m'(\sigma)| &= \big|\Gamma(m+1+\hat{\sigma}) \Psi^{(0)}(m+1+\hat{\sigma}) - (m+1)\Gamma(m+2-\hat{\sigma}) \Psi^{(0)}(m+2-\hat{\sigma})\big|\\[2\jot]
    &\lesssim \big|\Gamma(m+1+\hat{\sigma})\big|\big| \ln (m+1+\hat{\sigma})\big| + (m+1)\big|\Gamma(m+2-\hat{\sigma})\big| \big|\ln (m+2-\hat{\sigma})\big|.
\end{align*}
Inserting in (\ref{4.1:eq:temp1}) implies $I_m(y)$ is bounded uniformly in $y, \hat{\sigma},$ and $m$. 

Differentiating in $y$ and running the same argument gives the remaining cases of (\ref{eq:BesselK_estimates}).
\end{proof}

\subsection{Analysis of approximate spectral family $P_k^{(0)}(\sigma)$}
\label{sec:modeleqn}

We begin by defining two families of differential operators, each depending on a complex parameter $\sigma \in \mathbb{C}$. These so-called spectral families play a key role in motivating the construction of a scattering resolvent for solutions to (\ref{eq:1}).

\begin{definition}
For ${\sigma \in \mathbb{C}}$, define the following operators $P_{k}^{(0)}(\sigma)$, $P_k(\sigma)$ acting formally on $C^2_x(\mathbb{R}_{+})$ functions on the half-line:
\begin{align}
    \big(P_{k}^{(0)}(\sigma)f\big)(x) &\doteq -f''(x) + (\underbrace{4q_k \gamma_k}_{\doteq w_k^2} k^2 e^{-2q_k x} -\sigma^2)f(x), \\[\jot]
    \big(P_{k}(\sigma)f\big)(x) &\doteq -f''(x) + (4q_k \gamma_k e^{-2q_k x}V_k(x) -\sigma^2)f(x).
\end{align}
Here, $V_k(x)$ is as in (\ref{eq:wavehyper}), and $\gamma_k$ the constant in Lemma \ref{lem:propertiesofV}. For any $\epsilon > 0$ fixed,
\begin{equation}
    \label{eq:fullop_vsmodelop}
    P_k(\sigma) = P_k^{(0)}(\sigma) + O_{L^\infty}\big(k^2 e^{-4q_k(1-\epsilon)x}\big).
\end{equation}
\end{definition}

\vspace{1em}
\noindent
In this section we develop a complete understanding of the space of solutions to the \textit{approximate spectral equation}
\begin{equation}    
    \label{eq:modeleq}
    \big(P_{k}^{(0)}(\sigma)f\big)(x) = 0,
\end{equation}
as a function of $\sigma$ and asymptotic behavior for $x \rightarrow \infty$. To see the relation between (\ref{eq:modeleq}) and Bessel's equation (\ref{eq:modifiedbessel}) as considered in the previous section, we change coordinates from $x$ to $y \doteq w_k k q_k^{-1} e^{-q_k x}$. A computation gives that $f(y)$ obeys Bessel's equation on the domain $y \in (0,w_k k q_k^{-1}]$ with order
\begin{equation*}
    \hat{\sigma} \doteq i \frac{\sigma}{q_k}.
\end{equation*} 
Provided $\hat{\sigma} \notin \mathbb{Z}$, it follows from the discussion above that natural bases of solutions for (\ref{eq:modeleq}) are given in terms of $\{I_{\hat{\sigma}}, K_{\hat{\sigma}} \}$ or $\{I_{\pm \hat{\sigma}}\}$. We find it convenient to use the latter representation. Translating back to the original coordinates and normalizing gives the following pair of solutions to (\ref{eq:modeleq}):

\begin{lemma}
    Define 
    \begin{equation}
        \label{eq:defof_f_pm}
        f_{+,\sigma}(x) \doteq (w_k k)^{ \hat{\sigma}} I_{-\hat{\sigma}}\Big(\frac{w_k k}{q_k} e^{-q_k x} \Big), \quad f_{-,\sigma}(x) \doteq (w_k k)^{-\hat{\sigma}} I_{\hat{\sigma}}\Big(\frac{w_k k}{q_k} e^{-q_k x} \Big),
    \end{equation}
    where the principal branch of the logarithm is used in the definition of $a^z$, $z \in \mathbb{C}$. For all $\sigma \in \mathbb{C}$ the pair $\{f_{\pm,\sigma}(x)\}$ are solutions to (\ref{eq:modeleq}). Moreover, the solutions enjoy the following properties:
\begin{itemize}
    \item Provided $\sigma \notin \mp i q_k \mathbb{N}$, we have the asymptotic behavior
    \begin{align}
        f_{\pm,\sigma}(x) \sim e^{\pm i\sigma x} \ \textrm{as} \ x \rightarrow \infty,
    \end{align}
    in the sense that there exist non-zero constants $c_{\pm,\sigma}$ such that 
    \begin{align*}
        e^{\mp i \sigma x}f_{\pm,\sigma}(x) = c_{\pm,\sigma} + O(e^{-2q_k x}).
    \end{align*}
    \item The pair $\{f_{\pm,\sigma}(x)\}$ are linearly independent provided $\sigma \notin iq_k \mathbb{Z}$. The Wronskian is independent of $x$, and given by  
    \begin{equation}
        \label{eq:wronskianf+f-}
        W[f_{+,\sigma}, f_{-,\sigma}] = -\frac{2q_k}{\pi}\sin(\pi \hat{\sigma}).
    \end{equation}
\end{itemize}
\end{lemma}

Although $\{f_{\pm,\sigma}(x)\}$ are a priori defined for all $\sigma \in \mathbb{C}$, we will be interested in their behavior restricted to a neighborhood of the real axis. For $a_1, a_2 \in \mathbb{R}$ with $a_1 < a_2$, define  
\begin{align*}
    \mathbb{I}_{[a_1,a_2]} \doteq \{z \in \mathbb{C}: \Im z \in [a_1,a_2]\},
\end{align*}
and similarly $\mathbb{I}_{(a_1,a_2)}$. 

In the remainder of the section we collect estimates on the pair $\{f_{\pm,\sigma}(x)\}$. The series expansion (\ref{eq:model1}) is our main tool for deriving bounds that track the joint behavior in $(x,\sigma)$, and $k$. 

\begin{lemma}
    Fix a parameter $\eta \in (0,\frac12)$. There exist functions $g_{\pm,\sigma}(x)$ smooth in $(\sigma,x)$, holomorphic in $\sigma$, and with $C^2_x(\mathbb{R}_{+})$ norm bounded uniformly in $(\sigma,x,k) \in \mathbb{I}_{[-1-\eta,\eta]} \times \mathbb{R}_{+} \times [0,\epsilon)$ for $\epsilon$ sufficiently small, such that 
        \begin{align}
            \label{f+sigmaexpansion}
            e^{-i\sigma x}f_{+,\sigma}(x) &= (2q_k)^{\hat{\sigma}}\frac{1}{\Gamma(2 - \hat{\sigma})}\Big(1 - \hat{\sigma} + w_k^2 k^2 \Big(\frac{1}{2q_k} \Big)^{2} e^{-2 q_k x} + w_k^4 k^4 g_{+,\sigma}(x) e^{-4 q_k x}\Big), \\[3\jot]
            \label{f-sigmaexpansion}
            e^{i\sigma x}f_{-,\sigma}(x) &= (2q_k)^{-\hat{\sigma}}\frac{1}{\Gamma(2 + \hat{\sigma})}\Big(1 + \hat{\sigma} + w_k^2 k^2 \Big(\frac{1}{2q_k} \Big)^{2} e^{-2 q_k x} + w_k^4 k^4 g_{-,\sigma}(x) e^{-4 q_k x}\Big),
        \end{align}
        holds for $(\sigma,x,k) \in \mathbb{I}_{[-1-\eta,\eta]} \times \mathbb{R}_{+} \times [0,\epsilon)$.
    
    \end{lemma}
    \begin{proof}
        We show the statement for $f_{+,\sigma}(x).$ By the series expansion (\ref{eq:model1}) for $I_{\hat{\sigma}}$, it follows that 
        \begin{align}
            e^{-i\sigma x}f_{+,\sigma}(x) &= \sum_{m=0}^{\infty}\frac{(w_k k)^{2m}}{m!\Gamma(m + 1 - \hat{\sigma})}\Big(\frac{1}{2q_k} \Big)^{2m -\hat{\sigma}} e^{-2m q_1 x}  \nonumber\\[2\jot]
            &= (2q_k)^{\hat{\sigma}} \frac{1}{\Gamma(1-\hat{\sigma})}\Big(1 +  \frac{w_k^2 k^2}{1-\hat{\sigma}}(2q_k)^{-2} e^{-2q_1 x} +  \frac{w_k^4 k^4}{1-\hat{\sigma}} g_{+,\sigma}(x) e^{-4q_k x} \Big),\label{eq:lemma4.3temp1}
        \end{align}
        for a function $g_{+,\sigma}(x)$ satisfying all the stated properties. Note $g_{+,\sigma}(x)$ has poles for $\sigma \in \{-2,-3,-4,\ldots\}$, but by our choice of $\eta$ is uniformly bounded in the region $\mathbb{I}_{[-1-\eta,\eta]}$. To arrive now at (\ref{f+sigmaexpansion}), it remains to factor out $(1-\hat{\sigma})^{-1}$ in (\ref{eq:lemma4.3temp1}) and apply the multiplication identities for gamma functions.
        \end{proof}

\begin{rmk}
    \label{rmk:etafixed}
    In the remainder of the paper, we will assume a parameter $\eta \in (0, \frac12)$ has been fixed. For such a choice, and $k$ sufficiently small, we have $\mathbb{I}_{[-1-\eta,\eta]} \cap iq_k \mathbb{Z} = \{0\}\cup \{-iq_k\}.$  
\end{rmk}

The next lemma considers an integral kernel $G_{\sigma}(x,x')$ built out of $\{f_{\pm,\sigma}(x)\}$. This will be required in the Volterra iteration in the following section. 

\begin{lemma}
    For $\sigma \notin i q_k \mathbb{Z}$, define  
\begin{equation}
    G_\sigma(x,x') \doteq \frac{f_{-,\sigma}(x')f_{+,\sigma}(x) - f_{-,\sigma}(x)f_{+,\sigma}(x')}{W[f_{-,\sigma}, f_{+,\sigma}]}.
\end{equation}
For fixed $x,x' \geq 0$, $G_\sigma(x,x')$ extends to an analytic function of $\sigma$ in $\mathbb{I}_{[-1-\eta,\eta]}$. For fixed $x'>0$ and $\eta \in (0,\frac12)$, we additionally have the estimates 
\begin{equation}
    \label{eq:model2}
   \sup_{\sigma \in \mathbb{I}_{[-1-\eta,\eta]}} \sup_{x \in [0,x']} \bigg| \frac{1}{1+x'}e^{-|\Im \sigma| x' } G_\sigma(x,x') \bigg| \lesssim 1,
\end{equation}
\begin{equation}
    \label{eq:model2.5}
   \sup_{\sigma \in \mathbb{I}_{[-1-\eta,\eta]}} \sup_{x \in [0,x']} \bigg| \bigg(\frac{1}{1+|\sigma|}\bigg)^j \frac{1}{1+x'}e^{-|\Im \sigma| x' } \partial_{x}^j G_\sigma(x,x') \bigg| \lesssim 1,  \quad (1 \leq j \leq 2).
\end{equation}

\noindent
The $x$-dependence of these estimates may be improved at the cost of uniformity in $\sigma$. There exists a locally bounded function $c(\sigma)$ defined for $\sigma \neq 0$, such that 
\begin{equation}
    \label{eq:model2.6}
    \sup_{x \in [0,x']} \bigg|  \partial_x^{j} G_\sigma(x,x') \bigg| \lesssim c(\sigma) e^{|\Im \sigma| x' }, \quad (0 \leq j \leq 2).
\end{equation} 
\end{lemma}
\begin{proof}
To verify analyticity, it suffices by (\ref{eq:wronskianf+f-}) to check that $G_\sigma(x,x')$ has removable singularities at $\sigma \in \{0\} \cup \{-i q_k \}$ for fixed $(x,x')$. Near either point, employ (\ref{eq:KvsI_Bessel}) to write the kernel as  
\begin{align}
    G_\sigma(x,x') &= \frac{\pi}{2q_k \sin(\pi \hat{\sigma})}\bigg(I_{-\hat{\sigma}}\Big(\frac{w_k k}{q_k}e^{-q_k x} \Big) I_{\hat{\sigma}}\Big(\frac{w_k k}{q_k}e^{-q_k x'} \Big) -  I_{-\hat{\sigma}}\Big(\frac{w_k k}{q_k}e^{-q_k x'} \Big) I_{\hat{\sigma}}\Big(\frac{w_k k}{q_k}e^{-q_k x} \Big)     \bigg)\nonumber \\[2\jot]
    &= \frac{1}{q_k}\bigg(K_{\hat{\sigma}}\Big(\frac{w_k k}{q_k}e^{-q_k x} \Big) I_{\hat{\sigma}}\Big(\frac{w_k k}{q_k}e^{-q_k x'} \Big) - I_{\hat{\sigma}}\Big(\frac{w_k k}{q_k}e^{-q_k x} \Big)K_{\hat{\sigma}}\Big(\frac{w_k k}{q_k}e^{-q_k x'} \Big)    \bigg).\label{4.2:eq:temp1}
\end{align} 
The analyticity of $I_{\hat{\sigma}}, K_{\hat{\sigma}}$ implies $G_{\sigma}(x,x')$ extends to $\sigma \in \{0\} \cup \{-i q_k \}$ via this formula as an analytic function.

To estimate $G_\sigma(x,x')$, it is convenient to subdivide $\mathbb{I}_{[-1-\eta,\eta]}$ into two regions:
\begin{align*}
    R_1 =\mathbb{B}_{\frac12}(-iq_k), \ \ R_2 = \mathbb{I}_{[-1-\eta,\eta]} \setminus R_1.
\end{align*}
In $R_1$ we appeal to the regular expression (\ref{4.2:eq:temp1}). The series expansion (\ref{eq:model1}) for $I_{\hat{\sigma}}$ implies
\begin{align*}
    \sup_{x \in [0,x']} \Big| \partial_x^j I_{\hat{\sigma}} \Big(\frac{w_k k}{q_k}e^{-q_k x} \Big) \Big| \lesssim \Big| \Big(\frac{w_k k}{q_k} \Big)^{\hat{\sigma}}\Big| .
\end{align*}
The bound (\ref{eq:BesselK_estimates}) for $K_\sigma$ similarly implies 
\begin{align*}
     \sup_{x \in [0,x']} \Big| \partial_x^j K_{\hat{\sigma}} \Big(\frac{w_k k}{q_k}e^{-q_k x} \Big) \Big| \lesssim \Big| \Big(\frac{w_k k}{q_k} \Big)^{\hat{\sigma}}\Big| e^{q_k |\Im \sigma| x'} .
 \end{align*}
 The stated bounds (\ref{eq:model2})--(\ref{eq:model2.6}) in $R_1$ now directly follow by differentiating (\ref{4.2:eq:temp1}).

 \vspace{1em}
 \noindent
 In the complement we apply the expansions (\ref{f+sigmaexpansion})--(\ref{f-sigmaexpansion}) to give 
\begin{align}
    \label{4.2:eq:temp2}
     G_\sigma(x,x') &= \big(e^{-i \sigma x}e^{i \sigma x'} - e^{-i \sigma x'}e^{i \sigma x} \big)\frac{\pi\big(1-\hat{\sigma} + O_{\sigma,x,x'}(1) \big)\big(1+\hat{\sigma} + O_{\sigma,x,x'}(1) \big)}{2q_k \sin(\pi \hat{\sigma})\Gamma(2-\hat{\sigma})\Gamma(2+\hat{\sigma})},
\end{align}
where $ O_{\sigma,x,x'}(1)$ denotes terms that are bounded uniformly in $(\sigma, x ,x')$ for $\sigma \in \mathbb{I}_{[-1-\eta,\eta]}.$ A similar expression holds for the derivatives $\partial_x^j G_\sigma(x,x'), j=1,2.$ 

The ratio appearing in (\ref{4.2:eq:temp2}) is not defined when $\hat{\sigma} = 0$, and so we must look to the first factor for additional vanishing. In $\{\Im \sigma \geq 0\}$, rewrite this factor as  
\begin{align}
    e^{-i \sigma x}e^{i \sigma x'} - e^{-i \sigma x'}e^{i \sigma x} &= e^{i\sigma x}\big(e^{i \sigma x'} - e^{-i \sigma x'} \big) + e^{i\sigma x'}\big (e^{-i\sigma x} - e^{i \sigma x} \big) \nonumber \\
    &= 2i \big( e^{i\sigma x}\sin(\sigma x') - e^{i \sigma x'}\sin(\sigma x) \big). \label{eq:model3}
\end{align}
In the upper half plane the functions $|e^{i\sigma x}|, |e^{i \sigma x'}|$ are \textit{decreasing} with respect to $x,x'$, and bounded for $x,x' \geq 0$. Using the estimate 
\begin{equation}
    \Big| \frac{\sin(\sigma x) }{\sigma}\Big| \lesssim x e^{|\Im \sigma| x},
\end{equation} 
we find 
\begin{align*}
   \sup_{x \in [0,x']} |e^{-i \sigma x}e^{i \sigma x'} - e^{-i \sigma x'}e^{i \sigma x}| \lesssim |\sigma| x' e^{|\Im \sigma| x'}.
\end{align*}
We have thus gained a factor consistent with vanishing as $\sigma \rightarrow 0$, at the cost of a small polynomial loss in $x$. A similar argument works in the lower half plane, writing (\ref{eq:model3}) in terms of exponentials $e^{-i \sigma x}, e^{-i \sigma x'}$ which decay in $x$ for $\sigma$ lying below the real axis. 

We may estimate in $R_2$ to give 
\begin{align*}
    \sup_{x \in [0,x']}|G_\sigma(x,x')| &\lesssim \Big|\frac{\big(1-\hat{\sigma} + O_{\sigma,x,x'}(1) \big)\big(1+\hat{\sigma} + O_{\sigma,x,x'}(1) \big)}{ \Gamma(2-\hat{\sigma})\Gamma(2+\hat{\sigma})}\sigma \Gamma(\hat{\sigma})\Gamma(1-\hat{\sigma})\Big| x' e^{|\Im \sigma|x'} \\
    &\lesssim x' e^{|\Im \sigma|x'}.
\end{align*}
The differentiated estimates follow similarly, with additional powers of $\sigma$ arising from derivatives falling on the exponential factors $e^{\pm i \sigma x}$. We note also that the polynomial loss in $x'$ is associated only to a neighborhood of $\sigma = 0$, and may be disregarded away from the origin to give (\ref{eq:model2.6}).
\end{proof}

\subsection{Construction of $R(\sigma)$}
\label{sec:resolventconstr}
In this section we apply our understanding of solutions to (\ref{eq:modeleq}) to construct families of solutions to $P_k(\sigma)f = 0$. By (\ref{eq:fullop_vsmodelop}) the two spectral families $P_k^{(0)}(\sigma), P_k(\sigma)$ differ by a potential term that is rapidly decaying at infinity. We may therefore hope to construct solutions to the desired problem perturbatively, by solving 
\begin{equation}
    \label{eq:fulleqn}
    P_k^{(0)}(\sigma)f = -4q_k k^2 E_k(x)e^{-4q_k(1-\epsilon)} f,
\end{equation} 
where $\epsilon > 0$ is a fixed small parameter and $E_k(x)$ the function defined in Proposition \ref{lem:propertiesofV}. The main result of this section is that an outgoing family of solutions $f_{(out)}(\sigma)$ to (\ref{eq:fulleqn}) may indeed be constructed this way for $\sigma \in \mathbb{I}_{[-1-\eta,\eta]}$, leading to the definition of the scattering resolvent in Definition \ref{dfn:scatteringresolv} below.

We give a construction of $C^2_x(\mathbb{R}_+)$ solutions to (\ref{eq:fulleqn}) with prescribed asymptotic structure by solving the associated Volterra integral equation. The following result is adapted from \cite{wave_schrod}.

\begin{prop}
    \label{prop:existencevolterraperturb}
    For $\sigma \in \mathbb{I}_{[-1-\eta,\eta]},$ there exists a unique solution $f_{(out),\sigma}(x)$ to the Volterra equation
    \begin{equation}
        \label{perturb:volt1}
        f_{(out),\sigma}(x) = f_{+,\sigma}(x) + 4q_k k^2  \int_x^{\infty} G_\sigma(x,x')  E_k(x')e^{-4q_k(1-\epsilon)  x'}f_{(out),\sigma}(x')dx'.
    \end{equation}
Similarly, there exists a unique solution $f_{(in),\sigma}(x)$ to the Volterra equation 
\begin{equation}
        \label{perturb:volt2}
        f_{(in),\sigma}(x) = f_{-,\sigma}(x) + 4q_k k^2\int_x^{\infty} G_\sigma(x,x')  E_k(x')e^{-4q_k(1-\epsilon) x'}f_{(in),\sigma}(x')dx',
    \end{equation}
    The pair $f_{(out),\sigma}(x), f_{(in),\sigma}(x)$ are $C^2_x(\mathbb{R}_+)$ solutions to (\ref{eq:fulleqn}), and are linearly independent provided  $\sigma \in \mathbb{I}_{[-1-\eta,\eta]} \setminus \{0, -i q_k\}$. Moreover, for fixed $x \geq 0$, $f_{(out),\sigma}(x)$ is a holomorphic function of $\sigma \in \mathbb{I}_{[-1-\eta,\eta]}$.
\end{prop}
\begin{proof}
To simplify notation, define the kernel 
\begin{equation*}
    P_\sigma(x,x') \doteq 4q_k k^2 G_\sigma(x,x')E_k(x')e^{-4q_k(1-\epsilon) x'}.
\end{equation*} 
We will exhibit the solution to (\ref{perturb:volt1}) as a pointwise convergent sum 
\begin{equation}
    f_{(out),\sigma}(x) = f_{+,\sigma}(x) + \sum_{n=1}^{\infty}M_{n,\sigma}(x),
\end{equation}
where 
\begin{equation*}
    M_{n,\sigma}(x) = \underbrace{\int_{x}^{\infty} \int_{x}^{\infty} \ldots \int_{x}^{\infty} }_{n} \Big( \prod_{i=1}^n \chi_{\{x_i \geq x_{i-1}\}}P_\sigma(x_{i-1},x_{i})\Big) f_{+,\sigma}(x_n) dx_n dx_{n-1}\ldots dx_1.
\end{equation*}
Here, $\chi_{U}(x_1,\ldots,x_n)$ is the characteristic function of a set $U$, and by convention we set $x_0 \doteq x$. The goal will be to show convergence of $\displaystyle \sum_{n=1}^{\infty}M_{n,\sigma}(x)$ in $L^\infty(\mathbb{R}_{+})$, uniformly for $\sigma \in \mathbb{I}_{[-1-\eta,\eta]}$.  By the estimate (\ref{eq:model2}) for $G_\sigma$, and the $L^\infty$ bound on $E_k$ contained in Proposition \ref{lem:propertiesofV}, it follows that 
\begin{align*}
    \sup_{\sigma \in \mathbb{I}_{[-1-\eta,\eta]}} \sup_{x' \in [x,x'']}\big| P_\sigma(x',x'')\big| &\lesssim k^2 (1+x'') e^{-4q_k(1-\epsilon) x''} e^{|\Im \sigma|x''}.
\end{align*} 
At the cost of increasing $\epsilon$ by an arbitrarily small amount, we may drop the linear dependence on $x''$ in the above estimate. Therefore,
\begin{align*}
    |M_{n,\sigma}&(x)| \\
    &\lesssim  \|e^{\Im \sigma x}f_{+,\sigma} \|_{L^\infty(\mathbb{R}_{+})}\underbrace{\int_{x}^{\infty} \int_{x}^{\infty} \ldots \int_{x}^{\infty} }_{n} \Big( \prod_{i=1}^n \chi_{\{x_i \geq x_{i-1}\}}P_\sigma(x_{i-1},x_{i})\Big) e^{-\Im \sigma x_n} dx_n dx_{n-1}\ldots dx_1 \\[\jot]
    &\lesssim \frac{k^2}{(2q_k(1-\epsilon) - |\Im \sigma|)}\|e^{\Im \sigma x}f_{+,\sigma} \|_{L^\infty(\mathbb{R}_{+})} e^{-4q_k(1-\epsilon)x}e^{|\Im \sigma|x}e^{- \Im \sigma x} \\[\jot]
    &\hspace{14em} \underbrace{\int_{x}^{\infty} \int_{x}^{\infty} \ldots \int_{x}^{\infty} }_{n-1} \Big( \prod_{i=1}^n \chi_{\{x_i \geq x_{i-1}\}}P_\sigma(x_{i-1},x_{i})\Big) dx_{n-1}\ldots dx_1 \\[\jot]
    &\lesssim \frac{1}{(n-1)!}\frac{k^{2n}}{(2q_k(1-\epsilon) - |\Im \sigma|)}\|e^{\Im \sigma x}f_{+,\sigma} \|_{L^\infty(\mathbb{R}_{+})} e^{-4q_k(1-\epsilon)x}e^{|\Im \sigma|x}e^{- \Im \sigma x} \\[\jot]
    &\hspace{15.7em} \underbrace{\int_{x}^{\infty} \int_{x}^{\infty} \ldots \int_{x}^{\infty} }_{n-1} \Big( \prod_{i=1}^n e^{-4q_k(1-\epsilon)x_i}e^{|\Im \sigma|x_i}\Big) dx_{n-1}\ldots dx_1 \\[\jot]
    &\lesssim \frac{1}{(n-1)!}\frac{k^{2n}}{(2q_k(1-\epsilon) - |\Im \sigma|)}\frac{1}{\big(4q_k(1-\epsilon) - |\Im \sigma| \big)^{n-1}}\|e^{\Im \sigma x}f_{+,\sigma} \|_{L^\infty(\mathbb{R}_{+})} e^{-n\big(4q_k(1-\epsilon) - |\Im \sigma| \big)x}e^{- \Im \sigma x}.
\end{align*}
By the M-test, we conclude that the sum converges uniformly for $x \in  \mathbb{R}_{+}$, and locally uniformly for $\sigma \in \mathbb{I}_{[-1-\eta,\eta]}$ to a function $\sum_{n=1}^{\infty}M_{n,\sigma}(x) \in L^\infty(\mathbb{R}_{+})$ satisfying the tail bound
    \begin{equation}
        \label{eq:tailbound1}
        \bigg| \sum_{n=N}^{\infty}M_{n,\sigma}(x)\bigg| \lesssim k^{2N} \|e^{\Im \sigma x}f_{+,\sigma} \|_{L^\infty(\mathbb{R}_+)} e^{-\Im \sigma x } e^{-N(4q_k(1-\epsilon) - |\Im \sigma| )x}.
    \end{equation}
For fixed $x$ it is clear this sum defines a holomorphic function of $\sigma$. The construction of $f_{(in),\sigma}(x)$ is analogous, and we have an expansion 
\begin{equation*}
    f_{(in),\sigma}(x) = f_{-,\sigma}(x) + \sum_{n=1}^{\infty}L_{n,\sigma}(x),
\end{equation*}
where 
\begin{equation*}
    |L_{n,\sigma}(x)| \lesssim\frac{1}{(n-1)!}\frac{k^{2n}}{(2q_k(1-\epsilon)-|\Im \sigma|)}\frac{1}{(4q_k(1-\epsilon) -|\Im \sigma|)^{n-1}}\|e^{-\Im \sigma x}f_{-,\sigma} \|_{L^\infty(\mathbb{R}_+)}  e^{-n\big(4q_k(1-\epsilon) - |\Im \sigma| \big)x}e^{ \Im \sigma x }.
\end{equation*}

We conclude by arguing that $f_{(out),\sigma}(x), f_{(in),\sigma}(x)$ are indeed solutions to (\ref{eq:fulleqn}), and compute their Wronskian. By (\ref{4.2:eq:temp1}) it follows that the kernel $G_\sigma(x,x')$ is $C^2_{x,x'}$ for all $\sigma \in \mathbb{I}_{[-1-\eta,\eta]}$. Directly applying the operator $P_k^{(0)}(\sigma)$ to the integral equations (\ref{perturb:volt1})--(\ref{perturb:volt2}) and integrating by parts, we conclude that $f_{(out),\sigma}(x), f_{(in),\sigma}(x)$ are classical solutions to (\ref{eq:fulleqn}).

By the bounds established for $M_{n,\sigma}, L_{n,\sigma}$, we have for $\sigma \in \mathbb{I}_{[-1-\eta,\eta]}$,
    \begin{align}
        f_{(out),\sigma}(x) &= f_{+,\sigma}(x) + e^{-\Im \sigma x} O_{L^{\infty}}(e^{-2q_k(1-\epsilon) x}), \label{eq:fout_exp_1}\\ 
        f_{(in),\sigma}(x) &= f_{-,\sigma}(x) + e^{\Im \sigma x} O_{L^{\infty}}(e^{-2q_k(1-\epsilon)x}).\label{eq:fout_exp_2}
    \end{align}
    The same expansion holds for $f'_{(out),\sigma}(x), f'_{(in),\sigma}(x)$, as can be seen by considering the Volterra equation for these differentiated quantities and using (\ref{eq:model2.5}). Note the $O(\cdot)$ terms may depend on $\sigma$ here. It follows that 
    \begin{align*}
        W[f_{(out),\sigma}, f_{(in),\sigma}] &= \lim_{x \rightarrow \infty }\big(f_{(out),\sigma}(x) f'_{(in),\sigma}(x) - f'_{(out),\sigma}(x) f_{(in),\sigma}(x) \big)\\
        &= \lim_{x \rightarrow \infty } \big( f_{+,\sigma}(x) f'_{-,\sigma}(x) - f'_{+,\sigma}(x) f_{-,\sigma}(x) + O(e^{-2q_k(1-\epsilon) x}) \big) \\
        &= W[f_{+,\sigma}, f_{-,\sigma}].
    \end{align*}
    By (\ref{eq:wronskianf+f-}), it follows $f_{(out),\sigma}(x),f_{(in),\sigma}(x)$ are linearly independent away from $ \{0\} \cup \{-i q_k\}$.
\end{proof}

For $\sigma \notin \{0\} \cup \{-iq_k\}$, we refer to $f_{(out),\sigma}(x), f_{(in),\sigma}(x)$ as \textbf{outgoing solutions} and \textbf{ingoing solutions} respectively. These solutions model asymptotically free scattering states as $x \rightarrow \infty$, and are characterized by their asymptotic behavior, termed \textbf{outgoing boundary conditions} and \textbf{ingoing boundary conditions} respectively:
\begin{align*}
    f_{(out),\sigma}(x) \sim e^{i \sigma x} + O(e^{i\sigma x -2q_k x}), \quad f_{(in),\sigma}(x) \sim e^{-i \sigma x}+ O(e^{-i\sigma x -2q_k x}).
\end{align*}
We may also define a solution to $P_k(\sigma)f=0$ associated to \textbf{Dirichlet boundary conditions} $f(0)=0$, which we in turn label the \textbf{Dirichlet solution} $f_{(dir),\sigma}(x)$. The following definition makes this precise.

\begin{definition}
    \label{def:dirichletsoln}
    For any $\sigma \in \mathbb{C}$, the Dirichlet solution $f_{(dir),\sigma}(x)$ is the unique solution to 
    \begin{equation}
        \label{eq:dirichlet1}
        \begin{cases}
            P_k(\sigma)f_{(dir),\sigma} = 0, \\[\jot]
            (f_{(dir),\sigma}(0), f'_{(dir),\sigma}(0)) = (0,1).
        \end{cases}
    \end{equation}
    Equivalently, $f_{(dir),\sigma}(x)$ is the unique solution to the Volterra integral equation (for $\sigma \neq 0$)
    \begin{equation}
        \label{eq:voltdirichlet1}
        f_{(dir),\sigma}(x) = \frac{\sin(\sigma x)}{\sigma}  +  \int_0^{x} m_\sigma(x,x') (4q_k V_k(x') e^{-2q_k x'})  f_{(dir),\sigma}(x')dx',
    \end{equation} 
    where 
    \begin{align*}
        m_\sigma(x,x') = i \frac{e^{i\sigma x'} e^{-i \sigma x} - e^{i \sigma x} e^{-i \sigma x'}}{\sigma},
    \end{align*}
    and $V_k(x)$ is the potential defined in Proposition \ref{lem:propertiesofV}.
    For $\sigma = 0$, $f_{(dir),\sigma}(x)$ solves
    \begin{equation}
        \label{eq:voltdirichlet2}
        f_{(dir),0}(x) = x  +  \int_0^{x} (x-x') (4q_k V_k(x') e^{-2q_k x'})  f_{(dir),0}(x')dx'.
    \end{equation}
\end{definition}

As a consequqence of the Volterra equations (\ref{perturb:volt1}), (\ref{eq:voltdirichlet1}), (\ref{eq:voltdirichlet2}) and kernel estimates, we may derive estimates for the outgoing and Dirichlet solutions that are uniform in $(\sigma,x,k)  \in   \mathbb{I}_{[-1-\eta,\eta]}\times \mathbb{R}_{+} \times [0,\epsilon)$. The following lemma collects the relevant estimates. Observe that the bounds for $f_{(dir),\sigma}$ gain a \textit{decaying} power of $\sigma$ as $|\sigma| \rightarrow \infty$, a consequence of the $\sigma^{-1}$ dependence of the leading order term in (\ref{eq:voltdirichlet1}).

\begin{lemma}
    \label{lem:recordedbounds_outgoingdirichlet}
    The outgoing solution $f_{(out),\sigma}$ satisfies the following bound uniformly for $(\sigma,x,k) \in \mathbb{I}_{[-1-\eta,\eta]} \times \mathbb{R}_+ \times [0,\epsilon)$:
 \begin{equation}
    \label{eq:outgoinghighenergy}
    \bigg|\frac{d^j}{dx^j} f_{(out),\sigma}(x)\bigg| \lesssim(1+|\sigma|)^{j}\bigg(1 + \frac{1}{\Gamma(1-\hat{\sigma})}\bigg) (1+x) e^{(-\Im \sigma) x}, \quad (0 \leq j \leq 2).
 \end{equation}
Moreover, there exists a locally bounded function $c(\sigma)$ defined for $\sigma \neq 0$ such that 
 \begin{equation}
    \label{eq:outgoinggrowthbd}
    \bigg|\frac{d^j}{dx^j}f_{(out),\sigma}(x)\bigg| \lesssim c(\sigma) e^{(-\Im \sigma) x}, \quad (0 \leq j \leq 2).
 \end{equation}
 Similarly, the Dirichlet solution $f_{(dir),\sigma}$ satisfies the following bound uniformly in $(\sigma,x)$ for $\sigma \in \mathbb{I}_{[-1-\eta,\eta]} \times \mathbb{R}_+ \times [0,\epsilon)$:
 \begin{equation}
     \label{eq:dirichlethighenergy}
     \bigg|\frac{d^j}{dx^j} f_{(dir),\sigma}(x)\bigg| \lesssim (1+|\sigma|)^{-1+j} (1+x) e^{|\Im \sigma| x}, \quad (0 \leq j \leq 2).
  \end{equation}
Moreover, there exists a locally bounded function $c(\sigma)$ defined for $\sigma \neq 0$ such that 
\begin{equation}
 \label{eq:dirichletgrowthbd}
 \bigg|\frac{d^j}{dx^j}f_{(dir),\sigma}(x)\bigg| \lesssim c(\sigma) e^{|\Im \sigma| x}, \quad (0 \leq j \leq 2).
\end{equation}
\end{lemma}

\vspace{.5cm}
We next introduce function spaces encoding regularity and asymptotic decay for functions $f(x): [0,\infty) \rightarrow \mathbb{C}$.
\begin{definition}
    Let $\beta \in (1,2)$, $l \in \mathbb{N}$ be given parameters. Define 
    \begin{align}
        \mathcal{D}^{\beta,l}_{\infty}(\mathbb{R}_+) \doteq &\{f(x): [0,\infty)\rightarrow \mathbb{C}\mid  f(x) \in C^{l}_x(\mathbb{R}_+), \  e^{\beta q_k x}\frac{d^j}{dx^j}f(x) \in L^\infty(\mathbb{R_+}) \ \textrm{for} \ 0 \leq j \leq l \}, \label{defn:Dinftyspaces}
    \end{align}
    with the associated norm 
    \begin{equation}
        \| f \|_{\mathcal{D}^{\beta,l}_{\infty}(\mathbb{R}_+)} \doteq \sum_{j=0}^{l} \big\|e^{\beta q_k x}\frac{d^j}{dx^j}f(x) \big\|_{L^\infty(\mathbb{R}_+)}.
    \end{equation}
    For $\delta > 0$ and $c_0 \in \mathbb{R}$, introduce the space 
    \begin{align}
        \wt{\mathcal{D}}^{\beta,l,\delta,c_0}_{\infty}(\mathbb{R}_+) \doteq &\{f(x): [0,\infty)\rightarrow \mathbb{C}\mid  f(x) \in C^{l}(\mathbb{R}_+), \ f(x) - c_0 e^{- \beta q_k x} \in \mathcal{D}^{\beta+\delta,l}_\infty(\mathbb{R}_+)\},\label{defn:Dtildeinftyspaces}
    \end{align}
    with the associated norm 
    \begin{equation}
        \| f \|_{\wt{\mathcal{D}}^{\beta,l,\delta,c_0}_{\infty}(\mathbb{R}_+)} \doteq  |c_0| + \|f - c_0 e^{-x} \|_{\mathcal{D}^{\beta + \delta,l}_{\infty}(\mathbb{R}_+)}.
    \end{equation}
    Observe $\wt{\mathcal{D}}^{\beta,l,\delta,c_0}_{\infty}(\mathbb{R}_+) \subset \mathcal{D}^{\beta,l}_{\infty}(\mathbb{R}_+)$.
\end{definition}

We next define the resolvent operator $R(\sigma)$, which serves as a right inverse to $P_k(\sigma)$ on $\mathcal{D}^{\beta,l}_{\infty}(\mathbb{R}_+)$ spaces, for $\sigma$ in an appropriate complex strip (modulo a discrete set). A priori, the domain of $\sigma$ for which $R(\sigma)f$ is defined depends heavily on the assumed decay of $f$, i.e. on the value $\beta$ of the domain space. The larger $\beta$, the further \textit{down} into the complex plane the operator $R(\sigma)$ will extend.

\begin{definition}
    \label{dfn:scatteringresolv}
For fixed $\sigma \in \mathbb{I}_{[-1-\eta,\eta]}$, define the Wronskian
\begin{equation}
    \mathcal{W}(\sigma) \doteq W[f_{(dir),\sigma}, f_{(out),\sigma}] = - f_{(out),\sigma}(0)
\end{equation}
of the Dirichlet and outgoing solutions. The second equality follows by explicit evaluation of the Wronskian at ${x=0}$. Observe $\mathcal{W}(\sigma)$ is a holomorphic function of $\sigma \in \mathbb{I}_{[-1-\eta,\eta]}$ by Proposition \ref{prop:existencevolterraperturb}.

Let $\beta \in (1,2)$ be a given parameter, and define $\beta_* \doteq \min(\beta q_k, 1+\eta)$. Assume $\sigma \in \mathbb{I}_{(-\beta_*, \eta]}$ satisfies $\mathcal{W}(\sigma) \neq 0$. Define the \textbf{scattering resolvent} $R(\sigma): \mathcal{D}^{\beta,l}_\infty(\mathbb{R}_+)\rightarrow L^2_{loc}(\mathbb{R}_+)$ as an operator with integral kernel 
\begin{equation}
    R_\sigma(x,x') = \begin{cases}
        \frac{f_{(dir),\sigma}(x)f_{(out),\sigma}(x') }{\mathcal{W}(\sigma)}, & x < x' \\[3\jot]
        \frac{f_{(dir),\sigma}(x')f_{(out),\sigma}(x) }{\mathcal{W}(\sigma)}, & x > x'.
    \end{cases}
    \label{eq:resolvkernel}
\end{equation} 
This operator is well defined by (\ref{eq:outgoinghighenergy}), (\ref{eq:dirichlethighenergy}). 

\vspace{1em}
\noindent
For ${x_0 > 0}$, let $\rho_{x_0}(x)$ denote a smooth, decreasing cutoff function identically equal to $1$ for $x \leq x_0$, and identically equal to $0$ for $x \geq x_0 + 1$. Define the \textbf{cutoff resolvent} $\rho_{x_0}R(\sigma): \mathcal{D}_\infty^{\beta,l}(\mathbb{R}_+) \rightarrow L^2_x(\mathbb{R}_+)$
as an operator with integral kernel $ \rho_{x_0}(x)R_\sigma(x,x').$ For $h \in \mathcal{D}^{\beta,l}_\infty(\mathbb{R}_+)$ and fixed $x, x_0 \geq 0$, the function $(\rho_{x_0}R(\sigma) h)(x)$ is a meromorphic function of $\sigma \in \mathbb{I}_{(-\beta_*,\eta]},$ with poles at the zeroes of $\mathcal{W}(\sigma)$ lying in $\mathbb{I}_{(-\beta_*,\eta]}$. 

\vspace{1em}
\noindent
Finally, label any point ${\sigma_* \in \mathbb{I}_{[-1-\eta, \eta]}}$ with $\mathcal{W}(\sigma_*) = 0$ a \textbf{scattering resonance}, and the associated $f_{(dir),\sigma_*}(x)$ a \textbf{resonance function}. The order of vanishing is the \textbf{multiplicity} of the resonance.
\end{definition}

\subsection{Analytic properties of $R(\sigma)$}
\label{sec:resolventconstr_meromorphic}
Decay of solutions to the wave equation (\ref{eq:wavehyper}) is intimately related to the analytic properties of the resolvent (on appropriate spaces). In this section we study a) the location and multiplicity of scattering resonances, and b) the large $|\sigma|$ behavior of $R_\sigma(x,x')$, otherwise known as high-energy estimates. 

We first establish the existence of a large region free of scattering resonances, provided $k$ is sufficiently small.

\begin{lemma}
    \label{lem:polefreeregion}
    There exists $k$ sufficiently small such that for $\sigma \in \mathbb{I}_{[0,\eta]} \cup \Big( \{\Re \sigma \geq 1\} \cap \mathbb{I}_{[-1-\eta,\eta]} \Big),$ we have 
    \begin{equation}
        \label{eq:wronskianlowerbdfarregion}
        |\mathcal{W}(\sigma)| \gtrsim \Big|\frac{1}{\Gamma(1-\hat{\sigma})}\Big|.
    \end{equation} 
    In particular, the region $\mathbb{I}_{[0,\eta]} \cup \Big( \{\Re \sigma \geq 1\} \cap \mathbb{I}_{[-1-\eta,\eta]} \Big)$ is free of scattering resonances.
    \end{lemma}
    \begin{proof}
    Given $\mathcal{W}(\sigma) = - f_{(out),\sigma}(0)$, it suffices to control the value of $f_{(out),\sigma}(0)$. This will be possible for $k$ small. By the tail bound (\ref{eq:tailbound1}) and (\ref{f+sigmaexpansion}), for $\sigma \in \mathbb{I}_{[0,\eta]} \cup \Big( \{\Re \sigma \geq 1\} \cap \mathbb{I}_{[-1-\eta,\eta]} \Big)$ we have 
    \begin{equation*}
        | \mathcal{W}(\sigma) - f_{+,\sigma}(0)| \lesssim k^2 \big|(2q_k)^{\hat{\sigma}}\big|
        \Big|\frac{1}{\Gamma(1-\hat{\sigma})}\Big| \lesssim k^2 \Big|\frac{1}{\Gamma(1-\hat{\sigma})}\Big|,
    \end{equation*}
    and 
    \begin{equation*}
        \Big| f_{+,\sigma}(0) - (2q_k)^{\hat{\sigma}}\frac{1}{\Gamma(1-\hat{\sigma})} \Big|\lesssim k^2 \big|(2q_k)^{\hat{\sigma}}\big|\Big|\frac{1}{\Gamma(1-\hat{\sigma})}\Big| \lesssim k^2 \Big|\frac{1}{\Gamma(1-\hat{\sigma})}\Big|.
    \end{equation*}
    Given that $\Im \sigma$ is bounded above and below, we moreover have $\big|(2q_k)^{\hat{\sigma}} \big| \gtrsim 1$, and so 
    \begin{align*}
        |\mathcal{W}(\sigma)| &\geq |f_{+,\sigma}(0)| - O(k^2)\Big|\frac{1}{\Gamma(1-\hat{\sigma})}\Big| \\[\jot]
        &\gtrsim (1- O(k^2))\Big|\frac{1}{\Gamma(1-\hat{\sigma})}\Big|,
    \end{align*}
    implying the desired result for $k$ sufficiently small.
    \end{proof}

    Having established control on the Wronskian away from a bounded region in the lower half plane, we now give a useful interpretation of the resolvent operator. Provided $\Im \sigma >0$, $R(\sigma)$ gives the unique right inverse to $P_k(\sigma)$ which is bounded on $L^\infty(\mathbb{R}_+)$.

\begin{lemma}
    \label{lem:resolventuniquness}
    Let $h(x) \in L^\infty(\mathbb{R}_+)$ be given, and $k$ chosen sufficiently small. For all $\sigma \in \mathbb{I}_{(0,\eta]}$, $R(\sigma)$ extends to a bounded map on $L^\infty(\mathbb{R_+})$. Moreover, $u(x) = (R(\sigma)h)(x) \in L^\infty(\mathbb{R}_+) \cap C^2_x(\mathbb{R}_+)$ is a solution to  
    \begin{equation}
        \label{eq:forcedeqndefnofresolvent}
        \begin{cases}
            P(\sigma)u = h \\[\jot]
            u(0)=0,\\[\jot]
        \end{cases}
    \end{equation}
    and is the unique such solution with at most polynomial (in $x$) growth as $x\rightarrow \infty$.
    \end{lemma}

\begin{proof}
    We first show that provided $\sigma \in \mathbb{I}_{(0,\eta]}$, the resolvent is bounded on $L^\infty(\mathbb{R}_+)$. Expressing $(R(\sigma) h)(x)$ via the integral kernel (\ref{eq:resolvkernel}) and estimating yields 
    \begin{align*}
        |(R(\sigma) h)(x)| &\lesssim_\sigma \Big|\frac{f_{(out),\sigma}(x)}{\mathcal{W}(\sigma)} \Big| \int_0^x | f_{(dir),\sigma}(x') h(x')|dx' + \Big|\frac{f_{(dir),\sigma}(x)}{\mathcal{W}(\sigma)} \Big| \int_x^\infty | f_{(out),\sigma}(x') h(x')|dx' \\
        &\lesssim \|h\|_{L^\infty(\mathbb{R_+})}.
    \end{align*}
We have used the $\sigma$-dependent growth/decay bounds (\ref{eq:outgoinggrowthbd}), (\ref{eq:dirichletgrowthbd}), as well as the Wronskian lower bound (\ref{eq:wronskianlowerbdfarregion}) for $k$ sufficiently small. It follows that $R(\sigma)h$ is well-defined in $L^\infty(\mathbb{R}_+)$. Moreover, a consequence of $h \in L^\infty(\mathbb{R}_+)$ and the continuity of $f_{(dir),\sigma}, f_{(out),\sigma}$ is that $R(\sigma)h \in C^0_x(\mathbb{R}_+)$. It thus makes sense to consider $(R(\sigma)h)(0)$. 

To show that Dirichlet boundary conditions are satisfied, write 
\begin{align*}
    (R(\sigma)h)(x) = \frac{f_{(out),\sigma}(x)}{\mathcal{W}(\sigma)}\int_{0}^x f_{(dir),\sigma}(x')h(x')dx' + \frac{f_{(dir),\sigma}(x)}{\mathcal{W}(\sigma)}\int_{x}^\infty f_{(out),\sigma}(x')h(x')dx'.
\end{align*}
By the same bounds (\ref{eq:outgoinggrowthbd}), (\ref{eq:dirichletgrowthbd}), for fixed $\sigma$ estimate
\begin{align*}
    |(R(\sigma)h)(x)| &\lesssim \frac{1}{|\mathcal{W}(\sigma)|}\|h\|_{L^\infty([0,x])} + \frac{|f_{(dir),\sigma}(x)|}{|\mathcal{W}(\sigma)|}\|f_{(out),\sigma} \|_{L^1_x([0,\infty))}\|h \|_{L^\infty([0,\infty))} \xrightarrow{x \rightarrow 0} 0.
\end{align*}

To see that $R(\sigma)h$ provides the \textit{unique} solution to (\ref{eq:forcedeqndefnofresolvent}), it suffices to observe that the homogeneous equation has no non-trivial solutions with polynomial growth that satisfy Dirichlet boundary conditions. Such a non-trivial solution, denoted $u_0(x)$, would vanish at ${x=0}$---implying $W[u_0,f_{(dir),\sigma}]=0$---and would be asymptotic to $f_{(out),\sigma}$---implying $W[u_0, f_{(out),\sigma}]=0$. It would in turn follow that $\mathcal{W}(\sigma)=0$, contradicting the lower bound (\ref{eq:wronskianlowerbdfarregion}).
\end{proof}

We complete the discussion of the resonances of $R(\sigma)$ (i.e. zeroes of $\mathcal{W}(\sigma)$) for $\sigma \in  \mathbb{I}_{[-1-\eta,\eta]}.$  By Lemma \ref{lem:polefreeregion}, the region $\sigma \in  \mathbb{I}_{[0,\eta]} \cup \Big( \{\Re \sigma \geq 1\} \cap \mathbb{I}_{[-1-\eta,\eta]} \Big)$ is devoid of resonances for all $k$ sufficiently small. In the complement we will find that there is a \textit{unique}, \textit{simple} resonance at ${\sigma = -i}$. 

To prove $\mathcal{W}(-i) = 0$ (setting aside the simplicity of the zero), we must show that the Dirichlet and outgoing solutions with ${\sigma = -i}$ are linearly dependent. To facilitate this we appeal to the following lemma, which provides a sufficient condition for a function $f(x)$ to satisfy outgoing boundary conditions. 

\begin{lemma}[Regularity perspective on outgoing boundary conditions]
    \label{lem:regularityperspective}
    Fix $\sigma \in \mathbb{I}_{[-1-\eta,\eta]} \setminus \{0, -iq_k\}$, and let $f(x) \in L^\infty_{loc}(\mathbb{R}_+)$ solve $P_k(\sigma)f=0$. Define the associated function $F(z) \doteq e^{-i \sigma t}f(x)\big|_{\{s=0\}}$, where $x,t$ are expressed as functions of similarity coordinates.
     
    Then the following correspondence holds between outgoing boundary conditions for $f(x)$ with parameter $\sigma$, and regularity of $F(z)$ as $z \rightarrow 0$.
    \begin{itemize}
        \item If \underline{$\Im \sigma > 0$}, then $f(x)$ satisfies outgoing boundary conditions iff $F(z)$ remains \underline{bounded} up to $\{z=0\}$.
        \item If \underline{$\Im \sigma = 0$}, then $f(x)$ satisfies outgoing boundary conditions iff $F(z)$ extends \underline{continuously} to $\{z=0\}$.
        \item If \underline{$\Im \sigma < 0$}, define $n$ such that ${-(n+1)q_k < \Im \sigma \leq -n q_k}$. Then $f(x)$ satisfies outgoing boundary conditions iff $F(z)$ \underline{lies in $C^\alpha_z$}, $\alpha = \min(n+1,2^{-})$ up to $\{z=0\}$. 
    \end{itemize} 
    \end{lemma}

    \begin{proof}
        The proof relies on the observation that for $\sigma \in \mathbb{I}_{[-1-\eta,\eta]} \setminus \{0, -iq_k\}$, we have a basis of solutions for $P_k(\sigma)f=0$ given by $f_{(out),\sigma}, f_{(in),\sigma}$. Computing $f_{(out),\sigma}e^{-i\sigma t}, f_{(in),\sigma}e^{-i\sigma t}$ using (\ref{eq:fout_exp_1})--(\ref{eq:fout_exp_2}) and (\ref{f+sigmaexpansion})--(\ref{f-sigmaexpansion}), we schematically find
            \begin{align}
                e^{-i \sigma t} f_{(out),\sigma} &= c_{(out),\sigma}e^{-i \sigma s}\big(1 +\mathcal{E}(z) \big), \label{lem4.8:temp1}\\[\jot] 
                e^{-i \sigma t} f_{(in),\sigma} &= c_{(in),\sigma} e^{-i \sigma s}e^{i\frac{\sigma}{q_k}\ln|z|}\big(1 +\mathcal{E}(z) \big),\label{lem4.8:temp2}
            \end{align}
        for non-zero constants $c_{(out),\sigma}, c_{(in),\sigma}$. Here, $\mathcal{E}(z) \in C^{2-}_z([-1,0])$.  Thus, if we restrict attention to $C^{2-}_z$, then the regularity of $e^{-i \sigma t} f_{(out),\sigma}$ and $e^{-i \sigma t} f_{(in),\sigma}$ is dictated by the complex exponential factors. 

        For example, if $\Im \sigma > 0$, then (\ref{lem4.8:temp1}) is manifestly bounded in $z$; however, the factor $e^{-i \sigma s} e^{i\frac{\sigma}{q_k}\ln|z|}$ appearing in (\ref{lem4.8:temp2}) fails to be bounded as $z \rightarrow 0$. It follows that if $f(x)$ is a solution to $P_k(\sigma)f=0$ for $\sigma \notin \{0, -i q_k\}$, then $F(z)$ is bounded up to $\{z=0\}$ if and only if $f(x)$ and $f_{(out),\sigma}$ are linearly dependent, i.e. if $f(x)$ satisfies outgoing boundary conditions.
        
        The statements for $\Im \sigma = 0$, $\Im \sigma < 0$ are proved similarly. Note the significance of avoiding ${\sigma = 0}$ or ${\sigma = -i q_k}$. In both cases, the complex exponential factor, responsible for the limited regularity of the ingoing solution, becomes anomalously regular. The above argument does not apply to these cases.
        \end{proof}

        \begin{prop}
            \label{prop:residueofR}
            Let $k$ be sufficiently small. There exists a unique, simple zero of $\mathcal{W}(\sigma)$ on the domain $\sigma \in  \mathbb{I}_{[-1-\eta,\eta]}$ at ${\sigma_*=-i}$. The associated resonance function is given by $f_{(dir),-i} = a_* e^{x}\mr{r}(x)$, for a nonzero constant $a_*$.
            
            For $x_0 > 0$ fixed and $\beta > p_k$, define the operator $\rho_{x_0}R_{*}:\mathcal{D}^{\beta,0}_\infty(\mathbb{R}_+) \rightarrow L^2_{x}(\mathbb{R}_+)$ with integral kernel
            \begin{equation*}
                (\rho_{x_0}R_*)(x,x') \doteq \rho_{x_0}(x)e^{x+x'}\mr{r}(x)\mr{r}(x').
            \end{equation*} 
            For fixed $x \leq x_0$, and $h(x) \in \mathcal{D}^{\beta,0}_{\infty}(\mathbb{R}_+)$ with $\beta > p_k$, $(\rho_{x_0}R(\sigma)h)(x)$ admits the following expansion in a neighborhood of ${\sigma = -i}$:
            \begin{align}
                \label{eq:residueatpole}
                (\rho_{x_0}R(\sigma)h)(x) = \frac{A_1(\sigma,x)}{i + \sigma} + A_2(\sigma,x),
            \end{align}
            where $A_i(\sigma,x)$ are smooth in $(\sigma,x)$, holomorphic in $\sigma$ for fixed $x$, and bounded uniformly in terms of the $\mathcal{D}^{\beta, 0}_\infty(\mathbb{R}_+)$ norm of $h$. There exists a non-zero constant $c_*$ such that the residue is given by  
            \begin{equation}
                \label{eq:residueatpole1}
                A_1(-i,x) =  c_* \int_{0}^\infty (\rho_{x_0}R_*)(x,x')h(x')dx'.
            \end{equation}
            Finally, for fixed $x_0 > 0$, $h(x) \in \mathcal{D}^{\beta,0}_{\infty}(\mathbb{R}_+)$ with $\beta > p_k$, and $\sigma \in \mathbb{I}_{(-\beta_*,\eta]} $, the cutoff resolvent maps $\mathcal{D}^{\beta,0}_{\infty}(\mathbb{R}_+) \rightarrow C^2_x(\mathbb{R}_+)$ with bounds 
            \begin{equation}
                \label{eq:c0toc2bounds}
                \Big\| \frac{d^j}{dx^j}\big(\rho_{x_0} R(\sigma)h \big)\Big\|_{L^\infty(\mathbb{R}_+)} \lesssim (1+x_0)e^{|\Im \sigma| x_0}\Big(1 + \frac{1}{|i+\sigma|} \Big) \frac{(1+|\sigma|)^{-1+j}}{(\beta_* - |\Im\sigma|)}\|h\|_{\mathcal{D}^{\beta,0}_\infty(\mathbb{R}_+)}, \quad (0\leq j \leq 2). 
        \end{equation}
        \end{prop}

        \begin{proof}
            
            \vspace{1em}
            \noindent
            \textbf{Step 1: Existence of a unique, simple zero}
            Define the piecewise smooth contour $\Omega(t)$ to be the boundary of the region $\sigma \in  \mathbb{I}_{(-1-\eta,\eta)} \cap \{\Re \sigma < 1\},$ oriented counterclockwise. Arguing as in the proof of Lemma \ref{lem:polefreeregion}, we compute
                \begin{align}
                    \Big|\mathcal{W}(\sigma) - (2q_k)^{\hat{\sigma}}\frac{1}{\Gamma(1-\hat{\sigma})}\Big| \bigg|_{\Omega(t)} \lesssim k^2 \Big| (2q_k)^{\hat{\sigma}}\frac{1}{\Gamma(1-\hat{\sigma})}\Big| \bigg|_{\Omega(t)}. \label{eq:prop42.temp1}
                \end{align}
            We have used that along $\Omega(t)$, the two-sided bound $0<c_0 \leq |1-\hat{\sigma}| \leq c_1$ holds, and thus $\Gamma(1-\hat{\sigma}) \sim \Gamma(2-\hat{\sigma})$.
            
            The Wronskian $\mathcal{W}(\sigma)$ is a holomorphic function of $\sigma$, and thus by Rouché's theorem applied to (\ref{eq:prop42.temp1}), $\mathcal{W}(\sigma)$ has the same multiplicity of roots in the interior of $\Omega(t)$ as does $(2q_k)^{\hat{\sigma}}\frac{1}{\Gamma(1-\hat{\sigma})}$. The latter has precisely one, simple root.
            
            \vspace{1em}
            \noindent
            \textbf{Step 2: Identification of resonance function}
            Having shown the existence of a unique, simple root of $\mathcal{W}(\sigma)$, it is sufficient to show ${\mathcal{W}(-i) =0}$. We will show that up to a constant multiple, the quantity $e^x \mr{r}(x)$ is the associated resonance function, i.e. it is a solution to $P_k(-i)f=0$ satisfying both Dirichlet and outgoing boundary conditions. 
            
            To prove that this quantity solves $P_k(-i)f=0$, we observe that one family of solutions to the wave equation (\ref{eq:1}) is simply given by $\varphi = \text{const}$. In terms of the variable $r_k \varphi$, it follows that in hyperbolic coordinates, $r_k(t,x)$ must constitute a solution to (\ref{eq:wavehyper}), i.e. 
            \begin{align}
                \label{prop:scattheorytemp1}
                \partial_t^2 r_k(t,x) + P_k(0)r_k(t,x) = 0.
            \end{align}
            By the self-similar relations we have $r_k(u,v) = |u|\mr{r}(z) = e^{x-t}\mr{r}(x).$ Inserting into (\ref{prop:scattheorytemp1}) shows this to be equivalent to $P_k(-i)(e^{x}\mr{r}(x))=0$.
            
            Moreover, as $\mr{r}(0)=0$, it follows that Dirichlet boundary conditions hold. It remains to prove that outgoing boundary conditions hold, i.e. $W[e^{x}\mr{r}(x),f_{(out),-i}]=0.$ As $-i \notin i q_k \mathbb{Z}$, we may apply Lemma \ref{lem:regularityperspective}. It suffices to check that $e^{-i\sigma t}e^{x}\mr{r}(x)\big|_{\{s=0\}}$, viewed in similarity coordinates, lies in $C^{2-}_z$. But we recognize $e^{-i\sigma t}e^{x}\mr{r}(x) = e^{-s}\mr{r}(z)$, and the latter is in $C^2_z$ by Lemma \ref{lemma:backgroundprelims1}. 
            
            \vspace{1em}
            \noindent
            \textbf{Step 3: Computing (\ref{eq:residueatpole})}
            We have shown that $\mathcal{W}(\sigma)$ has a simple zero at ${\sigma_* = -i}$. Therefore, we may write $\mathcal{W}(\sigma) = (\sigma + i)\mathcal{W}_r(\sigma)$ for an analytic function $\mathcal{W}_r(\sigma)$ which is nonvanishing in $\mathbb{I}_{[-1-\eta,\eta]}$.
            
            Let $\gamma_\delta$ denote the contour $-i + \delta e^{-i \theta}, \ \theta \in [0,2\pi)$ of radius $\delta$, with ${\delta \ll \eta}$. For fixed $x_0, x$ we have $(\rho_{x_0}R(\sigma)h)(x)$ is a meromorphic function of $\sigma \in \text{Int}(\gamma_\delta)$, with a simple pole at ${\sigma_* = -i}$. The residue at the pole may be computed as 
            \begin{align*}
            \lim_{\sigma \rightarrow -i} (\sigma + i)(\rho_{x_0}R(\sigma)h)(x) &= \lim_{\sigma \rightarrow -i} \frac{\rho_{x_0}(x)f_{(out),\sigma}(x)}{\mathcal{W}_r(\sigma)}\int_0^x f_{(dir),\sigma}(x')h(x')dx' \\[\jot]
            &\hspace{3em}+ \frac{\rho_{x_0}(x)f_{(dir),\sigma}(x)}{\mathcal{W}_r(\sigma)}\int_x^{\infty} f_{(out),\sigma}(x')h(x')dx'
            \end{align*}
            By (\ref{eq:outgoinghighenergy}), (\ref{eq:dirichletgrowthbd}), we have that $\rho_{x_0} f_{(out),\sigma}$, $\rho_{x_0} f_{(dir),\sigma} $ converge uniformly in $L^\infty(\mathbb{R}_+)$ to their values at ${\sigma_* = -i}$, which by Step 2 are given by $a_* \rho_{x_0} e^{x} \mr{r}(x)$, $b_* \rho_{x_0} e^{x} \mr{r}(x)$ respectively for nonzero constants $b_*, a_*$. Similarly, by the dominated convergence theorem the integrals are seen to converge uniformly in $L^\infty(\mathbb{R}_+)$ to their values at ${\sigma_* = -i}$. 
            
            Combining these statements gives (\ref{eq:residueatpole}), with $c_* = \frac{a_* b_* }{W_r(-i)}$. 
            
            \vspace{1em}
            \noindent
            \textbf{Step 4: Resolvent estimates (\ref{eq:c0toc2bounds})}
            By definition of the resolvent, we need to estimate the terms 
                    \begin{align}
                        &\max_{0\leq j \leq 2}\sup_{x \leq x_0+1}\Big|\int_0^x \frac{f_{(dir),\sigma}(x')\partial_x^j f_{(out),\sigma}(x)}{\mathcal{W}(\sigma)} h(x')dx' \Big|, \label{eq:step4temp1}\\
                        &\max_{0 \leq j \leq 2}\sup_{x \leq x_0+1}\Big|\int_x^{\infty}  \frac{\partial_x^j f_{(dir),\sigma}(x)f_{(out),\sigma}(x')}{\mathcal{W}(\sigma)} h(x')dx' \Big|.\label{eq:step4temp2}
                    \end{align}
                    The required bounds follow from (\ref{eq:outgoinghighenergy}), (\ref{eq:dirichlethighenergy}), and (\ref{eq:wronskianlowerbdfarregion}). The convergence of the integrals follows by the decay encoded in the the space $\mathcal{D}^{\beta,0}_\infty(\mathbb{R}_+)$. Integrability becomes borderline when ${\Im \sigma = -\beta_*}$, accounting for the degeneration of (\ref{eq:c0toc2bounds}) as $\Im \sigma \rightarrow -\beta_*$.
                    
                    We also observe that the dependence on the cutoff scale, $x_0$, has been retained in (\ref{eq:c0toc2bounds}). This exponential loss results from the terms in (\ref{eq:step4temp1})--(\ref{eq:step4temp2}) which are not integrated, and can grow exponentially as stated in Lemma \ref{lem:recordedbounds_outgoingdirichlet}.
            \end{proof}

    In the remainder of the section, we extend the definition of the resolvent to a class of more weakly decaying data, and derive appropriate bounds. First, we prove a preliminary estimate on the function $F_{\sigma}(x) \doteq \partial_{\sigma'} f_{(dir),\sigma'}(x)\big|_{\sigma'=\sigma}$ for ${\sigma = -i}$.
    \begin{lemma}
        \label{lem:F-iEST}
        For the constant $a_*$ defined in Proposition \ref{prop:residueofR},
        \begin{equation}
            \label{eq:F_-iidentified}
            F_{-i}(x) = -\frac{a_* i}{k} e^x \mr{r}(x) (\mr{\phi}(x) - k x). 
        \end{equation}
        For any $x_0 > 0$ we have the estimate
        \begin{equation}
            \label{eq:F_-1est}
            \Big\|\frac{F_{-i}}{\mr{r}} \Big\|_{L^\infty([0,x_0])} \lesssim (1+x_0) e^{x_0}.
        \end{equation}
    \end{lemma}
    \begin{proof}
        Differentiating (\ref{eq:dirichlet1}) with respect to $\sigma$ and evaluating at ${\sigma = -i}$ yields the following equation:
        \begin{equation}
            \label{eq:definingF-i}
            \begin{cases}
                -F_{-i}'' + (V_k(x) + 1)F_{-i} = -2i f_{(dir),-i}(x) \\[\jot]
                (F_{-i}(0), F_{-i}'(0)) = (0,0).
            \end{cases}
        \end{equation}
        We claim that the solution to (\ref{eq:definingF-i}) is given by the expression (\ref{eq:F_-iidentified}). This follows by a direct calculation using the form of $f_{(dir),-i}(x)$ found in Proposition \ref{prop:residueofR}, the self-similar equation (\ref{ss:SSESF5}), and the coordinate transformations in Table \ref{table:1}. Having established the form of $F_{-i}(x)$, the bound (\ref{eq:F_-1est}) follows.
    \end{proof}

    In the following proposition, we consider applying $R(\sigma)$ to functions that decay with a precise exponential tail $h(x) = c_0 e^{-\beta q_k x} + O(e^{-(\beta+\delta)q_k x})$, for some $\beta \leq p_k$. In the scale of $\mathcal{D}^{\beta,l}_\infty(\mathbb{R}_+)$ spaces such functions lie in $\wt{\mathcal{D}}^{\beta, 0, \delta,c_0}_\infty(\mathbb{R}_+) \subset \mathcal{D}^{\beta,0}_\infty(\mathbb{R}_+)$, and $R(\sigma)h$ is a priori defined only for $\sigma \in \mathbb{I}_{(-\beta_*,\eta]}$. We show that the range of $\sigma$ may be extended to include a neighborhood of $\{\Im \sigma = -\beta_*\}$.

    \begin{prop}
        \label{lem:extR_sigma_2}
    Fix $\beta \in (1,p_k]$, and small parameters $\delta, \delta_1 > 0$ satisfying ${\delta_1 < \delta < \eta}$. For any $x_0 > 0$, $c_0 \in \mathbb{R}$, the cutoff resolvent extends meromorphically to a map $\rho_{x_0}R(\sigma): \wt{\mathcal{D}}^{\beta, 0, \delta,c_0}_\infty(\mathbb{R}_+) \rightarrow C^2_x(\mathbb{R}_+)$ for $\sigma \in \mathbb{I}_{[-\beta_*-\delta_1 q_k, \eta]}$. We distinguish between the cases $\beta < p_k$ and $\beta = p_k$ below.
    
    \vspace{.5em}
    \noindent
    \underline{$\beta < p_k$:} The extension satisfies the bounds 
    \begin{equation}
        \label{eq:c0toc2bounds_extension_2}
        \Big\|\frac{d^j}{dx^j}\big(\rho_{x_0} R(\sigma)h \big)\Big\|_{L^\infty(\mathbb{R}_+)} \lesssim (1+x_0)e^{|\Im \sigma| x_0} \Big(\frac{1}{|i+\sigma|}+ \frac{1}{|\beta q_k  - i\sigma|}\Big)(1+|\sigma|)^{-1+j}\|h\|_{\wt{\mathcal{D}}^{\beta,0,\delta,c_0}_\infty(\mathbb{R}_+)}, \quad (0 \leq j \leq 2). 
    \end{equation}
    For any $h(x) \in \wt{\mathcal{D}}^{\beta, 0, \delta,c_0}_\infty(\mathbb{R}_+)$, there exists an expansion in a neighborhood of $\sigma = -i \beta q_k$:
    \begin{align}
    \label{eq:resolvexp_extension_2}
    (\rho_{x_0}R(\sigma)h)(x) = \frac{B_1(\sigma,x)}{\beta q_k  - i\sigma} + B_2(\sigma,x),
    \end{align}
where $B_i(\sigma,x)$ are smooth in $(\sigma,x)$, holomorphic in $\sigma$ for fixed $x$, and bounded uniformly in terms of the $\wt{\mathcal{D}}^{\beta, 0, \delta,c_0}_\infty(\mathbb{R}_+)$ norm of $h$. Provided $c_0 \neq 0$, the quantity $B_1(-i \beta q_k, x)$ does not vanish identically.

    \vspace{.5em}
    \noindent
    \underline{$\beta = p_k$:} The extension satisfies the bounds 
    \begin{equation}
        \label{eq:c0toc2bounds_extension}
        \Big\|\frac{d^j}{dx^j}\big(\rho_{x_0} R(\sigma)h \big)\Big\|_{L^\infty(\mathbb{R}_+)} \lesssim (1+x_0) e^{|\Im \sigma| x_0}\Big(1+ \frac{1}{|i+\sigma|^2}\Big)(1+|\sigma|)^{-1+j}\|h\|_{\wt{\mathcal{D}}^{p_k,0,\delta,c_0}_\infty(\mathbb{R}_+)}, \quad (0 \leq j \leq 2). 
    \end{equation}
    For any $h(x) \in \wt{\mathcal{D}}^{p_k, 0, \delta,c_0}_\infty(\mathbb{R}_+)$, there exists an expansion
    \begin{align}
        \label{eq:resolvexp_extension}
        (\rho_{x_0}R(\sigma)h)(x) = \frac{C_1(\sigma,x)}{(1-i\sigma)^2} + \frac{C_2(\sigma,x)}{1-i\sigma} + C_3(\sigma,x),
    \end{align}
    where $C_i(\sigma,x)$ are smooth in $(\sigma,x)$, holomorphic in $\sigma$ for fixed $x$, and bounded uniformly in terms of the $\wt{\mathcal{D}}^{p_k, 0, \delta,c_0}_\infty(\mathbb{R}_+)$ norm of $h$. 
    
    \vspace{.5em}
    \noindent
    For fixed $x_0$, and all $x \leq x_0$, there exist constants $d_{*,j}$, such that  
    \begin{equation}
        \label{eq:resolventexp_A1}
        C_{j}(-i,x) = d_{*,j} \rho_{x_0}(x) e^{x}\mr{r}(x), \quad (1 \leq j \leq 2),
    \end{equation}
    and constants $e_{*}, f_{*}$ such that 
    \begin{equation}
        \label{eq:resolventexp_A2}
        \partial_{\sigma} C_1 (\sigma,x) \big|_{\sigma=-i} = e_* \rho_{x_0}(x) e^x \mr{r}(x) (\mr{\phi}(x) - k x) + f_* \rho_{x_0}(x) e^x \mr{r}(x)
    \end{equation}
    Finally, there is the bound 
    \begin{equation}
        \label{eq:resolventexp_vanishingextension}
       \sup_{x \leq x_0}|C_1(-i,x)|  + \sup_{x \leq x_0}|\partial_\sigma C_1(-i,x)|+  \sup_{x \leq x_0}|C_2(-i,x)|\lesssim (1+x_0) e^{x_0} \mr{r}(x) \|h\|_{\wt{\mathcal{D}}^{p_k, 0, \delta,c_0}_\infty(\mathbb{R}_+)}.
    \end{equation}
    
    \end{prop}

    \begin{proof}
    We discuss the case $\beta = p_k$ here; the remaining case $\beta < p_k$ follows by a similar computation.

    Let $h \in \wt{\mathcal{D}}^{p_k, 0, \delta,c_0}_\infty(\mathbb{R}_+)$ be given. This regularity implies we may write $h(x) = h_0(x) + h_1(x)$, where $h_0(x)=c_0 e^{-x}$ and $h_1(x) \in \mathcal{D}^{p_k+\delta, 0}_\infty(\mathbb{R}_+)$. By linearity, the rapid decay of $h_1(x)$, and the estimates (\ref{eq:residueatpole})--(\ref{eq:c0toc2bounds}), we reduce to to defining the action of the cutoff resolvent on $h_0(x)$.
    
    For $\sigma \in \mathbb{I}_{(-1,\eta]}$, the quantity $(\rho_{x_0}R(\sigma)h_0)(x)$ is unambiguously defined by the expression
    \begin{equation}
        \label{eq:Prop43:temp-.1}
        \underbrace{\frac{\rho_{x_0}(x)f_{(out),\sigma}(x)}{\mathcal{W}(\sigma)} \int_0^x f_{(dir),\sigma}(x')h_0(x')dx'}_{I_1(\sigma)} + \underbrace{\frac{\rho_{x_0}(x)f_{(dir),\sigma}(x)}{\mathcal{W}(\sigma)} \int_x^\infty f_{(out),\sigma}(x')h_0(x')dx'}_{I_2(\sigma)}.
     \end{equation}
     We treat the terms $I_1(\sigma)$ and $I_2(\sigma)$ separately. Observe that the integral quantity in $I_1(\sigma)$ extends holomorphically to $\sigma \in \mathbb{I}_{[-1-\delta_1 q_k, \eta]}$, and thus $I_1(\sigma)$ has poles only at the zeros of $\mathcal{W}(\sigma)$. To extend $I_2(\sigma)$, we require the precise leading order behavior of $f_{(out),\sigma}$. Applying (\ref{f+sigmaexpansion}) and the tail bound (\ref{eq:tailbound1}) gives 
\begin{equation}
    \label{eq:Prop43:temp0}
    f_{(out),\sigma}(x) = (2q_k)^{\hat{\sigma}}\frac{1}{\Gamma(1-\hat{\sigma})}e^{i\sigma x} + \mathcal{E}(\sigma,x),
\end{equation}
where
\begin{equation}
    \label{eq:Prop43:temp1}
    |\mathcal{E}(\sigma,x)| \lesssim k^2 \big(1+\frac{1}{\Gamma(1-\hat{\sigma})}\big)e^{(-2q_k + |\Im \sigma|)x}.
\end{equation}
Recall the definition $\mathcal{W}_r(\sigma) \doteq \frac{\mathcal{W}(\sigma)}{\sigma+i}$. Inserting (\ref{eq:Prop43:temp0}) and evaluating $I_2(\sigma)$ for $\sigma \in \mathbb{I}_{(-1,\eta]}$ gives 
\begin{align}
    \frac{\rho_{x_0}(x)f_{(dir),\sigma}(x)}{\mathcal{W}(\sigma)} &\int_x^\infty f_{(out),\sigma}(x')h_0(x')dx' \nonumber\\[\jot]
    &= c_0(2q_k)^{\hat{\sigma}}\frac{1}{\Gamma(1-\hat{\sigma})}\frac{\rho_{x_0}(x)f_{(dir),\sigma}(x)}{\mathcal{W}(\sigma)} \int_x^{\infty} e^{(i\sigma -1 )x'} dx' + \frac{\rho_{x_0}(x)f_{(dir),\sigma}(x)}{\mathcal{W}(\sigma)}\wt{\mathcal{E}}(\sigma,x) \nonumber\\[\jot]
    &= -i c_0(2q_k)^{\hat{\sigma}}\frac{1}{\Gamma(1-\hat{\sigma})}\frac{1}{(i+\sigma)^2}\frac{\rho_{x_0}(x)f_{(dir),\sigma}(x)}{\mathcal{W}_r(\sigma)}e^{(i\sigma-1)x} + \frac{\rho_{x_0}(x)f_{(dir),\sigma}(x)}{\mathcal{W}(\sigma)}\wt{\mathcal{E}}(\sigma,x),\label{eq:resolventextension_temp1}
\end{align}
where we have introduced $\wt{\mathcal{E}}(t,x)$, a holomorphic function of $\sigma \in {[-1-\delta_1 q_k, \eta]}$ for fixed $x$. The expression (\ref{eq:resolventextension_temp1}) defines a meromorphic extension to $\sigma \in \mathbb{I}_{[-1-\delta_1 q_k,\eta]}$, with a double pole at ${\sigma = -i}$ arising from the first term. The terms $C_1(-i,x), C_2(-i,x)$ are proportional to $f_{(dir),-i}$, and thus (\ref{eq:resolventexp_A1}) follows from Proposition \ref{prop:residueofR}, in which the resonance function $f_{(dir),-i}$ was explicitly identified.

The identity (\ref{eq:resolventexp_A2}) follows by differentiating the first term of (\ref{eq:resolventextension_temp1}), and employing (\ref{eq:F_-iidentified}). Finally, the resolvent estimates (\ref{eq:c0toc2bounds_extension}) and the bounds (\ref{eq:resolventexp_vanishingextension}) follow from the identities (\ref{eq:resolventexp_A1})--(\ref{eq:resolventexp_A2}), and an analogous argument as for (\ref{eq:c0toc2bounds}).
\end{proof}

\begin{lemma}
    \label{lem:resolventon_eqx}
    For fixed $x \leq x_0$, the function $\big(\rho_{x_0}R(\sigma)e^{-q_k x'}\big)(x)$, a priori defined for $\sigma \in \mathbb{I}_{(-q_k, \eta]}$, extends meromorphically to $\sigma \in \mathbb{I}_{[-1-\eta, \eta]}$ as a function with a single, simple pole at ${\sigma = -i}$. The following bounds hold:
    \begin{equation}
        \Big\| \frac{d^j}{dx^j}\big(\rho_{x_0} R(\sigma)e^{-q_k x'} \big)\Big\|_{L^\infty(\mathbb{R}_+)} \lesssim (1+x_0)e^{|\Im \sigma| x_0}\Big(1 + \frac{1}{|i+\sigma|} \Big) (1+|\sigma|)^{-1+j},  \quad (0 \leq j \leq 2)
    \end{equation}
\end{lemma}
\begin{proof}
    The calculation is similar to that in Proposition \ref{lem:extR_sigma_2}. Evaluating $\rho_{x_0} R(\sigma) e^{q_k x'}$ via the integral expression (\ref{eq:Prop43:temp-.1}) and the expansion (\ref{eq:Prop43:temp0}), the corresponding term $I_2(\sigma)$ has the pole at ${\sigma = -iq_k}$ due to the exponential integral cancelled by the factor $\frac{1}{\Gamma(1-\hat{\sigma})}$ appearing in (\ref{eq:Prop43:temp0}). The result is that $I_2(\sigma)$ extends holomorphically to a neighborhood of ${\sigma = -i q_k}$, leaving only the single pole at the zero of $\mathcal{W}(\sigma)$. 

    The estimates follow by an analogous argument as for (\ref{eq:c0toc2bounds}).
\end{proof}

\section{Proof concluded}
\label{sec:proof}
In order to prove Theorem \ref{theorem:mainthm}, we require two additional preliminaries. The first, discussed in Section \ref{subsec:proofcompleted_1}, is a physical space scattering result relating \textit{regularity} for null data to \textit{decay} for spacelike data. The second is a sharp decay result for inhomogeneous wave equations with spacelike data, in regions $\{x \leq \text{const}.\}$ close to the axis. This is proved in Section \ref{subsec:proofcompleted_2}, relying on the spectral theory constructions of Section \ref{sec:scattering}. Finally, in Section \ref{subsec:proofcompleted_3} we combine these results with muliplier estimates to complete the proof

Theorem \ref{theorem:angularthm} is proved in Section \ref{subsec:proofcompleted_4}, and is logically independent of the other results in this section. 

\subsection{Backwards scattering}
\label{subsec:proofcompleted_1}
The main result of this section establishes a correspondence between null data $(r\varphi)_0(v) \in \mathcal{C}^{\alpha}_{(hor)}([-1,0])$ on $\{u=-1\}$ and spacelike data $((r\varphi)_0(x), \partial_t (r\varphi)_0(x))$ on $\{t=0\}$. We will assume throughout that all data is spherically symmetric, and thus drop angular terms in (\ref{eq:wavesim}).

\begin{prop}
    \label{prop:backwardsdecay}   
    Fix $\alpha \in (1,2)$, and outgoing spherically symmetric null initial data $h_0(v) = (r\varphi)_0(v) \in C^{\alpha}_{(hor)}([-1,0])$. Let $\psi(u,v) = (r\varphi)(u,v)$ denote the unique solution to the linear wave equation (\ref{sec2.5:eq2}) in $\{u \geq -1\}$. There exists spacelike data $(f_0(x),f_1(x)) \in \mathcal{D}_\infty^{\alpha, 5}(\mathbb{R}_+) \times \mathcal{D}_\infty^{\alpha, 4}(\mathbb{R}_+)$ for $(r\varphi,\partial_t(r\varphi))$ on $\{t=0\}$ such that the unique solution $\wt{\psi}(t,x)=(r\varphi)(t,x)$ to the linear wave equation (in hyperbolic coordinates) (\ref{eq:wavehyper}) on $\{t \geq 0\}$ satisfies the following:
    \begin{itemize}
        \item $\wt{\psi}(t,x) \in C^{5}_{t,x}(\{t \geq 0\})$.
        \item Written in double-null coordinates, $\wt{\psi}(u,v)$ coincides with $\psi(u,v)$ in their common domain of definition.
        \item For any $T > 0$ fixed, the estimate holds 
        \begin{equation}
            \label{eq:propbackscat_1}
            \sup_{t\in[0,T]}\sum_{j=0}^5\|\partial_t^j\wt{\psi}(t,\cdot)\|_{\mathcal{D}^{\alpha,5-j}_\infty(\mathbb{R}_+)} \lesssim_T \|h_0 \|_{\mathcal{C}^{\alpha}_{(hor)}([-1,0])}.
        \end{equation}
    \end{itemize}
    If we moreover assume $h_0(z) \in C^{\alpha,\delta}_{(hor)}([-1,0])$ for some $\delta \in (0,1)$, then there exists a $c_0 \in \mathbb{R}$ and $\delta' > 0$ such that the induced data $(f_0(x),f_1(x))$ lies in $\wt{\mathcal{D}}^{\alpha, 5, \delta',c_0}_\infty(\mathbb{R}_+) \times  \wt{\mathcal{D}}^{\alpha, 4, \delta',c_0}_\infty(\mathbb{R}_+)$. For any $T>0$ and all $t \in [0,T]$, $\wt{\psi}(t,x) - c_0 e^{\alpha q_k (t-x)} \in \mathcal{D}^{\alpha+\delta',5}_{\infty}(\mathbb{R}_+),$ and we have the estimate 
    \begin{equation}
        \label{eq:propbackscat_1.5}
       |c_0| +  \sup_{t\in[0,T]} \sum_{j=0}^5 \|\partial_t^j\big(\wt{\psi}(t,\cdot) - c_0 e^{\alpha q_k (t-x)}\big)\|_{\mathcal{D}^{\alpha+\delta',5-j}_\infty(\mathbb{R}_+)} \lesssim_T \|h_0 \|_{\mathcal{C}^{\alpha,\delta}_{(hor)}([-1,0])}.
    \end{equation}
    The constant $c_0$ is given by
    \begin{equation}
        \label{eq:backwardsexplicitc_0}
        c_0 = \frac{1}{\alpha(\alpha-1)}\lim_{v\rightarrow 0}|v|^{2-\alpha}\frac{d^2}{dv^2}h_0(v).
    \end{equation}
    \end{prop}

\begin{proof}

\vspace{.5em}
\noindent
\textbf{Posing compatible ingoing data:} 
We construct spacelike data $(f_0(x),f_1(x))$ by solving the backwards characteristic problem for (\ref{sec2.5:eq2}) with data along $\{u=-1\} \cup \{v=0, u \leq -1\}$. To render this problem well-posed, we specify appropriate data for $r\varphi$ along the ingoing null component, denoted $h_1(u)$. We claim this ingoing data can be chosen such that the following conditions hold:
\begin{enumerate}
    \item $\text{supp} \ h_1(u) \in [-2,-1]$, and $h_1 \in C^5_{u}([-2,-1])$, with $\|h_1 \|_{C^5_u([-2,-1])} \lesssim \|h_0 \|_{C^\alpha_{(hor)}([-1,0])}$
    \item $h_1(-1) = h_0(0)$
    \item $\partial_u^{j} h_1(-1) = \partial_u^j \psi(-1,0), \quad (0 \leq j \leq 5)$
    \item $\text{supp} \ \partial_v \wt{\psi}(u,0) \in [-2,-1]$
    \item  $\big|\lim_{v\rightarrow 0}|v|^{j- \alpha} \partial_v^j \wt{\psi}(u,v)\big| \lesssim 1, \quad (2 \leq j \leq 5)$
\end{enumerate}
To briefly comment on the significance of these conditions, we note that (1)--(2) are natural requirements of compatibility with outgoing data and compact support. Condition (4) requires not only $r\varphi(u,0)$ to be compactly supported in $u$, but additionally\footnote{Note this is not possible for every wave equation of the form (\ref{sec2.5:eq2}). In Minkowski space, the wave equation reduces to a conservation law $\partial_u \big(\partial_v (r\varphi)\big) = 0$, implying the derivative can only be compactly supported if additional conditions are imposed on $h_0(v)$.} $\partial_v(r\varphi)(u,0)$. Condition (3) is forced by the requirement of $C^5$ gluing of $\psi$ and $\wt{\psi}$ across $\{u=-1\}$, which is in turn forced by the presence of an axis. Finally, condition (5) asserts that the singular bounds on $h_0(u)$ are propagated to the past.

Continuing with the proof, it is direct that we can choose $h_1(u)$ to satisfy conditions (1)--(3), which only constrain the jet of $h_1$ at the point $u=-1$. The quantitative estimate in (1) moreover follows from the local well-posedness estimates for $\psi$, i.e. (\ref{eq:localexist_bound}). 

To see condition (4), we restrict (\ref{sec2.5:eq2}) to $\{v=0\}$ and integrate in $u$, giving
\begin{equation}
    \label{eq:backscatproof:temp1}
    \partial_v \wt{\psi}(u,0) = h_0'(0) + \int_{u}^{-1} \bigg(\frac{\lambda(-\nu)\mu}{(1-\mu)r^2}\bigg)(u',0) h_1(u')du',
\end{equation}
By assumption on $(\epsilon_0,k)$-admissible spacetimes, we can expand
\begin{equation*}
    \bigg(\frac{\lambda(-\nu)\mu}{(1-\mu)r^2}\bigg)(u,0) = \frac{c_k + \epsilon_0 k^2 g(u)}{|u|^{2-k^2}},
\end{equation*}
for a non-zero constant $c_k$ and function $g(u)$ supported in $[-2,-1]$. It now suffices to choose $h_1$ such that (\ref{eq:backscatproof:temp1}) evaluates to zero when $u=-2$. 

Finally, condition (5) will follow from the set of conservation laws
\begin{align*}
    \partial_u \big(\lim_{v\rightarrow 0}|v|^{j-\alpha}\partial_v^j \wt{\psi}(u,v)\big) = 0, \quad (2\leq j \leq 5).
\end{align*}
These follow inductively in $j$ by commuting (\ref{sec2.5:eq2}) with $|v|^{j-\alpha}\partial_v^{j-1}$, taking the limit as $v \rightarrow 0$, and applying the regularity (\ref{eq:potentialest3.5}).

\vspace{1em}
\noindent
\textbf{Pointwise bounds:}
We turn to estimating the solution in a characteristic rectangle $\mathcal{B} \doteq \{u \leq -1, \ v \geq -\frac12 \}.$ In the following, let $C(h_1,h_2)$ denote any constant depending on the $C^{\alpha}_{(hor)}([-1,0])$ norm of $h_1$, and the $C^5_u([-2,-1])$ norm of $h_2$. We first show the bound
\begin{equation}
    \label{eq:backscatproof:temp2}
    \|\wt{\psi}\|_{C^1_{u,v}(\mathcal{B})} + \||u| \partial_u \wt{\psi}\|_{L^\infty(\mathcal{B})} \leq C(h_1,h_2),
\end{equation} 
which will follow from a multiplier estimate. We choose to work in similarity coordinates for convenience, and estimate (\ref{eq:wavesimfull}). Multiplying by $\partial_s \wt{\psi}$ and integrating by parts in $\mathcal{R}(s_0,0)$ for an arbitrary $s_0 < 0$ yields 
\begin{align*}
    \frac{1}{2}\int_{s_0}^{0}&(\partial_s \wt{\psi})^2(s,0)ds  + \frac{1}{2}q_k \int_{\{s=0\}} |z| (\partial_z \wt{\psi})^2(0,z)dz +  \frac{1}{2}\int_{\{s=0\}}V(0,z)  \wt{\psi}^2(0,z) dz \\[\jot]
    &+ \frac{1}{2}\iint_{\mathcal{R}(s_0,0)}\partial_s(V-V_k) \wt{\psi}^2(s,z) dsdz = \frac{1}{2}q_k \int_{\{s=s_0\}} |z| (\partial_z \wt{\psi})^2(s_0,z)dz+  \frac{1}{2}\int_{\{s=s_0\}}V(s_0,z)  \wt{\psi}^2(s_0,z) dz.
\end{align*}
We have used that $\partial_s V_k = 0$ to drop the corresponding bulk term. By the assumptions on $(\epsilon_0,k)$-admissible spacetimes, the remaining bulk term is supported in $\{u \in [-2,-1]\}$ and satisfies a pointwise bound in terms of $\epsilon_0$. Choosing this parameter sufficiently small, taking the supremum over $s \in [s_0,0]$, and absorbing the bulk term, we conclude that the multiplier estimate controls $\|\wt{\psi} \|_{L^2_z(\{s=s_0\})}$. In double-null coordinates we have, for all $u' \leq -1$,
\begin{equation}
    \label{eq:backscatproof:temp3}
    \|\wt{\psi} \|_{L^2_v(\{u=u'\})} \leq C(h_0,h_1)|u'|^{\frac12 q_k}.    
\end{equation}
This bound, in combination with integration along characteristics, will give (\ref{eq:backscatproof:temp2}). We sketch the argument here.

Denoting the zeroth order coefficient in (\ref{sec2.5:eq2}) by $P(u,v),$ we can estimate $|P(u,v)| \lesssim |u|^{-2+k^2}.$ In particular, this term rapidly decays as $|u| \rightarrow \infty$. We now integrate (\ref{sec2.5:eq2}) as an equation for $\partial_u \wt{\psi}$ and estimate $P(u,v)\wt{\psi}$ via Cauchy-Schwarz. It follows that $\partial_u \wt{\psi}$ decays at an integrable rate, and thus $\wt{\psi}$ and $\partial_u \wt{\psi}$ are bounded. Integrating (\ref{sec2.5:eq2}) again as an equation for $\partial_v \wt{\psi}$ gives boundedness of this quantity as well.

\vspace{.5em}
\noindent
A corollary of (\ref{eq:backscatproof:temp2}) are the higher order bounds 
\begin{equation}
    \label{eq:backscatproof:temp4}
    \big\||u|^{j}\partial_u^j \wt{\psi}\big\|_{L^\infty(\mathcal{B})}+\big\||v|^{j-\alpha}\partial_v^j \wt{\psi}\big\|_{L^\infty(\mathcal{B})} \leq C(h_1,h_2), \quad  (2\leq j \leq5).
\end{equation}
These follow inductively in $j$ by commuting (\ref{sec2.5:eq2}) with $\partial_u^j, \ |v|^{j-\alpha} \partial_v^{j-1}$, applying (\ref{eq:backscatproof:temp2}), and integrating in the $v$, $u$ directions respectively. The initial data terms along $\{u=-1\}$ limit the regularity of $\partial_v^j \wt{\psi}$ in (\ref{eq:backscatproof:temp4}).

The sharp bounds (\ref{eq:propbackscat_1}) now follow. Integrating the ${j=2}$ estimate in (\ref{eq:backscatproof:temp4}) in $v$ implies that for all ${u \leq -2}$, 
\begin{equation}
    \label{eq:backscatproof:temp5}
    |\partial_v \wt{\psi}(u,v)| \leq C(h_0,h_1)|v|^{\alpha-1},
\end{equation}
and after integrating again,
\begin{equation}
    \label{eq:backscatproof:temp6}
    |\wt{\psi}(u,v)| \leq C(h_0,h_1)|v|^{\alpha}.
\end{equation}

It now suffices to collect (\ref{eq:backscatproof:temp2}), (\ref{eq:backscatproof:temp4}), (\ref{eq:backscatproof:temp5}), (\ref{eq:backscatproof:temp6}), and apply the coordinate transformations between double-null and hyperbolic coordinates to conclude (\ref{eq:propbackscat_1}) in $\mathcal{B}$. The complement  $\{t \in [0,T]\} \setminus \mathcal{B}$ is contained in a set of the form $\{u \geq u_0\},$ and thus the estimate (\ref{eq:propbackscat_1}) follows from local existence theory, cf. Proposition \ref{prop:localexist}.

\vspace{1em}
\noindent
\textbf{Refinement for data in $\mathcal{C}^{\alpha,\delta}_{(hor)}$:}
Finally we examine the case of initial data $h_0(v) \in \mathcal{C}^{\alpha,\delta}_{(hor)}([-1,0])$, and show the sharper (\ref{eq:propbackscat_1.5}). By Lemma \ref{lem:holderspacedecomp} and linearity, it is sufficient to consider $h_0(v) = c_0|v|^{\alpha}$ for a constant $c_0$. Define $c_j \doteq |v|^{j-\alpha}\frac{d^j}{dv^j}h_0(v),$ $j \geq 0$. 

By the same argument that leads to (\ref{eq:backscatproof:temp4}), we may commute (\ref{sec2.5:eq2}) inductively by $|v|^{j-\alpha}\partial_v^{j-1}$ and estimate in $\mathcal{B}$, giving 
\begin{align*}
    \partial_u \big(|v|^{j-\alpha}\partial_v^j \wt{\psi}\big) = O_{L^\infty}(|v|^{2-\alpha}|u|^{-2+k^2}), \quad (2\leq j \leq 5).
\end{align*}
Integrating in $u$ from data, we see that the contribution from the initial data term dominates the expansion for $\partial_v^j \wt{\psi}$:
\begin{equation}
    \label{eq:backscatproof:temp7}
    |v|^{j-\alpha}\partial_v^j \wt{\psi}(u,v) = c_j + O_{L^\infty}(|v|^{2-\alpha}), \quad (2\leq j \leq 5).
\end{equation}
Integrating the $j=2$ expansion further yields 
\begin{align}
    \partial_v \wt{\psi}(u,v) &= \partial_v \wt{\psi}(u,0) + c_1 |v|^{\alpha-1} + O_{L^\infty}(|v|) \label{eq:backscatproof:temp8}\\
    \wt{\psi}(u,v) &= h_1(u) + c_0  |v|^{\alpha} + O_{L^\infty}(|v|^2). \label{eq:backscatproof:temp9}
\end{align}
If we now define $\wt{\psi}_{reg}  \doteq \wt{\psi} - c_0 |v|^{\alpha}$, it follows from the estimates on $u$-derivatives (\ref{eq:backscatproof:temp2})--(\ref{eq:backscatproof:temp4}), as well as the expansions (\ref{eq:backscatproof:temp7})--(\ref{eq:backscatproof:temp9}), that we have for some $\delta' > 0$
\begin{equation}
   \sup_{t' \in [0,T]} \| e^{(\alpha+\delta')x}\wt{\psi}_{reg} \|_{C^5_{t,x}(\{t=t'\})} \lesssim_T \|h_0\|_{C^{\alpha,\delta}_{(hor)}([-1,0])}.
\end{equation}
\end{proof}

\subsection{Leading order expansion in the near-axis region}
\label{subsec:proofcompleted_2}
In this section we establish a leading order expansion for the solution to the following inhomogeneous problem:
\begin{equation}
    \label{eq:forcedeqn1}
    \begin{cases}
    &\partial_t^2 \psi - \partial_x^2 \psi  + 4q_k e^{-2q_k x}V_k(x) \psi =  F(t,x), \\[\jot]
        &\psi(t,0) =0,  \\[\jot]
        &(\psi(0,x), \partial_t \psi(0,x)) = (0,0),
    \end{cases}
\end{equation}
where $F(t,x)$ is a given function satisfying 
\begin{equation}
    \label{52:forcingassTOTAL}
    \begin{cases}
    &F(t,x) \in C^3_{t,x}(\mathbb{R}_+ \times \mathbb{R}_+),   \\[\jot]
    &\text{supp}\ F(\cdot,x) \in [1,\infty) \ \text{for all} \ x,   \\[\jot]
    &\sup_{t \geq 0} \sum_{j=0}^3 e^{\frac32 t} \|\partial_t^j F\|_{\mathcal{D}^{\beta,0}_\infty(\mathbb{R}_+)} \leq C, 
    \end{cases}
\end{equation}
for constants $\beta \in (1,2)$ and $C > 0$. The analysis of (\ref{eq:forcedeqn1}) will depend on the spatial decay of $F(t,x)$, manifest in the constant $\beta$. In the range $\beta > p_k$, corresponding (by Proposition \ref{prop:backwardsdecay}) to null data with regularity \textit{above threshold}, we show that the solution to (\ref{eq:forcedeqn1}) converges to a constant multiple of the $k$-self-similar radius function $r_k(t,x)$, in spatially compact sets $\{x \leq \text{const}.\}$.

Convergence to constants will not in general hold for the \textit{threshold regularity} case ${\beta = p_k}$, or the \textit{below threshold regularity} case ${1<\beta < p_k}$. For these latter cases, we impose an additional assumption, denoted $(A_\beta)$.
\begin{enumerate}
    \item[$(A_{\beta})$] There exists a decomposition $F(t,x) = c_0 \chi_b(t) e^{\beta q_k(t-x)} + F_2(t,x)$, with $c_0$ a constant, $\chi_b(t)$ a smooth, non-negative bump function with support in $[1,2]$, and $F_2(t,x)$ a function satisfying the assumptions (\ref{52:forcingassTOTAL}) for some $\beta' = \beta+\delta$.
\end{enumerate}

For initial data with $\beta = p_k$, which moreover satisfies $(A_{p_k})$, we establish exponential convergence in spatially compact sets to a linear combination of terms of the form $r_k, g(t,x)r_k$, where $g(x) = \mr{\phi}(x) - k (x-t)$ is the restriction of the $k$-self-similar scalar field to $\{t=0\}$. The first term corresponds to a bounded scalar field; the second, however, is new to the threshold regularity case, and corresponds in double-null coordinates to a scalar field growing like $-\log |u|$. 

Finally, when $1<\beta <p_k$ and the assumption $(A_\beta)$ holds, the solution converges exponentially in spatially compact sets to a term which grows as $e^{-\beta q_k t}$; in particular, the solution is unstable with a rate strictly between the self-similar and blue-shift rate.

\begin{prop}
    \label{prop:resonanceexp}
    Fix an $(\epsilon_0,k)$-admissible background, with $k$ sufficiently small. Choose parameters $\beta \in (1,2)$, $x_0>0$, and let $\eta \in (0,\frac12)$ be the parameter defined in Section \ref{sec:scattering} (cf. Remark \ref{rmk:etafixed}). Let $\psi(t,x)$ denote the solution to (\ref{eq:forcedeqn1})--(\ref{52:forcingassTOTAL}). We distinguish various cases below:

    \vspace{.5em}
    \noindent
    \underline{$\beta > p_k$: } Define $\beta_* = \min(\beta q_k, 1+\eta),$ and fix any $\beta_*' < \beta_*$. Then there exists a constant $c_{\infty}$ independent of $x_0$, and a constant $C_{x_0, \beta_*'} \lesssim_{\beta_*'} (1+x_0) e^{\beta_*' x_0},$ such that
    \begin{equation}
        \label{eq:resonanceexp}
        \sup_{ t \geq 0 } \sum_{j=0}^2 \|e^{\beta_*' t}\partial_t^j(\psi - c_{\infty} r_k) \|_{C^{2-j}_x([0,x_0])} \lesssim C_{x_0, \beta_*'} \sup_{t \geq 0} \sum_{j=0}^3 e^{\frac32 t}\|\partial_t^j F\|_{\mathcal{D}^{\beta,0}_\infty(\mathbb{R}_+)}.
    \end{equation}
    Moreover, $|c_\infty|$ can be estimated by the right hand side of (\ref{eq:resonanceexp}).

    \vspace{.5em}
    \noindent
    \underline{$\beta = p_k$, $(A_{p_k})$: } Assume $F(t,x)$ satisfies $(A_{p_k})$ with constants $c_0, \delta$. For any $\delta' < \delta$, there exist constants $d_{\infty}^{(i)},$ $1 \leq i \leq 2$, independent of $x_0$, and a constant $C_{x_0, \delta'} \lesssim_{\delta'} (1+x_0)e^{(1+\delta')x_0}$, such that 
    \begin{align}
        \sup_{ t \geq 0 } \sum_{j=0}^2 \| e^{(1+\delta')t}\partial_t^j \big(\psi &-   d_\infty^{(1)} r_k -  d_\infty^{(2)} (\mr{\phi}(x) - k (t-x)) r_k \big)\|_{C^{2-j}_x([0,x_0])} \nonumber \\[\jot]
        &\leq C_{x_0, \delta'} \bigg(|c_0| + \sup_{t \geq 0} \sum_{j=0}^3 e^{\frac32 t}\| \partial_t^j F_2 \|_{\mathcal{D}^{p_k+\delta,0}_\infty(\mathbb{R}_+)} \bigg). \label{eq:resonanceexp1}
    \end{align}
    Moreover, $|d^{(i)}_\infty|$, $1 \leq i \leq 2$ can be estimated by the right hand side of (\ref{eq:resonanceexp1}).

    \vspace{.5em}
    \noindent
    \underline{$1< \beta < p_k$, $(A_{\beta})$: } Assume $F(t,x)$ satisfies $(A_{\beta})$ with constants $c_0, \delta$. For any $\delta' < \delta$, there exists a $L^\infty_{loc}(\mathbb{R}_+)$ function $g_\infty(x)$ and a constant $C_{x_0,\delta'} \lesssim_{\delta'} (1+x_0)e^{(\beta+\delta')x_0}$ such that 
    \begin{align}
        \sup_{t>0}\sum_{j=0}^2 \|e^{(\beta+\delta')q_k t} \partial_t^j (\psi - &g_\infty(x)e^{-\beta q_k t}) \|_{C^{2-j}_x([0,x_0])} \nonumber\\[\jot]
        &\leq  C_{x_0,\delta'}\bigg(|c_0| +  \sup_{t \geq 0} \sum_{j=0}^3 e^{\frac32 t} \|\partial_t^j F_2 \|_{\mathcal{D}^{\beta+\delta,0}_\infty(\mathbb{R}_+)} \bigg). \label{eq:resonanceexp1.1}
    \end{align}
    Moreover, $\|g_{\infty}\|_{L^\infty([0,x_0])}$ can be estimated by the right hand side of (\ref{eq:resonanceexp1.1}), and provided $c_0 \neq 0$, the function $g_\infty(x)$ does not identically vanish.
    \end{prop}

\begin{proof}[Proof of Proposition \ref{prop:resonanceexp}]
By a density argument, it suffices to establish (\ref{eq:resonanceexp})--(\ref{eq:resonanceexp1.1}) for $F(t,x) \in C^\infty_{t,x}(\mathbb{R}\times \mathbb{R}_+)$. Given that $F(t,x)$ is supported in $\{t \geq 1\}$, we may then extend $\psi(t,x)$ to $\mathbb{R}\times \mathbb{R}_+$ by $\psi(t,x) = 0$ on $\{t \leq 0\}$, and then $\psi(t,x) \in C^\infty_{t,x}(\mathbb{R}\times \mathbb{R}_+)$.

For $\sigma \in \mathbb{C}$, let $\wh{\psi}(\sigma,x)$ denote the Fourier-Laplace transform of $\psi$ in the $t$ variable, where we use the sign convention
\begin{equation}
    \wh{\psi}(\sigma,x) = \int_{\mathbb{R}}e^{i \sigma t} \psi(t,x) dt.
\end{equation}
We have established boundedness in $t$ for $\psi(t,x)$ in Lemma \ref{lem:pointwiseboundforeasywaveequation}, and thus this transform is well-defined for all $\Im \sigma > 0$. For fixed $\sigma$ in the upper half plane we have the regularity $\wh{\psi}(\sigma,x) \in C^\infty_x(\mathbb{R}_+)$, as well as the bound $\sup_{x}|(1+\sqrt{x})^{-1}\wh{\psi}(\sigma,x)| < \infty$. By (\ref{52:forcingassTOTAL}) the transform $\hat{F}(\sigma,x)$ is well-defined for $\{\Im \sigma > -\frac32\}$, and $\sup_{x}|e^{\beta q_k x}\hat{F}(\sigma,x)| < \infty$ holds. In the following, we write $\wh{\psi}_{\sigma}(x) \doteq \wh{\psi}(\sigma,x)$, and similarly $\wh{F}_{\sigma}(x) \doteq \wh{F}(\sigma,x)$.

It follows that $\wh{\psi}_{\sigma}(x)$ solves 
\begin{align*}
    \begin{cases}
    -\frac{d^2}{dx^2}\wh{\psi}_{\sigma}(x) + (4q_k e^{-2q_k x}V_k(x) -\sigma^2)\wh{\psi}_{\sigma}(x) = \wh{F}_\sigma(x) \\[\jot]
    \wh{\psi}_{\sigma}(0)=0,
    \end{cases}
\end{align*}
with at most polynomial growth at infinity. By Lemma \ref{lem:resolventuniquness} we must have
\begin{equation*}
    \wh{\psi}_{\sigma}(x) = (R(\sigma) \wh{F}_\sigma)(x).
\end{equation*}
The Fourier-Laplace inversion formula now provides the following representation formula for $\psi$ on $\{x \leq x_0\}$: 
\begin{align}
    \label{eq:resonanceexp1.5}
    \rho_{x_0}(x)\psi(t,x) = \frac{1}{2\pi}\int_{\{\Im \sigma = \eta \}} e^{-i\sigma t}(\rho_{x_0}R(\sigma)\wh{F}_\sigma)(x) d\sigma.
\end{align}
The integration takes place on a horizontal contour in the upper half plane.
As $F(t,x)$ is smooth in the $t$ coordinate, $\wh{F}_\sigma(x)$ is rapidly decaying in $\sigma$ for $|\Re \sigma| \gg 1$. Therefore, the integral converges pointwise in $x$. 

We now specialize to the case $\beta > p_k$. Arguments for the remaining cases follow similarly, and are sketched at the end of the proof. In the high regularity setting with $\beta > p_k$ (and therefore $\beta_* > 1$), the integrand in (\ref{eq:resonanceexp1.5}) is defined and meromorphic in $\sigma$ for $\sigma \in \mathbb{I}_{(-\beta_*,\eta]}$. Choosing $\epsilon_*$ small such that $-\beta_* + \epsilon_* < -1$, the goal is to deform the contour of integration in (\ref{eq:resonanceexp1.5}) to $\{\Im \sigma = -\beta_* + \epsilon_*\}$, picking up a contribution from the unique pole at ${\sigma = -i}$.

Fix a small constant $\epsilon_1$ and a large constant $R > 0$, and define the oriented contours
\begin{align}
    \label{eq:resonanceexp_prooftemp0}
    \begin{cases}
    \gamma_{\epsilon_1} = -i +  \epsilon_1 e^{-i \theta} \hspace{4em} \theta \in [0,2\pi) \\
    \Gamma^{\pm}_{1,R} = \eta i  \pm R \pm t \hspace{4em} t \in [0,\infty) \\
    \Gamma^{\pm}_{2,R} = \eta i \pm R - i t \hspace{3.65em} t \in [0, \eta + \beta_* - \epsilon_1] \\ 
    \Gamma_{3,R} = (-\beta_* + \epsilon_*)i + t \hspace{1.85em}   t \in [-R,R]
    \end{cases}
\end{align}
as well as the path $P_R \subset \mathbb{C}$
\begin{align*}
    P_R \doteq -\Gamma^{-}_{1,R} \cup \Gamma^{-}_{2,R} \cup \Gamma_{3,R} \cup -\Gamma^{+}_{2,R} \cup \Gamma^{+}_{1,R}.
\end{align*}
We may deform the contour in (\ref{eq:resonanceexp1.5}) in a compact subset of $\mathbb{C}$, giving 
\begin{align}
    \label{eq:resonanceexp_prooftemp.5}
    \rho_{x_0}(x)\psi(t,x) &= \frac{1}{2\pi}\int_{\gamma_{\epsilon_1}}e^{-i\sigma t}(\rho_{x_0}R(\sigma)\wh{F}_\sigma)(x) d\sigma + \frac{1}{2\pi}\int_{P_R}e^{-i\sigma t}(\rho_{x_0}R(\sigma)\wh{F}_\sigma)(x) d\sigma.
\end{align}
Appealing to the analysis (\ref{eq:residueatpole})--(\ref{eq:residueatpole1}) of $R(\sigma)$ near ${\sigma=-i}$, the integral over the closed loop $\gamma_{\epsilon_1}$ is explicitly computable, giving  
\begin{equation}
    \label{eq:resonanceexp_prooftemp1}
    -i c_* \rho_{x_0}(x)e^{x}e^{-t}\mr{r}(x)\int_0^\infty e^{x'}\mr{r}(x') \wh{F}_{-i}(x') dx' \doteq c^{(0)}_{\infty} e^{x-t}\rho_{x_0}(x) \mr{r}(x),
\end{equation}
where we have defined the constant $c_{\infty}^{(0)}$ appearing in (\ref{eq:resonanceexp}). Moreover, we can recognize $e^{x-t}\mr{r}(x)$ as $r_k(t,x)$. The integral over $P_R$ will contribute a faster decaying error, which we now estimate. The resolvent estimates (\ref{eq:c0toc2bounds}) are instrumental in bounding derivatives of the resolvent; however, each derivative loses a power of $\sigma$, obstructing convergence of the integrals for $|\Re \sigma| \gg 1$. In this region we regain favorable powers of $|\sigma|$ by exploiting the $t$ regularity of $F(t,x)$. The relevant statement is given in Lemma \ref{lem:associatedlemma_resonanceexp}, proved below. For $0 \leq m + j \leq 2$ apply (\ref{eq:c0toc2bounds}) and (\ref{eq:resonanceexp2}) to give
\begin{align*}
    \sup_{x \in [0,x_0]}\Big| \frac{1}{2\pi} \partial_t^m \partial_x^j \bigg( \int_{ P_R \setminus \Gamma_{3,R}}e^{-i\sigma t}(\rho_{x_0}R(\sigma)\wh{F}_\sigma)(x) d\sigma  \bigg) \Big| &\lesssim_{x_0}  \int_{ P_R \setminus \Gamma_{3,R}}  (1+|\sigma|) \|\hat{F}_\sigma\|_{\mathcal{D}^{\beta,0}_\infty(\mathbb{R}_+)} d\sigma  \\[\jot]
     &\lesssim_{x_0}  \sup_{t \geq 0}\sum_{j=0}^3e^{\frac32 t} \| \partial_t^j F\|_{\mathcal{D}^{\beta,0}_\infty(\mathbb{R}_+)} \int_{ P_R \setminus \Gamma_{3,R}} (1+|\sigma|)^{-2}  d\sigma \\[\jot]
    &\lesssim_{x_0} R^{-1} \sup_{t \geq 0}\sum_{j=0}^3e^{\frac32 t} \| \partial_t^j F\|_{\mathcal{D}^{\beta,0}_\infty(\mathbb{R}_+)} .
\end{align*}
For fixed $x \leq x_0$, the $C^2_{t,x}$ norm of these boundary integrals with $|\Re \sigma| \geq R$ vanishes as $R \rightarrow \infty$. It remains to estimate the term along $\Gamma_{3,R}$. For this term we also track the dependence on the cutoff $x_0$. For $0 \leq m + j \leq 2$ compute 
\begin{align*}
    \sup_{x \in [0,x_0]}\Big| \frac{1}{2\pi} \partial_t^m \partial_x^j \bigg(\int_{\Gamma_{3,R}} e^{-i\sigma t}(\rho_{x_0}R(\sigma)\wh{F}_\sigma)(x) d\sigma \bigg) \Big|
    &\lesssim (1+x_0)e^{(\beta_*-\epsilon_*)x_0} e^{-(\beta_*-\epsilon_*)t} \int_{\Gamma_{3,R}} (1+|\sigma|) \|\hat{F}_\sigma\|_{\mathcal{D}^{\beta,0}_\infty(\mathbb{R}_+)} d\sigma  \\[\jot]
    &\lesssim (1+x_0)e^{(\beta_*-\epsilon_*)x_0} e^{-(\beta_*-\epsilon_*)t} \sup_{t \geq 0} \sum_{j=0}^3e^{\frac32 t} \| \partial_t^j F\|_{\mathcal{D}^{\beta,0}_\infty(\mathbb{R}_+)}.
\end{align*} 
This completes the proof of (\ref{eq:resonanceexp}). To complete the case $\beta > p_k$, we observe that an estimate for $c_{\infty}^{(0)}$ follows from the integral expression (\ref{eq:resonanceexp_prooftemp1}), and a pointwise estimate for $\hat{F}_{-i}(x)$ given by (\ref{eq:resonanceexp2}).

We next consider the cases $\beta = p_k$ and $1< \beta < p_k$. Note that $\beta_* = \beta q_k \leq 1$. Under the assumption $(A_{\beta})$, the inhomogeneity $F(t,x)$ decomposes as $F(t,x) = c_0 \chi_b(t)e^{\beta q_k (t-x)} + F_2(t,x)$, for an error $F_2(t,x)$ with strictly better spatial decay. By linearity we may separately estimate the solution to (\ref{eq:forcedeqn1}) with right hand side $F_2(t,x)$. The required estimates (\ref{eq:resonanceexp1}), (\ref{eq:resonanceexp1.1}) follow by analogous techniques as above, precisely because $F_2(t,x)$ has more than enough decay to deform the contour to $\{\Im \sigma = -\beta q_k - \delta'\}$, for some $\delta'$ small enough. 

Assume that $F(t,x) = c_0 \chi_b(t)e^{\beta q_k (t-x)}$, and compute $\hat{F}_\sigma(x) = c_0 \hat{G}(\sigma) e^{-\beta q_k x}$ for an appropriate holomorphic function $\hat{G}(\sigma)$. Let the contours be as in (\ref{eq:resonanceexp_prooftemp0}), with the modification that $\Gamma^{\pm}_{2,R}, \Gamma_{3,R}$ extend only to $\{\Im \sigma = -\beta q_k - \delta' \}$. Equation (\ref{eq:resonanceexp_prooftemp.5}) continues to hold, and the integral over $P_R$ is estimated to give $t$ decay of $e^{-(\beta q_k+\delta')t},$ as expected. 

The remaining contribution arises from the poles of $\rho_{x_0}R(\sigma)e^{-\beta q_k x}$ in $\{\Im \sigma \geq -\beta q_k - \delta'\}$, which were studied in Proposition \ref{lem:extR_sigma_2}. There is either a double pole at $\sigma = -i$ (in the case $\beta = p_k$), or a simple pole at $\sigma = -\beta q_k i$ (in the case $1< \beta < p_k$). The appropriate resolvent expansions are given by (\ref{eq:resolvexp_extension}), (\ref{eq:resolvexp_extension_2}) respectively, along with the identities (\ref{eq:resolventexp_A1})--(\ref{eq:resolventexp_A2}) in the threshold case. Applying the residue theorem concludes the proof. We remark that in the threshold case, the residue at the double pole involves the precise form (including constants) of the differentiated resonance function (\ref{eq:F_-iidentified}).
\end{proof} 

The following lemma will complete the proof.
\begin{lemma}
    \label{lem:associatedlemma_resonanceexp}
    Assume $F(t,x)$ satisfies the conditions (\ref{52:forcingassTOTAL}). Then we have
        \begin{equation} 
            \label{eq:resonanceexp2}
            \sup_{\sigma \in \mathbb{I}_{(-\beta_*,\eta]}}(1+|\sigma|)^3\|\wh{F}_\sigma(x)\|_{\mathcal{D}^{\beta,0}_\infty(\mathbb{R}_+)} \lesssim \sup_{ t \geq 0 } \sum_{j=0}^{3}e^{\frac32 t}\|\partial_t^j F\|_{\mathcal{D}^{\beta,0}_\infty(\mathbb{R}_+)}.
        \end{equation}
    \end{lemma}
    \begin{proof}
      For $\sigma \in \mathbb{I}_{(-\beta_*,\eta]}, \ 0 \leq m \leq 3$,  estimate 
    \begin{align*}
        |\sigma|^m  |e^{\beta q_k x}\wh{F}_{\sigma}(x)| &=  |\sigma|^m \bigg| \int_{\mathbb{R}} e^{i\sigma t}e^{\beta q_k x} F(t,x)dt\bigg| \\[\jot] 
        &= |\sigma|^m \bigg| \int_{0}^\infty (i \sigma)^{-m}\partial_t^m e^{i\sigma t} e^{\beta q_k x} F(t,x)dt\bigg|\\[\jot]
        &=  \bigg|  \int_{0}^\infty  e^{i\sigma t} e^{\beta q_k x} \partial_t^m  F(t,x)dt\bigg| \\[\jot] 
        &\lesssim \sup_{ t \geq 0 }e^{\frac32 t} \|\partial_t^m F\|_{\mathcal{D}^{\beta,0}_\infty(\mathbb{R}_+)}.
    \end{align*}
    We have used (\ref{52:forcingassTOTAL}) to drop boundary terms arising from integration by parts in $t$. (\ref{eq:resonanceexp2}) now follows.
    \end{proof}

\subsection{Proof of Theorem \ref{theorem:mainthm}}
\label{subsec:proofcompleted_3}

\subsubsection*{Theorem \ref{theorem:mainthm}(a)}

The proof will essentially be an amalgamation of the techniques developed thus far. To clarify the conceptual structure, we split the argument into various steps.

\vspace{.5em}
\underline{Step 1: Reduction to spacelike problem} \
Let $(r\varphi)_0(v) \in C^{\alpha}_{(hor)}([-1,0])$, ${\alpha > p_k}$ denote given spherically symmetric, null initial data along $\{u=-1\}$, and let $\psi(u,v) \doteq (r\varphi)(u,v)$ denote the associated solution to (\ref{sec2.5:eq2}) in $\{u \geq -1\}$. By Proposition \ref{prop:backwardsdecay}, $\psi(u,v)$ may equivalently be described as the solution to (\ref{eq:wavehyper}) for some spacelike initial data $(r\varphi(x), \partial_t(r\varphi)(x)) = (f_0(x),f_1(x)) \in {D}^{\alpha,5}_\infty(\mathbb{R}_+) \times {D}^{\alpha,4}_\infty(\mathbb{R}_+)$ along $\{t=0\}$. The a priori estimate (\ref{eq:propbackscat_1}) holds in any compact (in $t$) region. 

\vspace{.5em}
\underline{Step 2: Application of leading order resonance expansion} \
To apply Proposition \ref{prop:resonanceexp}, we need to relate this initial data problem with a forcing problem of the form (\ref{eq:forcedeqn1}). Let $\chi(t) \in C^\infty(\mathbb{R})$ denote an increasing cutoff function, with $\chi(t) = 0$ for $t < 1$, and $\chi(t)=1$ for $t \geq 2$. Define
$$\psi_c \doteq \chi(t) \psi(t,x),$$
as well as the operator 
\begin{equation*}
    \mathcal{T}_k \doteq \partial_t^2 - \partial_x^2 + 4q_k e^{-2 q_k x}V_k(x).
\end{equation*}
A direct computation yields that $\psi_c$ obeys (\ref{eq:forcedeqn1}) with right hand side
\begin{align}
    F(t,x) &= \mathcal{T}_k(\chi \psi) - 4q_k e^{-2q_k x}\chi \mathcal{E}_{p,0}(t,x) \nonumber \\
    &= \underbrace{[\mathcal{T}_k, \chi]\psi}_{g_0(t,x)} - \underbrace{4q_k e^{-2q_k x}\chi \mathcal{E}_{p,0}(t,x)}_{g_1(t,x)}. \label{eq:forcingtemp}
\end{align}
The forcing decomposes as a sum of terms $g_0(t,x)$, $g_1(t,x)$. The former is supported in $t \in [1,2]$, and by the estimate (\ref{eq:propbackscat_1}) it follows that the assumptions (\ref{52:forcingassTOTAL}) hold for $\beta = \alpha > p_k$. 

The term $g_1(t,x)$ arises due to the deviation of the geometry from exact $k$-self-similarity. Although rapidly decaying in $t$ on compact sets $\{x \leq x_0\}$, this decay is not uniform in $x$, and the final assumption of (\ref{52:forcingassTOTAL}) is not satisfied for any $\beta > p_k$. However, we claim Proposition \ref{prop:resonanceexp} remains valid with this term included, and leave a justification to the end of the proof.

Fix a constant $\beta_*' \in (1, \beta_*)$. By Proposition \ref{prop:resonanceexp}, there exists a constant $c_{\infty}$ independent of $x_0$ such that for all $\{x \leq x_0, \ t \geq 2 \}$ we have\footnote{Note we may replace $r_k$ with $r$ in (\ref{eq:resonanceexp}), by the assumed rapid decay of ${r - r_k}$.}
\begin{equation}
    \label{eq:completingprooftemp1}
    r(t,x) \big(\varphi(t,x) - c_{\infty}\big)= e^{-\beta_*' t}\mathcal{E}_{x_0}(t,x),
\end{equation}
for an error $\mathcal{E}_{x_0}(t,x)$ satisfying
\begin{equation}
    \label{eq:completingprooftemp2}
   \sup_{ t \geq 2 } \| \mathcal{E}_{x_0}(t,x) \|_{C^2_{x}([0,x_0])} \leq C e^{\beta_*' x_0},
\end{equation}
where $C$ depends on the $C^{\alpha}_{(hor)}([-1,0])$ norm of initial data, but not on $x_0$. 

\vspace{.5em}
\underline{Step 3: Near-axis decay} \
Next, we consider the decay (\ref{eq:completingprooftemp1})--(\ref{eq:completingprooftemp2}) in similarity coordinates. For fixed $z_0 \in (-1,0)$, define the near-axis region $\mathcal{S}_{z_0}$ and near-horizon region $\mathcal{D}_{z_0}$
\begin{align*}
    \mathcal{S}_{z_0} \doteq \{s \geq 2, \ -1 \leq z \leq z_0 \}, \quad \mathcal{D}_{z_0} \doteq \{s \geq 2, \ z_0 \leq z \leq 0\}.
\end{align*}
Defining $x_0 = -\frac{1}{2q_k}\ln |z_0|$, we have $\mathcal{S}_{z_0} \subset \{t \geq 2, \ x \leq x_0 \}$, and from (\ref{eq:completingprooftemp1})--(\ref{eq:completingprooftemp2}), the averaging estimate (\ref{eq:averagingest}), as well as direct integration of (\ref{eq:wavesim}) we conclude
\begin{equation}
    \label{eq:completingprooftemp_est1}
    \|e^{\beta_*' s} r\big(\varphi - c_{\infty}\big)  \|_{C^2_{s,z}(\mathcal{S}_{z_0})} + \|e^{(\beta_*'-1) s} \varphi \|_{C^1_{s,z}(\mathcal{S}_{z_0})} \leq_{z_0} C.
\end{equation}
The main tool in propagating control to $\mathcal{D}_{z_0}$ will be multiplier estimates. There is, however, a pointwise estimate on $\varphi$ which is uniform up to $\{z=0\}$, and which we may already close. To see this, we retain the explicit dependence on $x_0$, and estimate the right hand side of (\ref{eq:completingprooftemp1}) in similarity coordinates. Observing $e^{\beta_*' x} e^{-\beta_*' t} = e^{-\beta_*' s}$, we find that 
\begin{equation}
    \label{eq:completingprooftemp_est2}
    \|e^{\beta_*' s} r\big(\varphi - c_{\infty}\big) \|_{L^{\infty}(\mathcal{D}_{z_0}) } \leq C.
\end{equation} 

\vspace{.5em}
\underline{Step 4: Near-horizon decay} \
We are now in a position to apply our multiplier estimates. Let $\rho \in (q_k-\beta_*', -k^2)$, and define $\psi_\rho \doteq e^{(q_k-\rho)s}r (\varphi - c_\infty)$, which satisfies the (commuted) equation (\ref{eq:wavesim_alpha_comm0}). We repeat a form of this equation here for convenience, where $V$ is the potential defined in (\ref{eq:combinedpotentials}):
\begin{equation}
    \label{eq:wavesim_alpha_comm0_rewritten}
    \partial_s \partial_z^2 \psi_\rho - q_k |z| \partial_z^3 \psi_\rho + (q_k + \rho)\partial_z^2 \psi_\rho + V \partial_z \psi_\rho + V' \psi_\rho = 0.
\end{equation}
Introduce a parameter $\omega$, satisfying $\omega = 0$ in the case $\alpha \geq \frac32$, and $\omega = \frac32 - \alpha + \epsilon$ in the case $p_k < \alpha \leq \frac32$, where $\epsilon \ll 1$ is sufficiently small. In the former, we will be able to directly close the multiplier estimate (\ref{eq:energyest2_0}). Note the coefficient of the bulk term in this estimate can be made non-negative for $\rho$ sufficiently small, but \textit{strictly less}\footnote{Recall that establishing boundedness of $\psi_\rho$ with ${\rho = -k^2}$ corresponds to proving self-similar bounds. As we wish to propagate better than self-similar bounds, ${\rho < -k^2}$ is a necessary condition.} than $-k^2$. Moreover, the boundary integral appearing on the right hand side of (\ref{eq:energyest2_0}) is supported in $\mathcal{S}_{z_0}$, and is thus already controlled. It follows that for $\alpha \geq \frac32$ we have 
\begin{equation}
    \label{eq:completingprooftemp_est3}
    \sup_{s' \geq 0 }\|\partial_z^2 \psi_\rho \|_{L^2_z(\{s=s'\})} \leq C, \quad \big(\alpha \geq \frac32\big).
 \end{equation}
We turn now to establishing (\ref{eq:completingprooftemp_est3}) in the case $p_k < \alpha \leq \frac32$. The above argument does not apply uniformly in $\alpha$ as $\alpha-p_k \rightarrow 0$, as the bulk term in (\ref{eq:energyest2_0}) becomes of size $O(k^2)$, without a definite sign.

However, we proceed by an analogous multiplier estimate, exploiting the additional control (\ref{eq:completingprooftemp_est2}). Multiplying (\ref{eq:wavesim_alpha_comm0_rewritten}) by $|z|^{2\omega}\partial_z^2 \psi_\rho$ and integrating by parts in $\mathcal{R}(s_0)$ yields (compare with (\ref{eq:temp2ndordermultiplier}))
\begin{align}
        \int_{\{s=s_0\}}&|z|^{2\omega}(\partial_z^2\psi_\rho)^2 dz + (q_k(1-2\omega)+2\rho)\iint_{\mathcal{R}(s_0)} |z|^{2\omega}(\partial_z^2 \psi_\rho)^2 dsdz \nonumber  \\[\jot]
        &\leq \int_{\{s=0\}}|z|^{2\omega}(\partial_z^2\psi_\rho)^2 dz - \underbrace{\iint_{\mathcal{R}(s_0)} 2 V |z|^{2\omega}\partial_z \psi_\rho \partial_z^2 \psi_\rho dzds}_{\text{I}} - \underbrace{\iint_{\mathcal{R}(s_0)} 2 V' |z|^{2\omega}\psi_\rho \partial_z^2 \psi_\rho dzds}_{\text{II}} . \label{eq:completingprooftemp2.5}
\end{align}
As $\alpha > p_k$, one can choose $\rho < -k^2$ and $\epsilon \ll 1$ such that the bulk term on the left hand side of the estimate is positive. For the term denoted $\text{I}$, integrate by parts in $z$ to give
\begin{align}
    \text{I} &= - \int_{\Gamma_{s_0}} V (\partial_z \psi_\rho)^2 ds - \iint_{\mathcal{R}(s_0)} \partial_z \big(V |z|^{2\omega}  \big) (\partial_z \psi_\rho)^2 dzds \nonumber \\[\jot]
    &= - \int_{\Gamma_{s_0}} V (\partial_z \psi_\rho)^2 ds - \iint_{\{z \leq z_k \}} \partial_z \big(V |z|^{2\omega}  \big) (\partial_z \psi_\rho)^2 dzds - \underbrace{\iint_{\{z \geq z_k \}} \partial_z \big(V |z|^{2\omega}  \big) (\partial_z \psi_\rho)^2 dzds}_{\text{I}_{good}}, \label{eq:completingprooftemp3}
\end{align}
where $z_k$ is the value defined in Proposition \ref{lem:propertiesofV}. From the repulsivity estimate (\ref{eq:potentialest_repulsivity}), it follows that the only term in (\ref{eq:completingprooftemp3}) which is supported near the horizon, namely $\text{I}_{good}$, has a good sign! The remaining terms are supported near the axis, and are controlled by (\ref{eq:completingprooftemp_est1}).

For $\text{II}$ we exploit the estimate 
\begin{equation}
    \| \psi_{\rho} \|_{L^2_{s,z}(\mathcal{R}(s_0))} \leq C,
\end{equation}
which follows by the choice of $\rho$ and (\ref{eq:completingprooftemp_est2}) above. Therefore, for a small constant $\wt{\epsilon} \ll 1$ estimate
\begin{align*}
    |\text{II}| & \lesssim \wt{\epsilon}  \iint_{\mathcal{R}(s_0)}  |z|^{2\omega} (\partial_z^2 \psi_{\rho})^2 dzds +  \wt{\epsilon}^{-1} \iint_{\mathcal{R}(s_0)} |z|^{2\omega} \psi_{\rho}^2 dzds \\[\jot]
    & \lesssim \wt{\epsilon}  \iint_{\mathcal{R}(s_0)}  |z|^{2\omega} (\partial_z^2 \psi_{\rho})^2 dzds +  C \wt{\epsilon}^{-1} .
\end{align*}
The first term is controlled by the remaining positive bulk term appearing on the left hand side of (\ref{eq:completingprooftemp2.5}). We therefore conclude
\begin{equation}
    \label{eq:completingprooftemp_est4}
    \sup_{s' \geq 0 } \big\||z|^{\omega}\partial_z^2 \psi_\rho \big\|_{L^2_z(\{s=s'\})} \leq C, \quad \big(\alpha \leq \frac32\big),
 \end{equation}
 for some ${\rho < -k^2}$. Pointwise bounds on $\partial_z \psi_\rho$ in $\mathcal{D}_{z_0}$ now follow by integrating (\ref{eq:completingprooftemp_est3}), (\ref{eq:completingprooftemp_est4}), and using the control in $\mathcal{S}_{z_0}$. Similarly, $\partial_s \psi_\rho$ is estimated by integrating the equation (\ref{eq:wavesim}) itself in the $z$ direction. After translating to double-null coordinates, we conclude the proof.

 \vspace{.5em}
\underline{Step 5: Analysis of forcing term $g_1(t,x)$ } \
It remains to revisit the proof of Proposition \ref{prop:resonanceexp}, and argue that the inhomogeneous term $g_1(t,x)$ does not affect the validity of the resonance expansion. This will justify (\ref{eq:completingprooftemp1})--(\ref{eq:completingprooftemp2}), completing the proof of the theorem.

We may write 
\begin{equation*}
    g_1(t,x) = c_k e^{-2q_k x} \chi(t) V_{k,p}(t,x) \psi(t,x) \mathbbm{1}_{\{t - x \geq -s_0\}},
\end{equation*}
for constants $c_k, s_0$ and where the potential $V_{k,p}(t,x)$ is defined in (\ref{eq:defn_Vkp}). The characteristic function of the set $S \doteq \{t - x \geq -s_0\}$ captures the support of $V_{k,p}(t,x)$. The goal will be to further decompose $g_1(t,x) = g_{1,a}(t,x) + g_{2,a}(t,x)$, where $g_{1,a}$ captures the leading order behavior near $\{z=0\}$. More precisely, let $\psi_h(s)$, $V_{k,p,h}(s)$ denote the restrictions of $\psi, V_{k,p}$ to $\{z=0\}$, and for any $(t,x) \in S$ write 
\begin{align*}
    \psi(t,x) &= \psi_h(t-x) + \psi_{e}(t,x) e^{-2q_k x}, \\
    V_{k,p}(t,x) &= V_{k,p,h}(t-x) + V_{k,p,e}(t,x)e^{-2 q_k x},
\end{align*}
for functions $\psi_{e}(t,x), V_{k,p,e}(t,x)$, which satisfy bounds 
\begin{align}
    \label{eq:completingprooftemp3.5}
    \sum_{j=0}^3 \| \partial_t^j \psi_e \|_{L^\infty(\{t \geq 0 \})} + \sum_{j=0}^3 \|  \partial_t^j V_{k,p,e} \|_{L^\infty(\{t \geq 0 \})} \leq C.
\end{align}
Define 
$$g_{1,a}(t,x) \doteq c_k e^{-2q_k x}\chi(t) V_{k,p,h}(t-x)\psi_e(t-x) \mathbbm{1}_{\{t - x \geq -s_0\}}, $$
and $g_{2,a}(t,x) \doteq g_1(t,x) - g_{1,a}(t,x)$. By the rapid spatial decay of $g_{2,a}(t,x)$ and the $t$ decay of $V_{k,p}$, it follows that $g_{2,a}$ in fact satisfies the assumptions of Proposition \ref{prop:resonanceexp} directly. It remains to consider $g_{1,a}$. 

Computing the Fourier-Laplace transform of $g_{1,a}(t,x)$ gives for $\sigma \in \mathbb{I}_{[-1-\eta,\eta]}$ 
\begin{align}
    \hat{g}_{1,a}(\sigma,x) &= c_k e^{-2q_k x} \mathbbm{1}_{\{x - s_0 \leq 1\}} \int_{\{t \geq 1\}} e^{i \sigma t} V_{k,p,h}(t-x)\psi_e(t-x) dt \nonumber \\[\jot]
    &\hspace{2em}+ c_k e^{-2q_k x} \mathbbm{1}_{\{x - s_0 \geq 1\}} \int_{\{t \geq x-s_0\}} e^{i \sigma t} V_{k,p,h}(t-x)\psi_e(t-x) dt \nonumber \\[\jot]
    &= c_k e^{-2q_k x}e^{i \sigma x} \mathbbm{1}_{\{x - s_0 \leq 1\}} \int_{\{x+s' \geq 1\}} e^{i \sigma s'} V_{k,p,h}(s')\psi_e (s') ds'\nonumber \\[\jot]
    &\hspace{2em}+ c_k e^{-2q_k x}e^{i \sigma x}  \mathbbm{1}_{\{x - s_0 \geq 1\}} \int_{\{s' \geq -s_0\}} e^{i \sigma s'} V_{k,p,h}(s')\psi_e(s') ds' \nonumber \\[\jot]
    &\doteq d(\sigma) e^{-2q_k x}e^{i \sigma x} + f(\sigma,x)\mathbbm{1}_{\{x - s_0 \leq 1\}}, \label{eq:completingprooftemp4}
\end{align}
where $d(\sigma)$ is rapidly decaying in $\sigma$ (i.e. $|d(\sigma)| \lesssim |\sigma|^{-3}$), due to the regularity (\ref{eq:completingprooftemp3.5}) and the decay of $V_{k,p,h}(s')$ as $s' \rightarrow \infty$. The remaining term in (\ref{eq:completingprooftemp4}) is compactly supported in $x$, and rapidly decaying in $\sigma$.

In order to verify that the argument of Proposition \ref{prop:resonanceexp} goes through for forcing terms $g_{1,a}$, we must verify that the application of the cutoff resolvent $\rho_{x_0} R(\sigma) \hat{g}_{1,a}(\sigma,x)$ defines, for fixed $x$, a meromorphic function of $\sigma$ with the pole structure and estimates guaranteed by Proposition \ref{prop:residueofR}. That this holds for the first term in (\ref{eq:completingprooftemp4}) is a consequence of Lemma \ref{lem:resolventon_eqx}. For the latter term, we have the requisite decay in $x$ to apply Proposition \ref{prop:residueofR} directly.

\subsubsection*{Theorem \ref{theorem:mainthm}(b)}
We now turn to the case of threshold initial data $(r\varphi)_0(v) \in \mathcal{C}^{p_k k^2, \delta}_{(hor)}([-1,0])$, $\delta > 0$. As above, let $\psi(u,v) \doteq (r\varphi)(u,v)$ denote the associated solution to (\ref{sec2.5:eq2}) in $\{u \geq -1\}$. In the case that the background geometry is exactly $k$-self-similar, the proof of Theorem \ref{theorem:mainthm}(b) will follow essentially as a corollary of Theorem \ref{theorem:mainthm}(a) above. We present this argument first, and only after turn to the case of general geometries.

\vspace{.5em}
\underline{Proof on $k$-self-similar backgrounds}: \
The key observation is Lemma \ref{lem:holderspacedecomp}, asserting that general threshold data splits as the union of a one-dimensional subspace (spanned by $\mr{\phi}(z)$) and spaces of strictly higher regularity. In particular, there exists a constant $c_0$ such that $\varphi_0(z) - c_0 \mr{\phi}(z) \in \mathcal{C}^{p_k k^2 + \delta'}_{(hor)}([-1,0])$, for a $0 < \delta' \ll 1$ sufficiently small. By linearity and Theorem \ref{theorem:mainthm}(a), the evolution of the more regular component of data satisfies the conclusions of Theorem \ref{theorem:mainthm}(b), and it suffices to assume $\varphi_0(z) = \mr{\phi}(z)$. On the $k$-self-similar background, however, the solution to (\ref{sec2.5:eq2}) with such data is known explicitly to be $\phi_k(s,z) = \mr{\phi}(z) + k s$. We are now able to read off the desired conclusion, i.e. the convergence of $\varphi$ to a multiple of $\phi_k$ (up to constants).

\vspace{.5em}
\underline{Proof on a general admissible background:} \ 
The above argument relied on the knowledge of an exact solution $\phi_k$ to the wave equation with the prescribed regularity. On general $(\epsilon_0,k)$-admissible backgrounds such a solution is not available, and we instead approach the problem as in the proof of Theorem \ref{theorem:mainthm}(a).

As above, it suffices to take as initial data $\varphi_0(z) = \mr{\phi}(z)$. Applying Proposition \ref{prop:backwardsdecay} reduces the problem to the study of spacelike data $(f_0(x), f_1(x)) \in \wt{\mathcal{D}}^{p_k,5,\delta',c_0}(\mathbb{R}_+) \times \wt{\mathcal{D}}^{p_k,4,\delta',c_0}(\mathbb{R}_+)$ for constants $\delta, c_0$, with the latter determined explicitly from $\mr{\phi}$ via (\ref{eq:backwardsexplicitc_0}). Transforming to a forcing problem as in Step 2 above (the argument of Step 5 again applies here), we apply the resonance expansion (\ref{eq:resonanceexp1}) for threshold regularity data to conclude that in the region $\{x \leq x_0, \ t \geq 2 \}$, 
\begin{equation}
    \label{eq:completingprooftemp5}
    r(t,x)\big(\varphi(t,x)  - d_{\infty}^{(1)} - d_{\infty}^{(2)}(\mr{\phi}(x)-k(x-t)) \big) = e^{-(1+\delta'')t} \mathcal{E}_{x_0}(t,x),
\end{equation}
for appropriate constants $d_{\infty}^{(i)}$, $\delta''$, and an error $\mathcal{E}_{x_0}(t,x)$ satisfying the analagous estimate to (\ref{eq:completingprooftemp2}). The decay of $\varphi$ in the near-axis region is precisely as in Step 3 above, with the more general expansion (\ref{eq:completingprooftemp5}). 

We claim that for our specific choice of initial data, the constant $d_{\infty}^{(2)} = 1$. This is true on a $k$-self-similar background, as in that case the solution is exactly given by $\varphi(t,x) = \mr{\phi}(x) - k(x-t).$ In general, observe that the constant $d_{\infty}^{(2)}$ depends on the initial data only through the coefficient $c_0$ of $f_0(x)$ in an expansion $f_0(x) = c_0 e^{- x} + O(e^{-(1+\delta)x})$. This coefficient, given by (\ref{eq:backwardsexplicitc_0}), is clearly independent of the choice of background geometry.

Given this, we apply multiplier estimates to the quantity $\psi_{\rho} \doteq e^{(q_k-\rho)s}r(\varphi - d_{\infty}^{(1)} - \mr{\phi}(x) + k(x-t)).$ The corresponding equation (\ref{eq:wavesim_alpha_comm0_rewritten}) in this case has additional forcing term which rapidly decays in $|u|$, due to $r \mr{\phi}_k$ not being an exact solution to the wave equation; however, the multipliers are easily adapted to include this term. Crucially, the data for $\psi_{\rho}$ along $\{s=0\}$ is more regular than that of $\mr{\phi}(z)$, as the lowest regularity piece of initial data was removed in the definition of $\psi_\rho$. The argument now runs as in Step 4 above, as we are safely in the setting of above-threshold regularity.

\subsubsection*{Theorem \ref{theorem:mainthm}(c)}
We first prove the existence of data leading to the lower bound (\ref{eq:mainthmbound2.7}). Fix a parameter $\alpha \in (1,p_k)$, and recall the definition of Dirichlet solutions $f_{(dir),\sigma}(x)$ to $P_k(\sigma)f=0$ given in Section \ref{sec:resolventconstr}, cf. (\ref{def:dirichletsoln}). Let ${\rho_0 \doteq \alpha q_k} $, which satisfies $\rho_0 \in (q_k, 1)$ by assumption, and let ${\sigma_0 = -i \rho_0}$. Defining
\begin{equation}
    \psi_{mode}(t,x) \doteq e^{ - i \sigma_0 t} f_{(dir), \sigma_0 }(x),
\end{equation}
it follows that $\psi_{mode}(t,x)$ is a mode solution to (\ref{eq:wavehyper}) on a $k$-self-similar background. To determine the regularity of this solution in similarity coordinates, observe that as $\Im \sigma_0 \in (-1,-q_k)$, the set $\{f_{(out),\sigma_0}, f_{(in),\sigma_0}\}$ form a basis of solutions to $P_k(\sigma)f = 0$. Expanding $f_{(dir), \sigma_0 }(x)$ with respect to this basis and applying the argument in the proof of Lemma \ref{lem:regularityperspective} yields
\begin{equation}
    \label{eq:completingprooftemp6}
    \psi_{mode}(0,z) = c_{(in),\sigma_0}|z|^{\alpha} + \mathcal{E}(z),
\end{equation}
where $\mathcal{E}(z) \in C^{2-}([-1,0]) \cap \mathcal{C}^{\alpha'}_{(hor)}$ for any $\alpha' < 2$, and $c_{(in),\sigma_0}$ is non-vanishing. It follows that $ \psi_{mode}(0,z) \in \mathcal{C}^{\alpha,\delta}_{(hor)}([-1,0])$ for some $\delta >0$. 

It is immediate that $\psi_{mode}(0,z)$ is the desired data on a $k$-self-similar background. To show that the same is true on an $(\epsilon_0,k)$-admissible background, it suffices to apply the argument of Steps 1--3 in the proof of Theorem \ref{theorem:mainthm}(a), appealing now to the resonance expansion (\ref{eq:resonanceexp1.1}). That the unstable term in the resonance expansion is non-trivial is guaranteed by Proposition \ref{prop:backwardsdecay} and the non-triviality of the constant $c_{(in),\sigma_0}$ in (\ref{eq:completingprooftemp6}).

\vspace{.5em}
\noindent
We next discuss the estimate (\ref{eq:mainthmbound2.5}), holding for general data $\varphi_0(z) \in \mathcal{C}^{\alpha}_{(hor)}([-1,0])$, $\alpha \in (1,p_k)$. Applying Steps 1--3 of the proof of Theorem \ref{theorem:mainthm}(a) and the resonance expansion (\ref{eq:resonanceexp1.1}), we conclude the stated decay in the near-axis region. To propagate this decay towards $\{z=0\}$, we estimate the wave equation (\ref{eq:wavesim_alpha_comm0_rewritten}) for the quantity $\psi_{\rho} \doteq e^{(q_k -\rho)}r \varphi$, $\rho \in (q_k(1-\alpha),0)$, and proceed as in Step 4 above.

\subsection{Proof of Theorem \ref{theorem:angularthm}}
\label{subsec:proofcompleted_4}
Let $B \gg 1$ denote a large constant independent of $k$, and set $\gamma = B k^2$. Assume initial data $(r\varphi)_0(z,\omega) \in \mathcal{H}^{1,\gamma}_{(hor)}([-1,0]\times \mathbb{S}^2)$ to (\ref{eq:wavesimfull}) is given, which is without loss of generality supported on angular modes $1 \leq \ell \leq N$, for some $N \in \mathbb{Z}_{>0}$. The extension to general solutions follows from the quantitative estimates shown below.

It follows that the solution $\varphi(s,z,\omega)$ in $\{s \geq 0\}$ decomposes as a finite sum
\begin{equation*}
    \varphi(s,z,\omega) = \sum_{ \substack{1 \leq \ell \leq N \\ |m| \leq \ell }} \varphi_{m\ell}(s,z)Y_{m \ell}(\omega),
\end{equation*}
in which the individual terms have regularity $\varphi_{m\ell}(s',z) \in \mathcal{H}^{1,\gamma}_{(hor)}(\{s=s'\}).$ Fixing an $(m,\ell)$-mode, define the variable $\psi_{m\ell, \rho}(s,z) \doteq e^{(q_k -\rho)s} (r \varphi_{m \ell})(s,z)$. Along initial data $\psi_{m\ell, \rho}(0,z) = r \varphi_{m\ell}(0,z)$, and by the assumed regularity we have
\begin{equation*}
    \psi_{m \ell,\rho}(0,z) \in W^{1,2}_z([-1,0]), \ \  |z|^{\frac12 - \gamma}\partial_z^2\psi_{m \ell,\rho}(0,z) \in L^2_z([-1,0]).
\end{equation*} 
The strategy will be to apply the second order multiplier estimate (\ref{eq:energyest2_1}) to close a weighted bound on $\partial_z^2 \psi_{m \ell,\rho}$ in ${L^2_z(\{s=s'\})}$. By Sobolev inequalities and the fundamental theorem of calculus, pointwise bounds for lower order quantities will follow. Provided the estimates for individual modes have at most a polynomial dependence on $\ell$, we will be able to sum and conclude the desired estimates for the spacetime solution $\varphi(s,z,\omega)$.

Proceeding to the argument, we set $\omega = \frac12 - \gamma = \frac12 - B k^2$, and $\delta \ll 1$ let $\rho = -(1+ \delta)k^2$. Examining the multiplier estimate (\ref{eq:energyest2_1}), compute the coefficient of the bulk term to be  
\begin{equation}
    q_k(1-2\omega) - C_1 k^2 + 2\rho = (2 B q_k - C_1 - 2(1+\delta))k^2,
\end{equation}
which for $B \gtrsim C_1$ and $\delta \lesssim 1$, can be made non-negative. Fix choices of $B, \delta$ satisfying these constraints, and thus by (\ref{eq:energyest2_1}) (and a Hardy inequality) we control 
\begin{equation}
    \label{eq:thm1d:temp0}
    \sup_{s' \geq 0}\||z|^{\frac12 - \gamma} \br{r}^{-2}\partial_z \psi_{m \ell,\rho}^{(1)} \|_{L^2_z(\{s=s'\})}+ \||z|^{\frac12 - \gamma} \br{r}^{-2}\partial_z \psi_{m \ell,\rho}^{(1)} \|_{L^2_{s,z}(\{s \geq 0\})}\lesssim C_{m \ell},
\end{equation} 
where the constant $C_{m \ell}$ is controlled by the $\mathcal{H}^{1,\gamma}_{(hor)}$ norm of $\varphi_{m \ell}(0,z)$. Integrating (\ref{eq:thm1d:temp0}) repeatedly, applying Cauchy-Schwarz, and using the lower bound on $\br{r}$ away from the axis gives 
\begin{align}
    \sup_{s' \geq 0}\|\br{r}^{-\frac12}\partial_z \psi_{m \ell,\rho}\|_{L^\infty(\{s=s'\})} &+  \sup_{s' \geq 0} \|\br{r}^{-\frac32}\psi_{m \ell,\rho}\|_{L^\infty(\{s=s'\})} \lesssim C_{m \ell}. \label{eq:thm1d:temp1}
\end{align}
It remains to estimate $\partial_s \psi_{m \ell, \rho}$. It suffices to integrate the wave equation (\ref{eq:wavesim_alpha}) as a transport equation for $\partial_s  \psi_{m \ell, \rho}$, utilizing the control in (\ref{eq:thm1d:temp1}) and integration by parts. The result is  
\begin{align}
    \label{eq:thm1d:temp3}
    \sup_{s' \geq 0}  \|\partial_s \psi_{m \ell, \rho} \|_{L^\infty_z(\{s=s'\})} \lesssim (1 + \ell(\ell+1)) C_{m \ell}.
\end{align}
By directly integrating the angular potential term, we have acquired a constant depending on $\ell$. It remains to sum these estimates over $\ell$. For the spacetime solution $(r\varphi)(s,z,\omega)$ we compute, for any $s' \geq 0$,
\begin{align}
    e&^{2(q_k - \rho)s'}\Big(\sum_{0\leq i \leq 1}\| \partial_z^i  (r\varphi) \|^2_{L^\infty W^{2,2}_\omega(\{s=s'\}\times\mathbb{S}^2)} + \| \partial_s(r\varphi) \|^2_{L^\infty W^{2,2}_\omega(\{s=s'\}\times\mathbb{S}^2)}\Big) \nonumber \\[\jot]
    &\lesssim \sup_{z \in [-1,0]} \bigg(\Big\| \sum_{\substack{ 1 \leq \ell \leq N \\ |m| \leq \ell  }} \sum_{0 \leq i \leq 1} \partial_z^i\psi_{m \ell, \rho}(s',z)Y_{m \ell}(\omega) \Big\|^2_{W^{2,2}_\omega(\mathbb{S}^2)} +  \Big\| \sum_{\substack{ 1 \leq \ell \leq N \\ |m| \leq \ell  }} \partial_s\psi_{m \ell, \rho}(s',z)Y_{m \ell}(\omega) \Big\|^2_{W^{2,2}_\omega(\mathbb{S}^2)}  \bigg) \nonumber\\[\jot]
    &\lesssim  \sum_{\substack{ 1 \leq \ell \leq N \\ |m| \leq \ell  }}  (1+\ell(\ell+1))^4 C_{m \ell}^2 \nonumber \\[\jot]
    &\lesssim \|r\varphi_0 \|^2_{\mathcal{H}^{1,\gamma}_{(hor)}(\{s=0\}\times\mathbb{S}^2)}. \label{eq:thm1d:temp4}
\end{align}
A similar calculation gives the desired estimate for $\slashed{\nabla} (r\varphi)(s,z,\omega)$. To conclude the proof, it now suffices to rewrite these bounds in double-null coordinates.

\appendix

\section{Proof of Proposition \ref{prop:extinst}}
\label{section:app1}
The statements of stability and instability will both follow from a leading order expansion of the solution in a domain $\mathcal{Q}^{(ex),\delta} \doteq \mathcal{Q}^{(ex)} \cap \{\frac{v}{|u|^{q_k}} \leq \delta\}$, for $\delta \ll 1$. The leading order term is determined \textit{explicitly} as a function of the outgoing initial data $\partial_v \varphi (-1,v)$, and will obey either self-similar bounds (for threshold and above threshold regularities) or will be unstable (for below threshold regularity). 

We will require a weak statement of regularity for the $(\epsilon_0,k)$-admissible spacetime, namely 
\begin{equation}
    \label{app:eq1.5}
    \bigg|\frac{|\nu|}{r} - \frac{1}{|u|}\bigg| \lesssim \frac{v}{|u|^{q_k}}.
\end{equation} 
This estimate relies on (a) the gauge normalization of $r$ along $\{v=0\}$, and (b) the fact that $r, \nu \in C^1_{v}$ in a self-similar neighborhood of $\{v=0\}$, away from the singular point. 

We now turn to the argument. Let $f_0(v)$ denote spherically symmetric data for $\partial_v \varphi(-1,v)$, and define a quantity $\Psi(u,v)$ by solving the equation 
\begin{equation}
    \label{app:eq0}
    \begin{cases}
    \partial_u \Psi - \frac{1}{|u|}\Psi =0  \\[\jot]
    \Psi (-1,v) = f_0(v).
    \end{cases}
\end{equation}
Formally, this equation follows by restricting the coefficients of the wave equation (\ref{sec2.5:eq1.5}) to $\{v=0\}$, and setting ingoing derivatives (e.g. $\partial_u \varphi$) to zero.
Integrating this equation yields
\begin{equation}
    \label{app:eq0.5}
    \Psi(u,v) = |u|^{-1}f_0(v).
\end{equation}
Define $(\partial_v \varphi)_p \doteq \partial_v \varphi - \Psi$, which by (\ref{sec2.5:eq1.5}) satisfies the inhomogeneous equation:
    \begin{equation}
        \label{app:eq1}
        \begin{cases}
        \partial_u (\partial_v \varphi)_p - \frac{1}{|u|}(\partial_v \varphi)_p = \Big(\frac{|\nu|}{r} - \frac{1}{|u|}\Big)\partial_v \varphi - \frac{\lambda}{r} \partial_u \varphi \\[2\jot]
        (\partial_v \varphi)_p(-1,v) = 0 \\
        (\partial_u \varphi)_p(u,0) = 0.
        \end{cases}
\end{equation}
For a parameter $\eta \ll 1$, we propagate the bootstrap assumption 
\begin{align}
    \label{eq:appa:bootstrap}
        |\partial_v \varphi_p| + |\partial_v \varphi| \leq C \Big(\frac{1}{|u|^{1-k^2-\eta}} +  \frac{f_0(v)}{|u|}\Big).
\end{align}
For large enough $C$, this holds in a neighborhood of $(u,v) = (-1,0)$. We aim to improve the assumption in $\mathcal{Q}^{(ex),\delta}$, for $\delta$ small enough. Of course, it suffices to improve the estimate for $\partial_v \varphi_p$, given the explicit form of $\Psi$.

We first establish a bound for $\partial_u \varphi$. Integrating (\ref{sec2.5:eq1.5}) in $v$ from data, dropping zeroth order terms with favorable signs, and using the assumption (\ref{eq:appa:bootstrap}) gives
\begin{equation}
    \label{app:eq2}
    |\partial_u \varphi(u,v)| \leq \wt{C} \Big(\frac{v}{|u|^{2-k^2-\eta}} + \frac{f_0(v) v}{|u|^{2}}  \Big),
\end{equation}
for a constant $\wt{C} \gtrsim C$. Returning to (\ref{app:eq1}), we conjugate by $w(u,v) \doteq |u|^{1-\eta}\big(\frac{|u|^{q_k}}{v}\big)^{p_k \eta} = |u| v^{-p_k \eta}$ and insert the estimates (\ref{app:eq1.5}), (\ref{eq:appa:bootstrap}), (\ref{app:eq2}), giving 
\begin{equation}
    \label{app:eq3}
    |\partial_u \big( w(\partial_v \varphi)_p\big) | \leq \wt{C}\Big( \frac{v^{1-p_k \eta}}{|u|^{2-2k^2-\eta}} +  \frac{f_0(v)v^{1-p_k \eta}}{|u|^{2-k^2}} \Big).
\end{equation}
Provided $k^2 < 1$ and $\eta$ is chosen small, the powers of $|u|$ appearing in the denominators of (\ref{app:eq3}) are strictly greater than one. Integrating in $u$, we may drop terms along $\{u=-1\}$ to give
\begin{align*}
    |(\partial_v \varphi)_p| &\lesssim \wt{C}\Big( \frac{v}{|u|^{2-2k^2-\eta}} + \frac{f_0(v)v}{|u|^{2-k^2}}\Big) \\
    &\lesssim \wt{C} \delta \Big( \frac{1}{|u|^{1-k^2-\eta}} +  \frac{f_0(v)}{|u|}\Big),
\end{align*}
improving the bootstrap assumption for $\delta$ small enough. Note that by the choice of $\Psi$, no data terms appear above. We conclude that in $\mathcal{Q}^{(ex),\delta} $, there is an expansion\footnote{In this expression, $A \sim B$ means that $A \lesssim B$ and $B \lesssim A$ both hold.} 
\begin{align}
    \frac{1}{\lambda_k} \partial_v \varphi (u,v)&=  \frac{1}{\lambda_k}\Psi(u,v) + \frac{1}{\lambda_k}( \partial_v \varphi)_p(u,v) \nonumber \\[\jot]
    & \sim \frac{f_0(v)}{|u|^{1+k^2}} + O\big(\delta\frac{f_0(v)}{|u|^{1+k^2}}+ \delta \frac{1}{|u|^{1-\eta}} \big) \nonumber \\[\jot]
    &\sim \frac{f_0(v)}{|u|^{1+k^2}} + O \big(\frac{1}{|u|^{1-\eta}} \big). \label{app:eq4}
\end{align}
The approximation holds up to constants and functions uniformly bounded above and below in $\mathcal{Q}^{(ex),\delta} $. Moreover, we have reduced the question of self-similar bounds to understanding the term $\frac{f_0(v)}{|u|^{1+k^2}}$. Specialize to data of the form $f_0(v) = g(v)|v|^{\alpha-1}$, for a non-trivial $g(v)$ bounded above and below. Evaluating on the curve $\gamma(u) \doteq  (u,\delta |u|^{1-k^2})$ for $u \in [-1,0)$, it follows that if $\alpha < p_k $ (and thus $k^2 - (\alpha-1)(1-k^2) > 0$), 
\begin{align*}
    \bigg( |u| \frac{1}{\lambda_k} \partial_v \varphi\bigg)\bigg|_{\gamma(u)} \sim \frac{g(\delta|u|^{1-k^2})}{|u|^{k^2-(\alpha-1)(1-k^2)}} \rightarrow \infty,
\end{align*}
The instability statement of Proposition \ref{prop:extinst} follows.

If $\alpha = p_k$, the above expansion implies $\big| \frac{1}{\lambda_k} \partial_v \varphi \big| \lesssim |u|^{-1},$ i.e. this derivative satisfies self-similar bounds. From (\ref{app:eq2}), the same is true for $\partial_u \varphi$. Finally, if $\alpha > p_k$ then both terms in (\ref{app:eq4}) are better than self-similar; integrating from $\{v=0\}$ shows $\varphi \rightarrow 0$. 

The stability statement of Proposition \ref{prop:extinst} thus holds in a domain $\mathcal{Q}^{(ex),\delta}$ for $\delta$ sufficiently small. Extending self-similar (or better than self-similar) bounds to the region $\{\delta \leq \frac{v}{|u|^{q_k}} \leq 1 \} $ is a straightforward application of the methods of \cite{igoryak2, singh} (cf. the analysis of Region II in the latter).

\section{Asymptotic properties of $k$-self-similar interiors in the small-$k$ limit}
\label{section:app3}
The primary aim of this appendix is to develop a quantitative understanding of the $k$-self-similar metric quantities $r_k, m_k, \Omega_k$ in the interior region of the spacetime. These functions enter into the spherical and angular potentials $V_k(z)$, $L_{k,\ell}(z)$, respectively, which appear in the wave equation (\ref{eq:wavesim}). The estimates derived here are ultimately applied in the proofs of Propositions \ref{lem:propertiesofV}--\ref{lem:propertiesofL_kl}, which contain the information on the background spacetimes required for the analysis of the paper. 

The study of the small-$k$ limit is motivated by \cite{igoryak2,yakov1}, which show that a corresponding small-$\kappa$ regime in the setting of $\kappa$-self-similar vacuum spacetimes is well behaved. Key insights of the latter, which have direct analogs in spherical symmetry, include 1) quantitative convergence to the Minkowski metric as $\kappa \rightarrow 0$ \textit{away} from the similarity horizons $\{v=0\} \cup \{u=0\}$, and 2) the concentration of particular transversal double-null quantities\footnote{In vacuum, the relevant ``large'' quantity across $\{v=0\}$ is the outgoing shear, $\Omega^{-1}\hat{\chi}$. In our context it is $\partial_v \phi_k$.} in small self-similar neighborhoods of $\{v=0\} \cup \{u=0\}$. 

The strong monotonicity built into the Einstein-scalar field system, as well as the study of novel monotonic quantities associated to $k$-self-similarity, allow for a precise asymptotic expansion of all metric and scalar field quantities as $k\rightarrow 0$. Although not strictly necessary for this paper, we provide the first terms of such an expansion here, which may be of independent interest. As a byproduct of the expansions we are able to show the sharp regularity of $\phi_k$ near $\{v=0\}$, stated in Lemma \ref{lemma:backgroundprelims1}. This justifies the interpretation of the regularity $C^{1,p_k k^2}_v$ for null data as a threshold set by the regularity of $\phi_k$. 

We now turn to the analysis. First, we record the self-similar system of differential equations for the quantities $\mr{\Psi}(\hat{z})$ (cf. Lemma \ref{lem:meaningof_Psi_selfsimilar}). For a derivation of the system, see \cite{singh}. As it plays an important role in this section, recall the difference between the self-similar coordinate $\hat{z}$ and its renormalization $z$, related by $|z| \doteq |\hat{z}|^{q_k}.$ We often switch between coordinates; $\hat{z}$ is more useful for the study of the system (\ref{ss:SSESF1})--(\ref{ss:SSESF7}) below, whereas $z$ is the appropriate coordinate for expressing regular derivatives near $\{z = \hat{z} = 0\}$.

\begin{prop}
    The self-similar functions $\mr{\Psi}(\hat{z})$ satisfy the system of differential equations for $\hz \in [-1,\infty) \setminus \{0\}$
    \begin{align}
        \mr{r} \hat{z} \frac{d}{d\hat{z}}\mr{\lambda} &= -\mr{\nu}\mr{\lambda} - \frac{1}{4}\mr{\Omega}^2 \label{ss:SSESF1}\\ 
        \mr{r}  \frac{d}{d\hat{z}}\mr{\nu} &= -\mr{\nu}\mr{\lambda} - \frac{1}{4}\mr{\Omega}^2 \label{ss:SSESF2}\\ 
        2\mr{\Omega}^{-1}\mr{\nu}\hat{z} \frac{d}{d\hat{z}}\mr{\Omega} &= \hat{z} \frac{d}{d\hat{z}}\mr{\nu} + \mr{r}(\hat{z} \mr{\phi}' + k)^2 \label{ss:SSESF3}\\
        2\mr{\Omega}^{-1}\mr{\lambda} \frac{d}{d\hat{z}}\mr{\Omega} &= \frac{d}{d\hat{z}}\mr{\lambda} + \mr{r}(\mr{\phi}')^2 \label{ss:SSESF4}\\
        \mr{r}\hat{z} \frac{d}{d\hat{z}}\mr{\phi}' &= -2\mr{\lambda}\hat{z}\mr{\phi}' - k \mr{\lambda} \label{ss:SSESF5}\\
        2 \mr{\lambda} \frac{d}{d\hat{z}}\mr{m} &= (1-\mr{\mu}) \mr{r}^2 (\mr{\phi}')^2 \label{ss:SSESF6}\\
        2 \mr{\nu} \hat{z} \frac{d}{d\hat{z}} \mr{m} &= 2\mr{\nu}\mr{m} + (1-\mr{\mu})\mr{r}^2(\hat{z} \mr{\phi}' + k)^2 \label{ss:SSESF7},
    \end{align} 
    subject to initial conditions posed at the axis: 
    \begin{align}
        \label{ss:SSESFinitval}
        (\mr{r}, \mr{\nu}, \mr{\lambda}, \mr{\Omega}, \mr{m}, \mr{\phi}')(-1) = \big(0,-\frac{1}{2}, \frac{1}{2}, 1, 0, \frac{k}{2}\big).
    \end{align}
    We also have the algebraic identities 
    \begin{align}
        \mr{\nu} + \mr{r} &= \hat{z}\mr{\lambda} \label{ss:alg1}\\ 
        \frac{1}{4}\mr{\Omega}^2 + \mr{\nu}\mr{\lambda} &= \mr{r}^2 \hat{z} (\mr{\phi}')^2 + 2k\hat{z}\mr{\lambda}\mr{r}\mr{\phi}' + k^2\mr{r}\mr{\lambda}. \label{ss:alg2}
    \end{align}
\end{prop}

The following lemma establishes soft bounds on the geometric quantities. Without any a priori quantitative control on the solution to (\ref{ss:SSESF1})--(\ref{ss:SSESF7}), we rely heavily on monotonicity and the explicit initial conditions (\ref{ss:SSESFinitval}), as well as the qualitative regularity expressed in Lemma \ref{lemma:backgroundprelims1}.
\begin{lemma}
    The following bounds hold for $\hz \in [-1,0]$:
    \begin{align}
        \frac{1}{2} \leq \ &\mr{\lambda}(\hz) \leq \frac{1}{2}|\hz|^{-k^2}, \label{backgroundbd:1}\\
         \frac{1}{2}(1-|\hz|) \leq \ &\mr{r}(\hz) \leq \frac{1}{2q_k}(1-|\hz|^{q_k}) , \label{backgroundbd:2}\\
        1 \leq \ &\mr{\Omega}^2(\hz) \leq |\hz|^{-k^2} \label{backgroundbd:3},\\[\jot]
         0 \leq \ &\mr{m}(\hz) \leq k^2. \label{backgroundbd:4}
    \end{align}
\end{lemma}
\begin{proof}
By the definition (\ref{eq:defnofm}) of Hawking mass, we may reexpress the right hand sides of (\ref{ss:SSESF1})--(\ref{ss:SSESF2}) as $-\mr{\nu}\mr{\lambda} - \frac{1}{4}\mr{\Omega}^2 = -\frac{1}{4}\mr{\Omega}^2 \mr{\mu} \leq 0$. Therefore, by comparison with the values obtained along the axis (\ref{ss:SSESFinitval}), we have the inequalities
\begin{equation*}
    \mr{\lambda}(\hat{z}) \geq \frac{1}{2}, \quad \mr{\nu}(\hat{z}) \leq  -\frac{1}{2}.
\end{equation*}
Similarly, direct inspection of (\ref{ss:SSESF4}) shows $\frac{d}{d\hz}\mr{\Omega} \geq 0$, giving the one sided bound $\mr{\Omega}(\hz) \geq 1$.

We next show that for all $\hz \in [-1,0)$, 
\begin{equation}
    \label{eq:signofphi'}
    \mr{\phi}'(\hz) \geq 0.
\end{equation}

By (\ref{ss:SSESFinitval}) and continuity, $\mr{\phi}'(\hz) > 0$ in a neighborhood of ${\hz = -1}$. Supposing (\ref{eq:signofphi'}) is not true, let $\hz_*$ denote the greatest lower bound of the subset of $[-1,0)$ on which $\mr{\phi}'(\hz) < 0$. It follows $\hz_* > -1 $, and by continuity that $\mr{\phi}'(\hz_*)=0$. However, inspecting (\ref{ss:SSESF5}) shows $\frac{d}{d\hz}\mr{\phi}'(\hz_*) > 0$, and therefore $\mr{\phi}'(z)>0$ on an interval $(\hz_*, \hz_* +\epsilon)$, contradicting the choice of $\hz_*$.

Returning to upper bounds on $\mr{\lambda}, \mr{\Omega}$, we consider the weighted quantities $|\hz|^{k^2}\mr{\lambda}$, $|\hz|^{\frac{k^2}{2}}\mr{\Omega}(\hz)$, which satisfy the equations
\begin{align}
    \frac{d}{d \hz} \big(|\hz|^{k^2}\mr{\lambda}\big) 
    &= -\mr{r}|\hz|^{k^2}(\mr{\phi}')^2 - 2k|\hz|^{k^2}\mr{\lambda}\mr{\phi}',\label{eq:weightedlambda} \\[\jot]
    \frac{d}{d\hat{z}}\big(|\hat{z}|^{\frac{k^2}{2}}\mr{\Omega} \big) &= -k \mr{\phi}' \big(|\hat{z}|^{\frac{k^2}{2}}\mr{\Omega} \big) \label{eq:monotonic_omega}.
\end{align}
We have used the algebraic identity (\ref{ss:alg2}) in addition to (\ref{ss:SSESF1}), (\ref{ss:SSESF4}). From (\ref{eq:signofphi'}) it follows that these equations have a sign, which immediately gives the upper bounds in (\ref{backgroundbd:1}), (\ref{backgroundbd:3}). To conclude the estimates (\ref{backgroundbd:2}) for $\mr{r}$, it now suffices to integrate $\partial_z \mr{r} = \mr{\lambda}$ and apply (\ref{backgroundbd:1}).

Finally consider $\mr{m}$. Recall $\mr{\mu}(\hat{z}=0) = \frac{k^2}{1+k^2}.$ It follows that $\mr{m}(0) = \frac{1}{2}\mr{\mu}(0)\mr{r}(0) \leq \frac{1}{4(1-k^2)}\frac{k^2}{1+k^2} \leq k^2$. Since by direct inspection of (\ref{ss:SSESF6}) we have that $\mr{m}$ is an increasing function, we conclude
\begin{equation*}
    \mr{m} \leq k^2.
\end{equation*}
\end{proof}

Before proceeding we record some inequalities relating $\hz$-weights. \begin{lemma}
    \label{lemma:inequality}
    Let $a\in (0,1)$, $b \in (\frac12,1)$ be given parameters. Then for all $\hz \in [-1,0]$ we have 
    \begin{equation}
        \label{eq:ineqineq1}
        a(1-|\hz|) \leq 1-|\hz|^a \leq 1-|\hz|,
    \end{equation}
    and
    \begin{equation}
        \label{eq:ineqineq3}
        \big| (1-|\hz|)  - b^{-1}(1-|\hz|^b)\big| \leq  8 (1-b) (1-|\hz|)^2.
    \end{equation}
    If $\hz \in [-1,-\delta]$ for fixed $\delta \in (0,1)$, then there exists a constant $c_\delta = \frac{-\ln \delta}{1-\delta}$ such that  
    \begin{equation}
        \label{eq:ineqineq2}
        1-|\hz|^a \leq c_\delta a (1-|\hz|).
    \end{equation}
    Finally, for fixed $\delta_1 \in (0,1)$, there exists a constant $c_{\delta_1}$ such that for all $\hz \in [-1,0]$, 
    \begin{equation}
        \label{eq:ineqineq4.5}
        |\hz|^{\delta_1} ( 1 - |\hz|^a) \leq c_{\delta_1} a. 
    \end{equation}
\end{lemma}
\begin{proof}
    We begin with (\ref{eq:ineqineq1}). Defining $w = |\hz|-1 \in [-1,0]$, we may apply Bernoulli's inequality to conclude 
    \begin{align*}
        (1+w)^a \leq 1 + aw \implies |\hz|^a \leq 1 + a(|\hz|-1).
    \end{align*}
    This is equivalent to the first inequality in (\ref{eq:ineqineq1}). To see the second inequality, it suffices to observe that $|\hz| \leq 1$ and $a \leq 1$ imply $|\hz| \leq |\hz|^a.$

    \vspace{1em}
    Turning to (\ref{eq:ineqineq3}), we first consider $\hz \in [-1,-\frac12]$, and define $f_b(\hz) = (1-|\hz|)  - b^{-1}(1-|\hz|^b)$. A direct computation gives $f_b(-1) = f_b'(-1) = 0$, and $\sup_{\hz \in [-1,-\frac12]}|f_b''(\hz)| \leq 4(1-b)$ holds. Taylor's theorem then implies the desired inequality. In $\hz \in [-1,-\frac{1}{2}]$, we observe that $f_b(\hz)$ is a strictly decreasing function of $\hz$, and thus the left hand side of (\ref{eq:ineqineq3}) can be estimated above by $b^{-1}-1 $. Similarly, the right hand side of (\ref{eq:ineqineq3}) can be bounded below in $\hz \in [-1,-\frac{1}{2}]$ by $2(1-b)$. These estimates are consistent provided $b > \frac12$. 

    \vspace{1em}
    We next consider (\ref{eq:ineqineq2}), and define the function $g_a(\hz) = \frac{1-|\hz|^a}{1-|\hz|}$. Computing the derivative explicitly shows that provided $a < 1$, the function $g_a(\hz)$ is non-decreasing in $\hz$. It follows that $g_a(\hz) \leq g_a(-\delta)$, which by Taylor's theorem can be seen to satisfy the estimate $g_a(-\delta) \leq \frac{-\ln \delta}{1-\delta} a \doteq c_\delta a$.

    \vspace{1em}
    Finally we show (\ref{eq:ineqineq4.5}). By (\ref{eq:ineqineq2}) with explicit constant we have $1 - |\hz|^{a} \leq -a \ln|\hz|$ for all $\hz \in [-1,0]$. Therefore, $|\hz|^{\delta_1}(1-|\hz|^{a}) \leq -a \ln |\hz| |\hz|^{\delta_1} \leq c_{\delta_1} a,$ as desired.
\end{proof}

We also require the following integral bound.
\begin{lemma}
    Let $a \in (0,1)$, $p \in [1,\infty)$. Then 
    \begin{equation}
        \label{eq:integralLPest}
        \int_{0}^1 (1-t^a)^p dt \lesssim_p a^p.
    \end{equation}
\end{lemma}
\begin{proof}
    For $p \in \mathbb{N}$ this integral may be evaluated explicitly as 
    \begin{equation*}
        \int_{0}^1 (1-t^a)^p dt = \frac{\Gamma(1+\frac{1}{a})\Gamma(1+p)}{\Gamma(1+\frac{1}{a}+p)}.
    \end{equation*}
    Applying the product rule for gamma functions yields
    \begin{align*}
        \frac{\Gamma(1+\frac{1}{a})\Gamma(1+p)}{\Gamma(1+\frac{1}{a}+p)} &= p! \prod_{j=1}^{p} \frac{1}{1+\frac{1}{a} + (p-j)} \\
        & \leq p! a^p.
    \end{align*}
    In the case $p \notin \mathbb{N}$, the estimate follows by log-convexity of $L^p$ norms.
\end{proof}

With these inequalities in hand, we give the leading order behavior for weighted geometric quantities as $k \rightarrow 0$, as well as $L^p_z$ estimates for the renormalized scalar field derivative $\partial_z \mr{\phi}(z)$.
\begin{prop}
    \label{prop:appendixLpbound}
    The following estimate holds for any $p \in [1,\infty)$:
    \begin{align} 
        \|\partial_z \mr{\phi}\|_{L^p_{z}([-1,0])} &\lesssim_p k. \label{backgroundbd:6}
    \end{align}
    Moreover, we have the leading order expansions
    \begin{align}
        &\||\hat{z}|^{k^2}\mr{\Omega}^2 - 1 \|_{L^\infty([-1,0])} \lesssim k^2, \label{eq:smallkbound7}\\
        &\| \frac{|\hat{z}|^{k^2}\mr{\lambda} - \frac{1}{2}}{1-|\hz|} \|_{L^\infty([-1,0])} + \||\hat{z}|^{k^2}\mr{\lambda} - \frac{1}{2}\|_{C^1_{z}([-1,-\frac{1}{2}])} + \|\partial_z (|\hz|^{k^2}\mr{\lambda}) \|_{L^p_z([-1,0])} \lesssim_p k^2, \label{eq:smallkbound8}\\
        &|\mr{r}(\hz) - \frac{1}{2}(1-|\hz|)| \lesssim  k^2 (1-|\hz|)^2, \label{eq:smallkbound9}\\
        &\|\mr{\nu} + \frac{1}{2} \|_{L^\infty([-1,0])} + \| \partial_z \mr{\nu}\|_{L^p_z([-1,0])} \lesssim_p k^2, \label{eq:smallkbound10} \\
        &\|\frac{\mr{m}}{\mr{r}^3} \|_{L^\infty([-1,0])} \lesssim k^2. \label{eq:smallkbound10.5} 
    \end{align}
    On compact subintervals of $[-1,0)$, the bounds (\ref{eq:smallkbound7})--(\ref{eq:smallkbound10}) moreover hold in $C^1_{z}$.
\end{prop}
\begin{proof}
    As a starting point, we show that $\mr{\phi}'(\hz)$ satisfies the pointwise estimate
    \begin{equation}
        \label{eq:smallktemp1}
        \sup_{\hz \in [-1,-\frac{1}{2}]}|\mr{\phi}'(\hz)| \leq k.
    \end{equation}
    Recall we have (\ref{eq:signofphi'}), and so it is enough to show an upper bound for $\mr{\phi}'(\hz)$. Rewriting (\ref{ss:SSESF5}) as an equation for $\mr{r}^2\mr{\phi}'$, integrating on $\hz' \in [-1,-\frac12]$, and applying (\ref{backgroundbd:1}) gives
    \begin{align*}
    \mr{\phi}'(\hz) &\leq \frac{1}{\mr{r}^2(\hz)}\int_{-1}^{\hz} \frac{k}{|\hz'|}\mr{r}(\hz')\mr{\lambda}(\hz') d\hz' \leq  k,
    \end{align*}
    as desired.
Let $d_k \doteq \frac{1}{k}\mr{\phi}'(-\frac{1}{2}) \leq 1$. We now propagate the control on $\mr{\phi}'(\hz)$ towards $\hz=0$. (\ref{ss:SSESF5}) is equivalent to
\begin{equation}
    \label{eq:smallktemp2}
    \frac{d}{d\hz}\mr{\phi}' = -\frac{2\mr{\lambda}}{\mr{r}}\mr{\phi}' + \frac{k\mr{\lambda}}{\mr{r}|\hz|}.
\end{equation} 
By (\ref{eq:signofphi'}) the first term is non-positive, and thus integrating on $\hz' \in [-\frac12,\hz]$ and applying (\ref{backgroundbd:1}), (\ref{backgroundbd:2}) gives 
\begin{align*}
    0 \leq \mr{\phi}'(\hz) &\leq k d_k + \int_{-\frac{1}{2}}^{\hz}\frac{k\mr{\lambda}}{\mr{r}|\hz'|} d\hz' \\
    &\leq k d_k + 2k \int_{-\frac{1}{2}}^{\hz} \frac{1}{|\hz'|^{1+k^2}}d\hz' \\
    &\leq k d_k + \frac{2}{k|\hz|^{k^2}}\big(1-|2\hz|^{k^2} \big).
\end{align*}
It follows that 
\begin{equation*}
    \partial_z \mr{\phi}(z) = p_k|\hz|^{k^2}\mr{\phi}'(\hz) \leq k p_k d_k |z|^{p_k k^2} +  \frac{2 p_k}{k}\big(1-|2z|^{p_k k^2} \big).
\end{equation*}
The first term has $L^\infty_z$ norm of size $k$, so it suffices to consider the second term. In $L^\infty_z$ this term has size $k^{-1}$, but in $L^p_z$, $p < \infty$ we may apply the integral estimate (\ref{eq:integralLPest}) to conclude (\ref{backgroundbd:6}):
\begin{align*}
    \|\partial_z \mr{\phi}(z)\|_{L^p_z([-1,0])} &\lesssim k \||z|^{p_k k^2} \|_{L^p_z([-1,0])} + k^{-1}\|1-|2z|^{p_k k^2} \|_{L^p_z([-1,0])} \\
    & \lesssim_p k.
\end{align*}
A useful consequence of (\ref{backgroundbd:6}) is a bound for the (un-renormalized) derivative $\mr{\phi}'(\hz)$:
\begin{equation}
    \label{eq:l2boundunrenormalized}
    \|\mr{\phi}' \|_{L^2_{\hz}([-1,0])} \lesssim k.
\end{equation}
To see this, observe $|z|^{-p_k k^2} \in L^2_{z}([-1,0])$, and thus Cauchy-Schwarz gives
\begin{align*}
    \int_{-1}^{0} \big(\mr{\phi}'(\hz)\big)^2 d\hz &= p_k \int_{-1}^{0}|z|^{-p_k k^2} \big(\partial_z \mr{\phi}(z) \big)^2 dz  \\
    &\lesssim \|\partial_z \mr{\phi}\|^2_{L^4_z} \lesssim k^2.
\end{align*}
The expansions for the geometric quantities will now follow readily. Starting with the first bound in (\ref{eq:smallkbound8}), we integrate (\ref{eq:weightedlambda}) for $\hz \in [-1,-\frac12]$, apply the pointwise estimate (\ref{eq:smallktemp1}), and compute 
\begin{align*}
    \big| |\hz|^{k^2}\mr{\lambda}(\hz) - \frac{1}{2}\big| \lesssim \| \mr{\phi}' \|^2_{L^\infty([-1,0])}(1-|\hz|) + k \||\hz|^{k^2}\mr{\lambda} \|_{L^\infty([-1,0])}\|\mr{\phi}'\|_{L^\infty([-1,0])}(1-|\hz|) \lesssim k^2 (1-|\hz|). 
 \end{align*} 
 In the region $\hz \in [-\frac12,0]$ we again integrate (\ref{eq:weightedlambda}) and apply (\ref{eq:l2boundunrenormalized}) to give
 \begin{align*}
    \big| |\hz|^{k^2}\mr{\lambda}(\hz) - \frac{1}{2}\big| \lesssim \| \mr{\phi}' \|^2_{L^2_{\hz}([-1,0])} + k \||\hz|^{k^2}\mr{\lambda} \|_{L^\infty_{\hz}([-1,0])}\|\mr{\phi}'\|_{L^1_{\hz}([-1,0])} \lesssim k^2.
 \end{align*}
The first bound of (\ref{eq:smallkbound8}) follows. For the remaining two bounds in (\ref{eq:smallkbound8}) it suffices to estimate the terms appearing in right hand side of (\ref{eq:weightedlambda}). Similarly, (\ref{eq:smallkbound7}) follows by integrating the weighted equation (\ref{eq:monotonic_omega}) and using (\ref{eq:l2boundunrenormalized}).

To conclude the expansion (\ref{eq:smallkbound9}) for $\mr{r},$ we integrate the first bound in (\ref{eq:smallkbound8}), and apply the inequality (\ref{eq:ineqineq3}) with parameter $b \doteq 1-k^2 > \frac12$.

Turn now to $\mr{\nu}$. The identity (\ref{ss:alg1}) implies 
\begin{align*}
    \frac{d}{d \hz}\mr{\nu} &= \frac{d}{d \hz} (\hz \mr{\lambda} - \mr{r} ) = -|\hz|^{1-k^2}\frac{d}{d \hz}\big(|\hz|^{k^2}\mr{\lambda}\big) + \frac{k^2}{|\hz|^{k^2}}\big(|\hz|^{k^2}\mr{\lambda}\big).
\end{align*}
It follows that
\begin{align*}
    \big|\mr{\nu}(\hz)+\frac{1}{2}\big| \lesssim \|\frac{d}{d \hz}\big(|\hz|^{k^2}\mr{\lambda}\big) \|_{L^1_{\hz}([-1,0])} + k^2 \||\hz|^{k^2}\mr{\lambda} \|_{L^\infty([-1,0])}\||\hz|^{-k^2} \|_{L^1_{\hz}([-1,0])} \lesssim k^2.
\end{align*}

Finally we consider (\ref{ss:SSESF6}). It suffices to prove (\ref{eq:smallkbound10.5}) in $\hz \in [-1,-\frac{1}{2}]$, given the lower bound on $\mr{r}(\hz)$ outside of this domain. In the region $\hz \in [-1,-\frac{1}{2}]$ there is the pointwise bound on $\mr{\phi}'$ (\ref{eq:smallktemp1}), and thus
\begin{equation}
    \label{eq:tempderivativeofm}
    \sup_{\hz \in [-1,-\frac{1}{2}]}|\frac{d}{dz}\mr{m}| \lesssim k^2 \mr{r}^2(\hz).
\end{equation}
Integrating in $\hz$ from the axis and applying (\ref{eq:smallkbound9}), we conclude (\ref{eq:smallkbound10.5}).
\end{proof}

By more careful bookkeeping we may improve the result of the previous proposition, and compute the leading order term as $k \rightarrow 0$ for the rescaled scalar field quantity $k^{-1}\partial_z \mr{\phi}(z)$. In general, this quantity does not have a pointwise limit at $z=0$; however, the  $L^p_z([-1,0])$ limit is well defined, for any $p < \infty$. The following proposition records this leading order term, as well as relevant higher order estimates.

\begin{prop}
    \label{prop:expansionPhiM}
    We have the estimates
    \begin{align}
        \big\| \frac{1}{k}|\hz|^{k^2}\mr{\phi}'(\hz) - \frac{-1-\hz - \ln |\hz|}{(1-|\hz|)^2} \big\|_{L^p_{z}([-1,0])} &\lesssim_p k^2, \label{eq:smallkbound12}\\[\jot]
        \big\| \frac{1}{k^2}\mr{m}(\hz) - \frac{1}{4}\bigg(1+\hz- \frac{|\hz| \ln^2 |\hz|}{1+\hz} \bigg)\big\|_{L^\infty([-1,0])} &\lesssim k^2, \label{eq:smallkbound13} \\[\jot]
         \|  \mr{\phi}'' \|_{L^\infty([-1,-\frac{1}{2}])} + \||z|\partial_z^2 \mr{\phi} \|_{L^p_z([-1,0])}  &\lesssim_p k, \label{eq:smallkbound14} \\[\jot]
         \| |z|\partial_z^3 \mr{r} \|_{L^p_z([-1,0])} + \| |z|\partial_z^2 \mr{m} \|_{L^p_z([-1,0])} &\lesssim_p k^2. \label{eq:smallkbound15}
     \end{align}
\end{prop}

We will require a computational lemma:
\begin{lemma}
    \label{lem:lpconvergencetoLOG}
    Define the functions $f_k(\hz), g(\hz): [-1,0) \rightarrow \mathbb{R}$ 
    \begin{align*}
        f_k(\hz) &\doteq \frac{1}{k^2 q_k}\frac{(1-|\hz|^{k^2}) - k^2(1-|\hz|)   }{(1-|\hz|)^2}, \\[\jot]
        g(\hz) &\doteq \frac{-1-\hz - \ln |\hz|}{(1-|\hz|)^2}.
    \end{align*}
    Then $f_k \xrightarrow{k \rightarrow 0} g$ in $L^p_{\hz}([-1,0])$, for any $p\in[1,\infty)$.
\end{lemma}
\begin{proof}
In the domain $[-1,-\frac12]$, it is straightforward to show that convergence holds pointwise almost everywhere. An application of the dominated convergence theorem then implies the statement in $L^p_{\hz}([-1,-\frac12])$.

In the near-horizon region $[-\frac12,0]$, the factors proportional to $1-|\hz|$ converge pointwise, and it is enough to show that 
\begin{equation}
    \label{eq:Lpconvergencetolog:temp1}
    \frac{1}{k^2}(1-|\hz|^{k^2}) \xrightarrow{L^p_{\hz}([-\frac{1}{2},0])} -\ln |\hz|. 
\end{equation}
The strategy is to use an integral representation of both functions, and estimate the $L^p_{\hz}$ norm of the differences using Minkowski's integral inequality, extracting a small $k$-dependence. Write the difference of the functions in (\ref{eq:Lpconvergencetolog:temp1}) as 
\begin{align}
    \frac{1}{k^2}(1-|\hz|^{k^2}) + \ln|\hz| = \int_{-1}^{\hz}  \frac{|\hz'|^{k^2}-1}{|\hz'|}d\hz', 
\end{align}
and estimate (here, $\delta < \frac{1}{p}$ is arbitrary)
\begin{align*}
   \bigg( \int_{-1}^{0} \bigg( \int_{-1}^{\hz}  \frac{|\hz'|^{k^2}-1}{|\hz'|}d\hz'  \bigg)^p d\hz \bigg)^{\frac{1}{p}} &\leq  \int_{-1}^{0}  \bigg(\int_{\hz'}^{0}  \bigg(\frac{|\hz'|^{k^2}-1}{|\hz'|}\bigg)^p d\hz \bigg)^{\frac{1}{p}} d\hz' \\[\jot]
   &= \int_{-1}^{0} \frac{1-|\hz'|^{k^2}}{|\hz'|^{1-\frac{1}{p}}}         d\hz' \\[\jot]
    &\leq \sup_{\hz\in[-1,0]}\big||\hz|^{\delta}(1-|\hz|^{k^2}) \big|  \int_{-1}^{0} \frac{1}{|\hz'|^{1-\frac{1}{p}+\delta}} d\hz' \\[\jot]
    &\lesssim_{p,\delta} \sup_{\hz\in[-1,0]}\big||\hz|^{\delta}(1-|\hz|^{k^2}) \big| \\[\jot]
    &\lesssim_{p,\delta} k^2.
\end{align*}
In the final inequality we have used (\ref{eq:ineqineq4.5}).

\end{proof}

\begin{proof}[Proof of Proposition \ref{prop:expansionPhiM}]
We will show that the scalar field equation (\ref{eq:smallktemp2}) can be directly integrated up to error terms. First, given the expansions (\ref{eq:smallkbound8}), (\ref{eq:smallkbound9}) we write 
\begin{align}
    \frac{\mr{\lambda}}{\mr{r}} &= \frac{1}{|\hz|^{k^2}(1-|\hz|)} + \frac{1}{|\hz|^{k^2}}\mathcal{E}_k(\hz) \label{eq:lambdaoverr_exp0} \\ 
    &= \frac{1}{1-|\hz|} + \frac{1}{|\hz|^{k^2}} (\mathcal{G}_k(\hz)+\mathcal{E}_k(\hz)),\label{eq:lambdaoverr_exp}
\end{align}  
where $\mathcal{G}_k(\hz) \doteq \frac{1-|\hz|^{k^2}}{1-|\hz|}, $ and $\mathcal{E}_k(\hz)$ is an error term satisfying 
\begin{equation}
    \label{eq:lpbounderrornew1}
    \sup_{\hz' \in [-1,0]}|\mathcal{E}_k(\hz')| +
    \sup_{\hz' \in [-1,-\frac12]}|(1-|\hz|)\mathcal{E}'_k(\hz')| + \|\mathcal{E}'_k(\hz) \|_{L^2_{\hz}([-\frac12,0])} \lesssim k^2.
\end{equation}
Primed quantities $\mathcal{G}_k'$, $\mathcal{E}_k'$ denote derivatives with respect to $\hz$. As a consequence of the $L^p_{\hz}$ regularity of $\mathcal{E}_k'$, we can further write $\mathcal{E}_k(\hz) = \mathcal{E}_k(0) + \mathcal{F}_k(\hz)$ for a constant $|\mathcal{E}_k(0)| \lesssim k^2$ and a function $\mathcal{F}_k(\hz)$ satisfying    
\begin{align}
   |\mathcal{F}_k(\hz)| &\lesssim \|\mathcal{E}_k'(\hz) \|_{L^1_{\hz}([\hz,0])}  \nonumber \\
   &\lesssim \||\hz|^{k^2}\mathcal{E}_k'(\hz) \|_{L^2_{\hz}([\hz,0])}\||\hz|^{-k^2} \|_{L^2_{\hz}([\hz,0])} \nonumber \\
   &\lesssim k^2 |\hz|^{\frac12 - k^2}. \label{eq:estmathcalF}
\end{align}
We will need the following additional estimates, which are consequences of (\ref{eq:smallktemp1}), (\ref{eq:l2boundunrenormalized}), (\ref{eq:integralLPest}), and Lemma \ref{lemma:inequality}.
\begin{equation}
    \label{eq:lpbounderrornew3}
    \sup_{\hz \in [-1,-\frac{1}{2}]}|\mr{\phi}'(\hz)| + \|\mr{\phi}'(\hz) \|_{L^2_{\hz}([-1,0])} \lesssim k,
\end{equation}
\begin{equation}
    \label{eq:lpbounderrornew6}
    \sup_{\hz\in [-\frac{1}{2}, \hz]}|G_k(\hz')| \lesssim 1-|\hz|^{k^2},
\end{equation}
\begin{equation}
    \label{eq:lpbounderrornew4}
    \sup_{\hz\in [-1,-\frac{1}{2}]}|G_k(\hz)| + \sup_{\hz\in [-1,-\frac{1}{2}]}|(1-|\hz|)G_k'(\hz')|+\|G_k(\hz)\|_{L^p_{\hz}([-1,0])} \lesssim_p k^2,
\end{equation}
Returning now to the study of (\ref{eq:smallktemp2}), we may insert the decomposition (\ref{eq:lambdaoverr_exp}) and write 
\begin{equation*}
    \frac{d}{d\hz}\mr{\phi}'= -\frac{2}{1-|\hz|}\mr{\phi}' +k \frac{1}{|\hz|(1-|\hz|)}  - \frac{2}{|\hz|^{k^2}}(\mathcal{G}_k(\hz)+\mathcal{E}_k(\hz))\mr{\phi}' + \frac{k}{|\hz|^{1+k^2}} (\mathcal{G}_k(\hz)+\mathcal{E}_k(\hz)).
\end{equation*}
Conjugating through by $w(\hz) \doteq (1-|\hz|)^2$ yields
\begin{equation}
    \frac{d}{d\hz} \big( w(\hz) \mr{\phi}'\big) = k \frac{1-|\hz|}{|\hz|} - \frac{2}{|\hz|^{k^2}}w(\hz) (\mathcal{G}_k(\hz)+\mathcal{E}_k(\hz))\mr{\phi}'+ \frac{k}{|\hz|^{1+k^2}}w(\hz) (\mathcal{G}_k(\hz)+\mathcal{E}_k(\hz)). \label{lpbounds:temp1.9}
\end{equation}
We now integrate this expression for $\hz' \in [-1,\hz]$. Observing that the contribution of $\mathcal{G}_k(\hz)$ to the final term of (\ref{lpbounds:temp1.9}) may be explicitly integrated, we find 
\begin{align}
    \label{lpbounds:temp2}
    \mr{\phi}'(\hz) =& \frac{1}{k q_k}\frac{(1-|\hz|^{k^2}) - k^2(1-|\hz|)   }{|\hz|^{k^2}(1-|\hz|)^2} - \underbrace{\frac{1}{w(\hz)}\int_{-1}^{\hz} \frac{2}{|\hz'|^{k^2}}w(\hz') (\mathcal{G}_k(\hz)+\mathcal{E}_k(\hz))\mr{\phi}'(\hz')d\hz'}_{\text{I}(\hz)} \nonumber \\
    &+ \underbrace{\frac{1}{w(\hz)} \int_{-1}^{\hz} \frac{k}{|\hz'|^{1+k^2}}w(\hz')\mathcal{E}_k(\hz) d\hz'}_{\text{II}(\hz)}.
\end{align}
We first estimate the function denoted $\text{I}(\hz)$. To begin, assume $\hz < -\frac{1}{2}$. In this region it suffices to apply the pointwise bounds (\ref{eq:lpbounderrornew1}), (\ref{eq:lpbounderrornew3}), (\ref{eq:lpbounderrornew4}) and calculate
\begin{align*}
    |\text{I}(\hz)| &\lesssim \frac{1}{w(\hz)}\int_{-1}^{\hz} \frac{2}{|\hz'|^{k^2}}w(\hz')|\mathcal{G}_k + \mathcal{E}_k|\mr{\phi}'(\hz')d\hz' \\[\jot]
    & \lesssim k^3 |\hz|^{-k^2} \lesssim k^3.
\end{align*}
Next consider $\hz \! \geq \! -\frac{1}{2}$. It follows $w(\hz) \! \sim \! 1$ independently of $k$. Employing the integrated bounds (\ref{eq:lpbounderrornew1}), (\ref{eq:lpbounderrornew3}), (\ref{eq:lpbounderrornew4}) and Cauchy-Schwarz, calculate
\begin{align}
    |\text{I}(\hz)| &\lesssim \frac{1}{w(\hz)}\int_{-1}^{\hz} \frac{2}{|\hz'|^{k^2}}w(\hz')|\mathcal{G}_k + \mathcal{E}_k|\mr{\phi}'(\hz')d\hz' \nonumber \\[\jot]
    &\lesssim  \frac{1}{w(\hz)}\int_{-1}^{-\frac{1}{2}} \frac{2}{|\hz'|^{k^2}}w(\hz')|\mathcal{G}_k + \mathcal{E}_k|\mr{\phi}'(\hz')d\hz' + \frac{1}{w(\hz)}\int_{-\frac{1}{2}}^{\hz} \frac{2}{|\hz'|^{k^2}}w(\hz')|\mathcal{G}_k + \mathcal{E}_k|\mr{\phi}'(\hz')d\hz' \nonumber \\[\jot]
    &\lesssim k^3 +   \| |\hz|^{-\frac{1}{2}k^2} \|_{L^4_{\hz}([-\frac{1}{2},0])} \| \mathcal{G}_k + \mathcal{E}_k \|_{L^4_{\hz}([-\frac{1}{2},0])} \| \mr{\phi}' \|_{L^2_{\hz}([-\frac{1}{2},0])} \nonumber\\[\jot]
    & \lesssim k^3. \label{eq:tempestI}
\end{align}
To estimate the term $\text{II}(\hz)$, a similar argument applies in $\hz \! < \! -\frac{1}{2}$ to give $|\text{II}(\hz)| \lesssim k^3$. In the remaining domain we decompose $\mathcal{E}_k(\hz) = \mathcal{E}_k(0) + \mathcal{F}_k(\hz)$ and apply (\ref{eq:estmathcalF}): 
\begin{align}
    \text{II}(\hz) &= k\frac{1}{w(\hz)}\int_{-1}^{-\frac{1}{2}}\frac{1}{|\hz'|^{1+k^2}}w(\hz')\mathcal{E}_k(\hz') d\hz' + k\frac{1}{w(\hz)}\int_{-\frac12}^{\hz}\frac{1}{|\hz'|^{1+k^2}}w(\hz')\mathcal{E}_k(\hz') d\hz' \nonumber \\[\jot]
    &= O_{L^\infty}(k^3) + k \mathcal{E}_k(0) \frac{1}{w(\hz)}\int_{-\frac12}^{\hz}\frac{1}{|\hz'|^{1+k^2}}w(\hz') d\hz' + k\frac{1}{w(\hz)}\int_{-\frac12}^{\hz}\frac{1}{|\hz'|^{1+k^2}}w(\hz')\mathcal{F}_k(\hz') d\hz' \nonumber \\[\jot] 
    &= O_{L^\infty}(k^3) + k \mathcal{E}_k(0) \frac{1}{w(\hz)}\int_{-\frac12}^{\hz}\frac{1}{|\hz'|^{1+k^2}}w(\hz') d\hz'. \label{eq:tempestII_0}
\end{align}
The final integral may be estimated to give 
\begin{equation}
    \label{eq:tempestII}
   |\text{II}(\hz)| \lesssim O_{L^\infty}(k^3) + k |\hz|^{-k^2}(1-|\hz|^{k^2}).
\end{equation}
It remains to take the $L^p_z$ limit as $k\rightarrow 0$ of $\frac{1}{k}|\hz|^{k^2}\mr{\phi}(\hz) $, by examining the various terms of (\ref{lpbounds:temp2}). The first term converges in $L^p_z([-1,0])$ to the desired limit, by Lemma \ref{lem:lpconvergencetoLOG}. For the integral expressions we apply (\ref{eq:tempestI}), (\ref{eq:tempestII}), giving
\begin{equation}
    \|\frac{1}{k}|\hz|^{k^2}\big(\text{I}(\hz) + \text{II}(\hz)\big)\|_{L^p_{\hz}([-1,0])} \lesssim_p k^2.
\end{equation}
We therefore conclude (\ref{eq:smallkbound12}). 

The expansion for $\mr{m}(\hz)$ will now follow as an immediate consequence. From (\ref{ss:SSESF6}) we compute 
\begin{align*}
   |\hz|^{k^2} \frac{d}{d\hz} \big(k^{-2}\mr{m}(\hz)\big) &= |\hz|^{k^2}\frac{\mr{r}^2}{2\mr{\lambda}}(1-\mr{\mu})(k^{-1}\mr{\phi}')^2 \\
    &= \frac{1}{4}\bigg(\frac{1+\hz+\ln |\hz|}{1+\hz}\bigg)^2 + O_{L^p_{\hz}}(k^2).
\end{align*}
Dividing by $|\hz|^{k^2}$ and integrating gives 
\begin{align*}
    k^{-2}\mr{m}(\hz) = \frac{1}{4}\bigg(1+\hz- \frac{|\hz| \ln^2 |\hz|}{1+\hz} \bigg)+ O_{L^\infty}(k^2).
\end{align*}
We next discuss the second derivative estimate (\ref{eq:smallkbound14}). We have the qualitative information $\mr{\phi}'' \in C^0_{\hz}([-1,-\frac{1}{2}])$, and attempt to derive a pointwise bound with proper $k$-dependence. Differentiating (\ref{lpbounds:temp2}) and estimating (observe we have pointwise bounds for $\mathcal{G}_k$, $\mathcal{E}_k$ in the region $\hz \in [-1,-\frac{1}{2}]$) yields
\begin{align*}
    |\mr{\phi}''(\hz)| \lesssim k\frac{1-|\hz|^2 - 2|\hz|\ln |\hz|}{|\hz|(1-|\hz|)^3} + O_{L^\infty}(k^3).
\end{align*}
We conclude (\ref{eq:smallkbound14}).

To see the $L^p_z$ bound on the weighted quantity $|z|\partial_z^2 \mr{\phi}$, it suffices to estimate in the near-horizon region $\hz \in [-\frac{1}{2},0]$. Expanding the $\partial_z$ derivatives, this quantity is schematically given by  
\begin{align*}
    |z|\partial_z^2 \mr{\phi} \sim |\hz|^{1+k^2}\mr{\phi}'' + |\hz|^{k^2}\mr{\phi}',
\end{align*}
and latter term has already been shown to be of size $k$ in $L^p_{\hz}$. By inspection of the right hand side of (\ref{ss:SSESF5}), and the fact that we are working in a region with a positive lower bound on $\mr{r}(\hz)$, a similar estimate for the first term follows. 

Finally, (\ref{eq:smallkbound15}) is a result of commuting (\ref{eq:weightedlambda}), (\ref{ss:SSESF6}) by $\partial_{\hz}\big(|\hz|^{k^2} \cdot \big) $ and inductively applying bounds.
\end{proof}

From the proof of Proposition \ref{prop:expansionPhiM} we may extract additional representation formulas for $\frac{\mr{\lambda}}{\mr{r}}(\hz)$ and $\mr{\phi}'(\hz)$, which will be useful in the following. 

\begin{cor}
    \label{cor:expnear0}
    There exist constants $C_{i,k}$, $i \leq 2$, bounded independently of $k$ small such that for $\hz \in [-\frac12, 0]$ we have 
    \begin{align}
        \frac{\mr{\lambda}}{\mr{r}}(\hz) &= \frac{C_{1,k}}{|\hz|^{k^2}} + O_{L^\infty}(|\hz|^{\frac12 - 2k^2}), \label{eq:newexpansions1}\\[\jot]
        k \mr{\phi}'(\hz) &= \frac{C_{1,k}}{|\hz|^{k^2}} - (p_k + k^2 C_{2,k}) + O_{L^\infty}(k^4) + O_{L^\infty}( |\hz|^{1-k^2}). \label{eq:newexpansions2}
    \end{align}
\end{cor}
\begin{proof}
    The expansion (\ref{eq:newexpansions1}) for some constant $C_{1,k}$ follows from (\ref{eq:lambdaoverr_exp0}), Taylor's theorem, and (\ref{eq:estmathcalF}).

    To see (\ref{eq:newexpansions2}), we return to the integrated wave equation (\ref{lpbounds:temp2}). As in the above proof, $|\text{I}(\hz)| \lesssim k^3$ holds. For $\text{II}(\hz)$ we consider (\ref{eq:tempestII_0}), and explicitly evaluate the integral term. It follows that we may write (recall $|\mathcal{E}_k(0)| \lesssim k^2$)
    \begin{equation}
        \text{II}(\hz) = \frac{d_{1,k}}{|\hz|^{k^2}}  + k d_{2,k} +   O_{L^\infty}( |\hz|^{1-k^2}),
    \end{equation}
    where $d_{i,k}$ are constants bounded independently of $k$. From (\ref{lpbounds:temp2}) we may write, after Taylor expanding the first term about $\hz =0$, 
    \begin{equation}
        \label{eq:expnew:temp1}
        k \mr{\phi}'(\hz) = \frac{c_{1,k}+d_{1,k}}{|\hz|^{k^2}} - (p_k - k^2 d_{2,k})+ O_{L^\infty_{\hz}}( |\hz|^{1-k^2}) + O_{L^\infty}(k^4),
    \end{equation}
    for a constant $c_{1,k}$ which is moreover bounded independently of $k$. In order to conclude (\ref{eq:newexpansions2}), we show that the coefficient of the singular $|\hz|^{-k^2}$ term appearing in (\ref{eq:expnew:temp1}) coincides with $C_{1,k}$. Label this currently undetermined coefficient $D_{1,k}$. We show that the equality $C_{1,k} = D_{1,k}$ is forced by $k$-self-similarity. By (\ref{ss:SSESF5}), it follows that $|\hz|^{k^2}\mr{\phi}'(\hz)$ satisfies the equation
    \begin{equation*}
        \partial_{\hz}\big(|\hz|^{k^2}\mr{\phi}'(\hz) \big) = \frac{k}{|\hz|}\Big(\frac{|\hz|^{k^2}\mr{\lambda}(\hz)}{\mr{r}(\hz)} - k |\hz|^{k^2}\mr{\phi}'(\hz) \Big) - \frac{2\mr{\lambda}(\hz)}{\mr{r}(\hz)}\mr{\phi}'(\hz).
    \end{equation*} 
    The latter term may be estimated by $|\hz|^{-2k^2}$ as $\hz \rightarrow 0$, and is therefore integrable. Inserting the expansions (\ref{eq:newexpansions1})--(\ref{eq:newexpansions2}) (with coefficient $D_{1,k}$) shows that  
    \begin{equation*}
        \partial_{\hz}\big(|\hz|^{k^2}\mr{\phi}'(\hz) \big) = \frac{k}{|\hz|}(C_{1,k} - D_{1,k}) + O_{L^{\infty}_{\hz}}(|\hz|^{-1+k^2}).
    \end{equation*}
    In order for $|\hz|^{k^2}\mr{\phi}(\hz)$ to have a finite limit as $\hz \rightarrow 0$, it follows that $C_{1,k} - D_{1,k}=0$, as desired.
\end{proof}

The significance of the expansions in Corollary (\ref{cor:expnear0}) is the identification of the leading order behavior of the constant term in (\ref{eq:newexpansions2}); in particular, it is non-vanishing for $k$ small. As discussed in \cite[Appendix A]{singh}, this coefficient appears in the calculation of the second order renormalized derivative $\partial_z^2 \mr{\phi}(z)$ near $z=0$. The analysis of that paper, which relied on local asymptotics for solutions to (\ref{ss:SSESF1})--(\ref{ss:SSESF7}) near its critical points, was unable to determine the value of this coefficient, which carries information about the \textit{global} shooting problem connecting regular solutions at the axis to those at $\{z=0\}$. With this new information, we establish the blowup (\ref{eq:dvsquaredblowup}). The remaining statements in Lemma \ref{lemma:backgroundprelims1} regarding sharp regularity of all double-null unknowns then follow directly from (\ref{ss:SSESF1})--(\ref{ss:SSESF7}).

\begin{prop}
    \label{prop:dvsquaredblowupproof}
    For $k$ sufficiently small, (\ref{eq:dvsquaredblowup}) holds, i.e. 
    \begin{equation}
        \partial_z^2 \mr{\phi}(z) \sim |z|^{-1+p_k k^2}.
    \end{equation}
\end{prop}
\begin{proof}
    By the computation in \cite[Lemma A.6]{singh}, up to constants bounded independently of $k$ we have
    \begin{equation*}
        \partial_z^2 \mr{\phi}(z) = k |z|^{-1+p_k k^2}\big(\frac{\mr{\lambda}}{\mr{r}}(\hz) - k \mr{\phi}'(\hz)\big) + \frac{(|\hz|^{k^2}\mr{\lambda}(\hz))(|\hz|^{k^2}\mr{\phi}'(\hz))}{\mr{r}(\hz)}.
    \end{equation*}
    The second term remains bounded up to $z = 0$, although it grows to size $k^{-1}$. The term responsible for the limited regularity is the first, which by Corollary \ref{cor:expnear0} has the form
    \begin{equation*}
         - k \big(p_k + O_{L^\infty}(k^2) + O_{L^\infty}(|\hz|^{\frac12-2k^2}) \big) |z|^{-1+p_k k^2}.
    \end{equation*}
    In particular, for $k$ sufficiently small this coefficient is non-zero as $z \rightarrow 0$.
\end{proof}

To conclude this section, we record various additional estimates.

\begin{lemma}
    The following bound holds:
    \begin{equation}
        \|\frac{\mr{m}}{\mr{r}^3}\|_{C^1_{z}([-1,-\frac{1}{2}])}  + \|\partial_z \big(\frac{\mr{m}}{\mr{r}^3} \big) \|_{L^p_z([-1,0])} \lesssim_p k^2 \label{eq:smallkbound10.6}.
    \end{equation}
    \end{lemma}
    \begin{proof}
        A pointwise bound on $\frac{\mr{m}}{\mr{r}^3}$ was given in (\ref{eq:smallkbound10.5}). To complete the $C^1_{z}$ bound for $z \in [-1,-\frac12]$, we employ a standard integration by parts trick to convert bounds on this derivative to higher order bounds on the scalar field. Note that $\partial_z$ and $\partial_{\hz}$ are comparable away from ${z=0}$, and thus the choice of derivative is immaterial. Write
    \begin{align*}
        \partial_{\hz}\bigg(\frac{\mr{m}}{\mr{r}^3} \bigg) &= \partial_{\hz} \bigg(\mr{r}^{-3} \int_{-1}^{\hz}\frac{\mr{r}^2}{2\mr{\lambda}}(1-\mr{\mu})(\mr{\phi}')^2 d\hz \bigg) \\
        &= -3\mr{r}^{-4}\mr{\lambda}\int_{-1}^{\hz}\frac{\mr{r}^2}{2\mr{\lambda}}(1-\mr{\mu})(\mr{\phi}')^2 d\hz  + \mr{r}^{-1} \frac{1}{2\mr{\lambda}}(1-\mr{\mu})(\mr{\phi}')^2 \\
        &= - \mr{r}^{-4}\mr{\lambda}\int_{-1}^{\hz}\partial_{\hz}(\mr{r}^3)\frac{1}{2\mr{\lambda}^2}(1-\mr{\mu})(\mr{\phi}')^2 d\hz + \mr{r}^{-1} \frac{1}{2\mr{\lambda}}(1-\mr{\mu})(\mr{\phi}')^2 \\ 
        &= \mr{r}^{-4}\mr{\lambda} \int_{-1}^{\hz} \mr{r}^3 \partial_{\hz}\bigg(\frac{1}{2\mr{\lambda}^2}(1-\mr{\mu})(\mr{\phi}')^2 \bigg)d\hz
    \end{align*}
    It now suffices to ensure $L^\infty_{\hz}([-1,-\frac{1}{2}])$ bounds with smallness on the differentiated quantity within the integrand. For $\mr{\lambda}, \mr{\mu}, \mr{\phi}'$ this follows from (\ref{eq:smallkbound8}), (\ref{eq:tempderivativeofm}), and (\ref{eq:smallkbound14}) respectively.
    
    For the $L^p_z$ bound, it suffices to estimate in a domain $z \in [-\frac12,0]$. Rewriting (\ref{ss:SSESF6}) as an equation in the $z$ coordinate, it follows that $\frac{d}{dz}\mr{m} \in L^p_z([-1,0])$ with $L^p_z([-1,0])$ norm of size roughly $k^2$. Thus directly differentiating $\frac{\mr{m}}{\mr{r}^3}$ (observe that $\mr{r}$ is bounded below in this domain) and estimating the resulting terms gives the result.
    \end{proof}

    Finally, we record higher derivative estimates. We shall not need to track the dependence of constants on $k$, and focus on the numerology of blowup rates near $\hz = 0$.
\begin{lemma}
    The following bounds hold, for $2 \leq j \leq 5$: 
    \begin{align}
        \big\||z|^{j-1 - p_k k^2} \partial_z^{j+1} \mr{r}\big\|_{L^\infty([-1,0])} &+ \big\||z|^{j -1- p_k k^2} \partial_z^{j} \mr{m}\big\|_{L^\infty([-1,0])}
        \nonumber \\[\jot]
        &+ \big\||z|^{j -1- p_k k^2} \partial_z^{j} \mr{\phi}\big\|_{L^\infty([-1,0])} \lesssim_k 1. \label{eq:smallkbound10.7}
    \end{align}
\end{lemma}
\begin{proof}
    By Lemma \ref{lemma:backgroundprelims1} we have for $0 \leq j \leq 1$
    \begin{align*}
        \mr{r}, \ \partial_z^{j+1} \mr{r}, \ \partial_z^{j} \mr{m}, \ \partial_z^j \mr{\Omega}, \ \partial_z^j \mr{\phi} &= O_{L^\infty}(1), \\[\jot]
        \partial_z^3 \mr{r}, \ \partial_z^2 \mr{m}, \ \partial_z^2 \mr{\Omega}, \ \partial_z^2 \mr{\phi} &= O_{L^\infty}(|z|^{-1+p_k k^2}).
    \end{align*}
    It now suffices to commute the system (\ref{ss:SSESF1})--(\ref{ss:SSESF7}) by $\partial_z \sim |\hz|^{k^2} \partial_{\hz}$ and apply the above bounds inductively. The point is that each commutation introduces at worst a single additional power of $|\hz|^{-q_k} \sim |z|^{-1}$. In principle cancellations could occur improving these blowup rates; however we do not consider this possibility here.
\end{proof}
\bibliographystyle{amsplain}
\bibliography{biblio}

\end{document}